\documentclass[10pt]{article}

\usepackage{fullpage}
\usepackage{enumitem}
\usepackage[page]{appendix}
\usepackage{amssymb}
\usepackage{mathrsfs} 
\usepackage{natbib}
\usepackage{algorithm}
\usepackage{algorithmic}
\usepackage{hyperref}
\usepackage{comment}
\usepackage{multirow}
\usepackage{amsthm}
\usepackage{mathtools}
\usepackage{epstopdf,xcolor}
\usepackage{bbm}
\usepackage{dsfont}
\usepackage{pifont}
\usepackage{subcaption}
\usepackage{tikz, subcaption}
\usetikzlibrary{shapes, arrows, positioning}
\usepackage{authblk}
\usepackage{listings}
\usepackage{framed}
\usepackage{titletoc}

\usepackage{tabularx,booktabs,array,rotating} 
\newcolumntype{Y}{>{\raggedright\arraybackslash}X} 
\allowdisplaybreaks

\definecolor{royalblue}{rgb}{0.25, 0.41, 0.88}

\newtheorem{theorem}{Theorem}
\newtheorem{lemma}{Lemma}

\newtheorem{proposition}{Proposition}
\newtheorem{corollary}{Corollary}
\newtheorem{definition}{Definition}
\newtheorem{example}{Example}
\newtheorem{remark}{Remark}

\newtheorem{assumption}{Assumption}

\newcommand{\E}{\mathbb{E}}

\usepackage{chngcntr}
\usepackage{apptools}
\AtAppendix{\counterwithin{theorem}{subsection}}
\AtAppendix{\counterwithin{lemma}{subsection}}
\AtAppendix{\counterwithin{conjecture}{subsection}}
\AtAppendix{\counterwithin{proposition}{subsection}}
\AtAppendix{\counterwithin{corollary}{subsection}}
\AtAppendix{\counterwithin{definition}{subsection}}
\AtAppendix{\counterwithin{example}{subsection}}
\AtAppendix{\counterwithin{remark}{subsection}}
\AtAppendix{\counterwithin{claim}{subsection}}
\AtAppendix{\counterwithin{assumption}{subsection}}

\newcommand{\independent}{\perp\mkern-9.5mu\perp}

\usepackage[normalem]{ulem}
\newcommand{\stkout}[1]{\ifmmode\text{\sout{\ensuremath{#1}}}\else\sout{#1}\fi}

\usepackage{hyperref}
\hypersetup{
    colorlinks= true,
    linkcolor= black,
    filecolor= cyan,      
    citecolor= royalblue
}

\title{\bf \Large Combining Experimental and Observational Data for Identification and Estimation of Long-Term Causal Effects\\~\\}  

\author[1]{AmirEmad Ghassami\thanks{Corresponding author (\texttt{ghassami@bu.edu)}}}
\author[1]{Chang Liu}
\author[2]{Alan Yang}
\author[3]{David Richardson}
\author[4]{Ilya Shpitser}
\author[5]{Eric Tchetgen Tchetgen}

\affil[1]{Department of Mathematics and Statistics, Boston University}
\affil[2]{Department of Electrical Engineering, Stanford University}
\affil[3]{Department of Environmental and Occupational Health, University of California, Irvine}
\affil[4]{Department of Computer Science, Johns Hopkins University}
\affil[5]{Department of Statistics and Data Science, University of Pennsylvania}

\date{\vspace{-3mm}First Version: January 26, 2022; Current Version: September 28, 2025\vspace{-3mm}}

\begin{document}

\maketitle

\begin{abstract}
We study identifying and estimating the causal effect of a treatment variable on a long-term outcome using data from an observational and an experimental domain. The observational data are subject to unobserved confounding. Furthermore, subjects in the experiment are only followed for a short period; thus, long-term effects are unobserved, though short-term effects are available. Consequently, neither data source alone suffices for causal inference on the long-term outcome, necessitating a principled fusion of the two. We propose three approaches for data fusion for the purpose of identifying and estimating the causal effect. The first assumes equal confounding bias for short-term and long-term outcomes. The second weakens this assumption by leveraging an observed confounder for which the short-term and long-term potential outcomes share the same partial additive association with this confounder. The third approach employs proxy variables of the latent confounder of the treatment-outcome relationship, extending the proximal causal inference framework to the data fusion setting. For each approach, we develop influence function-based estimators and analyze their robustness properties. We illustrate our methods by estimating the effect of class size on 8th-grade SAT scores using data from the Project STAR experiment combined with observational data from the Early Childhood Longitudinal Study.\\

\end{abstract}	

\section{Introduction}
\label{sec:intro}

The gold standard for estimating the causal effect of a treatment variable on an outcome variable of interest is to conduct a randomized experiment. This is due to the fact that randomization ensures (conditional) exchangeability (also known as ignorability), which implies that the treatment and control groups are comparable. However, conducting such experiments can be costly and time-consuming, often resulting in limited or incomplete experimental data. In contrast, large-scale observational data is frequently available, yet there may be unobserved confounders of the treatment-outcome relationship in the setting, rendering causal inference impossible without positing extra assumptions. Given the limitations of both data sources, a natural question arises: \emph{Can we improve our causal inference by combining experimental and observational data?}

In some applications, experiments can be run on the target variable of interest and data fusion is merely for the sake of improving estimation efficiency \citep{kallus2020role}. However, in many settings, the experiment is run with another target variable different from the primary target. Especially, the practicalities of conducting an experiment with human subjects dictates that observations on subjects (and their compliance with the study protocol) only extends over a relatively short period of time from enrollment in the experiment.  Hence, information is often missing on long-term outcomes (primary target) in randomized experiments. In this case, the experimental data alone cannot be used for estimating the causal effect of the treatment on the long-term outcome variable as it only includes the short-term outcome variable, and if the observational domain is confounded, the observational data alone cannot be used for causal inference either. This is the setup that we focus on in this work. An example of such a setup, discussed by \cite{athey2020combining}, is the estimation of the effect of class size on eighth-grade test scores in New York schools, where we have access to observational data from New York schools, and Project STAR \citep{DVN/SIWH9F_2008} serves as the experimental dataset, including test scores only through third grade.

In their work, \cite{athey2020combining} proposed a method for combining data from experimental and observational domains to enable causal inference. They showed that, under exchangeability-type assumptions for ensuring internal and external validity of the experimental data, along with an extra novel assumption termed latent unconfoundedness, the average treatment effect (ATE) on the long-term outcome in the observational data can be identified. In this paper, we first review their proposed approach and discuss the latent unconfoundedness assumption. We then propose three alternative approaches for data fusion to estimate ATE as well as the effect of treatment on the treated (ETT). 

Our first proposed data fusion approach is based on assuming equal confounding bias for the short-term and long-term outcomes, which we refer to as the \emph{equi-confounding assumption}. We consider both additive and quantile-quantile equi-confounding. Roughly speaking, the equi-confounding assumption posits that the magnitude of the confounding bias for the short-term and the long-term outcome variables are the same. This approach draws inspiration from the literature of difference-in-differences (DiD) framework \citep{card1990impact, angrist2008mostly}, change-in-changes framework \citep{athey2006identification}, and negative control-based causal inference \citep{lipsitch2010negative,tchetgen2014control,sofer2016negative}. 
Our second proposed data fusion approach is based on assuming the existence of an observed confounder in the system called the \emph{bespoke instrumental variable} (BSIV), such that the short-term and long-term potential outcome variables have the same partial additive association with that confounder. The existence of such a variable allows us to relax the equi-confounding assumption in the first approach by removing the need for assuming restrictions on the selection bias for the short-term and the long-term outcomes. Hence, in the presence of a BSIV in the setting, the researcher should prefer the use of this method over the equi-confounding method. This method builds on the BSIV causal inference framework recently introduced by \cite{richardson2022bespoke}.
Finally, our third proposed data fusion approach relies on the presence of a \emph{proxy variable} of the latent confounder, which is independent of the outcome variables conditional on the treatment and all observed and unobserved confounder variables. This approach extends the proximal causal inference framework \citep{miao2018identifying,tchetgen2020introduction,cui2020semiparametric} to the data fusion setting, but requires only a single proxy variable--unlike the standard proximal framework, which relies on two.

We formally establish that the unbiased information encoded in the experimental data \emph{enables the relaxation of key identification assumptions} in standard DiD methods, the original BSIV approach, and proximal causal inference methods:
Standard DiD methods applied to our setting would require that the treatment cannot causally impact the short-term outcome. However, as the short-term outcome is a post-treatment variable, this assumption may not hold. We demonstrate that by leveraging experimental data, we can relax this requirement by anchoring the short-term causal effect at that observed in the experimental sample, thus enabling identification of the treatment’s impact on the long-term outcome.
Similarly, in the original BSIV causal inference framework, it is assumed that there exists a reference domain in which the treatment is not applied and hence, the outcome in that domain in fact represents the potential outcome under no treatments. In our setup, we relax this assumption by utilizing the internal validity of the experimental domain.
Likewise, in the proximal approach, the standard framework would require exclusion restriction of no treatment effect on the short-term outcome. We demonstrate that by leveraging experimental data, such requirement can also be relaxed.

To the best of our knowledge, we are the first to propose nonparametric identification methods for long-term causal effects that allow for latent confounders influencing the treatment as well as both short- and long-term outcomes. After the release of the first draft of our work, \cite{imbens2022long} also proposed an approach for identifying long-term causal effects based on the proximal causal inference framework. We discuss that work in the Supplementary Materials. Its main identification result relies on stronger assumptions; however, the authors also present an extension to their setup that resembles our proximal data fusion approach.  Moreover, their method requires access to three short-term outcome variables, each influenced only by its immediate predecessor--an assumption that may not be feasible in many real-world applications.

The rest of the paper is organized as follows. We describe the model and parameters of interest in Section \ref{sec:desc}. In Section \ref{sec:athey}, we review the approach of \cite{athey2020combining} and discuss their latent unconfoundedness assumption. Our proposed alternative methods, the equi-confounding, BSIV, and proximal causal data fusion approaches, are presented in Sections \ref{sec:AltIDAss}, \ref{sec:bsiv}, and \ref{sec:proximal}, respectively. In Section \ref{sec:estimation}, we focus on the estimation aspect of the parameters of interest, and propose influence function-based estimation strategies for each of our data fusion approaches and study the robustness properties of the proposed estimators.
We evaluated our proposed methods on synthetic data in Section \ref{sec:simulation}. In Section \ref{sec:application}, we apply our proposed methods to estimate the effect of class size on long-term educational outcomes, measured by 8th-grade SAT scores, by combining data from the Project STAR experiment with observational data from the Early Childhood Longitudinal Study \citep{tourangeau2009eclsk}. Our concluding remarks are provided in Section \ref{sec:conc}. All the proofs are provided in the Supplementary Materials.

\section{Problem Description}
\label{sec:desc}

Let \( A \) denote a binary treatment variable, \( X \) the vector of pre-treatment covariates, \( M \) the short-term outcome, \( Y \) the long-term outcome, and \( U \) the set of unobserved (latent) confounders of the treatment--outcome relationship. Let \( G \in \{O, E\} \) indicate the data domain, with \( G = O \) corresponding to the observational domain and \( G = E \) to the experimental domain. We denote the set of all observed variables by \( V \). The potential short-term and long-term outcomes under treatment level \( a \in \{0,1\} \) are denoted by \( M^{(a)} \) and \( Y^{(a)} \), respectively. We observe independent and identically distributed (i.i.d.) data from \( \{A, X, M\} \) in the experimental domain, where treatment is (conditionally) randomized. In contrast, i.i.d. data from \( \{A, X, M, Y\} \) are available in the observational domain, where both short- and long-term outcomes are observed, but the treatment--outcome relationship is confounded by the unobserved variables \( U \). Thus, while the experimental data provide unconfounded information on the short-term outcome, only the observational data contain information about the long-term outcome, albeit subject to confounding.

Our objective is to identify two causal parameters of interest in the observational domain: the average treatment effect (ATE) in the observational population,
\[
\theta_{\text{ATE}} = \E\left[Y^{(1)} - Y^{(0)} \mid G = O\right],
\]
and the effect of treatment on the treated (ETT),
\[
\theta_{\text{ETT}} = \E\left[Y^{(1)} - Y^{(0)} \mid A = 1,\, G = O\right].
\]
Clearly, these parameters are not identifiable from the experimental data alone, as the long-term outcome \( Y \) is unobserved in that domain, and also the information regarding the treatment effect in the experimental domain may be not relevant to the observational domain.
Moreover, due to potential unobserved confounding, the parameters are also not identifiable from the observational data alone without additional assumptions. To proceed, we assume that the data-generating process satisfies the following conditions.

\begin{figure}[t!]
\centering
\includegraphics[scale=0.8]{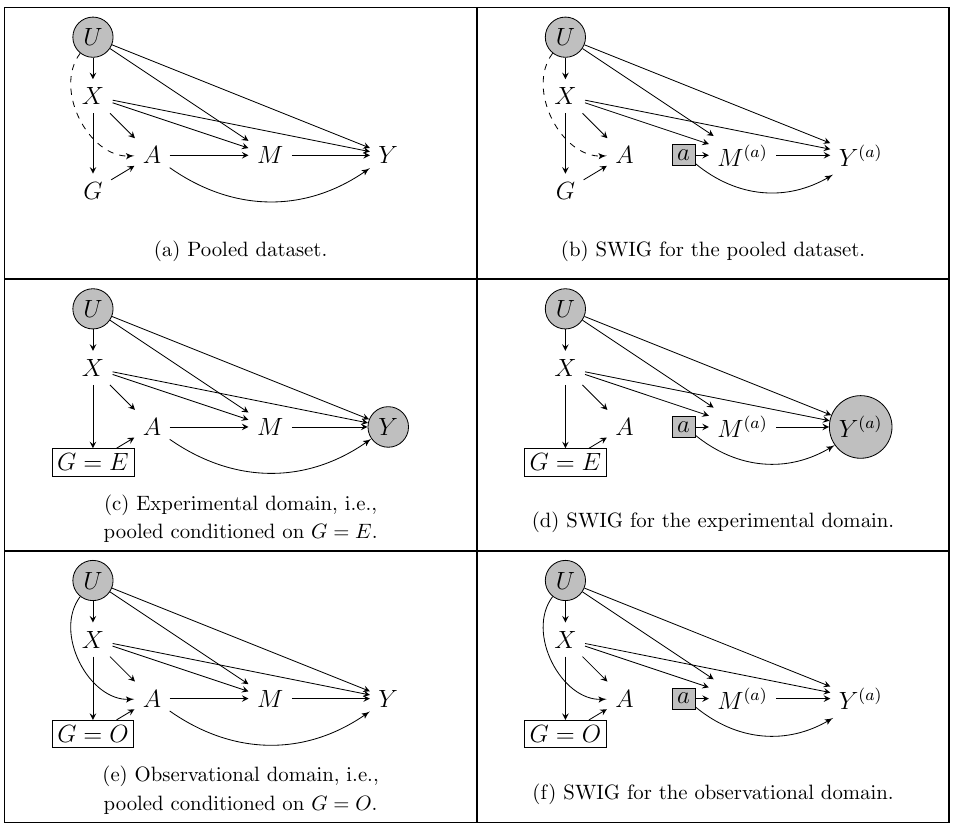}
\caption{A graphical model that satisfies Assumptions \ref{assumption:exchangeU}-\ref{assumption:ExVal}. Dashed edge is \emph{context-specific}: it only exists in the observational domain. Variables in gray circles are unobserved; variables in white boxes are conditioned on.}
\label{fig:GMP}
\end{figure}

\begin{assumption}[$\{X,U,G\}$--conditional exchangeability w.r.t. $A$]
\label{assumption:exchangeU}
\[
A\independent\{Y^{(a)},M^{(a)}\}\mid \{X,U,G\}\qquad \forall a\in\{0,1\}.
\]
\end{assumption}
\begin{assumption}[$X$--conditional exchangeability  w.r.t. $A$ in the experimental domain]
\label{assumption:IntVal}
\[
A\independent\{Y^{(a)},M^{(a)}\}\mid \{X,G=E\}\qquad \forall a\in\{0,1\}.
\]	
\end{assumption}
\begin{assumption}[$X$--conditional exchangeability  w.r.t. $G$]
\label{assumption:ExVal}
\[
G\independent\{Y^{(a)},M^{(a)}\}\mid X\qquad \forall a\in\{0,1\}.
\]	
\end{assumption}
Assumption \ref{assumption:exchangeU} is a much milder version of the standard conditional exchangeability assumption \citep{hernan2020causal} as it is stated conditional on any unobserved confounder $U$.
Assumption \ref{assumption:IntVal} is a context-specific conditional independence assumption. That is, a conditional independence which is only realized conditional on a certain event.
We also posit the standard consistency and positivity assumptions. Figure \ref{fig:GMP} represents a graphical model that satisfies Assumptions \ref{assumption:exchangeU}-\ref{assumption:ExVal}. Variables in gray circle are unobserved.
 Figures \ref{fig:GMP} (a) and \ref{fig:GMP} (b) represent the pooled dataset, Figures \ref{fig:GMP} (c) and \ref{fig:GMP} (d) represent the experimental dataset, which is the pooled data set conditioned on $G=E$, and Figures \ref{fig:GMP} (e) and \ref{fig:GMP} (f) represent the observational dataset, which is the pooled data set conditioned on $G=O$. Figures \ref{fig:GMP} (b), \ref{fig:GMP} (d), and \ref{fig:GMP} (f) represent single world intervention graphs (SWIGs), which are graphical models obtained from the original graphs, which include potential outcome variables, and hence, can be used for representing conditional independences involving potential outcome variables. See \citep{richardson2013single} for the definition and details.

Assumptions \ref{assumption:IntVal} and \ref{assumption:ExVal} are not sufficient for identification of the causal parameters of interest. In the following, we first review an approach proposed by \cite{athey2020combining} for identification based on an extra assumption called latent unconfoundedness, and then propose our alternative approaches.

\section{\cite{athey2020combining} Approach}
\label{sec:athey}

\cite{athey2020combining} introduced an approach for combining experimental and observational data to identify the average treatment effect parameter, $\theta_{\text{ATE}}$. Their approach relies on Assumptions \ref{assumption:IntVal} and \ref{assumption:ExVal}, along with an additional identifying condition stated below. As we show in this section, these assumptions also suffice to identify the effect of treatment on the treated, $\theta_{\text{ETT}}$.

\begin{assumption}[Latent Unconfoundedness]
\label{assumption:LaUn}
\[
A \independent Y^{(a)} \mid \{X, M^{(a)}, G = O\} \qquad \text{for all } a \in \{0,1\}.
\]	
\end{assumption}

\begin{theorem}
\label{thm:athey}
Under Assumptions \ref{assumption:IntVal}--\ref{assumption:LaUn}, the parameters $\theta_{\text{ATE}}$ and $\theta_{\text{ETT}}$ are identified. The identification formulae are presented in Supplementary Material \ref{sec:athey::supp}.
\end{theorem}
\begin{remark}
One can show that under Assumptions \ref{assumption:IntVal}--\ref{assumption:LaUn}, the full distribution \( p(Y^{(a)} \mid G = O) \) is identified. Hence, the latent unconfoundedness assumption in \citep{athey2020combining} is sufficiently strong to identify a broad class of causal estimands beyond mean contrasts, including functionals of the entire counterfactual distribution. The same holds for our proposed method in Section~\ref{sec:QQCondEqui} and the proximal data fusion approach in Section~\ref{sec:proximal}.
\end{remark}

\noindent
{\bf Discussion of Assumption \ref{assumption:LaUn}.}
Assume the data-generating process is governed by a nonparametric structural equation model. We consider two cases:
\begin{enumerate}
\item The distribution is \emph{faithful} \citep{spirtes2000causation} to the causal model. That is, all conditional independence relations are reflected in causal relations among the variables and no conditional independence arises from specific cancellations or alignment of the modules in the causal model. Under faithfulness, a graphical model is a reliable diagnostic tool for the existence of a \emph{natural sequential data generating processes}, in which edges represent direct causal relations. In this setting, Assumption~\ref{assumption:LaUn} implies the absence of a directed edge from the latent variable \( U \) to the outcome \( Y \). In other words, there must be no unobserved confounder that causally influences both \( A \) and \( Y \). Consequently, Assumption \ref{assumption:LaUn} may be overly strong for many real-world settings, where violations of faithfulness are unlikely.
\item The distribution violates faithfulness. In this case, it is possible for \( A \) and \( Y \) to share a latent confounder, yet still satisfy Assumption \ref{assumption:LaUn} due to structural cancellation. However, such scenarios are typically regarded as pathological or non-generic. An example of this case is discussed below.
\end{enumerate}
\begin{example}
\label{ex:1}
Suppose that in the observational domain, variables $M$ and $Y$ are generated via the following linear structural equation model:
\[
M = \theta_M A+\gamma_M X+U,\qquad\qquad Y = \theta_Y A+\zeta_Y M+\gamma_Y X	+\delta_Y U +\epsilon,
\]
such that $A\independent\epsilon\mid \{X,U,G=O\}$. 
Importantly, the non-generic assumption here is excluding an independent exogenous noise from the structural equation corresponding to the variable $M$.
Note that $Y^{(a)}=\theta_Y a+\zeta_Y M^{(a)}+\gamma_Y X+\delta_Y U+\epsilon = (\theta_Y+\zeta_Y\theta_M) a+(\gamma_Y+\zeta_Y\gamma_M) X+(\delta_Y+\zeta_Y) U+\epsilon$.
Therefore, the assumption $A\independent\epsilon\mid \{X,U,G=O\}$ implies that $A\independent Y^{(a)}\mid \{X,U,G=O\}$. Also, $M^{(a)}=\theta_M a+\gamma_M X+U$ and hence $U=\theta_M a+\gamma_M X-M^{(a)}$. Therefore, $A\independent Y^{(a)}\mid \{X,U,G=O\}$ implies that $A\independent Y^{(a)}\mid \{X,M^{(a)},G=O\}$, which is Assumption \ref{assumption:LaUn}.
Hence, the assumption basically implies that adjusting for $\{M^{(a)},X\}$ is equivalent to adjusting for $\{U,X\}$, i.e., roughly speaking, observing $\{M^{(a)},X\}$ is equivalent to observing $\{U,X\}$.
Note that excluding an independent exogenous noise from the structural equation corresponding to the variable $M$ is essential for this equivalence to hold. Here, the deterministic dependence of $M$ on $X$, $A$ and $U$ violates faithfulness.
\end{example}

In the following three sections, we present our alternative identification approaches. In the main text, we focus on results for the parameter $\theta_{\text{ETT}}$, while the corresponding results for $\theta_{\text{ATE}}$ are provided in the Supplementary Materials.

\section{Approach 1: Equi-Confounding Data Fusion}
\label{sec:AltIDAss}

Our first proposed approach is based on the assumption that the conditional confounding bias is equal for the short-term and long-term outcomes. This restriction is inspired by assumptions used in identification strategies involving negative outcome controls \citep{lipsitch2010negative,tchetgen2014control,sofer2016negative}, and is conceptually similar to the parallel trends assumption in the difference-in-differences (DiD) framework \citep{card1990impact,angrist2008mostly}.

\begin{assumption}[Conditional Additive Equi-Confounding Bias]
\label{assumption:CondEqui}
With probability one,
\begin{align*}
\E[M^{(0)}\mid X,A=0,&G=O]-\E[M^{(0)}\mid  X,A=1,G=O]\\
&=\E[Y^{(0)}\mid  X,A=0,G=O]-\E[Y^{(0)}\mid  X,A=1,G=O].
\end{align*}	
\end{assumption}
Unlike the standard conditional exchangeability assumption, Assumption~\ref{assumption:CondEqui} allows for the presence of latent confounders. It does not require that the treated and control groups have identical conditional potential outcome distributions. Rather, it asserts that the difference of the expected value of the short-term potential outcome across these two groups (that is the bias due to confounding on an additive scale) is the same as that of the long-term potential outcome variable. Equivalently, it implies that the expected change from the short-term to long-term potential outcome is the same in each stratum of $X$ across the treated and control groups.
\begin{example}
Assumption~\ref{assumption:CondEqui} is satisfied if the data are generated from the following model:
\[
M = g_M(A, X) + f_M(X, U) + \epsilon_M, 
\qquad \qquad Y = g_Y(A, X) + f_Y(X, M, U) + \epsilon_Y,
\]
where $\epsilon_M$ and $\epsilon_Y$ are independent noise terms and we have $\E[f(X,M,U)\mid X,A=1]=\E[f(X,M,U)\mid X,A=0]$, where $f(X,M,U)=f_Y(X,M,Y)-f_M(X,U)$, and $g_M$, $f_M$, $g_Y$, and $f_Y$ can be stochastic functions.
\end{example}

\begin{remark}
As noted above, Assumption~\ref{assumption:CondEqui} is similar in flavor to the parallel trends assumption in the DiD framework. However, in that setting, the counterpart of $M$ is assumed to be a pre-treatment variable, and thus cannot be causally affected by the treatment. In contrast, in our setting, $M$ is post-treatment and may be influenced by $A$.
\end{remark}

We now present our identification result for $\theta_{\text{ETT}}$ under conditional equi-confounding assumption.

\begin{theorem}
\label{thm:equiETTCond}
Under Assumptions~\ref{assumption:IntVal}, \ref{assumption:ExVal}, and~\ref{assumption:CondEqui}, the parameter $\theta_{\text{ETT}}$ is identified as:
{\footnotesize
\begin{equation}
\label{eq:param1}
\begin{aligned}
&\theta_{\text{ETT}} 
= \E[Y \mid A=1, G=O] 
+ \E\left[ \frac{1}{p(A=1 \mid X, G=O)} \E[M \mid X, A=0, G=O] \mid A=1, G=O \right] \\
&\quad - \E\left[ \frac{1}{p(A=1 \mid X, G=O)} \E[M \mid X, A=0, G=E] + \E[Y \mid X, A=0, G=O] \mid A=1, G=O \right].
\end{aligned}
\end{equation}
}
\end{theorem}
The corresponding identification result for the parameter $\theta_{\text{ATE}}$ is provided in Supplementary Material~\ref{sec:AltIDAss::supp}. 

\begin{remark}
We have provided the unconditional counterpart of the conditional equi-confounding bias assumption in Assumption~\ref{assumption:Equi} and the corresponding identification result in Supplementary Material~\ref{sec:AltIDAss::supp}. Importantly, Assumptions~\ref{assumption:Equi} and~\ref{assumption:CondEqui} are not nested and do not imply one another. Each allows for different forms of heterogeneity and interactions among variables. Assumption \ref{assumption:Equi} posits that the confounding bias for the short-term and the long-term outcome variables are the same. However, it might be the case that the researcher does not believe that the bias for the two outcomes are equal marginally, but this equality holds in each stratum of $X$. In this case, Assumption \ref{assumption:CondEqui} is more appropriate. In this sense, Assumption~\ref{assumption:CondEqui} may be viewed as a generally weaker assumption.
\end{remark}

\subsection{Quantile-Quantile Equi-Confounding Data Fusion}
\label{sec:QQCondEqui}

We note that the (conditional) additive equi-confounding bias assumption may be restrictive, as it requires the short-term and long-term outcomes to be measured on the same scale. While this is not a limitation in our specific application---where \( M \) and \( Y \) represent short- and long-term versions of the same outcome---it may be problematic in other settings. To address this, we propose a generalization of the additive equi-confounding framework, inspired by the changes-in-changes approach in panel data analysis \citep{athey2006identification} and its analogue in the negative control inference literature \citep{sofer2016negative}. This generalization also enables identification of causal parameters beyond mean-based quantities such as ATE and ETT. We present this extension in Supplementary Material~\ref{sec:AltIDAss::supp}.

\section{Approach 2: Bespoke IV Data Fusion}
\label{sec:bsiv}

In this section, we show that the (conditional) equi-confounding bias assumption can be relaxed in the presence of a variable \( B \) among the observed pre-treatment covariates, provided it satisfies a specific condition regarding its association with the potential outcomes. The key idea is that, under such conditions, \( B \) can play a role analogous to that of an instrumental variable. This approach is inspired by the recently proposed Bespoke Instrumental Variable (BSIV) framework of \cite{richardson2022bespoke}. Following that work, we refer to the variable \( B \) as a BSIV and refer to the resulting method as the BSIV data fusion approach. We present the framework for a binary BSIV, though as discussed in Remark~\ref{rmk:nonbinary_BSIV}, the approach naturally extends to non-binary variables. To maintain consistency with previous notation, we denote the remaining observed covariates by \( X \). The requirements on the BSIV variable \( B \) are as follows.
\begin{assumption}[BSIV Relevance]
\label{assumption:BSIVrelevance}
\[
\E[A \mid B=0, X, G=O] \neq \E[A \mid B=1, X, G=O].
\]
\end{assumption}
\begin{assumption}[BSIV Partial Additive Equi-Association]
\label{assumption:BSEqui}
With probability one,
\begin{align*}
\E[M^{(0)}\mid X,B=1,&G=O]-\E[M^{(0)}\mid  X,B=0,G=O]\\
&=\E[Y^{(0)}\mid  X,B=1,G=O]-\E[Y^{(0)}\mid  X,B=0,G=O].
\end{align*}	
\end{assumption}
Assumption~\ref{assumption:BSIVrelevance} is the standard, testable instrumental variable (IV) relevance condition, common to all IV frameworks. Assumption~\ref{assumption:BSEqui} posits that the short-term and long-term potential outcomes share the same partial additive association with the covariate \( B \). This condition is a weaker version of Assumption~\ref{assumption:CondEqui}. While Assumption~\ref{assumption:CondEqui} requires the selection bias to be equal for both \( M \) and \( Y \), Assumption~\ref{assumption:BSEqui} does not constrain the selection bias directly. Instead, it focuses solely on the additive association of the potential outcomes with the observed variable \( B \), thereby giving the researcher flexibility to select a variable they believe is most likely to satisfy this condition. A similar assumption appears in the framework of \cite{richardson2022bespoke}. However, their setting assumes that, in the reference population, treatment assignment is withheld through external intervention, so the counterfactual outcome under \( A = 0 \) coincides with the observed outcome.

Note that Assumption~\ref{assumption:BSEqui} can be equivalently expressed as:
\[
\E[M^{(0)} - Y^{(0)} \mid X, B=1, G=O] 
= \E[M^{(0)} - Y^{(0)} \mid X, B=0, G=O].
\]
This resembles the restriction imposed on an instrumental variable in a standard IV framework when the outcome is defined as the difference \( Y - M \). However, unlike a classical IV, the covariate \( B \) in our setting does not satisfy the two key assumptions of standard IV analysis---namely, unconfoundedness and the exclusion restriction. This motivates the term \emph{bespoke instrumental variable} for \( B \).

To investigate identification, we consider the following nonparametric reparametrization of the outcome regression function, which is a variation of the approach proposed in \citep{robins1994correcting, tchetgen2013alternative}:
\begin{align*}
&\E[Y - M \mid A = a, B = b, X, G = O] \\
& = \beta_1(b, X) a + \gamma_0(b, X) \{a - p(A = 1 \mid B = b, X, G = O)\} + \E[Y^{(0)} - M^{(0)} \mid B = b, X, G = O],
\end{align*}
where
\begin{align*}
\beta_1(B, X) &\coloneqq \E\big[ \{Y^{(1)} - M^{(1)}\} - \{Y^{(0)} - M^{(0)}\} \mid A = 1, B, X, G = O \big], \\
\gamma_0(B, X) &\coloneqq \E[Y^{(0)} - M^{(0)} \mid A = 1, B, X, G = O] - \E[Y^{(0)} - M^{(0)} \mid A = 0, B, X, G = O].
\end{align*}

We use this reparametrization to identify the causal effect contrast function \( \beta_1 \), and in turn the parameter \( \theta_{\text{ETT}} \). An additional reparametrization is used to identify \( \theta_{\text{ATE}} \), as discussed in Supplementary Material~\ref{sec:bsiv::supp}.
By Assumption~\ref{assumption:BSEqui}, the term \( \E[Y^{(0)} - M^{(0)} \mid B = b, X, G = O] \) does not depend on \( b \). Therefore, for each value of \( X \), the reparametrization yields four equations and five unknowns: two for \( \beta_1(b, X) \), two for \( \gamma_0(b, X) \), and one for \( \E[Y^{(0)} - M^{(0)} \mid B = b, X, G = O] \). To achieve point identification, we impose one of two additional restrictions: either \( \beta_1 \) does not depend on \( b \), or \( \gamma_0 \) does not depend on \( b \). These are formalized below.

\begin{assumption}[Partial Homogeneity of Causal Effect Contrast]
\label{assumption:homogeneityETT}
With probability one,
\begin{align*}
&\E[Y^{(1)} - Y^{(0)} \mid A = 1, B = 1, X, G = O] - \E[M^{(1)} - M^{(0)} \mid A = 1, B = 1, X, G = O] \\
&= \E[Y^{(1)} - Y^{(0)} \mid A = 1, B = 0, X, G = O] - \E[M^{(1)} - M^{(0)} \mid A = 1, B = 0, X, G = O].
\end{align*}
\end{assumption}

\begin{assumption}[Partial Homogeneity of Bias Contrast]
\label{assumption:homogeneityBias}
With probability one,
\begin{equation*}
\begin{aligned}
&\Big\{ \E[Y^{(0)} \mid A = 1, B = 1, X, G = O] - \E[Y^{(0)} \mid A = 0, B = 1, X, G = O] \Big\} \\
&\quad - \Big\{ \E[M^{(0)} \mid A = 1, B = 1, X, G = O] - \E[M^{(0)} \mid A = 0, B = 1, X, G = O] \Big\} \\
&= \Big\{ \E[Y^{(0)} \mid A = 1, B = 0, X, G = O] - \E[Y^{(0)} \mid A = 0, B = 0, X, G = O] \Big\} \\
&\quad - \Big\{ \E[M^{(0)} \mid A = 1, B = 0, X, G = O] - \E[M^{(0)} \mid A = 0, B = 0, X, G = O] \Big\}.
\end{aligned}
\end{equation*}
\end{assumption}
Note that neither Assumption~\ref{assumption:homogeneityETT} nor Assumption~\ref{assumption:homogeneityBias} restricts the set of causes or effects for the involved variables, and the variable \( B \) may act as a confounder of \( A \), \( M \), and \( Y \). 
Assumption~\ref{assumption:homogeneityETT} posits that the difference between the conditional in-group causal effects for outcomes \( Y \) and \( M \) is invariant to \( B \). For instance, the ETT for the outcome \( Y - M \) does not vary with \( B \). A special case where this holds is when the conditional in-group causal effect for each of \( Y \) and \( M \) is individually independent of \( B \). Assumption~\ref{assumption:homogeneityBias} requires that the selection bias for the contrast \( Y - M \) is constant across levels of \( B \). One special case where this assumption holds is when the selection bias for both \( Y \) and \( M \) is individually invariant to \( B \). Another special case is when the curly brackets on the right hand side are equal, and also the curly brackets on the left hand side are equal---i.e., both sides of the equation are zero. This situation corresponds precisely to the conditional equi-confounding bias condition stated in Assumption~\ref{assumption:CondEqui}. Therefore, Assumption~\ref{assumption:homogeneityBias} is strictly weaker than Assumption~\ref{assumption:CondEqui}. Hence, in the presence of a BSIV the researcher should prefer using the BSIV data fusion approach over the equi-confounding approach.
\begin{theorem}
\label{thm:BSIV_ett}
Define $\pi(B,X) \coloneqq p(A=1 \mid B,X,G=O)$, and for $a,b \in \{0,1\}$ define 
\[
E_{ab}^O(X) \coloneqq \E[Y - M \mid A=a, B=b, X, G=O], \quad
P_{ab}^O(X) \coloneqq p(A=a \mid B=b, X, G=O).
\]	

\begin{itemize}
\item[(a)] 
Under Assumptions \ref{assumption:IntVal}, \ref{assumption:ExVal}, \ref{assumption:BSIVrelevance}, \ref{assumption:BSEqui}, and \ref{assumption:homogeneityETT}, the parameter $\theta_{\text{ETT}}$ is identified by:
{\small
\begin{equation}
\label{eq:param2}
\begin{aligned}
\theta_{\text{ETT}}
&= \E\Bigg[
\frac{\E[Y - M \mid B=1, X, G=O] - \E[Y - M \mid B=0, X, G=O]}{P_{11}^O(X) - P_{10}^O(X)} \\
&\quad - \frac{\E[M \mid A=0, B, X, G=E]}{\pi(B, X)}
\ \Big| \ A=1, G=O
\Bigg]
+ \frac{\E[M \mid G=O]}{p(A=1 \mid G=O)}.
\end{aligned}
\end{equation}
}
\item[(b)]
Under Assumptions \ref{assumption:IntVal}, \ref{assumption:ExVal}, \ref{assumption:BSIVrelevance}, \ref{assumption:BSEqui}, and \ref{assumption:homogeneityBias}, the parameter $\theta_{\text{ETT}}$ is identified by:
{\small
\begin{equation}
\label{eq:param3}
\begin{aligned}
\theta_{\text{ETT}} 
&= \E\Bigg[
\big\{ E_{11}^O(X) - E_{01}^O(X) - E_{10}^O(X) + E_{00}^O(X) \big\} B + E_{10}^O(X) - E_{00}^O(X) \\
&\quad - \frac{E_{01}^O(X) \!-\! E_{00}^O(X)}{P_{01}^O(X) \!-\! P_{00}^O(X)}
\!+\! \E[M \!\mid \!A=1, B, X, G=O] 
\!-\! \frac{\E[M \!\mid \!A=0, B, X, G=E]}{\pi(B, X)} \\
&\quad + \frac{\E[M \mid A=0, B, X, G=O](1 - \pi(B, X))}{\pi(B, X)}
\ \Big| \ A=1, G=O
\Bigg].
\end{aligned}
\end{equation}
}
\end{itemize}
\end{theorem}
See Supplementary Material~\ref{sec:bsiv::supp} for the corresponding identification results for the ATE.
\begin{remark}
\label{rmk:real_ex_val}
Assumption~\ref{assumption:IntVal} is formulated as $A \independent \{Y^{(a)}, M^{(a)}\} \mid B, X, G=E$, and  Assumption~\ref{assumption:ExVal} is expressed as $G \independent M^{(a)} \mid B, X$. However, it is worth noting that the identification results in Theorem~\ref{thm:BSIV_ett} only require the weaker conditions, e.g., $\E[M^{(a)} \mid X, B, G=O] = \E[M^{(a)} \mid X, B, G=E]$, which is an equality of conditional expectations rather than a full conditional independence assumption.
\end{remark}
\begin{remark}
\label{rmk:nonbinary_BSIV}
As noted earlier, the bespoke instrumental variable \( B \) need not be binary. In Supplementary Material~\ref{sec:bsiv::supp}, we provide the counterparts of the assumptions in our BSIV data fusion framework for non-binary \( B \). The identification arguments extend straightforwardly to this setting with no substantive changes to the proofs.
\end{remark}

\section{Approach 3: Proximal Data Fusion}
\label{sec:proximal}

In this section, we present our third identification strategy, which leverages the presence of a variable that serves as a \emph{proxy} for the latent confounder. This approach is inspired by the proximal causal inference framework \citep{miao2018identifying, tchetgen2020introduction, cui2020semiparametric}. However, unlike that framework---which requires two distinct proxies---we only require a single proxy variable.\footnote{A comparison to other proximal setups is presented in Supplementary Material \ref{sec:proximal::supp}.} Specifically, we assume the existence of a variable \( Z \) in the observational domain that satisfies the following condition.

\begin{assumption}[Existence of Proxy Variable]
\label{assumption:proxy}
There exists a proxy variable \( Z \) for the latent confounder in the observational domain such that
\[
Z \independent \{M, Y\} \mid \{U, X, A, G=O\}.
\]
\end{assumption}

Figure~\ref{fig:prox} provides an example of a causal graph that satisfies this proxy variable condition. Importantly, no assumptions are required about \( Z \) in the experimental domain---indeed, \( Z \) need not even exist in that domain.

\begin{remark}
\label{rmk:proxyflexy}
Figure~\ref{fig:prox} illustrates just one of many possible graphical structures that can satisfy Assumption~\ref{assumption:proxy}. For example, the assumption still holds if the edge from \( Z \) to \( A \) is reversed, making \( Z \) a post-treatment variable. Alternatively, the edge between \( Z \) and \( A \) can be removed entirely, or \( Z \) and \( A \) can share a latent confounder. These alternatives illustrate the flexibility available to the researcher in selecting a variable to serve as the proxy \( Z \).
\end{remark}

In order to achieve identifiability, we impose the following additional assumption.

\begin{assumption}
\label{assumption:compexist1}
~
\vspace{-3mm}
\begin{itemize}
    \item[(i)] For any square-integrable function \( g \), and for all \( a \) and \( x \), if \( \E[g(U) \mid Z, a, x, G=O] = 0 \) almost surely, then \( g(U) = 0 \) almost surely.
    
    \item[(ii)] There exists a function \( h(m,a,x) \), called the \emph{outcome bridge function}, that solves the following integral equation:
    \begin{equation}
    \label{eq:ORproxexisth}
    \E[Y \mid Z, A, X, G=O] = \E[h(M, A, X) \mid Z, A, X, G=O].
    \end{equation}
\end{itemize}
\end{assumption}

We now state the identification result.
\begin{theorem}
\label{thm:POR}
Under Assumptions \ref{assumption:exchangeU}--\ref{assumption:ExVal}, \ref{assumption:proxy}, and \ref{assumption:compexist1}, the parameter \( \theta_{\text{ETT}} \) is identified as:
\begin{align}\label{eq:POR}
\theta_{\text{ETT}}
&= \frac{\E[Y \mid G=O]}{p(A=1 \mid G=O)} -
\frac{\E\big[\E[h(M,0,X) \mid A=0, X, G=E] \mid G=O\big]}{p(A=1 \mid G=O)}.
\end{align}
\end{theorem}
See Supplementary Material~\ref{sec:proximal::supp} for the corresponding identification results for the ATE.

\begin{figure}[t!]
\centering
		\tikzstyle{block} = [draw, circle, inner sep=2.5pt, fill=lightgray]
		\tikzstyle{input} = [coordinate]
		\tikzstyle{output} = [coordinate]
        \begin{tikzpicture}
            \tikzset{edge/.style = {->,> = latex'}}
            \node[] (a) at  (-2,0) {$A$};
            \node[] (n) at  (0,0) {$M$};
            \node[] (x) at  (-3,1) {$X$};
            \node[] (y) at  (2,0) {$Y$};   
            \node[] (z) at  (-3,-.5) {$Z$};                     
            \node[block] (u) at  (-3,2) {$U$};                       
            \draw[-stealth] (x) to (a);
            \draw[-stealth][edge, bend left=-65] (u) to (a);                        
            \draw[-stealth] (a) to (n);
            \draw[-stealth][edge, bend left=-35] (a) to (y);            
            \draw[-stealth] (u) to (y);
            \draw[-stealth][edge, bend left=-65] (u) to (z);            
            \draw[-stealth] (x) to (z);
            \draw[-stealth] (z) to (a);                                                
            \draw[-stealth] (u) to (n);
            \draw[-stealth] (n) to (y);            	  
            \draw[-stealth] (u) to (x);	              
            \draw[-stealth] (x) to (n);
            \draw[-stealth] (x) to (y);            	                                      
        \end{tikzpicture}
        \caption{Example of a model satisfying the Assumption \ref{assumption:proxy}.}
        \label{fig:prox}
\end{figure}

\subsection{An Alternative Proximal Identification Method}

We now present an alternative proximal identification strategy based on a different set of completeness and bridge function assumptions.

\newpage
\begin{assumption}
\label{assumption:compexist2}
~
\vspace{-3mm}
\begin{itemize}
    \item[(i)] For any square-integrable function \( g \), and for all \( a \) and \( x \), if \( \E[g(U) \mid M, a, x, G=O] = 0 \) almost surely, then \( g(U) = 0 \) almost surely.

    \item[(ii)] There exists a function \( q(z, a, x) \), called the \emph{treatment bridge function}, that solves the following integral equation:
    \begin{equation*}
    \E[q(Z, A, X) \mid M, A, X, G=O] = \frac{p(M \mid A, X, G=E)}{p(M \mid A, X, G=O) \cdot p(A \mid X, G=O)}.
    \end{equation*}
\end{itemize}
\end{assumption}

The corresponding identification result is as follows.
\begin{theorem}
\label{thm:PIPW}
Under Assumptions \ref{assumption:exchangeU}--\ref{assumption:ExVal}, \ref{assumption:proxy}, and \ref{assumption:compexist2}, the parameter \( \theta_{\text{ETT}} \) is identified as:
\begin{align}\label{eq:PIPW}
\theta_{\text{ETT}}
&= \frac{\E[Y \mid G=O]}{p(A=1 \mid G=O)} -
\frac{\E[I(A=0) Y q(Z, A, X) \mid G=O]}{p(A=1 \mid G=O)}.
\end{align}
\end{theorem}
See Supplementary Material~\ref{sec:proximal::supp} for the corresponding identification results for the ATE.

\section{Estimation Strategies}
\label{sec:estimation}

In this section, we turn to the problem of estimation and propose influence function (IF)-based estimators for the parameters of interest under each of our data fusion frameworks. IF-based estimators are advantageous in that their bias is of second order with respect to the error in estimating the nuisance functions—an important property in the presence of complex data-generating processes. Furthermore, under certain convergence rate conditions on the nuisance function estimators—which are weaker than those required for non-IF-based estimators—the resulting IF-based estimator is consistent and asymptotically normal. This enables valid inference without requiring resampling methods such as the bootstrap.

An additional benefit is that our IF-based estimators exhibit \emph{multiple robustness}: the estimator remains unbiased even when some of the nuisance functions are misspecified. While the asymptotic normality of IF-based estimators is now standard in the semiparametric literature (see, e.g., \citep{bickel1993efficient,newey1990semiparametric,tsiatis2007semiparametric,robins2017minimax,kennedy2024semiparametric}), the analysis of multiple robustness is more context-specific and depends on the parameter being estimated. Therefore, we will focus on establishing multiple robustness results specific to each framework.
As in the previous sections, 
we only provide the results regarding the ETT in the main text; the corresponding results for the ATE are provided in Supplementary Material~\ref{sec:estimation::supp}.

In the subsections that follow, we derive the influence functions and analyze the robustness properties of the proposed estimators under each data fusion framework. Given an IF-based moment function, we adopt the cross-fitting procedure of \cite{chernozhukov2018double} to decouple the estimation of nuisance functions from that of the target parameter. This procedure allows for weaker regularity conditions on the nuisance estimators while still ensuring asymptotic normality. The estimation procedure proceeds as follows:
\begin{itemize}
\item {\bf Step 1.} Partition the sample into \( L \) equally sized folds: \( \{I_1, \ldots, I_L\} \).
\item {\bf Step 2.} For each \( \ell \in \{1, \ldots, L\} \), estimate the set of nuisance functions \( \zeta^{\text{approach}} \) using the data from all folds except \( I_\ell \), yielding \( \hat\zeta_\ell^{\text{approach}} \).
\item {\bf Step 3.} For each \( \ell \in \{1, \ldots, L\} \), estimate the parameter of interest \( \hat\psi_\ell \) by solving:
\[
\frac{1}{|I_\ell|}\sum_{i \in I_\ell} \Phi^{\text{approach}}(V_i; \hat\zeta_\ell^{\text{approach}}, \hat\psi_\ell) = 0,
\]
where \( \Phi^{\text{approach}} \) is the IF-based moment function, which differs across our approaches.
\item {\bf Step 4.} Define the final estimator of the parameter as the average across folds:
\begin{center}
    $\hat\psi^{\text{approach}} = \frac{1}{L} \sum_{\ell=1}^L \hat\psi_\ell$.
\end{center} 
\end{itemize}

\subsection{Estimation under Equi-Confounding Data Fusion}

In this subsection, we consider estimation under the equi-confounding data fusion framework, focusing on the case of conditional equi-confounding. Recall from Theorem~\ref{thm:equiETTCond} that, under the assumptions of this framework, for the parameter $\theta_{\text{ETT}}$, we only need to focus on the right hand side of Equation \eqref{eq:param1}, which is a functional of the observed data distribution. We denote this functional by $\psi^{\text{equi}}_{\text{ETT}}$. One can use this functional to construct a simple plug-in estimator. However, such an estimator generally suffers from first-order bias with respect to the error in the estimation of the nuisance functions. Hence, as mentioned earlier, we use an IF-based estimator instead.

We begin by deriving the influence function for the parameter $\psi^{\text{equi}}_{\text{ETT}}$. Let the collection of nuisance functions be defined as:
$\zeta^{\text{equi}}:=
\Big\{
\mu^{G=E}_M(A=\cdot,X=\cdot)\coloneqq \E[M\mid X=\cdot,A=\cdot,G=E], 
\mu^{G=O}_M(A=\cdot,X=\cdot)\coloneqq \E[M\mid X=\cdot,A=\cdot,G=O], 
\mu^{G=O}_Y(A=\cdot,X=\cdot)\coloneqq \E[Y\mid X=\cdot,A=\cdot,G=O],
\pi^{G=E}(\cdot)\coloneqq p(A=1\mid X=\cdot,G=E), 
\pi^{G=O}(\cdot)\coloneqq p(A=1\mid X=\cdot,G=O),
p(G=E\mid X=\cdot)
\Big\}$.
\begin{theorem}
\label{thm:IF:equi-conf:ETT}
Under a nonparametric model, the efficient influence function for the parameter $\psi^{\text{equi}}_{\text{ETT}}$ is given by:
{\small
\begin{align*}
&IF_{\psi^{\text{equi}}_{\text{ETT}}}(V)=
\frac{1}{p(A=1,G=O)}\Big\{
\frac{I(G=O)I(A=0)}{1-\pi^{G=O}(X)}\{M-\mu^{G=O}_M(A=0,X)\}\\
&\qquad-\frac{I(G=E)I(A=0)}{1-\pi^{G=E}(X)}\left\{\frac{1}{p(G=E\mid X)}-1\right\}\{M-\mu^{G=E}_M(A=0,X)\}\\
&\qquad+\frac{I(G=O)I(A=0)\pi^{G=O}(X)}{1-\pi^{G=O}(X)}\{Y-\mu^{G=O}_Y(A=0,X)\}\\
&\qquad+I(G=O)\Big[
\mu^{G=O}_M(A=0,X)
\!-\!\mu^{G=E}_M(A=0,X)\!+\!I(A=1)\{Y\!-\!\mu^{G=O}_Y(A=0,X)\!-\!\psi^{\text{equi}}_{\text{ETT}}\}
\Big]\Big\}.
\end{align*}
}
\end{theorem}
Let $\hat\zeta^{\text{equi}}$ be an estimator for $\zeta^{\text{equi}}$, and define the moment function $\Phi^{\text{equi}}(V; \hat\zeta^{\text{equi}}, \hat\psi)$ to be the expression for $IF_{\psi^{\text{equi}}_{\text{ETT}}}(V)$ with nuisance components replaced by their corresponding estimates from $\hat\zeta^{\text{equi}}$, and $\psi^{\text{equi}}_{\text{ETT}}$ replaced by $\hat\psi$. We use this moment function in the cross-fitting procedure to obtain the estimator $\hat\psi^{\text{equi}}_{\text{ETT}}$. We have the following multiple robustness result for $\hat\psi^{\text{equi}}_{\text{ETT}}$.
\begin{proposition}
\label{prop:DR:equi-conf}
The IF-based estimator $\hat\psi^{\text{equi}}_{\text{ETT}}$ is multiply robust in the sense that it is unbiased if at least one of the following subsets of nuisance functions is correctly specified:
(i) $\{\mu^{G=E}_M,\mu^{G=O}_M,\mu^{G=O}_Y\}$; (ii) $\{\pi^{G=E},\pi^{G=O},p(G=E\mid X=\cdot)\}$; (iii) $\{\mu^{G=E}_M,\pi^{G=O}\}$; (iv) $\{\pi^{G=E},\mu^{G=O}_M,\mu^{G=O}_Y,p(G=E\mid X=\cdot)\}$.
\end{proposition}

\subsection{Estimation under Bespoke IV Data Fusion}

In this subsection, we consider estimation under the BSIV data fusion framework. Recall from Theorem~\ref{thm:BSIV_ett} that, under the assumptions of this framework, for the parameter $\theta_{\text{ETT}}$, we only need to focus on the right hand sides of Equations \eqref{eq:param2} and \eqref{eq:param3}, which are functionals of the observed data distribution. We denote these functionals by $\psi^{\text{bsiv1}}_{\text{ETT}}$ and $\psi^{\text{bsiv2}}_{\text{ETT}}$, respectively.

We begin by deriving the influence functions for the parameters. Let the collection of nuisance functions be defined as:
$\zeta^{bsiv}:=
\Big\{
E_{ab}^O(X):=\mathbb{E}[Y-M\mid A=a, B=b,X,G=O], 
e_{b}^O(X):=\mathbb{E}[Y-M\mid B=b,X,G=O], 
M_{a}^{E}(B,X):=\mathbb{E}[M\mid A=a,B,X,G=E], 
\mu_{b}^O(X):=\mathbb{E}[M\mid B=b,X,G=O], 
\pi^O(X):=p(A=1\mid X,G=O), 
P_{ab}^O(X):=p(A=a\mid B=b,X,G=O), 
P_{ab}^E(X):=p(A=a\mid B=b,X,G=E), 
\rho_{b}^O(X):=p(B=b\mid X,G=O), 
\rho_{b}^E(X):=p(B=b\mid X,G=E), 
\tau(B,X):=p(G=E\mid B,X)
\Big\}$.
\begin{theorem}
\label{thm:IF:BSIV-1_ett}
Under a nonparametric model, the efficient influence functions for the parameters $\psi^{\text{bsiv1}}_{\text{ETT}}$ and $\psi^{\text{bsiv2}}_{\text{ETT}}$ are given by $IF_{\psi^{\text{bsiv1}}_{\text{ETT}}}(V)$ and $IF_{\psi^{\text{bsiv2}}_{\text{ETT}}}(V)$, respectively, 
where the closed-form expressions for the influence functions are deferred to  Supplementary Material \ref{sec:estimation::supp}.
\end{theorem}

Let $\hat\zeta^{\text{bsiv}}$ be an estimator for $\zeta^{\text{bsiv}}$, and define the moment functions $\Phi^{\text{bsiv1}}(V; \hat\zeta^{\text{bsiv}}, \hat\psi)$ and $\Phi^{\text{bsiv2}}(V; \hat\zeta^{\text{bsiv}}, \hat\psi)$ to be the expressions for $IF_{\psi^{\text{bsiv1}}_{\text{ETT}}}(V)$ and $IF_{\psi^{\text{bsiv2}}_{\text{ETT}}}(V)$ with nuisance components replaced by their corresponding estimates from $\hat\zeta^{\text{bsiv}}$, and $\psi^{\text{bsiv1}}_{\text{ETT}}$ and $\psi^{\text{bsiv2}}_{\text{ETT}}$ replaced by $\hat\psi$. We use these moment functions in the cross-fitting procedure to obtain the estimators $\hat\psi^{\text{bsiv1}}_{\text{ETT}}$ and $\hat\psi^{\text{bsiv2}}_{\text{ETT}}$. We have the following multiple robustness result for these estimators.
\begin{proposition}
\label{prop:DR:BSIV_ett}
The IF-based estimator $\hat\psi^{\text{bsiv1}}_{\text{ETT}}$ is multiply robust in the sense that it is unbiased if at least one of the following subsets of nuisance functions is correctly specified:\\
        (i) $\{P_{AB}^O(X),\mathbb{E}[M\mid B,X,G=O],\mathbb{E}[Y\mid B,X,G=O],\mathbb{E}[M\mid A,B,X,G=E]\}$;\\
        (ii) $\{\tau(B,X),\rho_{B}^O(X),\rho_{B}^E(X),P_{AB}^O(X),P_{AB}^E(X)\}$;\\
        (iii) $\{\tau(B,X),P_{AB}^O(X),P_{AB}^E(X),\mathbb{E}[M\mid B,X,G=O],\mathbb{E}[Y\mid B,X,G=O]\}$;\\
        (iv) $\{\rho_{B}^O(X),\rho_{B}^E(X),P_{AB}^O(X),\mathbb{E}[M\mid A,B,X,G=E]\}$.\\
Same result holds for the estimator $\hat\psi^{\text{bsiv2}}_{\text{ETT}}$ 
but with replacing outcome regression functions with
$\mathbb{E}[M\mid A,B,X,G=E]$,
$\mathbb{E}[M\mid A,B,X,G=O]$, and
$\mathbb{E}[Y\mid A,B,X,G=O]\}$.
\end{proposition}

\subsection{Estimation under Proximal Data Fusion}

In this subsection, we consider estimation under the proximal data fusion framework. As in
our other two frameworks, we propose IF-based estimator. This estimator requires the estimation of both bridge functions $h$ and $q$. As seen earlier, these nuisance functions
are solutions to conditional moment equations, and hence they cannot be estimated by a
simple standard regression. Therefore, in Subsection \ref{sec:alt_est_strategy}, we also present simpler (yet not robust) estimation strategies which require the estimation of only one of the bridge functions. We will discuss the estimation of bridge functions in Subsection \ref{sec:est:bridge}.

Recall from Theorem \eqref{eq:POR} that, under the assumptions of this framework, for the parameter $\theta_{\text{ETT}}$, we only need to focus on the right hand side of Equation \eqref{eq:POR}, which is a functional of the observed data distribution. We denote this functional by $\psi_{\text{ETT}}^{\text{proxy}}$. We begin by deriving the influence function for this parameter. 
We assume that the integral equations in Assumptions \ref{assumption:compexist1} $(ii)$ and \ref{assumption:compexist2} $(ii)$ have unique solutions $h(\cdot)$ and $q(\cdot)$, and let the collection of nuisance functions be defined as $\zeta^{\text{proxy}}:=\Big\{
h,q, p(M=\cdot\mid A=\cdot,X=\cdot,G=E), p(A=1\mid X=\cdot,G=E),  p(G=E\mid X=\cdot)\Big\}$.

\begin{theorem}
\label{thm:IF:proximal-1_ett}	
Under a semiparametric model that Assumptions \ref{assumption:compexist1} $(ii)$ and \ref{assumption:compexist2} $(ii)$ have unique solutions, an influence function of the parameter $\psi_{\text{ETT}}^{\text{proxy}}$ is given by
{\small
\begin{align*}
&IF_{\psi_{\text{ETT}}^{\text{proxy}}}(V)
=\frac{1}{p(G=O)p(A=1\mid G=O)}\Big\{I(G=O)\big\{Y-I(A=0)q(Z,0,X)\{Y-h(M,0,X)\}\\
&\hspace{90mm}-\eta(0,X)-I(A=1)\psi_{\text{ETT}}^{\text{proxy}}\big\}\\
&\quad-\frac{I(G=E)I(A=0)}{1-p(A=1\mid X,G=E)}
\{h(M,0,X)-\eta(0,X)\}\{\frac{1}{p(G=E\mid X)}-1\}\Big\},
\end{align*}
}
where
$\eta(a,x)\coloneqq\E[h(M,A,X)\mid A=a,X=x,G=E]$.

\end{theorem}
\subsubsection{Estimation Strategies for ETT}
\label{sec:alt_est_strategy}
Let $\hat\zeta^{\text{proxy}}$ be an estimator for $\zeta^{\text{proxy}}$. Based on identification formulae \eqref{eq:POR} and \eqref{eq:PIPW}, along with Theorem \ref{thm:IF:proximal-1_ett}, we propose the following estimation strategies for the parameter $\theta_{\text{ETT}}$. Strategies 1–3 require estimation of only one of the bridge functions, $h$ or $q$, whereas Strategy 4, the IF-based estimator, requires estimation of both $h$ and $q$.
\begin{itemize}
	
\item{\bf Estimation Strategy 1.} 
From \eqref{eq:POR}, define the moment function as,
{\small
\begin{align*}
\Phi^{\text{proxy1}}(V;\hat{\zeta}^{\text{proxy}},\hat\psi)
\!:=\!
\frac{I(G=O)}{p(G=O)p(A=1\!\mid \!G=O)}\{Y
\!-\!\sum_m\hat{h}(m,0,X)\hat{p}(m\!\mid \!0,X,G=E)\}\!-\!\hat\psi.
\end{align*}
}
\item{\bf Estimation Strategy 2.}
Also from \eqref{eq:POR}, define the moment function as,
{\small
\begin{align*}
\Phi^{\text{proxy2}}(V;\hat{\zeta}^{\text{proxy}},\hat\psi)
:=&
\frac{1}{p(G=O)p(A=1\mid G=O)}\Big\{I(G=O)Y\\
&-\frac{I(G=E)I(A=0)}{1-\hat{p}(A=1\mid X,G=E)}
\hat{h}(M,0,X)\{\frac{1}{\hat{p}(G=E\mid X)}-1\}\Big\}-\hat\psi.
\end{align*}
}
\item{\bf Estimation Strategy 3.}
From \eqref{eq:PIPW}, define the moment function as,
{\small
\begin{align*}
\Phi^{\text{proxy3}}(V;\hat{\zeta}^{\text{proxy}},\hat\psi)
:=&
\frac{I(G=O)}{p(G=O)p(A=1\mid G=O)}\{Y-I(A=0)Y\hat{q}(Z,0,X)\}-\hat\psi.
\end{align*}
}
\begin{sloppypar}
\item{\bf Estimation Strategy 4 (IF-based Strategy).} 
Define the moment function $\Phi^{\text{proxy}}(V; \hat\zeta^{\text{proxy}}, \hat\psi)$ to be the expression for $IF_{\psi^{\text{proxy}}_{\text{ETT}}}(V)$ with nuisance components replaced by their corresponding estimates from $\hat\zeta^{\text{proxy}}$, and $\psi^{\text{proxy}}_{\text{ETT}}$ replaced by $\hat\psi$.
\end{sloppypar}
\end{itemize}
These moment functions are implemented in the cross-fitting procedure to obtain $\hat\psi^{\text{proxy1}}_{\text{ETT}}$, $\hat\psi^{\text{proxy2}}_{\text{ETT}}$, $\hat\psi^{\text{proxy3}}_{\text{ETT}}$, and the IF-based estimator $\hat\psi^{\text{proxy}}_{\text{ETT}}$.

\begin{proposition}
\label{prop:DR:proximal_ett}	
The IF-based estimator $\hat\psi^{\text{proxy}}_{\text{ETT}}$ is multiply robust in the sense that it is unbiased if at least one of the following subsets of nuisance functions is correctly specified:
        (i) $\{h,p(M=\cdot\mid A=0,X=\cdot,G=E))\}$;
        (ii) $\{h,p(A=1\mid X=\cdot,G=E),p(G=E\mid X=\cdot)\}$;
        (iii) $\{q,p(A=1\mid X=\cdot,G=E),p(G=E\mid X=\cdot)\}$.
\end{proposition}

\subsubsection{Estimating the Bridge Functions}
\label{sec:est:bridge}
In all proposed estimation strategies, estimation of at least one of the nuisance functions \(h\) or \(q\) is required to estimate the target parameter. However, these nuisance functions are defined as solutions to conditional moment (integral) equations and therefore cannot be obtained via simple regression. A recent line of work proposes nonparametric adversarial (minimax) estimators for such equations \citep{dikkala2020minimax}; these ideas have been adapted to the original semiparametric proximal causal inference framework in \citep{ghassami2022minimax,kallus2021causal}. We employ the same technique to estimate the bridge functions \(h\) and \(q\) in our proximal data fusion setup.

To proceed, note that the bridge function \(q\) satisfies the following conditional moment equation, which we use to design an estimator for \(q\).

\begin{proposition}
\label{prop:qcondmomeq}
The bridge function \(q\) satisfies
{\small
\begin{equation}
\label{eq:IPWcondmomeqq}
\E\!\left[ \frac{I(G=O)}{p(G=O\mid X)}\,q(Z,A,X)
      - \frac{I(G=E)}{p(A,G=E\mid X)}
      \,\Bigm|\, M,A,X \right] = 0.
\end{equation}
}
\end{proposition}

Let \(\mathcal{H}\), \(\mathcal{Q}\), and \(\mathcal{F}\) be normed function spaces. Based on the conditional moment equations \eqref{eq:ORproxexisth} and \eqref{eq:IPWcondmomeqq}, we propose the following regularized minimax estimators for the bridge functions \(h\) and \(q\):
{\small
\begin{align*}
\hat{h}
&= \arg\min_{h\in\mathcal{H}}\sup_{f\in\mathcal{F}}\;
\mathbb{P}_n\!\left[ \frac{I(G=O)}{p(G=O)}\{h(M,A,X)-Y\}\,f(Z,A,X) - f^2(Z,A,X) \right]
- \lambda^{h}_{\mathcal{F}}\|f\|_{\mathcal{F}}^{2} + \lambda^{h}_{\mathcal{H}}\|h\|_{\mathcal{H}}^{2},\\[4pt]
\hat{q}
&=\! \arg\min_{q\in\mathcal{Q}}\sup_{f\in\mathcal{F}}\;
\mathbb{P}_n\!\left[\! \left\{\!\frac{I(G=O)}{\hat{p}(G=O\!\mid \!X)}\,q(Z,A,X)
\!- \!\frac{I(G=E)}{\hat{p}(A,G=E\!\mid \!X)}\!\right\} f(M,A,X) \!-\! f^2(M,A,X) \right]
\!-\! \lambda^{q}_{\mathcal{F}}\|f\|_{\mathcal{F}}^{2} \!+\! \lambda^{q}_{\mathcal{Q}}\|q\|_{\mathcal{Q}}^{2}.
\end{align*}
}
We refer the reader to \citep{dikkala2020minimax,ghassami2022minimax} for convergence analyses of these minimax estimators.

\section{Simulation Studies}
\label{sec:simulation}

\begin{table}[t]
\centering
\begin{tabular}{c | c c c c c} 
 \hline
 Sample Size & Methods & CI (Coverage) & Bias & RMSE & SD   \\ [0.5ex] 
 \hline
n=1000 & BSIV & [-3.30, 4.34] (0.99)  & 0.16 & 3.03  & 3.02  \\
 & proximal & [0.64, 0.99] (0.71) & -0.13 & 0.16 & 0.09 \\
\hline
n=2000 & BSIV & [-0.98, 2.36] (0.98)  & -0.01 & 0.97  & 0.97 \\
 & proximal & [0.62, 0.88] (0.76) & -0.07 & 0.10 & 0.07 \\
 \hline
 n=4000 & BSIV & [-0.25, 1.62] (0.967)  & 0.00  & 0.49  & 0.49 \\
& proximal & [0.62, 0.81] (0.87) & -0.03 & 0.06 & 0.05 \\
 \hline
\end{tabular}
\caption{Inference results for influence-function-based estimators.}
\label{inference_table}
\end{table}

\begin{table}[t]\footnotesize
\centering
\begin{tabular}{c | c c c c c c c} 
 \hline
Sample size & estimators & all true & case 1 & case 2 & case 3 & case 4  & all false\\ [0.5ex] 
\hline
n=1000 & plug-in & -0.07 (0.89)  & -0.07 (0.89) & 0.56 (0.84) & 0.32 (0.96) & 0.16 (0.63) & 0.49 (0.78) \\
 & IF-based & 0.16 (3.03) & 0.17 (3.07) & -0.08 (2.13) & 0.17 (3.07) & -0.07 (2.12) & 0.61 (1.85) \\
\hline
n=2000 & plug-in & 0.02 (0.61) & 0.02 (0.61) & 0.61 (0.77) & 0.38 (0.74) & 0.25 (0.48) & 0.55 (0.71) \\
 & IF-based & -0.01 (0.97) & 0.03 (0.96) & -0.08 (0.77) & 0.03 (0.96) & -0.07 (0.76) & -0.44 (1.05)\\
 \hline
 n=4000 & plug-in & -0.02 (0.41) & -0.02 (0.41) & 0.58 (0.66) & 0.35 (0.55) & 0.22 (0.35) & 0.52 (0.60) \\
 & IF-based & -0.00 (0.49) & -0.00 (0.66) & -0.04 (0.43) & -0.01 (0.66) & -0.04 (0.43) & -0.42 (0.76) \\
\hline
\end{tabular}
\caption{Robustness results for BSIV estimators: Bias (RMSE).}
\label{BSIV_robustness_table}
\end{table}

\begin{table}[t]
\centering
\begin{tabular}{c | c c c c c c} 
 \hline
 sample size & estimators & all true & case 1 & case 2 & case 3 & all false\\ [0.5ex] 
 \hline
n=1000 & estimator 1 & -0.16 (0.18)  & -0.16 (0.18) & 0.21 (0.29) & -1.42 (1.42) & -1.42 (1.42) \\
 & estimator 2 & -0.11 (0.14) & -0.16 (0.24) & -0.11 (0.14) & -1.42 (1.42) & -1.42 (1.42) \\
 & estimator 3 & -0.63 (0.64) & -1.27 (1.28) & -1.27 (1.28) & -0.73 (0.73) & -1.27 (1.28)\\
 & IF-based & -0.13 (0.16) & -0.17 (0.18) & -0.18 (0.20) & -0.70 (0.70) & -1.21 (1.21)\\
\hline
n=2000 & estimator 1 & -0.09 (0.10) & -0.09 (0.10) & 0.42 (0.45) & -1.31 (1.31) & -1.31 (1.31) \\
 & estimator 2 & -0.07 (0.09) & -0.13 (0.17) & -0.07 (0.09) & -1.32 (1.32) & -1.31 (1.32) \\
 & estimator 3 & -0.46 (0.46) & -1.12 (1.12) & -1.12 (1.12) & -0.56 (0.56) & -1.12 (1.12)\\
 & IF-based & -0.07 (0.10) & -0.10 (0.12) & -0.11 (0.13) & -0.52 (0.52) & -1.00 (1.00)\\
 \hline
 n=4000 & estimator 1 & -0.05 (0.07) & -0.05 (0.07) & 0.54 (0.56) & -1.13 (1.13) & -1.13 (1.13) \\
 & estimator 2 & -0.04 (0.07) & -0.12 (0.15) & -0.04 (0.07) & -1.16 (1.16) & -1.16 (1.16) \\
 & estimator 3 & -0.32 (0.33) & -0.93 (0.94) & -0.93 (0.94) & -0.44 (0.44) & -0.93 (0.94)\\
 & IF-based & -0.03 (0.06) & -0.06 (0.07) & -0.06 (0.08) & -0.37 (0.37) & -0.77 (0.78)\\
\hline
\end{tabular}
\caption{Robustness results for proximal estimators: Bias (RMSE).}
\label{proximal_robustness_table}
\end{table}

We evaluated our BSIV and proximal data fusion frameworks on synthetic data for estimating the ETT. We omit the equi-confounding framework because, in the presence of a BSIV, it is nested as a special case of the BSIV framework. Binary variables were generated from logistic (expit) models, and continuous variables from normal models. To enable a comparison between the proximal and BSIV frameworks, data-generating parameters were chosen so that the assumptions of \emph{both} frameworks hold; see Supplementary Material~\ref{sec:simulation::supp} for full details.

We considered sample sizes \(n \in \{1000, 2000, 4000\}\) and conducted 300 Monte Carlo replications for each \(n\). For cross-fitting, we used four folds. Hyper-parameters for nuisance estimators were tuned via cross-validation with a 75/25 train–validation split. For the BSIV approach, we report both plug-in and IF-based estimators; for the proximal approach, we consider the four estimation strategies described in Section~\ref{sec:alt_est_strategy}.

Table~\ref{inference_table} compares the IF-based estimators for the BSIV and proximal approaches. The BSIV estimator exhibits slightly lower bias but notably higher variance. Coverage for both methods approaches the nominal \(95\%\) as \(n\) increases.

We also conducted a robustness study; results for the BSIV and proximal approaches appear in Tables~\ref{BSIV_robustness_table} and~\ref{proximal_robustness_table}, respectively. Each scenario corresponds to a setting in which specific subsets of nuisance functions (matching the cases in Propositions~\ref{prop:DR:BSIV_ett} and~\ref{prop:DR:proximal_ett}) are correctly specified; see Supplementary Material~\ref{sec:simulation::supp} for details. As expected, the two IF-based estimators display multiple robustness: whenever the corresponding subset of nuisance functions is correctly specified, the estimators have small bias that diminishes with increasing \(n\).

\section{Application: Effect of Class Size on SAT Scores}
\label{sec:application}

\begin{table}[t]
\centering
\begin{tabular}{c | c c} 
 \hline
  Method & Estimate & 95\% CI \\ [0.5ex] 
 \hline
Proximal estimator 1 & 87.85 & [28.94, 153.68] \\
Proximal estimator 2 & 69.89 & [5.40, 139.67]\\
 \hline
Estimator of \cite{athey2020combining} & -25.11 & [-76.01, 20.64] \\
Naive IF-based estimator & 0.36 & [-2.23, 2.97]\\
Equi-confounding IF-based estimator & -25.28 & [-30.13, -19.94]\\
\hline
\end{tabular}
\caption{Real data result for ETT.}
\label{realdata_table}
\end{table}

We applied our proposed methods to evaluate the effect of class size on long-term educational outcomes. The observational (target) data come from the Early Childhood Longitudinal Study, Kindergarten Class of 1998–99 (ECLS-K) \citep{tourangeau2009eclsk}, which tracks a nationally representative cohort from kindergarten through eighth grade. As the experimental domain, we use the Project STAR experiment \citep{DVN/SIWH9F_2008}—a large-scale randomized study conducted in Tennessee. Both datasets report class size; following the Project STAR design, we define \emph{small} classes (treatment) as those with at most 19 students and \emph{regular} classes (control) as those with more than 19 students. We use third-grade Mathematics SAT scores as short-term outcomes and eighth-grade Mathematics SAT scores as long-term outcomes. We used observed covariates gender, ethnicity, access to free lunch, and school location (rural, suburban, urban). Family socioeconomic status (SES) is a plausible unobserved confounder; school location is considered a potential proxy for SES, as higher-SES families are more likely to live in suburban or urban areas. To ensure support overlap across domains, we restrict the short-term score to the 570–630 range and discretize it into four categories. Restricting to complete cases yields 2,172 observations from ECLS-K and 2,584 from Project STAR.

\begin{sloppypar}
We apply only proximal estimation strategies 1 and 2. 
We do not employ our BSIV approach because we lack a covariate that plausibly satisfies the BSIV conditions; for example, free-lunch status and school location are correlated with factors that drive changes in achievement over time, making equal partial associations with short- and long-term outcomes unlikely. 
We also do not use proximal methods relying on the bridge function \(q(\cdot)\) because Assumption~\ref{assumption:compexist2}(i) requires the third-grade SAT score to be sufficiently informative about SES, which may not hold. Since all covariates and the proxy variable are discrete, the outcome bridge \(h\) admits a closed-form solution \citep{miao2018identifying}: 
$h(m,a,x)\;=\;\E\!\left[Y \,\middle|\, B, A=a, X=x, G=O\right]\; P\!\left(M \,\middle|\, B, A=a, X=x, G=O\right)^{\dagger}$,
where \({\dagger}\) denotes the matrix pseudoinverse; here \(h(m,a,x)\) and \(\E[Y \mid B, A=a, X=x, G=O]\) are row vectors of sizes \(|\mathcal{M}|\) and \(|\mathcal{B}|\), respectively, and \(P(M \mid B, A=a, X=x, G=O)\) is a \(|\mathcal{M}|\times|\mathcal{B}|\) probability matrix.
\end{sloppypar}

We compare our estimators with those of \cite{athey2020combining}, the equi-confounding estimator, and a naive IF-based estimator that assumes no unobserved confounding, and is based on the influence function of the parameter, $\psi_{\text{ETT, naive}}=\mathbb{E}[Y\mid A=1,G=O]-\mathbb{E}[\mathbb{E}[Y\mid X,A=0,G=O]\mid A=1,G=O]$ (which is equal to the true ETT if there were no latent confounders present in the system). Results are summarized in Table~\ref{realdata_table}. Our estimators indicate a positive effect of smaller class sizes on SAT scores, consistent with \cite{krueger1999experimental}, which found that smaller classes improve student test performance relative to regular-sized classes. By contrast, the estimator of \cite{athey2020combining} and the naive IF-based estimator yield a negative effect or an effect close to zero. Yet these are unlikely to be plausible due to the presence of unobserved confounder in the setting which is likely to be directly affecting the long-term outcome, and hence violating the assumptions of these two frameworks. Moreover, the equi-confounding IF-based estimator also yields a negative effect which is not plausible as the association with the potential outcome is likely to change over time.

\vspace{-3mm}
\section{Conclusion and Discussion}
\label{sec:conc}

In many real-world settings, available observational data are confounded by latent variables and therefore cannot be used to identify the causal effect of a treatment on an outcome of interest. At the same time, experimental data may be available, but due to practical constraints, the observed outcome may only reflect a short-term version of the long-term outcome of interest. Individually, neither the observational nor the experimental dataset suffices to identify the causal parameter. This raises the central question: can we combine information from the two sources to achieve identification?
In this work, we proposed three data fusion frameworks under which the long-term causal effect can be identified:
(1) Equi-confounding method, which assumes equal confounding bias for the short-term and long-term outcomes;
(2) BSIV method, which leverages an observed confounder for which the short-term and long-term potential outcomes share the same partial additive association;
(3) Proximal method, which relies on a proxy variable of the latent confounder in the treatment-outcome relationship, extending the proximal causal inference framework to the data fusion setting.
For each approach, we developed influence function-based estimation strategies and analyzed the robustness properties of the resulting estimators. 

A natural question is how a practitioner should choose among our three proposed data-fusion strategies. The main requirements of both equi-confounding and BSIV approaches for connecting \(M\) and \(Y\) are of the \emph{equal additive-association} nature. Equi-confounding approach requires this in terms of requiring equal association of the treatment variables with the short-term and long-term potential outcomes. On the other hand, the BSIV approach allows the treatment variable to be substituted with some confounder $B$ that the practitioner believes is more likely to satisfy the equal additive-association restriction. Moreover, as mentioned earlier, the partial homogeneity assumption of the BSIV approach is weaker than the assumption of the equi-confounding approach. Therefore, when a credible BSIV exists, practitioners should prefer BSIV over equi-confounding. The proximal data fusion approach takes on a more generalized method for connecting \(M\) and \(Y\) by modeling their potentially complex and non-linear relation beyond additive association via bridge functions that solve conditional moment (integral) equations. This yields substantial model flexibility and generality compared to the BSIV approach. The trade-offs are practical: identifying a valid proxy \(Z\) may be harder than finding a credible BSIV, and because bridge functions are not structural outcome models, parametric forms are difficult to justify—hence they are best estimated nonparametrically (as in Section~\ref{sec:est:bridge}). In practice, the main challenge is constructing and reliably estimating the bridge function. A table summarizing the comparison of the three proposed methods is provided in Supplementary Material~\ref{sec:summary::supp}.

\bibliographystyle{apalike}
\bibliography{Bibliography-MM-MC.bib}

\newpage
\begin{center}
{\Large\bf SUPPLEMENTARY MATERIALS for\\ ``Combining Experimental and Observational Data for Identification and Estimation of Long-Term Causal Effects''}
\end{center}

\vspace{10mm}

\allowdisplaybreaks

\appendix
\startcontents[appendix] 
\printcontents[appendix]{}{1}{\section*{Supplementary Materials Contents}}

\newpage
\section{Supplementary Materials for Section \ref{sec:athey}}
\label{sec:athey::supp}

\subsection{Identification Formulae}

\begin{theorem}
\label{thm:athey::supp}
Under Assumptions \ref{assumption:IntVal}--\ref{assumption:LaUn}, the parameters $\theta_{\text{ATE}}$ and $\theta_{\text{ETT}}$ are identified via the following expressions:
\begin{align*}
\theta_{\text{ATE}} &= \E\left[ \E\left[ \E[Y \mid M, A=1, X, G=O] \mid X, G=E \right] \mid G=O \right] \\
&\quad - \E\left[ \E\left[ \E[Y \mid M, A=0, X, G=O] \mid X, G=E \right] \mid G=O \right], \\
\theta_{\text{ETT}} &= \E[Y \mid A=1, G=O] \\
&\quad - \frac{\E\left[ \E\left[ \E[Y \mid M, A=0, X, G=O] \mid X, G=E \right] \mid G=O \right]}{p(A=1 \mid G=O)} \\
&\quad + \frac{p(A=0 \mid G=O)}{p(A=1 \mid G=O)} \E[Y \mid A=0, G=O].
\end{align*} 
\end{theorem}

\section{Supplementary Materials for Section \ref{sec:AltIDAss}}
\label{sec:AltIDAss::supp}

\subsection{Identification of ATE}
\label{sec:EqConf:ATE::supp}

\subsubsection{Equi-Confounding Data Fusion}

\begin{assumption}[Additive Equi-Confounding Bias]
\label{assumption:Equi::supp}
~
\begin{enumerate}[label=(\roman*)]
\item \makebox[\linewidth]{\(\begin{aligned}[t]
\E[M^{(1)}\mid  A=0, &G=O]-\E[M^{(1)}\mid  A=1, G=O]\\
&=\E[Y^{(1)}\mid  A=0, G=O]-\E[Y^{(1)}\mid  A=1, G=O],
\end{aligned}\)}
\item\makebox[\linewidth]{\(\begin{aligned}[t]
\E[M^{(0)}\mid A=0,&G=O]-\E[M^{(0)}\mid  A=1,G=O]\\
&=\E[Y^{(0)}\mid  A=0,G=O]-\E[Y^{(0)}\mid  A=1,G=O].
\end{aligned}	\)}
\end{enumerate}
\end{assumption}
See Figure \ref{fig:ours} for a schematic representation of Assumption \ref{assumption:Equi::supp}.

\begin{theorem}
\label{thm:equiATE}
Under Assumptions \ref{assumption:IntVal}, \ref{assumption:ExVal}, and \ref{assumption:Equi::supp}, for $a\in\{0,1\}$, the parameter $\E[Y^{(a)}\mid G=O]$ and hence the parameter $\theta_{\text{ATE}}$ are identified. $\theta_{\text{ATE}}$ is identified using the following formula
\begin{align*}
\theta_{\text{ATE}} 
&=
\E[Y\mid A=1,G=O]-\E[Y\mid A=0,G=O]\\
&\quad+\E[\E[M\mid X,A=1,G=E]\mid G=O]-\E[\E[M\mid X,A=0,G=E]\mid G=O]\\
&\quad-\E[M\mid A=1,G=O]+\E[M\mid A=0,G=O].
\end{align*}
\end{theorem}

\begin{figure}[t]
\centering
		\tikzstyle{block} = [draw, circle, inner sep=2.5pt, fill=lightgray]
		\tikzstyle{input} = [coordinate]
		\tikzstyle{output} = [coordinate]
        \begin{tikzpicture}
            \tikzset{edge/.style = {->,> = latex'}}
            \node[] (a1) at  (0,4.17) {$A\ref{assumption:Equi::supp}(ii)$};
            \node[] (a1) at  (-.1,2.77) {$A\ref{assumption:Equi::supp}(i)$};                        
            \node[] (axl) at  (-5.5,0) {};
            \node[] (axr) at  (5.5,0) {};
            \node[] (l1) at  (-1,1.5) {$\circ$};
            \node[] (l1t1) at  (-3.5,1.5) {$\E[M^{(1)}\mid A=0,G=O]$};
            \node[] (l1t2) at  (-3.5,2.1) {$\E[M^{(0)}\mid A=0,G=O]$};                                              
            \node[] (l2) at  (-1,3.5) {$\circ$};
            \node[] (l2t1) at  (-3.5,3.5) {$\E[M^{(0)}\mid A=1,G=O]$};
            \node[] (l2t2) at  (-3.5,4.5) {$\E[M^{(1)}\mid A=1,G=O]$};  
            \node[] (l3) at  (-1,2) {$\circ$};
            \node[] (l4) at  (-1,4.5) {$\circ$};                                              
            \node[] (l0) at  (-1,0) {$ $};            
            \node[] (r1) at  (1,2.5) {$\circ$};
            \node[] (r1t) at  (3.5,2.5) {$\E[Y^{(0)}\mid A=0,G=O]$};                        
            \node[] (r2) at  (1,3) {$\circ$};
            \node[] (r2t) at  (3.5,3.1) {$\E[Y^{(1)}\mid A=0,G=O]$};            
            \node[] (r3) at  (1,4) {$\circ$};
            \node[] (r3t) at  (3.5,4) {$\E[Y^{(0)}\mid A=1,G=O]$};            
            \node[] (r4) at  (1,6) {$\circ$};
            \node[] (r4t) at  (3.5,6) {$\E[Y^{(1)}\mid A=1,G=O]$};            
            \node[] (r0) at  (1,0) {$ $};
            \draw[-] (axl) to (axr);
            \draw[-,thick] (l3) to (r1);
            \draw[dashed,thick] (l1) to (r2);
            \draw[dashed,thick] (l2) to (r3);                                    
            \draw[-,thick] (l4) to (r4);            
            \draw[dotted] (r4) to (r3);
            \draw[dotted] (r3) to (r2);
            \draw[dotted] (r2) to (r1);
            \draw[dotted] (r1) to (r0);
            \draw[dotted] (l3) to (l1);
            \draw[dotted] (l1) to (l0);
			\draw[dotted] (l4) to (l2);
            \draw[dotted] (l2) to (l3);            
        \end{tikzpicture}
        \caption{Schematic representation of Assumption \ref{assumption:Equi::supp}. The dashed lines, representing the difference between two unobserved parameters, are the assumptions.}
        \label{fig:ours}
\end{figure}

\subsubsection{Conditional Equi-Confounding Data Fusion}

\begin{assumption}[Conditional Additive Equi-Confounding Bias]
\label{assumption:CondEqui::supp}
~
\begin{enumerate}[label=(\roman*)]
\item With probability one, we have
\begin{align*}
\E[M^{(1)}\mid  X,A=0,&G=O]-\E[M^{(1)}\mid  X,A=1,G=O]\\
&=\E[Y^{(1)}\mid  X,A=0,G=O]-\E[Y^{(1)}\mid  X,A=1,G=O],
\end{align*}
\item With probability one, we have
\begin{align*}
\E[M^{(0)}\mid X,A=0,&G=O]-\E[M^{(0)}\mid  X,A=1,G=O]\\
&=\E[Y^{(0)}\mid  X,A=0,G=O]-\E[Y^{(0)}\mid  X,A=1,G=O].
\end{align*}	
\end{enumerate}
\end{assumption}

\begin{theorem}
\label{thm:equiATECond}
Under Assumptions \ref{assumption:IntVal}, \ref{assumption:ExVal}, and \ref{assumption:CondEqui::supp}, for $a\in\{0,1\}$, the parameter $\E[Y^{(a)}\mid G=O]$  and hence the parameter $\theta_{\text{ATE}}$ are identified. $\theta_{\text{ATE}}$ is identified using the following formula
\begin{align*}
\theta_{\text{ATE}}
&=\E\big[
\E[Y\mid X,A=1,G=O]-\E[Y\mid X,A=0,G=O]\\
&\quad\quad+\E[M\mid X,A=1,G=E]-\E[M\mid X,A=0,G=E]\\
&\quad\quad+\E[M\mid X,A=0,G=O]-\E[M\mid X,A=1,G=O]
\big| G=O\big].
\end{align*}
\end{theorem}

\subsection{Unconditional Equi-Confounding Data Fusion}
\label{sec:UNCondEqui::supp}

In this section, we provide an unconditional version of Assumption \ref{assumption:CondEqui} as the basis for identification.

\begin{assumption}[Additive Equi-Confounding Bias]
\label{assumption:Equi}
{\small
\[
\E[M^{(0)} \mid A=0, G=O] \!-\! \E[M^{(0)} \mid A=1, G=O]
\!=\! \E[Y^{(0)} \mid A=0, G=O] \!-\! \E[Y^{(0)} \mid A=1, G=O].
\]
}
\end{assumption}
\begin{example}
Assumption~\ref{assumption:Equi} holds if the data are generated from the following model:
\[
M = \tau A + f_M(X, U) + \epsilon_M, 
\qquad \qquad Y = \theta A + f_Y(X, M, U) + \epsilon_Y,
\]
where $\epsilon_M$ and $\epsilon_Y$ are independent noise terms, and the function 
$f(X, M, U) := f_Y(X, M, U) - f_M(X, U)$ satisfies
$\E[f(X, M, U) \mid A=1] = \E[f(X, M, U) \mid A=0]$.
Here, $f_M$ and $f_Y$ may be stochastic functions.
\end{example}

To identify the parameter $\theta_{\text{ETT}}$, we begin by noting that under Assumption~\ref{assumption:Equi},
\begin{align*}
\theta_{\text{ETT}}
&= \E[Y^{(1)} \mid A=1, G=O] - \E[Y^{(0)} \mid A=0, G=O] \\
&\quad + \E[Y^{(0)} \mid A=0, G=O] - \E[Y^{(0)} \mid A=1, G=O] \\
&= \E[Y \mid A=1, G=O] - \E[Y \mid A=0, G=O] \\
&\quad + \E[M \mid A=0, G=O] - \E[M^{(0)} \mid A=1, G=O].
\end{align*}
Thus, identification of $\theta_{\text{ETT}}$ reduces to identification of $\E[M^{(0)} \mid A=1, G=O]$.
One might attempt to argue that
\begin{align*}
\E[M^{(0)} \mid A=1, G=O]
&= \E[ \E[M^{(0)} \mid A=1, X, G=O] \mid A=1, G=O ] \\
&\overset{(*)}{=} \E[ \E[M^{(0)} \mid A=1, X, G=E] \mid A=1, G=O ] \\
&\overset{\text{A\ref{assumption:IntVal}}}{=} \E[ \E[M^{(0)} \mid A=0, X, G=E] \mid A=1, G=O ] \\
&= \E[ \E[M \mid A=0, X, G=E] \mid A=1, G=O ],
\end{align*}
and thereby conclude identifiability. However, step $(*)$ is not generally valid. Although Assumption~\ref{assumption:ExVal} states that for all $a \in \{0,1\}$, $G \independent \{Y^{(a)}, M^{(a)}\} \mid X$, this does not imply that $G \independent \{Y^{(a)}, M^{(a)}\} \mid X, A=1$. The reason is that $A$ is a collider on the path between $G$ and the potential outcomes, and conditioning on it may induce spurious associations.

We now present our identification result for $\theta_{\text{ETT}}$ under equi-confounding assumption.

\begin{theorem}
\label{thm:equiETT}
Under Assumptions~\ref{assumption:IntVal}, \ref{assumption:ExVal}, and~\ref{assumption:Equi}, the parameter $\theta_{\text{ETT}}$ is identified by:
\begin{align*}
\theta_{\text{ETT}}
&= \E[Y \mid A=1, G=O] - \E[Y \mid A=0, G=O] + \E[M \mid A=0, G=O] \\
&\quad - \frac{ \E[ \E[M \mid X, A=0, G=E] \mid G=O ] - \E[M \mid A=0, G=O] \cdot p(A=0 \mid G=O) }{ p(A=1 \mid G=O) }.
\end{align*}
\end{theorem}

The corresponding identification result for the parameter $\theta_{\text{ATE}}$ is provided in Supplementary Material~\ref{sec:EqConf:ATE::supp}.

\subsection{Quantile-Quantile Equi-Confounding Data Fusion}
\label{sec:QQCondEqui::supp}

We note that the (conditional) additive equi-confounding bias assumption may be restrictive, as it requires the short-term and long-term outcomes to be measured on the same scale. While this is not a limitation in our specific application---where \( M \) and \( Y \) represent short- and long-term versions of the same outcome---it may be problematic in other settings. To address this, we propose a generalization of the additive equi-confounding framework, inspired by the changes-in-changes approach in panel data analysis \citep{athey2006identification} and its analogue in the negative control inference literature \citep{sofer2016negative}. This generalization also enables identification of causal parameters beyond mean-based quantities such as ATE and ETT. We only provide the result for the treated sub-population; it can be extended to marginal parameters similar to our approach in the previous subsections.

To proceed, we first introduce the quantile–quantile  association, as a measure of association between two variables, which we will use to encode confounding bias. 

\begin{definition}
The quantile–quantile association between any variable $W$ and the binary treatment variable $A$ conditional on $X$ is defined as 
\[
q_W(v\mid x)\coloneqq F_{W\mid A=0,X=x}\circ F^{-1}_{W\mid A=1,X=x}(v),\quad\quad v\in[0,1],
\]
where for random variables $X_1$ and $X_2$, we denote the cumulative distribution function of $X_1$ conditioned on $X_2$ by $F_{X_1\mid X_2}$, and the operator $\circ$ denotes function composition.
\end{definition}

\begin{assumption}[Quantile-Quantile Equi-Confounding Bias]
\label{assumption:QQEqui}
For all $v\in[0,1]$, with probability one, we have
\begin{align*}
q_{M^{(0)}}(v\mid X,G=O)=q_{Y^{(0)}}(v\mid X,G=O).
\end{align*}	
\end{assumption}

\begin{example}
Assumption \ref{assumption:QQEqui} is satisfied if the data is generated from the following structural equations.
\begin{align*}
&M=g_M(A,X,U),\\
&Y= g_Y(A,X,U),
\end{align*}
where $g_M$ and $g_Y$ are monotonically increasing functions of $U$ for any $A,X$.

To see this, we note that for $a\in\{0,1\}$,
\begin{align*}
F_{Y^{(a)}\mid A=1,X}(y)
&=Pr\big( Y^{(a)}\le y \big| A=1,X\big)\\
&=Pr\big( g_Y(a,X,U)\le y \big| A=1,X\big)\\
&=Pr\big( U\le g^{-1}_Y(a,X,y) \big| A=1,X\big)\\
&=F_{U\mid A=1,X}(g^{-1}_Y(a,X,y)).
\end{align*}
Also,
\[
F_{Y^{(a)}\mid A=1,X}(y)=v\Rightarrow F^{-1}_{Y^{(a)}\mid A=1,X}=y,
\]
and
\[
F_{U\mid A=1,X}(g^{-1}_Y(a,X,y))=v\Rightarrow g_Y(a,X,F^{-1}_{U\mid A=1,X}(v))=y.
\]
Therefore, we have
\begin{align*}
F_{Y^{(a)}\mid A=0,X}\circ F^{-1}_{Y^{(a)}\mid A=1,X}(v)
&=Pr\big( Y^{(a)}\le F^{-1}_{Y^{(a)}\mid A=1,X}(v) \big| A=0,X\big)\\
&=Pr\big( g_Y(a,X,U)\le F^{-1}_{Y^{(a)}\mid A=1,X}(v) \big| A=0,X\big)\\
&=Pr\big( g_Y(a,X,U)\le g_Y(a,X,F^{-1}_{U\mid A=1,X}(v)) \big| A=0,X\big)\\
&=Pr\big( U\le F^{-1}_{U\mid A=1,X}(v) \big| A=0,X\big)\\
&=F_{U\mid A=0,X}\circ F^{-1}_{U\mid A=1,X}(v).
\end{align*}
Similarly,
\begin{align*}
F_{M^{(a)}\mid A=0,X}\circ F^{-1}_{M^{(a)}\mid A=1,X}(v)
&=F_{U\mid A=0,X}\circ F^{-1}_{U\mid A=1,X}(v).
\end{align*}
Therefore,
\begin{align*}
q_{M^{(a)}}(v\mid X,G=O)
&=F_{M^{(a)}\mid A=0,X}\circ F^{-1}_{M^{(a)}\mid A=1,X}(v)\\
&=F_{Y^{(a)}\mid A=0,X}\circ F^{-1}_{Y^{(a)}\mid A=1,X}(v)\\
&=q_{Y^{(a)}}(v\mid X,G=O).
\end{align*}

\end{example}

\begin{theorem}
\label{thm:equiETTQQ}
Under Assumptions \ref{assumption:IntVal}, \ref{assumption:ExVal}, and \ref{assumption:QQEqui}, the conditional distribution $F_{Y^{(0)}\mid X,A=1,G=O}$ is identified using the following formula
\begin{align*}
F_{Y^{(0)}\mid A=1,X,G=O}(y)
&=\frac{F_{M\mid X,A=0,G=E}\circ F^{-1}_{M\mid A=0,X,G=O} \circ F_{Y\mid A=0,X,G=O}(y)}{p(A=1\mid X,G=O)}\\
&\quad-\frac{p(A=0\mid X,G=O)}{p(A=1\mid X,G=O)} F_{Y\mid A=0,X,G=O}(y).
\end{align*}
\end{theorem}

\section{Supplementary Materials for Section \ref{sec:bsiv}}
\label{sec:bsiv::supp}

\subsection{Identification of ATE}

\begin{assumption}[BSIV Partial Additive Equi-Association]
\label{assumption:BSEqui::supp}
~
\begin{enumerate}[label=(\roman*)]
\item With probability one, we have
\begin{align*}
\E[M^{(1)}\mid  X,B=1,&G=O]-\E[M^{(1)}\mid  X,B=0,G=O]\\
&=\E[Y^{(1)}\mid  X,B=1,G=O]-\E[Y^{(1)}\mid  X,B=0,G=O],
\end{align*}
\item With probability one, we have
\begin{align*}
\E[M^{(0)}\mid X,B=1,&G=O]-\E[M^{(0)}\mid  X,B=0,G=O]\\
&=\E[Y^{(0)}\mid  X,B=1,G=O]-\E[Y^{(0)}\mid  X,B=0,G=O].
\end{align*}	
\end{enumerate}
\end{assumption}

\begin{assumption}[Partial Homogeneity of Causal Effect Contrast]
\label{assumption:homogeneityETT::supp}
~
\begin{enumerate}[label=(\roman*)]
\item With probability one, we have
\begin{align*}
&\E[Y^{(1)}-Y^{(0)}\mid A=0,B=1,X,G=O]-\E[M^{(1)}-M^{(0)}\mid A=0,B=1,X,G=O]\\
&=\E[Y^{(1)}-Y^{(0)}\mid A=0,B=0,X,G=O]-\E[M^{(1)}-M^{(0)}\mid A=0,B=0,X,G=O].
\end{align*}
\item With probability one, we have
\begin{align*}
&\E[Y^{(1)}-Y^{(0)}\mid A=1,B=1,X,G=O]-\E[M^{(1)}-M^{(0)}\mid A=1,B=1,X,G=O]\\
&=\E[Y^{(1)}-Y^{(0)}\mid A=1,B=0,X,G=O]-\E[M^{(1)}-M^{(0)}\mid A=1,B=0,X,G=O].
\end{align*}
\end{enumerate}
\end{assumption}

\begin{assumption}[Partial Homogeneity of Bias Contrast]
\label{assumption:homogeneityBias::supp}
~
\begin{enumerate}[label=(\roman*)]
\item With probability one, we have
\begin{equation*}
\begin{aligned}
&\big\{\E[Y^{(1)}\mid A=1,B=1,X,G=O]-\E[Y^{(1)}\mid A=0,B=1,X,G=O]\big\}\\
&\quad-\big\{\E[M^{(1)}\mid A=1,B=1,X,G=O]-\E[M^{(1)}\mid A=0,B=1,X,G=O]\big\}\\
&=\big\{\E[Y^{(1)}\mid A=1,B=0,X,G=O]-\E[Y^{(1)}\mid A=0,B=0,X,G=O]\big\}\\
&\quad-\big\{\E[M^{(1)}\mid A=1,B=0,X,G=O]-\E[M^{(1)}\mid A=0,B=0,X,G=O]\big\}.
\end{aligned}
\end{equation*}
\item With probability one, we have
\begin{equation*}
\begin{aligned}
&\big\{\E[Y^{(0)}\mid A=1,B=1,X,G=O]-\E[Y^{(0)}\mid A=0,B=1,X,G=O]\big\}\\
&\quad-\big\{\E[M^{(0)}\mid A=1,B=1,X,G=O]-\E[M^{(0)}\mid A=0,B=1,X,G=O]\big\}\\
&=\big\{\E[Y^{(0)}\mid A=1,B=0,X,G=O]-\E[Y^{(0)}\mid A=0,B=0,X,G=O]\big\}\\
&\quad-\big\{\E[M^{(0)}\mid A=1,B=0,X,G=O]-\E[M^{(0)}\mid A=0,B=0,X,G=O]\big\}.
\end{aligned}
\end{equation*}
\end{enumerate}
\end{assumption}

\begin{theorem}
\label{thm:BSIV_ATE}
Define $\pi(B,X) \coloneqq p(A=1 \mid B,X,G=O)$, and for $a,b\in\{0,1\}$ define $E_{ab}^O(X)\coloneqq\E[Y-M\mid A=a,B=b,X,G=O]$ and $P_{ab}^O(X)\coloneqq p(A=a\mid B=b,X,G=O)$.	
\begin{itemize}
\item[(a)]
Under Assumptions \ref{assumption:IntVal}, \ref{assumption:ExVal},\ref{assumption:BSIVrelevance}, \ref{assumption:BSEqui::supp}, and \ref{assumption:homogeneityETT::supp}, the parameter $\theta_{ATE}$ is identified by
\begin{align*}
\theta_{ATE}
&=\E\Big[
\frac{\E[Y-M\mid B=1,X,G=O]-\E[Y-M\mid B=0,X,G=O]}{P_{11}^O(X)-P_{10}^O(X)}\\
&\qquad+\E[M\mid A=1,B,X,G=E]-\E[M\mid A=0,B,X,G=E]\Big|G=O\Big].
\end{align*}
\item[(b)]
Under Assumptions \ref{assumption:IntVal}, \ref{assumption:ExVal},\ref{assumption:BSIVrelevance}, \ref{assumption:BSEqui::supp}, and  \ref{assumption:homogeneityBias::supp}, the parameter $\theta_{ATE}$ is identified by
\begin{align*}
\theta_{ATE}
&=\E\Big[
\{E_{11}^O(X)-E_{01}^O(X)-E_{10}^O(X)+E_{00}^O(X)\}B+E_{10}^O(X)-E_{00}^O(X)\\
&\qquad-\frac{\{E_{01}^O(X)-E_{00}^O(X)\}\pi(B,X)+\{E_{11}^O(X)-E_{10}^O(X)\}(1-\pi(B,X))}{P_{01}^O(X)-P_{00}^O(X)}\\
&\qquad+\E[M\mid A=1,B,X,G=E]-\E[M\mid A=0,B,X,G=E]\Big|G=O\Big].
\end{align*}
\end{itemize}
\end{theorem}

\subsection{Non-Binary Bespoke Instrumental Variable}

The bespoke instrumental variable \( B \) need not be binary. In the following, we provide the counterparts of the assumptions in our BSIV data fusion framework for non-binary \( B \). The identification arguments extend straightforwardly to this setting with no substantive changes to the proofs.

\begin{itemize}
\item {\bf \emph{BSIV Relevance.}} For any $b\neq 0$,
\[
\E[A\mid B=b,X,G=O]\neq \E[A\mid B=0,X,G=O].
\]	

\item {\bf \emph{BSIV Partial Additive Equi-Association.}} 
\begin{enumerate}[label=(\roman*)]
\item With probability one, we have
\begin{align*}
\E[M^{(1)}\mid  X,B,&~G=O]-\E[M^{(1)}\mid  X,B=0,G=O]\\
&=\E[Y^{(1)}\mid  X,B,G=O]-\E[Y^{(1)}\mid  X,B=0,G=O],
\end{align*}
\item With probability one, we have
\begin{align*}
\E[M^{(0)}\mid X,B,&~G=O]-\E[M^{(0)}\mid  X,B=0,G=O]\\
&=\E[Y^{(0)}\mid  X,B,G=O]-\E[Y^{(0)}\mid  X,B=0,G=O].
\end{align*}	
\end{enumerate}

\item {\bf \emph{Partial Homogeneity of Causal Effect Contrast.}}
\begin{enumerate}[label=(\roman*)]
\item The following quantity is not a function of $B$.
\begin{align*}
&\E[Y^{(1)}-Y^{(0)}\mid A=0,B,X,G=O]-\E[M^{(1)}-M^{(0)}\mid A=0,B,X,G=O].
\end{align*}
\item The following quantity is not a function of $B$.
\begin{align*}
&\E[Y^{(1)}-Y^{(0)}\mid A=1,B,X,G=O]-\E[M^{(1)}-M^{(0)}\mid A=1,B,X,G=O].
\end{align*}
\end{enumerate}

\item {\bf \emph{Partial Homogeneity of Bias Contrast.}}
\begin{enumerate}[label=(\roman*)]
\item The following quantity is not a function of $B$.
\begin{align*}
&\big\{\E[Y^{(1)}\mid A=1,B,X,G=O]-\E[Y^{(1)}\mid A=0,B,X,G=O]\big\}\\
&-\big\{\E[M^{(1)}\mid A=1,B,X,G=O]-\E[M^{(1)}\mid A=0,B,X,G=O]\big\}.
\end{align*}
\item The following quantity is not a function of $B$.
\begin{align*}
&\big\{\E[Y^{(0)}\mid A=1,B,X,G=O]-\E[Y^{(0)}\mid A=0,B,X,G=O]\big\}\\
&-\big\{\E[M^{(0)}\mid A=1,B,X,G=O]-\E[M^{(0)}\mid A=0,B,X,G=O]\big\}.
\end{align*}
\end{enumerate}

\end{itemize}

\section{Supplementary Materials for Section \ref{sec:proximal}}
\label{sec:proximal::supp}

\subsection{Identification of ATE}

\begin{theorem}
\label{thm:POR::supp}
Under Assumptions \ref{assumption:exchangeU}-\ref{assumption:ExVal}, \ref{assumption:proxy}, and \ref{assumption:compexist1}, the parameter $\theta_{ATE}$ is identified by
\begin{equation*}
\begin{aligned}
&\theta_{ATE}=
\E\big[\E[h(M,A,X)\mid A=1,X,G=E]\big|  G=O\big]
-\E\big[\E[h(M,A,X)\mid A=0,X,G=E]\big|  G=O\big].
\end{aligned}
\end{equation*}
\end{theorem}

\begin{theorem}
\label{thm:PIPW::supp}
Under Assumptions \ref{assumption:exchangeU}-\ref{assumption:ExVal}, \ref{assumption:proxy}, and \ref{assumption:compexist2}, the parameter $\theta_{ATE}$ is 
identified by
\begin{equation*}
\begin{aligned}
&\theta_{ATE}=
\E[I(A=1)Yq(Z,A,X)\mid G=O]-\E[I(A=0)Yq(Z,A,X)\mid G=O].
\end{aligned}
\end{equation*}
\end{theorem}

\subsection{Comparison to Other Proximal Setups}

\subsubsection{Comparison to the Original Proximal Causal Inference Setup}
\label{sec:prox:orig}

Our approach can be viewed as an extension of the  proximal causal inference framework \citep{miao2018identifying,tchetgen2020introduction,cui2020semiparametric} to the data fusion setup.
The original proximal causal inference framework only considers data from one domain and besides assuming existence of a proxy variable $Z$ which satisfies $Z\independent Y\mid\{U,X,A\}$, it also assumes existence of a second proxy variable $M$ of the latent confounder in the system which satisfies $M\independent \{A,Z\}\mid\{U,X\}$. Figures \ref{fig:origprox} demonstrates an example of a graphical model which satisfies the proxy variables conditions in the original setup. Importantly, that setup requires no treatment effect on the variable $M$. Our work shows that if this condition is violated, experimental data still allows for identification of the causal effect of the treatment on $Y$ by essentially anchoring  the short-term causal impact at that observed in the experimental sample.

\begin{figure}[t]
\centering
		\tikzstyle{block} = [draw, circle, inner sep=2.5pt, fill=lightgray]
		\tikzstyle{input} = [coordinate]
		\tikzstyle{output} = [coordinate]
        \begin{tikzpicture}
            \tikzset{edge/.style = {->,> = latex'}}
            \node[] (a) at  (-2,0) {$A$};
            \node[] (n) at  (0,0) {$M$};
            \node[] (x) at  (-3,1) {$X$};
            \node[] (y) at  (2,0) {$Y$};   
            \node[] (z) at  (-3,-.5) {$Z$};                     
            \node[block] (u) at  (-3,2) {$U$};                       
            \draw[-stealth] (x) to (a);
            \draw[-stealth][edge, bend left=-65] (u) to (a);                        
            \draw[-stealth][edge, bend left=-35] (a) to (y);            
            \draw[-stealth] (u) to (y);
            \draw[-stealth][edge, bend left=-65] (u) to (z);            
            \draw[-stealth] (x) to (z);
            \draw[-stealth] (z) to (a);                                             
            \draw[-stealth] (u) to (n);
            \draw[-stealth] (n) to (y);            	  
            \draw[-stealth] (u) to (x);	              
            \draw[-stealth] (x) to (n);
            \draw[-stealth] (x) to (y);            	                                      
        \end{tikzpicture}
        \caption{The original proximal causal inference model.}
        \label{fig:origprox}	
\end{figure}

\subsubsection{Comparison to \citep{imbens2022long}}
\label{sec:comImbens}

After the release of the first draft of our work, \cite{imbens2022long} also proposing an approach for identification of long-term causal effects based on proximal causal inference framework. Its main identification result in that work relies on stronger assumptions. Specifically, as for the internal validity of the experimental domain, the authors assume that for $a\in\{0,1\}$,
\[
\{Y^{(a)},M^{(a)},U,X\}\independent A\mid G=E,
\]
for the external validity, the authors assume that for $a\in\{0,1\}$,
\[
\{M^{(a)},U,X\}\independent G.
\]
Our assumption for internal validity is weaker in the sense that we allow that the researcher assigns treatments in the trial based on observed covariates. Clearly choosing the treatment independent of the observed covariates is a special case.
More importantly, our assumption for external validity is much weaker as we allow the distribution of the covariates to be different in the experimental and observational datasets. The data in these two datasets may have been collected in completely geographically separated places, and hence it is important to have the capability of allowing different distributions for the covariates in the two domains. 
However, we note that the authors also present an extension to their setup that resembles our proximal data fusion approach.  

In the place of our Assumption \ref{assumption:proxy} (i.e., existence of proximal variable) the authors of that work assume there exists three short-term outcomes $S=(S_1,S_2,S_3)$ in the system, sorted in the temporal order, and posit the ``sequential outcomes'' assumption, which states that for $a\in\{0,1\}$,
\[
\{Y^{(a)},S_3^{(a)}\}\independent S_1^{(a)}\mid \{S_2^{(a)},U,X,G=O\}.
\]
For example, this assumption is satisfied if the short-term and long-term outcomes can be directly affected by \emph{only} outcomes immediately preceding them.
The assumption is designed in a way that if $S_2$ is in the conditioning set, $S_1$ becomes independent of $S_3$. Therefore, by including $S_2$ in the conditioning set, a setup similar to our requirement of conditional independence of $Z$ and $\{M,Y\}$ is created.
However, it may be challenging in real-world setups to find three post-treatment variables that satisfy the specific sequential conditional independence required by the sequential outcomes assumption, whereas as clarified in Remark \ref{rmk:proxyflexy}, our model allows for great flexibility in terms of choosing the proxy variable $Z$.

\section{Supplementary Materials for Section \ref{sec:estimation}}
\label{sec:estimation::supp}

\subsection{Influence Functions for the Bespoke IV Data Fusion}

Below is the complete statement of Theorem \ref{thm:IF:BSIV-1_ett}.

\begin{theorem}

Under a nonparametric model, the efficient influence function for the parameter $\psi^{\text{bsiv1}}_{\text{ETT}}$ is given by:
{\footnotesize
\begin{align*}
&IF_{\psi^{\text{bsiv1}}_{\text{ETT}}}(V)\\
&=\frac{I(G=O)}{P_{11}^O(X)-P_{10}^O(X)}\frac{1}{p(A=1,G=O)}\frac{\pi^O(X)}{\rho_B^O(X)}\bigg\{I(B=1)\{Y-M-e_1^O(X)\}-I(B=0)\{Y-M-e_0^O(X)\}\\
&\quad\quad+\frac{e_1^O(X)-e_0^O(X)}{P_{11}^O(X)-P_{10}^O(X)}\big\{I(B=0)\{I(A=1)-P_{10}^O(X)\}-I(B=1)\{I(A=1)-P_{11}^O(X)\}\big\}\bigg\}\\
&\quad+\frac{I(A=1)I(G=O)}{p(A=1,G=O)}\Big\{\frac{e_1^O(X)-e_0^O(X)}{P_{11}^O(X)-P_{10}^O(X)}\Big\}-\frac{I(A=0)I(G=E)}{p(A=1,G=O)}\cdot\frac{1-\tau(B,X)}{\tau(B,X)}\cdot
\frac{M-M_{0}^{E}(B,X)}{1-P_{1B}^E(X)}\\
&\quad+\frac{I(G=O)}{p(A=1,G=O)}\cdot \{M-M_{0}^{E}(B,X)-I(A=1)\psi^{\text{bsiv1}}_{\text{ETT}}\},
\end{align*}
}
the efficient influence function for the parameter $\psi^{\text{bsiv2}}_{\text{ETT}}$ is given by:
{\footnotesize
\begin{align*}
&IF_{\psi^{\text{bsiv2}}_{\text{ETT}}}(V)\\
&=\frac{I(G=O)}{p(A=1,G=O)}\Bigg\{I(A=1)I(B=1)\{Y-M-E_{01}^O(X)-E_{10}^O(X)+E_{00}^O(X)\}\\
&\quad-I(A=0)I(B=1)\frac{P_{11}^O(X)}{P_{01}^O(X)}\{Y-M-E_{01}^O(X)\}+I(A=1)\{E_{10}^O(X)-E_{00}^O(X)\}\\
&\quad-I(B=0)\Big\{\frac{\rho_{1}^O(X)P_{11}^O(X)}{\rho_{0}^O(X)P_{10}^O(X)}I(A=1)\{Y-M-E_{10}^O(X)\}-\frac{\rho_{1}^O(X)P_{11}^O(X)}{\rho_{0}^O(X)P_{00}^O(X)}I(A=0)\{Y-M-E_{00}^O(X)\}\Big\}\\
&\quad+\frac{I(B=0)}{P_{10}^O(X)}\frac{\pi^O(X)}{\rho_{0}^O(X)}\cdot I(A=1)\{Y-M-E_{10}^O(X)\}-\frac{I(B=0)}{P_{00}^O(X)}\frac{\pi^O(X)}{\rho_{0}^O(X)}\cdot I(A=0)\{Y-M-E_{00}^O(X)\}\\
&\quad+\frac{\pi^O(X)}{P_{01}^O(X)-P_{00}^O(X)}\Big\{I(A=0)\big\{\frac{I(B=1)}{P_{01}^O(X)\rho_1^O(X)}\{Y-M-E_{01}^O(X)\}-\frac{I(B=0)}{P_{00}^O(X)\rho_0^O(X)}\{Y-M-E_{00}^O(X)\}\big\}\\
&\quad\quad+\frac{E_{01}^O(X)-E_{00}^O(X)}{P_{01}^O(X)-P_{00}^O(X)}\big\{-\frac{I(B=1)}{\rho_1^O(X)}\{I(A=0)-P_{01}^O(X)\}+\frac{I(B=0)}{\rho_0^O(X)}\{I(A=0)-P_{00}^O(X)\}\big\}\Big\}\\
&\quad+I(A=1)\{\frac{E_{01}^O(X)-E_{00}^O(X)}{P_{01}^O(X)-P_{00}^O(X)}\}+M-M_0^E(B,X)-I(A=1)\psi^{\text{bsiv2}}_{\text{ETT}}\Bigg\}\\
&\quad-\frac{1}{p(A=1,G=O)}\cdot\frac{I(A=0)I(G=E)}{1-P_{1B}^E(X)}\cdot\frac{1-\tau(B,X)}{\tau(B,X)}\cdot\{M-M_0^E(B,X)\}.
\end{align*}
}
\end{theorem}

\subsection{Estimation of ATE}

In this section, we provide estimation strategies for ATE under the three data fusion frameworks. We derive the influence functions for ATE and analyze the robustness properties of the proposed estimators under each case. We adopt the cross-fitting procedure described in Section \ref{sec:estimation}.

\subsubsection{Estimation under Equi-Confounding Data Fusion}

From Theorem \ref{thm:equiATECond}, the task of inference for $\theta_{ATE}$ reduces to the estimation of the following functionals of the observed data distribution respectively.
\begin{align*}
\psi_{\text{ATE}}^{\text{equi}} 
&=\E\big[
\E[Y\mid X,A=1,G=O]-\E[Y\mid X,A=0,G=O]\\
&\quad\quad+\E[M\mid X,A=1,G=E]-\E[M\mid X,A=0,G=E]\\
&\quad\quad+\E[M\mid X,A=0,G=O]-\E[M\mid X,A=1,G=O]
\big| G=O\big].
\end{align*}
Again, the bias of the resulting estimator will be of first order with respect to the bias in the estimation of the nuisance functions. Hence, we use IF-based estimators instead.

We first derive the influence functions of the parameter $\psi_{\text{ATE}}^{\text{equi}}$. Based on the obtained influence functions, we propose new identification formulae for $\theta_{ATE}$ as well as a multiply robust estimation strategies for these parameters. Let the collection of nuisance unctions be defined as: $\zeta^{\text{equi}}:=\Big\{\mathbb{E}[M\mid X,A,G=E], \mathbb{E}[M\mid X,A,G=O], \mathbb{E}[Y\mid X,A,G=O], p(A\mid X,G=E), p(A\mid X,G=O),p(G=E\mid X)\Big\}$.

\begin{theorem}
\label{thm:IF:equi-conf:ATE}
Under a non-parametric model, the efficient influence function of the parameter $\psi_{\text{ATE}}^{\text{equi}} $ is given by
\begin{align*}
IF_{\psi_{\text{ATE}}^{\text{equi}} }(V)=
&\frac{(-1)^{1-A}}{p(A\mid X,G=E)}\cdot\frac{I(G=E)}{p(G=O)}
\{M-\E[M\mid A,X,G]\}\{\frac{1}{p(G=E\mid X)}-1\}\\
&+\frac{I(G=O)}{p(G=O)}\big\{
\frac{(-1)^{1-A}}{p(A\mid X,G=O)}\{Y-\E[Y\mid A,X,G]-M+\E[M\mid A,X,G]\}\\
&+
\E[Y\mid X,A=1,G=O]-\E[Y\mid X,A=0,G=O]\\
&+\E[M\mid X,A=1,G=E]-\E[M\mid X,A=0,G=E]\\
&+\E[M\mid X,A=0,G=O]-\E[M\mid X,A=1,G=O]
-\psi_{\text{ATE}}^{\text{equi}} \big\}.
\end{align*}
\end{theorem}

Let $\hat\zeta^{\text{equi}}$ be and estimator for $\zeta^{\text{equi}}$, and define the moment function $\Phi^{\text{equi}}(V; \hat\zeta^{\text{equi}}, \hat\psi)$ to be the expression for $IF_{\psi^{\text{equi}}_{\text{ATE}}}(V)$ with nuisance components replaced by their corresponding estimates from $\hat\zeta^{\text{equi}}$, and $\psi^{\text{equi}}_{\text{ATE}}$ replaced by $\hat\psi$. We use this moment function in the cross-fitting procedure to obtain the estimator $\hat\psi^{\text{equi}}_{\text{ATE}}$. We have the following multiple robustness result for $\hat\psi^{\text{equi}}_{\text{ATE}}$.

\begin{proposition}
\label{prop:DR:equi-conf_ATE}
The IF-based estimator $\hat\psi^{\text{equi}}_{\text{ATE}}$ is multiply robust, in the sense that it is unbiased if at least one of the following subsets of nuisance functions is correctly specified:
\begin{itemize}
\item $\{\mathbb{E}[M\mid X,A,G=E], \mathbb{E}[M\mid X,A,G=O], \mathbb{E}[Y\mid X,A,G=O]\}$
\item $\{p(A\mid X,G=E), p(A\mid X,G=O),p(G=E\mid X)\}$
\item $\{\mathbb{E}[M\mid X,A,G=E], p(A\mid X,G=O)\}$
\item $\{\mathbb{E}[M\mid X,A,G=O], \mathbb{E}[Y\mid X,A,G=O], p(A\mid X,G=E),p(G=E\mid X)\}$
\end{itemize}
\end{proposition}

\subsubsection{Estimation under Bespoke IV Data Fusion}

Recall from Theorem \ref{thm:BSIV_ATE} that under the assumptions of that framework, the task of inference for $\theta_{ATE}$ reduces to the estimation of the following  functionals of the observed data distribution.
\begin{align*}
\psi^{\text{bsiv1}}_{\text{ATE}}
&=\E\Big[
\frac{\E[Y-M\mid B=1,X,G=O]-\E[Y-M\mid B=0,X,G=O]}{p(A=1\mid B=1,X,G=O)-p(A=1\mid B=0,X,G=O)}\\
&\qquad+\E[M\mid A=1,B,X,G=E]-\E[M\mid A=0,B,X,G=E]\Big|G=O\Big],\\\\
\psi^{\text{bsiv2}}_{\text{ATE}}
&=\E\Big[
\{E_{11}^O(X)-E_{01}^O(X)-E_{10}^O(X)+E_{00}^O(X)\}B+E_{10}^O(X)-E_{00}^O(X)\\
&\qquad-\frac{\{E_{01}^O(X)-E_{00}^O(X)\}P_{1B}^O(X)+\{E_{11}^O(X)-E_{10}^O(X)\}P_{0B}^O(X)}{P_{01}^O(X)-P_{00}^O(X)}\\
&\qquad+\E[M\mid A=1,B,X,G=E]-\E[M\mid A=0,B,X,G=E]\Big|G=O\Big].
\end{align*}
We first derive the influence functions of the parameters above. Based on the obtained influence functions, we propose new identification formulae for $\theta_{ATE}$ as well a multiply robust estimation strategy for these parameters. Let the collection of nuisance unctions be defined as: $\zeta^{bsiv}:=
\Big\{
E_{ab}^O(X):=\mathbb{E}[Y-M\mid A=a, B=b,X,G=O], 
e_{b}^O(X):=\mathbb{E}[Y-M\mid B=b,X,G=O], 
M_{a}^{E}(B,X):=\mathbb{E}[M\mid A=a,B,X,G=E], 
\mu_{b}^O(X):=\mathbb{E}[M\mid B=b,X,G=O], 
P_{ab}^O(X):=p(A=a\mid B=b,X,G=O), 
P_{ab}^E(X):=p(A=a\mid B=b,X,G=E), 
\rho_{b}^O(X):=p(B=b\mid X,G=O), 
\rho_{b}^E(X):=p(B=b\mid X,G=E), 
\tau(B,X):=p(G=E\mid B,X)
\Big\}$.

\begin{theorem}
\label{thm:IF:BSIV-1}
Under a non-parametric model, the efficient influence function of the parameter $\psi^{\text{bsiv1}}_{\text{ATE}}$ is given by
\begin{align*}
&IF_{\psi^{\text{bsiv1}}_{\text{ATE}}}(V)\\
&=\frac{I(G=O)}{p(G=O)}\Bigg\{\frac{1}{P_{11}^O(X)-P_{10}^O(X)}
\frac{1}{\rho_B^O(X)}\bigg\{I(B=1)\{Y-M-e_1^O(X)\}
-I(B=0)\{Y-M-e_0^O(X)\}\\
&\quad\quad\quad\quad\quad+\frac{e_1^O(X)-e_0^O(X)}{P_{11}^O(X)-P_{10}^O(X)}\big\{I(B=0)\{I(A=1)-P_{10}^O(X)\}-I(B=1)\{I(A=1)-P_{11}^O(X)\}\big\}\bigg\}\\
&\quad\quad\quad\quad\quad+\frac{e_1^O(X)-e_0^O(X)}{P_{11}^O(X)-P_{10}^O(X)}+M_{1}^{E}(B,X)-M_{0}^{E}(B,X)-\psi^{\text{bsiv1}}_{\text{ATE}}\Bigg\}\\
&\quad+\frac{I(G=E)}{p(G=O)}\cdot\frac{1-\tau(B,X)}{\tau(B,X)}\bigg\{\frac{I(A=1)}{P_{1B}^E(X)}\cdot\{M-M_{1}^{E}(B,X)\}-\frac{I(A=0)}{1-P_{1B}^E(X)}\cdot\{M-M_{0}^{E}(B,X)\}\bigg\}
,
\end{align*}
and the efficient influence function of the parameter $\psi^{\text{bsiv2}}_{\text{ATE}}$ is given by
\begin{align*}
&IF_{\psi^{\text{bsiv2}}_{\text{ATE}}}(V)\\
&=\frac{I(G=O)}{p(G=O)}\Bigg\{\frac{I(A=1)I(B=1)}{P_{11}^O(X)}\{Y-M-E_{11}^O(X)\}-\frac{I(A=0)I(B=1)}{P_{01}^O(X)}\{Y-M-E_{01}^O(X)\}\\
    &\ \ \ \ \ \ \ \ \ \ \ \ \ \ \ \ \ \ -\frac{I(A=0)I(B=0)}{P_{00}^O(X)}\{Y-M-E_{00}^O(X)\}+\frac{I(A=1)I(B=0)}{P_{10}^O(X)}\{Y-M-E_{10}^O(X)\}\\
    &\ \ \ \ \ \ \ \ \ \ \ \ \ \ \ \ \ \ +I(B=1)\{E_{11}^O(X)-E_{01}^O(X)-E_{10}^O(X)+E_{00}^O(X)\}+E_{10}^O(X)-E_{00}^O(X)\\
    &\ \ \ \ \ \ \ \ \ \ \ \ \ \ \ \ \ \ -\frac{\{E_{01}^O(X)-E_{00}^O(X)\}P_{1B}^O(X)+\{E_{11}^O(X)-E_{10}^O(X)\}(1-P_{1B}^O(X))}{P_{01}^O(X)-P_{00}^O(X)}\\
    &\ \ \ \ \ \ \ \ \ \ \ \ \ \ \ \ \ \ +M_1^E(B,X)-M_0^E(B,X)-\psi^{\text{bsiv2}}_{\text{ATE}}\\
    &\ \ \ \ \ \ \ \ \ \ \ \ \ \ \ \ \ \ +\frac{1}{P_{01}^O(X)-P_{00}^O(X)}\bigg\{-\frac{I(A=1)I(B=1)P_{01}^O(X)}{P_{11}^O(X)\rho_1^O(X)}\{Y-M-E_{11}^O(X)\}\\
    &~~~~~~~~~~~~~~~~~~~~~~~~~~~-\frac{I(A=0)I(B=1)P_{11}^O(X)}{P_{01}^O(X)\rho_1^O(X)}\{Y-M-E_{01}^O(X)\}\\
    &~~~~~~~~~~~~~~~~~~~~~~~~~~~+\frac{I(A=1)I(B=0)P_{00}^O(X)}{P_{10}^O(X)\rho_0^O(X)}\{Y-M-E_{10}^O(X)\}\\
    &~~~~~~~~~~~~~~~~~~~~~~~~~~~+\frac{I(A=0)I(B=0)P_{10}^O(X)}{P_{00}^O(X)\rho_0^O(X)}\{Y-M-E_{00}^O(X)\}\\
    &~~~~~~~~~~~~~~~~~~~~~~~~~~~+\frac{E_{11}^O(X)-E_{01}^O(X)-E_{10}^O(X)+E_{00}^O(X)}{\rho_B^O(X)}\big\{I(A=1)-P_{1B}^O(X)\big\}\\
    &~~~~~~~~~~~~~~~~~~~~~~~~~~~ +\frac{\{E_{11}^O(X)-E_{01}^O(X)-E_{10}^O(X)+E_{00}^O(X)\}P_{1B}^O(X)}{P_{01}^O(X)-P_{00}^O(X)}\\
    &~~~~~~~~~~~~~~~~~~~~~\cdot\big\{-\frac{I(B=1)}{\rho_1^O(X)}\{I(A=0)-P_{01}^O(X)\}+\frac{I(B=0)}{\rho_0^O(X)}\{I(A=0)-P_{00}^O(X)\}\big\}\bigg\}\Bigg\}\\
    &+\frac{I(G=E)}{p(G=O)}\frac{1-\tau(B,X)}{\tau(B,X)}\bigg\{\frac{I(A=1)}{P_{1B}^E(X)}\{M-M_1^E(B,X)\}-\frac{I(A=0)}{1-P_{1B}^E(X)}\{M-M_0^E(B,X)\}\bigg\}
.
\end{align*}

\end{theorem}

Let $\hat\zeta^{\text{bsiv}}$ be and estimator for $\zeta^{\text{bsiv}}$, and define the moment functions $\Phi^{\text{bsiv1}}(V; \hat\zeta^{\text{bsiv}}, \hat\psi)$ and $\Phi^{\text{bsiv2}}(V; \hat\zeta^{\text{bsiv}}, \hat\psi)$ to be the expressions for $IF_{\psi^{\text{bsiv1}}_{\text{ATE}}}(V)$ and $IF_{\psi^{\text{bsiv2}}_{\text{ATE}}}(V)$ with nuisance components replaced by their corresponding estimates from $\hat\zeta^{\text{bsiv}}$, and $\psi^{\text{bsiv1}}_{\text{ATE}}$ and $\psi^{\text{bsiv2}}_{\text{ATE}}$ replaced by $\hat\psi$. We use these moment functions in the cross-fitting procedure to obtain the estimators $\hat\psi^{\text{bsiv1}}_{\text{ATE}}$ and $\hat\psi^{\text{bsiv2}}_{\text{ATE}}$. We have the following multiple robustness result for these estimators.

\begin{proposition}
\label{prop:DR:BSIV_ATE}
The IF-based estimator $\hat\psi^{\text{bsiv1}}_{\text{ATE}}$ is multiply robust in the sense that it is unbiased if at least one of the following subsets of nuisance functions is correctly specified:\\
(i) $\{P_{AB}^O(X),\mathbb{E}[M\mid B,X,G=O],\mathbb{E}[Y\mid B,X,G=O],\mathbb{E}[M\mid A,B,X,G=E]\}$;\\
(ii) $\{\tau(B,X),\rho_B^O(X),\rho_B^E(X),P_{AB}^O(X),P_{AB}^E(X)\}$;\\
(iii) $\{\tau(B,X),P_{AB}^O(X),P_{AB}^E(X),\mathbb{E}[M\mid B,X,G=O],\mathbb{E}[Y\mid B,X,G=O]\}$;\\
(iv) $\{\rho_B^O(X),\rho_E^O(X),P_{AB}^O(X),\mathbb{E}[M\mid A,B,X,G=E]\}$.\\
Same result holds for the estimator $\hat\psi^{\text{bsiv2}}_{\text{ATE}}$ 
but with replacing outcome regression functions with
$\mathbb{E}[M\mid A,B,X,G=E]$,
$\mathbb{E}[M\mid A,B,X,G=O]$, and
$\mathbb{E}[Y\mid A,B,X,G=O]\}$.
\end{proposition}

\subsubsection{Estimation under Proximal Data Fusion}
For $a\in\{0,1\}$, we define the parameter $\theta^{(a)}\coloneqq\E[Y^{(a)}\mid G=O]$, which with arguments similar to those in Theorems \ref{thm:POR::supp} and \ref{thm:PIPW::supp}, can be identified using the functional
\begin{align}
\psi^a
&=\E\big[\E[h(M,A,X)\mid A=a,X,G=E]~\big|~  G=O\big]\label{eq:theta-id1}\\
&=\E\big[I(A=a)Yq(Z,A,X)~\big|~G=O\big]\label{eq:theta-id2}.
\end{align}
From the identification formulae of the parameter $\theta_{ATE}$, we note that $\theta^{(a)}$ is the main piece which is needed to be estimated for obtaining an estimator for $\theta_{ATE}$. Therefore, we focus on the estimation of the functional $\psi^a$.
We first derive the influence function of the parameter $\psi^a$. Based on the obtained influence function, we propose a new identification formula for $\theta^{(a)}$ as well as several estimation strategies for this parameter. In the following, we assume that the integral equations in Assumptions \ref{assumption:compexist1} $(ii)$ and \ref{assumption:compexist2} $(ii)$ have unique solutions. Let the collection of nuisance functions be defined as $\zeta^{\text{proxy}}:=\Big\{
h(M,A,X),q(Z,A,X), p(m\mid A,X,G=E), p(A=1\mid X,G=E),  p(G=E\mid X)\Big\}$.

\begin{theorem}
\label{thm:IF:proximal}	
Under a semiparametric model that Assumptions \ref{assumption:compexist1} $(ii)$ and \ref{assumption:compexist2} $(ii)$ have unique solutions, for $a\in\{0,1\}$, an influence function of the parameter $\psi^a$ is given by
\begin{align*}
IF_{\psi^a}(V)&=
\frac{I(G=O)}{p(G=O)}I(A=a)q(Z,A,X)\{Y-h(M,A,X)\}\\
&\quad+\frac{I(G=E)}{p(G=O)}\cdot\frac{I(A=a)}{p(A=a\mid X,G=E)}
\{h(M,A,X)-\eta(A,X)\}\{\frac{1}{p(G=E\mid X)}-1\}\\
&\quad+\frac{I(G=O)}{p(G=O)}\{\eta(a,X)-\psi^a\},
\end{align*}
where
\begin{align*}
\eta(a,x)\coloneqq&\E[h(M,A,X)\mid A=a,X=x,G=E]\\
=&\E[ I(A=a)h(M,A,X)q(Z,A,X)\mid X=x,G=E ].
\end{align*}
\end{theorem}

Let $\hat\zeta^{\text{proxy}}$ be and estimator for $\zeta^{\text{proxy}}$. Based on identification formulae \eqref{eq:theta-id1} and \eqref{eq:theta-id2}, along with Theorem \ref{thm:IF:proximal}, we propose the following estimation strategies for the parameter $\theta^{(a)}$.

\begin{itemize}
	
\item{\bf Estimation Strategy 1.} From \eqref{eq:theta-id1}, define the moment function as,
\begin{align*}
\Phi^{\text{proxy1}}(V;\hat{\zeta}^{\text{proxy}},\hat\psi):= \sum_m\frac{I(G=O)}{p(G=O)}\hat{h}(m,a,X)\hat{p}(m\mid a,X,G=E)-\hat\psi.
\end{align*}

\item{\bf Estimation Strategy 2.} Also from \eqref{eq:theta-id1}, define the moment function as,
\begin{align*}
&\Phi^{\text{proxy2}}(V;\hat{\zeta}^{\text{proxy}},\hat\psi)\\
&:=\frac{I(G=E)I(A=a)}{p(G=O)}\cdot\frac{1}{1-a+(-1)^{1-a}\hat{p}(A=1\mid X,G=E)}
\hat{h}(M,a,X)\{\frac{1}{\hat{p}(G=E\mid X)}-1\}
-\hat\psi.
\end{align*}

\item{\bf Estimation Strategy 3.} From \eqref{eq:theta-id2}, define the moment function as,

\begin{align*}
&\Phi^{\text{proxy2}}(V;\hat{\zeta}^{\text{proxy}},\hat\psi):=\frac{I(G=O)I(A=a)}{p(G=O)}Y\hat{q}(Z,a,X)-\hat\psi.
\end{align*}

\item{\bf Estimation Strategy 4 (IF-based Strategy).} From Theorem \ref{thm:IF:proximal}, define the moment function $\Phi^{\text{proxy}}(V; \hat\zeta^{\text{proxy}}, \hat\psi)$ to be the expression for $IF_{\psi^a}(V)$ with nuisance components replaced by their corresponding estimates from $\hat\zeta^{\text{proxy}}$, and $\psi^a$ replaced by $\hat\psi$.

\end{itemize}

These moment functions are implemented in the cross-fitting procedure to obtain $\hat\psi^{a,\text{proxy1}}$, $\hat\psi^{a,\text{proxy2}}$, $\hat\psi^{a,\text{proxy3}}$, and the IF-based estimator $\hat\psi^{a,\text{proxy}}$.

\begin{proposition}
\label{prop:DR:proximal_ate}	
Estimation Strategy 4 is multiply robust, in the sense that it is unbiased if at least one of the following subsets of nuisance functions is correctly specified:
(i) $\{h,p(m\mid a,x,G=E)\}$;
(ii) $\{h,p(A=1\mid x,G=E),p(G=E\mid x)\}$;
(iii) $\{q,p(A=1\mid x,G=E),p(G=E\mid x)\}$.
\end{proposition}

\section{Supplementary Materials for Section \ref{sec:simulation}}
\label{sec:simulation::supp}

\paragraph{Data Generating Process.}
We consider a two-dimensional covariate $X=(X_1,X_2)$ generated from a multivariate normal distribution $\mathcal{N}(\mu,\Sigma)$, where $\mu=\begin{pmatrix}
    0.1 & -0.1
\end{pmatrix}^T$, and $\Sigma=\begin{pmatrix}
    1 & 0\\
    0 & 1
\end{pmatrix}$. We then generate the bespoke instrument variable $B$ conditional on $X$ from a Bernoulli distribution, with $P(B=1\mid X_1,X_2)=expit(-0.43+0.15X_1+0.18X_2)$, and $U$ conditional on $(X,B)$ from a normal distribution $\mathcal{N}(0.15-0.35X_1+0.8X_2+0.6B,1)$. Next, we generate domain assignment $G$ from a Bernoulli distribution, with $P(G=E\mid X_1,X_2,B)=expit(0.2+0.15X_1+0.1X_2-0.35B)$. For treatment assignment $A$ of observational data, we have
\begin{align*}
    P(A=1\mid X_1,X_2,B,U)=expit(-0.3-0.1U+1.3B+0.1X_1+0.15X_2),
\end{align*}
and for $A$ of experimental data, we have
\begin{align*}
    P(A=1\mid X_1,X_2,B)=expit(-0.23+0.68B-0.13X_1).
\end{align*}

Then we generate the proxy variable $Z$ from the following normal distribution,
\begin{align*}
    Z\mid U,A,X,B \sim \mathcal{N}(0.2+1.4U+1.5A+0.1X_1-0.5X_2+1.3B,1).
\end{align*}
and we generate $M$ from the following normal distribution,
\begin{align*}
    M\mid U,A,X,B\sim \mathcal{N}(0.75U+0.4A+0.2X_1-0.5X_2+0.12B,1).
\end{align*}

Finally, $Y$ is generated from $\mathcal{N}(E[M\mid A,B,X]+W,1)$, where
\begin{align*}
    &W=\beta_1(X_1,X_2,B)A+\gamma_0(X_1,X_2)(A-P(A=1\mid X,B))+E[Y^{(0)}-M^{(0)}\mid B,X_1,X_2,G=O]\\
    &\beta_1(X_1,X_2,B)=0.25+0.3X_1-0.6X_2-0.1B\\
    &\gamma_0(X_1,X_2)=0.12-0.5X_1+0.45X_2\\
    &E[Y^{(0)}-M^{(0)}\mid B,X_1,X_2,G=O]=0.54-0.28X_1+0.35X_2.
\end{align*}

\paragraph{Case Specification.}
\begin{itemize}
    \item BSIV
    \begin{itemize}
        \item All true: all models.
        \item Case 1: $\{P_{1B}^O(X),\mathbb{E}[M\mid A,B,X,G=O],\mathbb{E}[Y\mid A,B,X,G=O],\mathbb{E}[M\mid A,B,X,G=E]\}$.
        \item Case 2: $\{\tau(B,X),\rho_B^O(X),\rho_B^E(X),P_{1B}^O(X),P_{1B}^E(X)\}$.
        \item Case 3: $\{\tau(B,X),P_{1B}^O(X),P_{1B}^E(X),\mathbb{E}[M\mid A,B,X,G=O],\mathbb{E}[Y\mid A,B,X,G=O]\}$.
        \item Case 4: $\{\rho_B^O(X),\rho_B^E(X),P_{1B}^O(X),\mathbb{E}[M\mid A,B,X,G=E]\}$.
        \item All false: none.
    \end{itemize}
    \item Proximal
    \begin{itemize}
        \item All true: all models.
        \item Case 1: $\{h,p(m\mid A=0,x,G=E))\}$.
        \item Case 2: $\{h,p(A=1\mid x,G=E),p(G=E\mid x)\}$.
        \item Case 3: $\{q,p(A=1\mid x,G=E),p(G=E\mid x)\}$.
        \item All false: none.
    \end{itemize}
\end{itemize}

\section{Comparison of the Proposed Approaches}
\label{sec:summary::supp}

\begingroup
\setlength{\tabcolsep}{6pt}        
\setlength{\aboverulesep}{10pt}
\setlength{\belowrulesep}{10pt}
\renewcommand{\arraystretch}{1.1}
{\fontsize{9}{10}\selectfont       
\begin{sidewaystable}[p]\scriptsize
\centering
\caption{Comparison of the equi-confounding, BSIV, and proximal data fusion approaches.}
\begin{tabularx}{\textwidth}{@{} l Y Y Y Y Y @{}}
\toprule
\textbf{Approach} & \textbf{Key identifying assumption} & \textbf{Auxiliary variable} &
\textbf{Nature of assumption} & \textbf{Benefit} & \textbf{Challenge} \\
\midrule
Equi-\\confounding &
Equal additive association with the treatment \(A\) for \(M^{(0)}\) and \(Y^{(0)}\)&
None&
Equality of additive association&
No need for auxiliary variables&
Assuming equi-additive association w.r.t. $A$ may be hard to justify\\
\addlinespace\midrule
BSIV   &
Equal additive association with the BSIV \(B\) for \(M^{(0)}\) and \(Y^{(0)}\) and homogeneity&
Bespoke IV \(B\) &
Equality of additive association&
Weaker assumptions compared to equi-confounding approach&
Might be challenging to find credible $B$;
assuming equi-additive association w.r.t. $B$ may be hard to justify\\
\addlinespace\midrule
Proximal   &
Proxy variable assumption; completeness; existence of the bridge function&
Proxy variable \(Z\)&
Connecting \(M\) and \(Y\) via bridge functions&
Generalizes beyond additive association; allows greater modeling flexibility compared to equi-confounding/BSIV approaches&
Might be challenging to find credible $Z$;
parametric estimation approach is hard to justify;
might be challenging to estimate the bridge function nonparametrically\\
\bottomrule
\end{tabularx}
\end{sidewaystable}
} 
\endgroup

\newpage
\section{Proofs}

\subsection{Proofs of Section \ref{sec:athey}}

\begin{proof}[Proof of Theorem \ref{thm:athey}]
\cite{athey2020combining} only considered $\theta_{\text{ATE}}$ as the parameter of interest and proved the following proof:
\begin{align*}
\E[Y^{(a)}\mid G=O]
&=\E[ \E[Y^{(a)}\mid X,G=O] \mid G=O]\\
&\overset{A\ref{assumption:ExVal}}{=}\E[ \E[Y^{(a)}\mid X,G=E] \mid G=O]\\
&=\E[ \E[\E[Y^{(a)}\mid M^{(a)},X,G=E]\mid X,G=E] \mid G=O]\\
&\overset{A\ref{assumption:ExVal}}{=}\E[ \E[\E[Y^{(a)}\mid M^{(a)},X,G=O]\mid X,G=E] \mid G=O]\\
&\overset{A\ref{assumption:LaUn}}{=}\E[ \E[\E[Y^{(a)}\mid M^{(a)},A=a,X,G=O]\mid X,G=E] \mid G=O]\\
&=\E[ \E[\E[Y\mid M,A=a,X,G=O]\mid X,G=E] \mid G=O].
\end{align*}
$f(M,X)=\E[Y\mid M,A=a,X,G=O]$ is identified from the observational data.

Note that the middle expectation is over $M^{(a)}$. By Assumption \ref{assumption:IntVal}, 
\[
\E[f(M^{(a)},X)\mid X,G=E]
=\E[f(M^{(a)},X)\mid X,A=a,G=E]
=\E[f(M,X)\mid X,A=a,G=E]
\]
is identified from the experimental data.

To see the back-and-forths between the two domains better, it is easier to look at the following alternative presentation of the proof:
\begin{align*}
p(Y^{(a)}\mid G=O)	
&=\sum_{X} p(Y^{(a)}\mid X,G=O)p(X\mid G=O)\\
&=\sum_{X} p(Y^{(a)}\mid X,G=E)p(X\mid G=O)\\
&=\sum_{X,M^{(a)}} p(Y^{(a)}\mid M^{(a)},X,G=E)p(M^{(a)}\mid X,G=E)p(X\mid G=O)\\
&=\sum_{X,M^{(a)}} p(Y^{(a)}\mid M^{(a)},X,G=O)p(M^{(a)}\mid X,G=E)p(X\mid G=O)\\
&=\sum_{X,M} p(Y\mid A=a,M,X,G=O)p(M\mid A=a,X,G=E)p(X\mid G=O).
\end{align*}
Therefore, $\E[Y^{(1)}\mid G=0]$, $\E[Y^{(0)}\mid G=0]$, and hence ATE is identified.

Realizing that $\E[Y^{(0)}\mid G=O]$ is identified, we see that $\E[Y^{(0)}\mid A=1, G=O]$ is also identified as 
\begin{align*}
&\E[Y^{(0)}\mid G=O]\\
&=\E[Y^{(0)}\mid A=1, G=0]p(A=1\mid G=O)
+\E[Y^{(0)}\mid A=0, G=O]p(A=0\mid G=O)\\
&\Rightarrow
\E[Y^{(0)}\mid A=1, G=O]
=\frac{\E[Y^{(0)}\mid G=O]-\E[Y\mid A=0, G=O]p(A=0\mid G=O)}{p(A=1\mid G=O)}.
\end{align*}
Therefore, $\theta_{\text{ETT}}=\E[Y\mid A=1, G=O]-\E[Y^{(0)}\mid A=1, G=O]$ is also identified.

\end{proof}

\subsection{Proofs of Sections \ref{sec:AltIDAss} and \ref{sec:AltIDAss::supp}}

\begin{proof}[Proof of Theorem \ref{thm:equiETT}]

\begin{lemma}
\label{lem:idcond}	
Under Assumptions \ref{assumption:IntVal} and \ref{assumption:ExVal}, for $a\in\{0,1\}$, the parameter $\E[M^{(a)}\mid X,A=1-a,G=O]$ is identified as
\begin{align*}
\E[M^{(a)}\mid X,&A=1-a, G=O]\\
&=\frac{\E[M\mid X,A=a,G=E]-\E[M\mid X,A=a, G=O]p(A=a\mid X,G=O)}{p(A=1-a\mid X,G=O)}.
\end{align*}
\end{lemma}
\begin{proof}[Proof of Lemma \ref{lem:idcond}]
We only show that $\E[M^{(0)}\mid X,A=1,G=O]$ is identified.
\begin{align*}
\E[M^{(0)}\mid X,G=O]
&=\E[M^{(0)}\mid X,A=1, G=O]p(A=1\mid X,G=O)\\
&\quad+\E[M^{(0)}\mid X,A=0, G=O]p(A=0\mid X,G=O).
\end{align*}
Therefore,
\begin{align*}
&\E[M^{(0)}\mid X,A=1, G=O]\\
&=\frac{\E[M^{(0)}\mid X,G=O]-\E[M\mid X,A=0, G=O]p(A=0\mid X,G=O)}{p(A=1\mid X,G=O)}\\
&\overset{A\ref{assumption:ExVal}}{=}\frac{\E[M^{(0)}\mid X,G=E]-\E[M\mid X,A=0, G=O]p(A=0\mid X,G=O)}{p(A=1\mid X,G=O)}\\
&\overset{A\ref{assumption:IntVal}}{=}\frac{\E[M^{(0)}\mid X,A=0,G=E]-\E[M\mid X,A=0, G=O]p(A=0\mid X,G=O)}{p(A=1\mid X,G=O)}\\
&=\frac{\E[M\mid X,A=0,G=E]-\E[M\mid X,A=0, G=O]p(A=0\mid X,G=O)}{p(A=1\mid X,G=O)}.
\end{align*}

\end{proof}

\begin{lemma}
\label{lem:id1}
Under Assumptions \ref{assumption:IntVal} and \ref{assumption:ExVal}, for $a\in\{0,1\}$, the parameter $\E[M^{(a)}\mid A=1-a,G=O]$ is identified as
\begin{align*}
&\E[M^{(a)}\mid A=1-a, G=O]\\
&=\frac{\E[\E[M\mid X,A=a,G=E]\mid G=O]-\E[M\mid A=a, G=O]p(A=a\mid G=O)}{p(A=1-a\mid G=O)}.
\end{align*}
\end{lemma}
\begin{proof}[Proof of Lemma \ref{lem:id1}]
We only show that $\E[M^{(0)}\mid A=1,G=O]$ is identified.
\begin{align*}
\E[M^{(0)}\mid G=O]
&=\E[M^{(0)}\mid A=1, G=O]p(A=1\mid G=O)\\
&\quad+\E[M^{(0)}\mid A=0, G=O]p(A=0\mid G=O).
\end{align*}
Therefore,
\begin{equation}
\label{eq:lem1}
\E[M^{(0)}\mid A=1, G=O]
=\frac{\E[M^{(0)}\mid G=O]-\E[M\mid A=0, G=O]p(A=0\mid G=O)}{p(A=1\mid G=O)}.	
\end{equation}
Moreover,
\begin{equation}
\label{eq:lem2}
\begin{aligned}
\E[M^{(0)}\mid G=O]
&=\E[\E[M^{(0)}\mid X,G=O]\mid G=O]\\
&\overset{A\ref{assumption:ExVal}}{=}\E[\E[M^{(0)}\mid X,G=E]\mid G=O]\\
&\overset{A\ref{assumption:IntVal}}{=}\E[\E[M^{(0)}\mid X,A=0,G=E]\mid G=O]\\
&=\E[\E[M\mid X,A=0,G=E]\mid G=O].
\end{aligned}
\end{equation}
\eqref{eq:lem1} and \eqref{eq:lem2} imply that
\begin{align*}
&\E[M^{(0)}\mid A=1, G=O]\\
&=\frac{\E[\E[M\mid X,A=0,G=E]\mid G=O]-\E[M\mid A=0, G=0]p(A=0\mid G=O)}{p(A=1\mid G=O)}.
\end{align*}

\end{proof}

In order to prove Theorems \ref{thm:equiETT} and \ref{thm:equiATE}, we first state the following corollary of Lemma \ref{lem:id1}.

\begin{corollary}
\label{cor:id1}
Under Assumptions \ref{assumption:IntVal}, \ref{assumption:ExVal}, and \ref{assumption:Equi::supp} (\ref{assumption:Equi::supp} weakened to \ref{assumption:Equi} for ETT), for $a\in\{0,1\}$, the parameter $\E[Y^{(a)}\mid A=1-a,G=O]$ is identified. 
\end{corollary}
\begin{proof}[Proof of Corollary \ref{cor:id1}]
We only show that $\E[Y^{(0)}\mid A=1,G=O]$ is identified.
By Assumption \ref{assumption:Equi},
\begin{align*}
&\E[Y^{(0)}\mid  A=1,G=O]\\
&=\E[Y\mid  A=0,G=O]
-\E[M\mid A=0,G=O]+\E[M^{(0)}\mid A=1,G=O].
\end{align*}
Therefore, by Lemma \ref{lem:id1}, 
\begin{align*}
&\E[Y^{(0)}\mid  A=1,G=O]\\
&=\E[Y\mid  A=0,G=O]
-\E[M\mid A=0,G=O]\\
&\quad+\frac{\E[\E[M\mid X,A=0,G=E]\mid G=O]-\E[M\mid A=0, G=O]p(A=0\mid G=O)}{p(A=1\mid G=O)}.
\end{align*}

\end{proof}

Using Corollary \ref{cor:id1} we have
\begin{align*}
\theta_{\text{ETT}}&=\E[Y^{(1)}\mid A=1, G=O]-\E[Y^{(0)}\mid A=1, G=O]\\
&=\E[Y\mid A=1, G=O]-\E[Y\mid  A=0,G=O]
+\E[M\mid A=0,G=O]\\
&\quad-\frac{\E[\E[M\mid X,A=0,G=E]\mid G=O]-\E[M\mid A=0, G=O]p(A=0\mid G=O)}{p(A=1\mid G=O)}.
\end{align*}

\end{proof}

\begin{proof}[Proof of Theorem \ref{thm:equiATE}]
Using Corollary \ref{cor:id1} we have
\begin{equation}
\label{eq:thm31}
\begin{aligned}
\E[Y^{(1)}\mid G=O]
&=\E[Y^{(1)}\mid A=1, G=O]p(A=1\mid G=O)\\
&\quad +\E[Y^{(1)}\mid A=0, G=O]p(A=0\mid G=O)\\
&=\E[Y\mid A=1, G=O]p(A=1\mid G=O)\\
&\quad +\E[Y\mid A=1,G=O]p(A=0\mid G=O)\\
&\quad -\E[M\mid A=1,G=O]p(A=0\mid G=O)\\
&\quad +\E[\E[M\mid X,A=1,G=E]\mid G=O]\\
&\quad -\E[M\mid A=1,G=O]p(A=1\mid G=O).
\end{aligned}
\end{equation}
Similarly,
\begin{equation}
\label{eq:thm32}
\begin{aligned}
\E[Y^{(0)}\mid G=O]
&=\E[Y^{(0)}\mid A=1, G=O]p(A=1\mid G=O)\\
&\quad +\E[Y^{(0)}\mid A=0, G=O]p(A=0\mid G=O)\\
&=\E[Y\mid A=0,G=O]p(A=1\mid G=O)\\
&\quad -\E[M\mid A=0,G=O]p(A=1\mid G=O)\\
&\quad +\E[\E[M\mid X,A=0,G=E]\mid G=O]\\
&\quad -\E[M\mid A=0,G=O]p(A=0\mid G=O)\\
&\quad +\E[Y\mid A=0, G=O]p(A=0\mid G=O).
\end{aligned}
\end{equation}
\eqref{eq:thm31} and \eqref{eq:thm32} conclude that 
\begin{align*}
\theta_{\text{ATE}} &=\E[Y^{(1)}\mid G=O]-\E[Y^{(0)}\mid G=O]\\
&=
\E[Y\mid A=1,G=O]-\E[Y\mid A=0,G=O]\\
&\quad+\E[\E[M\mid X,A=1,G=E]\mid G=O]-\E[\E[M\mid X,A=0,G=E]\mid G=O]\\
&\quad-\E[M\mid A=1,G=O]+\E[M\mid A=0,G=O].
\end{align*}
\end{proof}

\begin{proof}[Proof of Theorem \ref{thm:equiETTCond}]
	
By Lemma \ref{lem:idcond} and Assumption \ref{assumption:CondEqui}, 
\begin{align*}
-\E[&Y^{(0)}\mid  X,A=1,G=O]\\
&=\E[M\mid X,A=0,G=O]
-\E[Y\mid  X,A=0,G=O]
-\E[M^{(0)}\mid  X,A=1,G=O]\\
&=\E[M\mid X,A=0,G=O]
-\E[Y\mid  X,A=0,G=O]\\
&\quad-\frac{\E[M\mid X,A=0,G=E]-\E[M\mid X,A=0, G=O]p(A=0\mid X,G=O)}{p(A=1\mid X,G=O)}\\
&=\frac{1}{p(A=1\mid X,G=O)}\E[M\mid X,A=0,G=O]\\
&\quad-\frac{1}{p(A=1\mid X,G=O)}\E[M\mid X,A=0,G=E]-\E[Y\mid  X,A=0,G=O].
\end{align*}

Therefore,
\begin{align*}
\theta_{\text{ETT}}
&=\E[Y^{(1)}\mid A=1,G=O]-\E[Y^{(0)}\mid A=1,G=O]\\
&=\E[Y\mid A=1,G=O]+\E[-\E[Y^{(0)}\mid X,A=1,G=O]\mid A=1,G=O]\\
&=\E[Y\mid A=1,G=O]
+\E[\frac{1}{p(A=1\mid X,G=O)}\E[M\mid X,A=0,G=O]\mid A=1,G=O]\\
&\quad-\E[\frac{1}{p(A=1\mid X,G=O)}\E[M\mid X,A=0,G=E]+\E[Y\mid  X,A=0,G=O]\mid A=1,G=O].
\end{align*}

\end{proof}

\begin{proof}[Proof of Theorem \ref{thm:equiATECond}]
As seen in the proof of Theorem \ref{thm:equiETTCond},
\begin{align*}
\E[&Y^{(0)}\mid  X,A=1,G=O]\\
&=-\frac{1}{p(A=1\mid X,G=O)}\E[M\mid X,A=0,G=O]\\
&\quad+\frac{1}{p(A=1\mid X,G=O)}\E[M\mid X,A=0,G=E]+\E[Y\mid  X,A=0,G=O].
\end{align*}
Similarly,
\begin{align*}
\E[&Y^{(1)}\mid  X,A=0,G=O]\\
&=-\frac{1}{p(A=0\mid X,G=O)}\E[M\mid X,A=1,G=O]\\
&\quad+\frac{1}{p(A=0\mid X,G=O)}\E[M\mid X,A=1,G=E]+\E[Y\mid  X,A=1,G=O].
\end{align*}

Therefore,
\begin{equation}
\label{eq:thm61}
\begin{aligned}
\E[Y^{(1)}\mid G=O]
&=\E[\E[Y^{(1)}\mid X,G=O]\mid G=O]\\
&=\E[\E[Y\mid X,A=1,G=O]p(A=1\mid X,G=O)\mid G=O]\\
&\quad+\E[\E[Y^{(1)}\mid X,A=0,G=O]p(A=0\mid X,G=O)\mid G=O]\\
&=\E[\E[Y\mid X,A=1,G=O]p(A=1\mid X,G=O)\mid G=O]\\
&\quad+\E[\E[M\mid X,A=1,G=E]-\E[M\mid X,A=1,G=O]\\
&\quad\quad\quad+\E[Y\mid X,A=1,G=O]p(A=0\mid X,G=O)\mid G=O].
\end{aligned}
\end{equation}
Similarly,
\begin{equation}
\label{eq:thm62}
\begin{aligned}
\E[Y^{(0)}\mid G=O]
&=\E[\E[Y^{(0)}\mid X,G=O]\mid G=O]\\
&=\E[\E[Y^{(0)}\mid X,A=1,G=O]p(A=1\mid X,G=O)\mid G=O]\\
&\quad+\E[\E[Y\mid X,A=0,G=O]p(A=0\mid X,G=O)\mid G=O]\\
&=\E[\E[M\mid X,A=0,G=E]-\E[M\mid X,A=0,G=O]\\
&\quad\quad\quad+\E[Y\mid X,A=0,G=O]p(A=1\mid X,G=O)\mid G=O]\\
&\quad+\E[\E[Y\mid X,A=0,G=O]p(A=0\mid X,G=O)\mid G=O].
\end{aligned}
\end{equation}
\eqref{eq:thm61} and \eqref{eq:thm62} conclude that 
\begin{align*}
\theta_{\text{ATE}} &=\E[Y^{(1)}\mid G=O]-\E[Y^{(0)}\mid G=O]\\
&=\E\big[
\E[Y\mid X,A=1,G=O]-\E[Y\mid X,A=0,G=O]\\
&\quad\quad+\E[M\mid X,A=1,G=E]-\E[M\mid X,A=0,G=E]\\
&\quad\quad+\E[M\mid X,A=0,G=O]-\E[M\mid X,A=1,G=O]
\big| G=O\big].
\end{align*}
\end{proof}

\begin{proof}[Proof of Theorem \ref{thm:equiETTQQ}]
By Assumption \ref{assumption:QQEqui},
\[
F_{Y^{(0)}\mid A=0,X,G=O}\circ F^{-1}_{Y^{(0)}\mid A=1,X,G=O}(v)
=F_{M^{(0)}\mid A=0,X,G=O}\circ F^{-1}_{M^{(0)}\mid A=1,X,G=O}(v),
\]	
which implies that
\[
F^{-1}_{Y^{(0)}\mid A=1,X,G=O}(v)
=F^{-1}_{Y\mid A=0,X,G=O}\circ F_{M\mid A=0,X,G=O}\circ F^{-1}_{M^{(0)}\mid A=1,X,G=O}(v).
\]
Note that
\begin{align*}
&F^{-1}_{Y^{(0)}\mid A=1,X,G=O}(v)=y\\
&\Rightarrow F_{Y^{(0)}\mid A=1,X,G=O}(y)=v,
\end{align*}
and
\begin{align*}
&F^{-1}_{Y\mid A=0,X,G=O}\circ F_{M\mid A=0,X,G=O}\circ F^{-1}_{M^{(0)}\mid A=1,X,G=O}(v)=y\\
&\Rightarrow 
F_{M^{(0)}\mid A=1,X,G=O}\circ F^{-1}_{M\mid A=0,X,G=O} \circ F_{Y\mid A=0,X,G=O}(y)=v.
\end{align*}
Therefore,
\[
F_{Y^{(0)}\mid A=1,X,G=O}(y)=F_{M^{(0)}\mid A=1,X,G=O}\circ F^{-1}_{M\mid A=0,X,G=O} \circ F_{Y\mid A=0,X,G=O}(y).
\]
Finally, we note that
\begin{align*}
F_{M^{(0)}\mid X,G=O}(m)
&=F_{M^{(0)}\mid X,A=1, G=O}(m)p(A=1\mid X,G=O)\\
&\quad+F_{M^{(0)}\mid X,A=0, G=O}(m)p(A=0\mid X,G=O),
\end{align*}
and hence,
\begin{align*}
F_{M^{(0)}\mid X,A=1, G=O}(m)
&=\frac{F_{M^{(0)}\mid X,G=O}(m)-F_{M\mid X,A=0, G=O}(m)p(A=0\mid X,G=O)}{p(A=1\mid X,G=O)}\\
&\overset{A\ref{assumption:ExVal}}{=}\frac{F_{M^{(0)}\mid X,G=E}(m)-F_{M\mid X,A=0, G=O}(m)p(A=0\mid X,G=O)}{p(A=1\mid X,G=O)}\\
&\overset{A\ref{assumption:IntVal}}{=}\frac{F_{M^{(0)}\mid X,A=0,G=E}(m)-F_{M\mid X,A=0, G=O}(m)p(A=0\mid X,G=O)}{p(A=1\mid X,G=O)}\\
&=\frac{F_{M\mid X,A=0,G=E}(m)-F_{M\mid X,A=0, G=O}(m)p(A=0\mid X,G=O)}{p(A=1\mid X,G=O)}.
\end{align*}

This concludes that
\begin{align*}
F_{Y^{(0)}\mid A=1,X,G=O}(y)
&=\frac{F_{M\mid X,A=0,G=E}\circ F^{-1}_{M\mid A=0,X,G=O} \circ F_{Y\mid A=0,X,G=O}(y)}{p(A=1\mid X,G=O)}\\
&\quad-\frac{p(A=0\mid X,G=O)}{p(A=1\mid X,G=O)} F_{Y\mid A=0,X,G=O}(y).
\end{align*}

\end{proof}

\subsection{Proofs of Sections \ref{sec:bsiv} and \ref{sec:bsiv::supp}}

\begin{proof}[Proof of Theorem \ref{thm:BSIV_ett} and \ref{thm:BSIV_ATE}]
Define
\begin{align*}
&b_0(X)\coloneqq \E[Y^{(0)}\mid B=0,X,G=O],\\
&b_1(X)\coloneqq \E[Y^{(1)}\mid B=0,X,G=O],\\
&\omega_0(X)\coloneqq \E[M\mid A=0,B=0,X,G=E],\\
&\omega_1(X)\coloneqq \E[M\mid A=1,B=0,X,G=E],\\
&\beta_0(B,X)\coloneqq\E[\{Y^{(1)}-M^{(1)}\}-\{Y^{(0)}-M^{(0)}\}\mid A=0,B,X,G=O],\\
&\beta_1(B,X)\coloneqq\E[\{Y^{(1)}-M^{(1)}\}-\{Y^{(0)}-M^{(0)}\}\mid A=1,B,X,G=O],\\
&\gamma_0(B,X)\coloneqq \E[Y^{(0)}-M^{(0)}\mid A=1,B,X,G=O]-\E[Y^{(0)}-M^{(0)}\mid A=0,B,X,G=O],\\
&\gamma_1(B,X)\coloneqq \E[Y^{(1)}-M^{(1)}\mid A=1,B,X,G=O]-\E[Y^{(1)}-M^{(1)}\mid A=0,B,X,G=O],\\
&\pi(B,X)\coloneqq p(A=1\mid B,X,G=O).
\end{align*}
We note that by Assumptions \ref{assumption:IntVal}, \ref{assumption:ExVal}, and \ref{assumption:BSEqui::supp}(ii), we have
\begin{align*}
\E[Y^{(0)}-M^{(0)}\mid B,X,G=O]
&\overset{A\ref{assumption:BSEqui::supp}(ii)}{=}\E[Y^{(0)}-M^{(0)}\mid B=0,X,G=O]\\
&=b_0(X)-\E[M^{(0)}\mid B=0,X,G=O]\\
&\overset{A\ref{assumption:ExVal}}{=}b_0(X)-\E[M^{(0)}\mid B=0,X,G=E]\\
&\overset{A\ref{assumption:IntVal}}{=}b_0(X)-\E[M^{(0)}\mid A=0,B=0,X,G=E]\\
&=b_0(X)-\omega_0(X).
\end{align*}
Similarly, by Assumptions \ref{assumption:IntVal}, \ref{assumption:ExVal}, and \ref{assumption:BSEqui::supp}(i), we have
\begin{align*}
\E[Y^{(1)}-M^{(1)}\mid B,X,G=O]
&=b_1(X)-\omega_1(X).
\end{align*}
Note that $\omega_0(X)$ and $\omega_1(X)$ are identified.	

Using a nonparametric reparametrization of the outcome conditional mean function similar to \citep{robins1994correcting,tchetgen2013alternative}, we have
\begin{equation}
\label{eq:thm9-1}
\begin{aligned}
&\E[Y-M\mid A=a,B,X,G=O]\\
&=\E[Y^{(a)}-M^{(a)}\mid A=a,B,X,G=O]-\E[Y^{(0)}-M^{(0)}\mid A=a,B,X,G=O]\\
&\quad+\E[Y^{(0)}-M^{(0)}\mid A=a,B,X,G=O]-\E[Y^{(0)}-M^{(0)}\mid A=0,B,X,G=O]\\
&\quad-\{\E[Y^{(0)}-M^{(0)}\mid A=1,B,X,G=O]\\
&\quad\quad-\E[Y^{(0)}-M^{(0)}\mid A=0,B,X,G=O]\}p(A=1\mid B,X,G=O)\\
&\quad+\E[Y^{(0)}-M^{(0)}\mid B,X,G=O]\\
&=\E[\{Y^{(1)}-M^{(1)}\}-\{Y^{(0)}-M^{(0)}\}\mid A=1,B,X,G=O]a\\
&\quad+\{\E[Y^{(0)}-M^{(0)}\mid A=1,B,X,G=O]-\E[Y^{(0)}-M^{(0)}\mid A=0,B,X,G=O]\}\\
&\quad\times\{a-p(A=1\mid B,X,G=O)\}\\
&\quad+\E[Y^{(0)}-M^{(0)}\mid B,X,G=O]\\
&=\beta_1(B,X)a+\gamma_0(B,X)\{a-\pi(B,X)\}+b_0(X)-\omega_0(X).
\end{aligned}
\end{equation}

Note that for every fixed $X$, the left hand side identifies 4 parameters.
Under Assumption \ref{assumption:homogeneityETT::supp}$(ii)$, $\beta_1$ is not a function of $B$. Therefore, we also have 4 unknown parameters on the right hand side: 1 corresponding to $\beta_1$, 2 corresponding to $\gamma_0$, and 1 corresponding to $b_0$.
Similarly, under Assumption \ref{assumption:homogeneityBias::supp}$(ii)$, $\gamma_0$ is not a function of $B$. Therefore, we also have 4 unknown parameters on the right hand side: 2 corresponding to $\beta_1$, 1 corresponding to $\gamma_0$, and 1 corresponding to $b_0$.
Therefore, under either of these two assumptions, the parameter $\beta_1(B,X)$ is identified. 

Formally, define
\begin{align*}
&E_{ab}^O(X)\coloneqq\E[Y-M\mid A=a,B=b,X,G=O],\\
&P_{ab}^O(X)\coloneqq p(A=a\mid B=b,X,G=O).	
\end{align*}

From \eqref{eq:thm9-1}, we have
\begin{align*}
&E_{00}^O(X)=-P_{10}^O(X)\gamma_0(0,X)+b_0(X)-\omega_0(X),\\
&E_{01}^O(X)=-P_{11}^O(X)\gamma_0(1,X)+b_0(X)-\omega_0(X),\\
&E_{10}^O(X)=\beta_1(0,X)+P_{00}^Ov(X)\gamma_0(0,X)+b_0(X)-\omega_0(X),\\
&E_{11}^O(X)=\beta_1(1,X)+P_{01}^O(X)\gamma_0(1,X)+b_0(X)-\omega_0(X).
\end{align*}

Under Assumption \ref{assumption:homogeneityETT::supp}$(ii)$, $\beta_1(X)\coloneqq\beta_1(0,X)=\beta_1(1,X)$, which can be obtained as follows. Noting that
\begin{align*}
\E[Y-M\mid B=1,X,G=O]
&=E_{11}^O(X)P_{11}^O(X)+E_{01}^O(X)P_{01}^O(X)\\
&=P_{11}^O(X)\beta_1(X)+b_0(X)-\omega_0(X),
\end{align*}
and
\begin{align*}
\E[Y-M\mid B=0,X,G=O]
&=E_{10}^O(X)P_{10}^O(X)+E_{00}^O(X)P_{00}^O(X)\\
&=P_{10}^O(X)\beta_1(X)+b_0(X)-\omega_0(X),
\end{align*}
we have
\begin{align*}
\beta_1(X)=\frac{\E[Y-M\mid B=1,X,G=O]-\E[Y-M\mid B=0,X,G=O]}{P_{11}^O(X)-P_{10}^O(X)}.
\end{align*}
Note that
\begin{align*}
\beta_1(B,X)
&=\E[Y^{(1)}-Y^{(0)}\mid A=1,B,X,G=O]\\
&\quad-\E[M\mid A=1,B,X,G=O]\\
&\quad+\E[M^{(0)}\mid A=1,B,X,G=O].
\end{align*}	
By Lemma \ref{lem:idcond}, the last term is identified. Therefore, the first term on the right hand side is identified, and hence, the parameter $\theta_{ETT}$ is also identified as follows.
\begin{align*}
\theta_{ETT}
&=\E\Big[
\frac{\E[Y-M\mid B=1,X,G=O]-\E[Y-M\mid B=0,X,G=O]}{P_{11}^O(X)-P_{10}^O(X)}\\
&\qquad+\E[M\mid A=1,B,X,G=O]-\frac{\E[M\mid A=0,B,X,G=E]}{p(A=1\mid B,X,G=O)}\\
&\qquad+\frac{\E[M\mid A=0,B,X,G=O]p(A=0\mid B,X,G=O)}{p(A=1\mid B,X,G=O)}\Big|A=1, G=O\Big].
\end{align*}
	
Under Assumption \ref{assumption:homogeneityBias::supp}$(ii)$, $\gamma_0(X)\coloneqq\gamma_0(0,X)=\gamma_0(1,X)$, which can be obtained as 
\[
\gamma_0(X)=\frac{E_{01}^O(X)-E_{00}^O(X)}{P_{01}^O(X)-P_{00}^O(X)}.
\]
We also note that
\begin{align*}
b_0(X)-\omega_0(X)
&=E_{00}^O(X)+P_{10}^O(X)\gamma_0(X)\\
&=E_{01}^O(X)+P_{11}^O(X)\gamma_0(X),
\end{align*}
by which we have
\begin{align*}
\beta_1(0,X)
&=E_{10}^O(X)-P_{00}^O(X)\gamma_0(X)-E_{00}^O(X)-P_{10}^O(X)\gamma_0(X)\\
&=E_{10}^O(X)-E_{00}^O(X)-\gamma_0(X),
\end{align*}
and	
\begin{align*}
\beta_1(1,X)
&=E_{11}^O(X)-P_{01}^O(X)\gamma_0(X)-E_{01}^O(X)-P_{11}^O(X)\gamma_0(X)\\
&=E_{11}^O(X)-E_{01}^O(X)-\gamma_0(X),
\end{align*}	
which implies that
\[
\beta_1(B,X)=\{E_{11}^O(X)-E_{01}^O(X)-E_{10}^O(X)+E_{00}^O(X)\}B+E_{10}^O(X)-E_{00}^O(X)-\gamma_0(X).
\]
Note that
\begin{align*}
\beta_1(B,X)
&=\E[Y^{(1)}-Y^{(0)}\mid A=1,B,X,G=O]\\
&\quad-\E[M\mid A=1,B,X,G=O]\\
&\quad+\E[M^{(0)}\mid A=1,B,X,G=O].
\end{align*}	
By Lemma \ref{lem:idcond}, the last term is identified. Therefore, the first term on the right hand side is identified, and hence, the parameter $\theta_{ETT}$ is also identified as follows.
\begin{align*}
\theta_{ETT}
&=\E\Big[
\{E_{11}^O(X)-E_{01}^O(X)-E_{10}^O(X)+E_{00}^O(X)\}B+E_{10}^O(X)-E_{00}^O(X)\\
&\qquad-\frac{E_{01}^O(X)-E_{00}^O(X)}{P_{01}^O(X)-P_{00}^O(X)}+\E[M\mid A=1,B,X,G=O]-\frac{\E[M\mid A=0,B,X,G=E]}{p(A=1\mid B,X,G=O)}\\
&\qquad+\frac{\E[M\mid A=0,B,X,G=O]p(A=0\mid B,X,G=O)}{p(A=1\mid B,X,G=O)}\Big|A=1,G=O\Big].
\end{align*}

In order to show the identifiability of $\theta_{ATE}$, we note that using a similar parametrization, we have
\begin{align*}
&\E[Y-M\mid A=a,B,X,G=O]\\
&=\beta_0(B,X)(a-1)+\gamma_1(B,X)\{a-\pi(B,X)\}+b_1(X)-\omega_1(X).
\end{align*}
A similar argument regarding counting the parameters as before shows that under either Assumption \ref{assumption:homogeneityETT::supp}$(i)$, or \ref{assumption:homogeneityBias::supp}$(i)$,
	 the parameter $\beta_0(B,X)$ and consequently, the parameter $\E[Y^{(1)}-Y^{(0)}\mid A=0,B,X,G=O]$	 
	  are identified. 
	  
Formally, we have
\begin{align*}
&E_{00}^O(X)=-\beta_0(0,X)-P_{10}^O(X)\gamma_1(0,X)+b_1(X)-\omega_1(X),\\
&E_{01}^O(X)=-\beta_0(1,X)-P_{11}^O(X)\gamma_1(1,X)+b_1(X)-\omega_1(X),\\
&E_{10}^O(X)=P_{00}^O(X)\gamma_1(0,X)+b_1(X)-\omega_1(X),\\
&E_{11}^O(X)=P_{01}^O(X)\gamma_1(1,X)+b_1(X)-\omega_1(X).
\end{align*}

Under Assumption \ref{assumption:homogeneityETT::supp}$(i)$, $\beta_0(X)\coloneqq\beta_0(0,X)=\beta_0(1,X)$, which can be obtained similar to the previous case as	  
\begin{align*}
\beta_0(X)=\frac{\E[Y-M\mid B=1,X,G=O]-\E[Y-M\mid B=0,X,G=O]}{P_{11}^O(X)-P_{10}^O(X)}.
\end{align*}	  
Note that
\begin{align*}
\beta_0(B,X)
&=\E[Y^{(1)}-Y^{(0)}\mid A=0,B,X,G=O]\\
&\quad+\E[M\mid A=0,B,X,G=O]\\
&\quad-\E[M^{(1)}\mid A=0,B,X,G=O].
\end{align*}	
By Lemma \ref{lem:idcond}, the last term is identified. Therefore, the first term on the right hand side is identified. The parameter $\E[Y^{(1)}-Y^{(0)}\mid A=0,B,X,G=O]$ is identified as follows.
\begin{align*}
\E[Y^{(1)}-&Y^{(0)}\mid A=0,B,X,G=O]\\
&=
\frac{\E[Y-M\mid B=1,X,G=O]-\E[Y-M\mid B=0,X,G=O]}{P_{11}^O(X)-P_{10}^O(X)}\\
&\qquad-\E[M\mid A=0,B,X,G=O]+\frac{\E[M\mid A=1,B,X,G=E]}{p(A=0\mid B,X,G=O)}\\
&\qquad-\frac{\E[M\mid A=1,B,X,G=O]p(A=1\mid B,X,G=O)}{p(A=0\mid B,X,G=O)}.
\end{align*}	  
  
Finally, note that
\begin{align*}
\theta_{ATE}&=\E\big[\E[Y^{(1)}-Y^{(0)}\mid B,X,G=O]\big| G=O\big]\\
&=\E\big[\E[Y^{(1)}-Y^{(0)}\mid A=1,B,X,G=O]\pi(B,X)\\
&\quad+\E[Y^{(1)}-Y^{(0)}\mid A=0,B,X,G=O](1-\pi(B,X))\big|G=O\big]\\
&=\E\Big[
\frac{\E[Y-M\mid B=1,X,G=O]-\E[Y-M\mid B=0,X,G=O]}{P_{11}^O(X)-P_{10}^O(X)}\\
&\qquad+\E[M\mid A=1,B,X,G=E]-\E[M\mid A=0,B,X,G=E]\Big|G=O\Big].
\end{align*}	

Under Assumption \ref{assumption:homogeneityBias::supp}$(i)$, $\gamma_1(X)\coloneqq\gamma_1(0,X)=\gamma_1(1,X)$, which can be obtained as
\[
\gamma_1(X)=\frac{E_{11}^O(X)-E_{10}^O(X)}{P_{01}^O(X)-P_{00}^O(X)}.
\]
This by an approach similar to the previous case leads to
\[
\beta_0(B,X)=\{E_{11}^O(X)-E_{01}^O(X)-E_{10}^O(X)+E_{00}^O(X)\}B+E_{10}^O(X)-E_{00}^O(X)-\gamma_1(X).
\]
Note that
\begin{align*}
\beta_0(B,X)
&=\E[Y^{(1)}-Y^{(0)}\mid A=0,B,X,G=O]\\
&\quad+\E[M\mid A=0,B,X,G=O]\\
&\quad-\E[M^{(1)}\mid A=0,B,X,G=O].
\end{align*}	
By Lemma \ref{lem:idcond}, the last term is identified. Therefore, the first term on the right hand side is identified. The parameter $\E[Y^{(1)}-Y^{(0)}\mid A=0,B,X,G=O]$ is identified as follows.
\begin{align*}
\E[Y^{(1)}-&Y^{(0)}\mid A=0,B,X,G=O]\\
&=
\{E_{11}^O(X)-E_{01}^O(X)-E_{10}^O(X)+E_{00}^O(X)\}B+E_{10}^O(X)-E_{00}^O(X)\\
&\qquad-\frac{E_{11}^O(X)-E_{10}^O(X)}{P_{01}^O(X)-P_{00}^O(X)}-\E[M\mid A=0,B,X,G=O]+\frac{\E[M\mid A=1,B,X,G=E]}{p(A=0\mid B,X,G=O)}\\
&\qquad-\frac{\E[M\mid A=1,B,X,G=O]p(A=1\mid B,X,G=O)}{p(A=0\mid B,X,G=O)}.
\end{align*}	
Finally, note that
\begin{align*}
\theta_{ATE}&=\E\big[\E[Y^{(1)}-Y^{(0)}\mid B,X,G=O]\big|G=O\big]\\
&=\E\big[\E[Y^{(1)}-Y^{(0)}\mid A=1,B,X,G=O]\pi(B,X)\\
&\quad+\E[Y^{(1)}-Y^{(0)}\mid A=0,B,X,G=O](1-\pi(B,X))\big]\\
&=\E\Big[
\{E_{11}^O(X)-E_{01}^O(X)-E_{10}^O(X)+E_{00}^O(X)\}B+E_{10}^O(X)-E_{00}^O(X)\\
&\qquad-\frac{\{E_{01}^O(X)-E_{00}^O(X)\}\pi(B,X)+\{E_{11}^O(X)-E_{10}^O(X)\}(1-\pi(B,X))}{P_{01}^O(X)-P_{00}^O(X)}\\
&\qquad+\E[M\mid A=1,B,X,G=E]-\E[M\mid A=0,B,X,G=E]\Big|G=O\Big].
\end{align*}	

\end{proof}

\subsection{Proofs of Sections \ref{sec:proximal} and \ref{sec:proximal::supp}}

\begin{proof}[Proof of Theorems \ref{thm:POR} and \ref{thm:POR::supp}]
For identifying ATE, we show that the parameter $\E[Y^{(a)}\mid G=O]$ is identified.
By Assumption \ref{assumption:compexist1} $(ii)$, for any choice of $Z$, $A=a$, $X$, we have
\begin{align*}
&\E[Y\mid Z,A=a,X,G=O]=\E[h(M,a,X)\mid Z,A=a,X,G=O]\\
\Rightarrow~
&\E[\E[Y\mid U,Z,A=a,X,G=O]\mid Z,A=a,X,G=O]\\
&\quad\quad=\E[\E[h(M,a,X)\mid U,Z,A=a,X,G=O]\mid Z,A=a,X,G=O]\\
\overset{A\ref{assumption:proxy}}{\Rightarrow}~
&\E[\E[Y\mid U,A=a,X,G=O]\mid Z,A=a,X,G=O]\\
&\quad\quad=\E[\E[h(M,a,X)\mid U,A=a,X,G=O]\mid Z,A=a,X,G=O]\\
\overset{A\ref{assumption:compexist1}}{\Rightarrow}~
&\E[Y\mid U,A=a,X,G=O]=\E[h(M,a,X)\mid U,A=a,X,G=O]\\
\Rightarrow~
&\E[\E[Y\mid U,A=a,X,G=O]\mid G=O]=\E[\E[h(M,a,X)\mid U,A=a,X,G=O]\mid G=O]\\
\Rightarrow~
&\E[\E[Y^{(a)}\mid U,A=a,X,G=O]\mid G=O]=\E[\E[h(M^{(a)},a,X)\mid U,A=a,X,G=O]\mid G=O]\\
\Rightarrow~
&\E[\E[Y^{(a)}\mid U,X,G=O]\mid G=O]=\E[\E[h(M^{(a)},a,X)\mid U,X,G=O]\mid G=O]\\
\Rightarrow~
&\E[Y^{(a)}\mid G=O]=\E[h(M^{(a)},a,X)\mid G=O].
\end{align*}
Therefore, we have
\begin{equation}
\label{eq:proxproofh1}
\begin{aligned}
\E[Y^{(a)}\mid G=O]
&=\E[\E[h(M^{(a)},a,X)\mid X,G=O]\mid G=O]\\
&\overset{A\ref{assumption:ExVal}}{=}
\E[\E[h(M^{(a)},a,X)\mid X,G=E]\mid G=O]\\	
&\overset{A\ref{assumption:IntVal}}{=}
\E[\E[h(M^{(a)},a,X)\mid A=a,X,G=E]\mid G=O]\\
&=\E[\E[h(M,A,X)\mid A=a,X,G=E]\mid G=O],
\end{aligned}
\end{equation}
which concludes the desired result.

Realizing that $\E[Y^{(0)}\mid G=O]$ is identified, it is easy to see that $\E[Y^{(0)}\mid A=1, G=O]$ is also identified as 
\begin{align*}
&\E[Y^{(0)}\mid G=O]\\
&=\E[Y^{(0)}\mid A=1, G=0]p(A=1\mid G=O)
+\E[Y^{(0)}\mid A=0, G=O]p(A=0\mid G=O)\\
&\Rightarrow
\E[Y^{(0)}\mid A=1, G=O]
=\frac{\E[Y^{(0)}\mid G=O]-\E[Y\mid A=0, G=O]p(A=0\mid G=O)}{p(A=1\mid G=O)}.
\end{align*}
Therefore, $\theta_{ETT}=\E[Y\mid A=1, G=O]-\E[Y^{(0)}\mid A=1, G=O]$ is also identified.

\end{proof}

\begin{proof}[Proof of Theorems \ref{thm:PIPW} and \ref{thm:PIPW::supp}]
For identifying ATE, we show that the parameter $\E[Y^{(a)}\mid G=O]$ is identified.
By Assumption \ref{assumption:compexist2} $(ii)$, for any choice of $M$, $A=a$, $X$, we have
\begin{align*}
\E[q(Z,A,X)\mid & M=m,A=a,X,G=O]\\
&=\frac{p(M=m\mid A=a,X,G=E)}{p(M=m\mid A=a,X,G=O)p(A=a\mid X,G=O)}\\
&=\frac{p(M^{(a)}=m\mid A=a,X,G=E)}{p(M^{(a)}=m\mid A=a,X,G=O)p(A=a\mid X,G=O)}\\
&\overset{A\ref{assumption:IntVal}}{=}\frac{p(M^{(a)}=m\mid X,G=E)}{p(M^{(a)}=m\mid A=a,X,G=O)p(A=a\mid X,G=O)}\\
&\overset{A\ref{assumption:ExVal}}{=}\frac{p(M^{(a)}=m\mid X,G=O)}{p(M^{(a)}=m\mid A=a,X,G=O)p(A=a\mid X,G=O)}\\
&=\frac{1}{p(A=a\mid M^{(a)}=m,X,G=O)}.
\end{align*}
Hence,
\begingroup
\allowdisplaybreaks
\begin{align*}
&\sum_z q(z,a,x)p(z\mid M=m,a,x,G=O)
=\frac{1}{p(A=a\mid M^{(a)}=m,x,G=O)}\\
&\overset{A\ref{assumption:proxy}}{\Rightarrow} \sum_u\sum_z q(z,a,x)p(z\mid u,a,x,G=O)p(u\mid M=m,a,x,G=O)\\
&\quad=\sum_u\frac{1}{p(A=a\mid M^{(a)}=m,x,G=O)}p(u\mid M^{(a)},X,G=O)\\
&\quad=\sum_u\frac{p(u\mid A=a,M^{(a)}=m,x,G=O)}{p(A=a\mid M^{(a)}=m,x,G=O)p(u\mid A=a,M^{(a)}=m,x,G=O)}p(u\mid M^{(a)},X,G=O)\\
&\quad=\sum_u\frac{1}{p(A=a\mid u,M^{(a)}=m,x,G=O)}p(u\mid A=a,M^{(a)}=m,x,G=O)\\
&\quad=\sum_u\frac{1}{p(A=a\mid u,M^{(a)}=m,x,G=O)}p(u\mid A=a,M=m,x,G=O)\\
&\quad=\sum_u\frac{1}{p(A=a\mid u,x,G=O)}p(u\mid A=a,M=m,x,G=O),
\end{align*}
\endgroup
where the last equality is due to the conditional independence $A\independent M^{(a)}\mid\{X,U,G=O\}$. Note that here we are conditioning on $U$, all possible latent confounders.

Therefore, by Assumption \ref{assumption:compexist2} $(i)$, we have
\begin{align*}
\sum_z q(z,a,x)p(z\mid u,a,x,G=O)
=\frac{1}{p(A=a\mid u,x,G=O)}.
\end{align*}

Therefore, we have
\begin{align*}
\E[Y^{(a)}\mid G=O]
&=\sum_y y p(Y^{(a)}=y\mid G=O)\\
&=\sum_{y,u,x} y p(Y^{(a)}=y\mid u,x,G=O)p(u,x,\mid G=O)\\
&=\sum_{y,u,x} y p(Y^{(a)}=y\mid a,u,x,G=O)p(u,x,\mid G=O)\\
&=\sum_{y,u,x} y p(Y=y\mid a,u,x,G=O)\frac{p(a\mid u,x,G=O)}{p(a\mid u,x,G=O)}p(u,x,\mid G=O)\\
&=\sum_{z,y,u,x} yq(z,a,x) p(z\mid a,u,x,G=O)p(y\mid a,u,x,G=O)p(a,u,x,\mid G=O)\\
&\overset{A\ref{assumption:proxy}}{=}\sum_{z,y,u,x} yq(z,a,x) p(y,z\mid a,u,x,G=O)p(a,u,x,\mid G=O)\\
&=\sum_{\tilde{a},z,y,u,x} I(\tilde{a}=a)yq(z,\tilde{a},x) p(y,z,\tilde{a},u,x\mid G=O)\\
&=\sum_{\tilde{a},z,y,x} I(\tilde{a}=a)yq(z,\tilde{a},x) p(y,z,\tilde{a},x\mid G=O)\\
&=\E[I(A=a)Yq(Z,A,X)\mid G=O].
\end{align*}
which concludes the desired result.

Realizing that $\E[Y^{(0)}\mid G=O]$ is identified, it is easy to see that $\E[Y^{(0)}\mid A=1, G=O]$ is also identified as 
\begin{align*}
&\E[Y^{(0)}\mid G=O]\\
&=\E[Y^{(0)}\mid A=1, G=0]p(A=1\mid G=O)
+\E[Y^{(0)}\mid A=0, G=O]p(A=0\mid G=O)\\
&\Rightarrow
\E[Y^{(0)}\mid A=1, G=O]
=\frac{\E[Y^{(0)}\mid G=O]-\E[Y\mid A=0, G=O]p(A=0\mid G=O)}{p(A=1\mid G=O)}.
\end{align*}
Therefore, $\theta_{ETT}=\E[Y\mid A=1, G=O]-\E[Y^{(0)}\mid A=1, G=O]$ is also identified.

\end{proof}

\subsection{Proofs of Sections \ref{sec:estimation} and \ref{sec:estimation::supp}}

\begin{proof}[Proof of Theorem \ref{thm:IF:equi-conf:ETT}]
	
Define
\begin{align*}
&\psi_1=
\E[Y\mid A=1,G=O],\\
&\psi_2=
\E\big[\frac{1}{p(A=1\mid X,G=O)}\E[M\mid X,A=0,G=O]~\big|~ A=1,G=O\big],\\
&\psi_3=
\E\big[\frac{1}{p(A=1\mid X,G=O)}\E[M\mid X,A=0,G=E]~\big|~ A=1,G=O\big],\\
&\psi_4=
\E\big[\E[Y\mid  X,A=0,G=O]~\big|~ A=1,G=O\big].
\end{align*}
We use the notation $\partial_tf(t)$ to denote $\frac{\partial f(t)}{\partial t}\big|_{t=0}$.
For parameter $\psi$, let $\psi_t$ be the parameter under a regular parametric sub-model indexed by $t$, that includes the ground-truth model at $t=0$.
Let $V$ be the set of all observed variables.
In order to obtain an influence function, we need to find a random variable $\Gamma$ with mean zero, that satisfies
\[
\partial_t\psi_t=\E[\Gamma S(V)],
\]
where $S(V)=\partial_t\log p_t(V)$.	

For $\psi_1$, note that
\begin{align*}
\partial_t\psi_{1_t} 
&= \sum_{y} y\partial_t p_t(y\mid A=1,G=O)\\
&= \sum_{y} yS(y\mid A=1,G=O)p(y\mid A=1,G=O)\\
&= \sum_{y,a,g}\frac{I(a=1)I(g=O)}{p(A=1,G=O)} yS(y\mid a,g)p(y,a,g)\\
&= \E\Big[\frac{I(A=1)I(G=O)}{p(A=1,G=O)} YS(Y\mid A,G)\Big]\\
&= \E\Big[\frac{I(A=1)I(G=O)}{p(A=1,G=O)} \{Y-\E[Y\mid A=1,G=O]\}S(Y\mid A,G)\Big]\\
&= \E\Big[\frac{I(A=1)I(G=O)}{p(A=1,G=O)} \{Y-\psi_1\}S(Y\mid A,G)\Big].
\end{align*}
Note that
\begin{align*}
\E\Big[\frac{I(A=1)I(G=O)}{p(A=1,G=O)} \{Y-\psi_1\}S(A,G)\Big]=0.
\end{align*}
Therefore,
\begin{equation*}
\partial_t\psi_{1_t}
=\E\Big[\frac{I(A=1)I(G=O)}{p(A=1,G=O)} \{Y-\psi_1\}S(V)\Big].
\end{equation*}
This implies that 
\begin{align*}
\frac{I(A=1)I(G=O)}{p(A=1,G=O)} \{Y-\psi_1\}	
\end{align*}
is the influence function of $\psi_1$.

For $\psi_2$, note that
\begin{equation}
\label{eq:IF:equi:ett:1}
\begin{aligned}
\partial_t\psi_{2_t} 
&=\partial_t \sum_{m,x} m\frac{1}{p_t(A=1\mid x,G=O)}p_t(m\mid x,A=0,G=O)p_t(x\mid A=1,G=O)\\
&=\partial_t \sum_{m,x} m\frac{1}{p_t(A=1,G=O)}p_t(m\mid x,A=0,G=O)p_t(x,G=O)\\
&= \sum_{m,x} m \partial_t\frac{1}{p_t(A=1,G=O)}p(m\mid x,A=0,G=O)p(x,G=O)\\
&\quad+\sum_{m,x} m \frac{1}{p(A=1,G=O)}\partial_t p_t(m\mid x,A=0,G=O)p(x,G=O)\\
&\quad+\sum_{m,x} m \frac{1}{p(A=1,G=O)}p(m\mid x,A=0,G=O)\partial_t p_t(x,G=O).
\end{aligned}
\end{equation}

For the first term in \eqref{eq:IF:equi:ett:1}, we have
\begin{equation}
\label{eq:IF:equi:ett:2}
\begin{aligned}
&\sum_{m,x} m \partial_t\frac{1}{p_t(A=1,G=O)}p(m\mid x,A=0,G=O)p(x,G=O)\\
&=-\sum_{m,x} m \frac{1}{p(A=1,G=O)}p(m\mid x,A=0,G=O)p(x,G=O)S(A=1,G=O)\\
&=-\psi_2S(A=1,G=O)\\
&=-\E\Big[\frac{I(A=1)I(G=O)}{p(A=1,G=O)}\psi_2S(A,G)\Big]\\
&=-\E\Big[\frac{I(A=1)I(G=O)}{p(A=1,G=O)}\psi_2S(V)\Big].
\end{aligned}
\end{equation}

For the second term in \eqref{eq:IF:equi:ett:1}, we have
\begin{align*}
&\sum_{m,x} m \frac{1}{p(A=1,G=O)}\partial_t p_t(m\mid x,A=0,G=O)p(x,G=O)\\
&=\sum_{m,x} m \frac{1}{p(A=1,G=O)} S(m\mid x,A=0,G=O)p(m\mid x,A=0,G=O)p(x,G=O)\\
&=\sum_{m,x} m \frac{1}{p(A=1,G=O)}\cdot\frac{1}{p(A=0\mid x,G=O)} S(m\mid x,A=0,G=O)p(m,x,A=0,G=O)\\
&=\sum_{m,a,x,g} m \frac{1}{p(A=1,G=O)}\cdot\frac{I(a=0)I(g=O)}{p(A=0\mid x,G=O)} S(m\mid x,a,g)p(m,x,a,g)\\
&=\E\Big[\frac{1}{p(A=1,G=O)}\cdot\frac{I(A=0)I(G=O)}{p(A=0\mid X,G=O)}M S(M\mid X,A,G) \Big]\\
&=\E\Big[\frac{1}{p(A=1,G=O)}\cdot\frac{I(A=0)I(G=O)}{p(A=0\mid X,G=O)}\{M-\E[M\mid X,A,G]\} S(M\mid X,A,G) \Big].
\end{align*}
Note that
\begin{align*}
\E\Big[\frac{1}{p(A=1,G=O)}\cdot\frac{I(A=0)I(G=O)}{p(A=0\mid X,G=O)}\{M-\E[M\mid X,A,G]\} S(X,A,G) \Big]=0.	
\end{align*}
Therefore
\begin{equation}
\label{eq:IF:equi:ett:3}
\begin{aligned}
&\sum_{m,x} m \frac{1}{p(A=1,G=O)}\partial_t p_t(m\mid x,A=0,G=O)p(x,G=O)\\
&=\E\Big[\frac{1}{p(A=1,G=O)}\cdot\frac{I(A=0)I(G=O)}{p(A=0\mid X,G=O)}\{M-\E[M\mid X,A,G]\} S(V) \Big].
\end{aligned}
\end{equation}

For the third term in \eqref{eq:IF:equi:ett:1}, we have
\begin{equation}
\label{eq:IF:equi:ett:4}
\begin{aligned}
&\sum_{m,x} m \frac{1}{p(A=1,G=O)}p(m\mid x,A=0,G=O)\partial_t p_t(x,G=O)\\
&=\sum_{x}\frac{1}{p(A=1,G=O)}\E[M\mid X=x,A=0,G=O] S(x,G=O)p(x,G=O)\\
&=\sum_{x,g}\frac{I(g=O)}{p(A=1,G=O)}\E[M\mid X=x,A=0,G=O] S(x,g)p(x,g)\\
&=\E\Big[\frac{I(G=O)}{p(A=1,G=O)}\E[M\mid X,A=0,G=O] S(X,G)\Big]\\
&=\E\Big[\frac{I(G=O)}{p(A=1,G=O)}\E[M\mid X,A=0,G=O] S(V)\Big].
\end{aligned}
\end{equation}

Combining \eqref{eq:IF:equi:ett:2}-\eqref{eq:IF:equi:ett:4} concludes that
\begin{align*}
\partial_t\psi_{2_t} 
&=\E\Big[
\frac{I(G=O)}{p(A=1,G=O)}\Big\{
\frac{I(A=0)}{p(A=0\mid X,G=O)}\{M-\E[M\mid X,A,G]\}\\
&+\E[M\mid X,A=0,G=O]-I(A=1)\psi_2
\Big\}S(V)
\Big].	
\end{align*}
The variable
\begin{align*}
\frac{I(G=O)}{p(A=1,G=O)}\Big\{&
\frac{I(A=0)}{p(A=0\mid X,G=O)}\{M-\E[M\mid X,A,G]\}\\
&+\E[M\mid X,A=0,G=O]-I(A=1)\psi_2
\Big\}
\end{align*}
is mean zero and hence is the influence function of $\psi_2$.

For $\psi_3$, note that
\begin{equation}
\label{eq:IF:equi:ett:5}
\begin{aligned}
\partial_t\psi_{3_t} 
&=\partial_t \sum_{m,x} m\frac{1}{p_t(A=1\mid x,G=O)}p_t(m\mid x,A=0,G=E)p_t(x\mid A=1,G=O)\\
&=\partial_t \sum_{m,x} m\frac{1}{p_t(A=1,G=O)}p_t(m\mid x,A=0,G=E)p_t(x,G=O)\\
&= \sum_{m,x} m \partial_t\frac{1}{p_t(A=1,G=O)}p(m\mid x,A=0,G=E)p(x,G=O)\\
&\quad+\sum_{m,x} m \frac{1}{p(A=1,G=O)}\partial_t p_t(m\mid x,A=0,G=E)p(x,G=O)\\
&\quad+\sum_{m,x} m \frac{1}{p(A=1,G=O)}p(m\mid x,A=0,G=E)\partial_t p_t(x,G=O).
\end{aligned}
\end{equation}

For the first term in \eqref{eq:IF:equi:ett:5}, we have
\begin{equation}
\label{eq:IF:equi:ett:6}
\begin{aligned}
&\sum_{m,x} m \partial_t\frac{1}{p_t(A=1,G=O)}p(m\mid x,A=0,G=E)p(x,G=O)\\
&=-\sum_{m,x} m \frac{1}{p(A=1,G=O)}p(m\mid x,A=0,G=E)p(x,G=O)S(A=1,G=O)\\
&=-\psi_3S(A=1,G=O)\\
&=-\E\Big[\frac{I(A=1)I(G=O)}{p(A=1,G=O)}\psi_3S(A,G)\Big]\\
&=-\E\Big[\frac{I(A=1)I(G=O)}{p(A=1,G=O)}\psi_3S(V)\Big].
\end{aligned}
\end{equation}

For the second term in \eqref{eq:IF:equi:ett:5}, we have
{\footnotesize \begin{align*}
&\sum_{m,x} m \frac{1}{p(A=1,G=O)}\partial_t p_t(m\mid x,A=0,G=E)p(x,G=O)\\
&=\sum_{m,x} m \frac{1}{p(A=1,G=O)} S(m\mid x,A=0,G=E)p(m\mid x,A=0,G=E)\{\frac{1}{p(G=E\mid x)}-1\}p(x,G=E)\\
&=\sum_{m,x} m \frac{1}{p(A=1,G=O)}\cdot\frac{1}{p(A=0\mid x,G=E)}\{\frac{1}{p(G=E\mid x)}-1\} S(m\mid x,A=0,G=E)p(m,x,A=0,G=E)\\
&=\sum_{m,a,x,g} m \frac{1}{p(A=1,G=O)}\cdot\frac{I(a=0)I(g=E)}{p(A=0\mid x,G=E)}\{\frac{1}{p(G=E\mid x)}-1\} S(m\mid x,a,g)p(m,x,a,g)\\
&=\E\Big[\frac{1}{p(A=1,G=O)}\cdot\frac{I(A=0)I(G=E)}{p(A=0\mid X,G=E)}\{\frac{1}{p(G=E\mid X)}-1\}M S(M\mid X,A,G) \Big]\\
&=\E\Big[\frac{1}{p(A=1,G=O)}\cdot\frac{I(A=0)I(G=E)}{p(A=0\mid X,G=E)}\{\frac{1}{p(G=E\mid X)}-1\}\{M-\E[M\mid X,A,G]\} S(M\mid X,A,G) \Big].
\end{align*}}
Note that
{\footnotesize \begin{align*}
\E\Big[\frac{1}{p(A=1,G=O)}\cdot\frac{I(A=0)I(G=E)}{p(A=0\mid X,G=E)}\{\frac{1}{p(G=E\mid X)}-1\}\{M-\E[M\mid X,A,G]\} S(X,A,G) \Big]=0.	
\end{align*}}
Therefore
\begin{equation}
\label{eq:IF:equi:ett:7}
\begin{aligned}
&\sum_{m,x} m \frac{1}{p(A=1,G=O)}\partial_t p_t(m\mid x,A=0,G=E)p(x,G=O)\\
&=\E\Big[\frac{1}{p(A=1,G=O)}\cdot\frac{I(A=0)I(G=E)}{p(A=0\mid X,G=E)}\{\frac{1}{p(G=E\mid X)}-1\}\{M-\E[M\mid X,A,G]\} S(V) \Big].
\end{aligned}
\end{equation}

For the third term in \eqref{eq:IF:equi:ett:5}, we have
\begin{equation}
\label{eq:IF:equi:ett:8}
\begin{aligned}
&\sum_{m,x} m \frac{1}{p(A=1,G=O)}p(m\mid x,A=0,G=E)\partial_t p_t(x,G=O)\\
&=\sum_{x}\frac{1}{p(A=1,G=O)}\E[M\mid X=x,A=0,G=E] S(x,G=O)p(x,G=O)\\
&=\sum_{x,g}\frac{I(g=O)}{p(A=1,G=O)}\E[M\mid X=x,A=0,G=E] S(x,g)p(x,g)\\
&=\E\Big[\frac{I(G=O)}{p(A=1,G=O)}\E[M\mid X,A=0,G=E] S(X,G)\Big]\\
&=\E\Big[\frac{I(G=O)}{p(A=1,G=O)}\E[M\mid X,A=0,G=E] S(V)\Big].
\end{aligned}
\end{equation}

Combining \eqref{eq:IF:equi:ett:6}-\eqref{eq:IF:equi:ett:8} concludes that
\begin{align*}
\partial_t\psi_{3_t} 
&=\E\Big[
\frac{1}{p(A=1,G=O)}\Big\{
\frac{I(A=0)I(G=E)}{p(A=0\mid X,G=E)}\{\frac{1}{p(G=E\mid X)}-1\}\{M-\E[M\mid X,A,G]\}\\
&+I(G=O)\E[M\mid X,A=0,G=E]-I(A=1)I(G=O)\psi_3
\Big\}S(V)
\Big].	
\end{align*}
The variable
\begin{align*}
&\frac{1}{p(A=1,G=O)}\Big\{
\frac{I(A=0)I(G=E)}{p(A=0\mid X,G=E)}\{\frac{1}{p(G=E\mid X)}-1\}\{M-\E[M\mid X,A,G]\}\\
&+I(G=O)\E[M\mid X,A=0,G=E]-I(A=1)I(G=O)\psi_3
\Big\}
\end{align*}
is mean zero and hence is the influence function of $\psi_3$.

For $\psi_4$, note that
\begin{equation}
\label{eq:IF:equi:ett:9}
\begin{aligned}
\partial_t\psi_{4_t} 
&=\partial_t \sum_{y,x} y p_t(y\mid x,A=0,G=O)p_t(x\mid A=1,G=O)\\
&= \sum_{y,x} y \partial_t p_t(y\mid x,A=0,G=O)p(x\mid A=1,G=O)\\
&\quad+\sum_{y,x} y p(y\mid x,A=0,G=O)\partial_t p_t(x\mid A=1,G=O).
\end{aligned}
\end{equation}

For the first term in \eqref{eq:IF:equi:ett:9}, we have
\begin{align*}
&\sum_{y,x} y \partial_t p_t(y\mid x,A=0,G=O)p(x\mid A=1,G=O)\\	
&=\sum_{y,x} y S(y\mid x,A=0,G=O)p(y\mid x,A=0,G=O)p(x\mid A=1,G=O)\\	
&=\sum_{y,x} \frac{1}{p(A=1,G=O)}\cdot \frac{p(A=1\mid x,G=O)}{p(A=0\mid x,G=O)}y S(y\mid x,A=0,G=O)p(y,x,A=0,G=O)\\	
&=\sum_{y,a,x,g} \frac{I(a=0)I(g=O)}{p(A=1,G=O)}\cdot \frac{p(A=1\mid x,G=O)}{p(A=0\mid x,G=O)}y S(y\mid x,a,g)p(y,x,a,g)\\
&=\E\Big[\frac{I(A=0)I(G=O)}{p(A=1,G=O)}\cdot \frac{p(A=1\mid X,G=O)}{p(A=0\mid X,G=O)}Y S(Y\mid A,X,G)\Big]\\	
&=\E\Big[\frac{I(A=0)I(G=O)}{p(A=1,G=O)}\cdot \frac{p(A=1\mid X,G=O)}{p(A=0\mid X,G=O)}\{Y-\E[Y\mid A,X,G]\} S(Y\mid A,X,G)\Big].
\end{align*}
Note that
\begin{align*}
\E\Big[\frac{I(A=0)I(G=O)}{p(A=1,G=O)}\cdot \frac{p(A=1\mid X,G=O)}{p(A=0\mid X,G=O)}\{Y-\E[Y\mid A,X,G]\} S(A,X,G)\Big]=0.
\end{align*}
Therefore,
\begin{equation}
\label{eq:IF:equi:ett:10}
\begin{aligned}
&\sum_{y,x} y \partial_t p_t(y\mid x,A=0,G=O)p(x\mid A=1,G=O)\\
&=\E\Big[\frac{I(A=0)I(G=O)}{p(A=1,G=O)}\cdot \frac{p(A=1\mid X,G=O)}{p(A=0\mid X,G=O)}\{Y-\E[Y\mid A,X,G]\} S(V)\Big].
\end{aligned}
\end{equation}
For the second term in \eqref{eq:IF:equi:ett:9}, we have
{\footnotesize \begin{align*}
&\sum_{y,x} y p(y\mid x,A=0,G=O)\partial_t p_t(x\mid A=1,G=O)\\
&=\sum_{x} \E[Y\mid x,A=0,G=O] S(x\mid A=1,G=O)p(x\mid A=1,G=O)\\
&=\sum_{a,x,g}\frac{I(a=1)I(g=O)}{p(A=1,G=O)} \E[Y\mid x,A=0,G=O] S(x\mid a,g)p(x,a,g)\\
&=\E\Big[\frac{I(A=1)I(G=O)}{p(A=1,G=O)} \E[Y\mid X,A=0,G=O] S(X\mid A,G)\Big]\\
&=\E\Big[\frac{I(A=1)I(G=O)}{p(A=1,G=O)} \{\E[Y\mid X,A=0,G=O]-\E[\E[Y\mid X,A=0,G=O]\mid A=1,G=O]\} S(X\mid A,G)\Big]\\
&=\E\Big[\frac{I(A=1)I(G=O)}{p(A=1,G=O)} \{\E[Y\mid X,A=0,G=O]-\psi_4\} S(X\mid A,G)\Big].
\end{align*}}
Note that
\begin{align*}
\E\Big[\frac{I(A=1)I(G=O)}{p(A=1,G=O)} \{\E[Y\mid X,A=0,G=O]-\psi_4\} S(A,G)\Big]=0.	
\end{align*}
Therefore,
\begin{equation}
\label{eq:IF:equi:ett:11}
\begin{aligned}
&\sum_{y,x} y p(y\mid x,A=0,G=O)\partial_t p_t(x\mid A=1,G=O)\\
&=\E\Big[\frac{I(A=1)I(G=O)}{p(A=1,G=O)} \{\E[Y\mid X,A=0,G=O]-\psi_4\} S(V)\Big].
\end{aligned}
\end{equation}
Combining \eqref{eq:IF:equi:ett:10} and \eqref{eq:IF:equi:ett:11} concludes that
\begin{align*}
\partial_t\psi_{4_t} 
&=\E\Big[
\frac{I(G=O)}{p(A=1,G=O)}
\Big\{
\frac{I(A=0)p(A=1\mid X,G=O)}{p(A=0\mid X,G=O)}\{Y-\E[Y\mid A,X,G]\}\\
&+I(A=1)\{\E[Y\mid X,A=0,G=O]-\psi_4\} 
\Big\}S(V)\Big].	
\end{align*}
Therefore
\begin{align*}
&\frac{I(G=O)}{p(A=1,G=O)}
\Big\{
\frac{I(A=0)p(A=1\mid X,G=O)}{1-p(A=1\mid X,G=O)}\{Y-\E[Y\mid A,X,G]\}\\
&+I(A=1)\{\E[Y\mid X,A=0,G=O]-\psi_4\} 
\Big\}
\end{align*}
is the influence function of $\psi_4$.

For $i\in\{1,2,3,4\}$, denote the obtained influence functions by $IF_{\psi_i}$.
The influence function for $\psi_{\text{ETT}}^{\text{equi}}$ can be obtained as
$IF_{\psi_{\text{ETT}}^{\text{equi}}}=
\sum_{i=1}^4IF_{\psi_i}$. Therefore,
\begin{align*}
\frac{1}{p(A=1,G=O)}\Big\{
&\frac{I(G=O)I(A=0)}{1-p(A=1\mid X,G=O)}\{M-\E[M\mid X,A,G]\}\\
&+\frac{I(G=E)I(A=0)}{1-p(A=1\mid X,G=E)}\{\frac{1}{p(G=E\mid X)}-1\}\{M-\E[M\mid X,A,G]\}\\
&+\frac{I(G=O)I(A=0)p(A=1\mid X,G=O)}{1-p(A=1\mid X,G=O)}\{Y-\E[Y\mid A,X,G]\}\\
&+I(G=O)\big\{
\E[M\mid X,A=0,G=O]
+\E[M\mid X,A=0,G=E]\\
&\qquad\qquad\qquad+I(A=1)\{Y+\E[Y\mid X,A=0,G=O]-\psi_{\text{ETT}}^{\text{equi}}\}\big\}\Big\}
\end{align*}
is the influence function of $\psi_{\text{ETT}}^{\text{equi}}$.

\end{proof}

\begin{proof}[Proof of Theorem \ref{thm:IF:equi-conf:ATE}]

For $a\in\{0,1\}$, define
\begin{align*}
&\psi_1^{(a)}=
\E\big[\E[M\mid A=a,X,G=E]\big| G=O\big],\\
&\psi_2^{(a)}=
\E\big[\E[M\mid A=a,X,G=O]\big| G=O\big],\\
&\psi_3^{(a)}=
\E\big[\E[Y\mid A=a,X,G=O]\big| G=O\big].
\end{align*}
We use the notation $\partial_tf(t)$ to denote $\frac{\partial f(t)}{\partial t}\big|_{t=0}$.
For parameter $\psi^{(a)}$, let $\psi^{(a)}_t$ be the parameter under a regular parametric sub-model indexed by $t$, that includes the ground-truth model at $t=0$.
Let $V$ be the set of all observed variables.
In order to obtain an influence function, we need to find a random variable $\Gamma$ with mean zero, that satisfies
\[
\partial_t\psi^{(a)}_t=\E[\Gamma S(V)],
\]
where $S(V)=\partial_t\log p_t(V)$.

For $\psi_1^{(a)}$, note that
\begin{equation}
\label{eq:IF:equi:ate:1}
\begin{aligned}
\partial_t\psi^{(a)}_{1_t} 
&=\partial_t \sum_{m,x} mp_t(m\mid a,x,G=E)p_t(x\mid G=O)\\
&= \sum_{m,x} m\partial_t p_t(m\mid a,x,G=E)p(x\mid G=O)\\
&\quad+ \sum_{m,x} m p(m\mid a,x,G=E)\partial_t p_t(x\mid G=O).
\end{aligned}
\end{equation}

For the first term in \eqref{eq:IF:equi:ate:1}, we have
\begin{align*}
&\sum_{m,x} m\partial_t p_t(m\mid a,x,G=E)p(x\mid G=O)\\
&=\sum_{m,x} m\partial_t p_t(m\mid a,x,G=E)\{\frac{1}{p(G=E\mid x)}-1\}\frac{1}{p(G=O)}p(x,G=E)\\
&=\sum_{m,x} mS(m\mid a,x,G=E)\{\frac{1}{p(G=E\mid x)}-1\}\frac{1}{p(A=a\mid x,G=E)}\cdot\frac{1}{p(G=O)}p(m,a,x,G=E)\\
&=\sum_{m,\tilde{a},x,g} mS(m\mid \tilde{a},x,g)\{\frac{1}{p(G=E\mid x)}-1\}\frac{I(\tilde{a}=a)}{p(A=a\mid x,G=E)}\cdot\frac{I(g=E)}{p(G=O)}p(m,\tilde{a},x,g)\\
&=\E\Big[
\frac{I(A=a)}{p(A=a\mid X,G=E)}\cdot\frac{I(G=E)}{p(G=O)}
M\{\frac{1}{p(G=E\mid X)}-1\}S(M\mid A,X,G)
\Big]\\
&=\E\Big[
\frac{I(A=a)}{p(A=a\mid X,G=E)}\cdot\frac{I(G=E)}{p(G=O)}
\{M-\E[M\mid A,X,G]\}\{\frac{1}{p(G=E\mid X)}-1\}S(M\mid A,X,G)
\Big],
\end{align*}
Note that
\begin{align*}
\E\Big[
\frac{I(A=a)}{p(A=a\mid X,G=E)}\cdot\frac{I(G=E)}{p(G=O)}
\{M-\E[M\mid A,X,G]\}\{\frac{1}{p(G=E\mid X)}-1\}S(A,X,G)
\Big]=0.
\end{align*}
Therefore,
\begin{equation}
\label{eq:IF:equi:ate:2}
\begin{aligned}
&\sum_{m,x} m\partial_t p_t(m\mid a,x,G=E)p(x\mid G=O)\\
&=\E\Big[
\frac{I(A=a)}{p(A=a\mid X,G=E)}\cdot\frac{I(G=E)}{p(G=O)}
\{M-\E[M\mid A,X,G]\}\{\frac{1}{p(G=E\mid X)}-1\}S(V)
\Big].
\end{aligned}
\end{equation}

For the second term in \eqref{eq:IF:equi:ate:1}, we have
\begin{align*}
&\sum_{m,x} m p(m\mid a,x,G=E)\partial_t p_t(x\mid G=O)\\
&=\sum_{x} \E[M\mid A=a,X=x,G=E]\partial_t p_t(x\mid G=O)\\
&=\sum_{x} \E[M\mid A=a,X=x,G=E] S(x\mid G=O)p(x\mid G=O)\\
&=\sum_{x,g}\frac{I(g=O)}{p(G=O)}\E[M\mid A=a,X=x,G=E]S(x\mid g)p(x,g)\\
&=\E\Big[\frac{I(G=O)}{p(G=O)}\E[M\mid A=a,X,G=E]S(X\mid G)\Big]\\
&=\E\Big[\frac{I(G=O)}{p(G=O)}\{\E[M\mid A=a,X,G=E]-\E[\E[M\mid A=a,X,G=E]\mid G=O]\}S(X\mid G)\Big]\\
&=\E\Big[\frac{I(G=O)}{p(G=O)}\{\E[M\mid A=a,X,G=E]-\psi_1^{(a)}\}S(X\mid G)\Big].
\end{align*}
Note that
\begin{align*}
&\E\Big[\frac{I(G=O)}{p(G=O)}\{\E[M\mid A=a,X,G=E]-\psi_1^{(a)}\}S(G)\Big]=0.
\end{align*}
Therefore,
\begin{equation}
\label{eq:IF:equi:ate:3}
\begin{aligned}
&\sum_{m,x} m p(m\mid a,x,G=E)\partial_t p_t(x\mid G=O)
=\E\Big[\frac{I(G=O)}{p(G=O)}\{\E[M\mid A=a,X,G=E]-\psi_1^{(a)}\}S(V)\Big].
\end{aligned}
\end{equation}

Combining \eqref{eq:IF:equi:ate:2} and \eqref{eq:IF:equi:ate:3} concludes that
\begin{align*}
\partial_t\psi^{(a)}_{1_t} 
&=\E
\Big[\Big\{
\frac{I(A=a)}{p(A=a\mid X,G=E)}\cdot\frac{I(G=E)}{p(G=O)}
\{M-\E[M\mid A,X,G]\}\{\frac{1}{p(G=E\mid X)}-1\}\\
&\qquad+\frac{I(G=O)}{p(G=O)}\{\E[M\mid A=a,X,G=E]-\psi_1^{(a)}\}
\Big\}S(V)\Big].
\end{align*}

Therefore, 
\begin{align*}
&\frac{I(A=a)}{p(A=a\mid X,G=E)}\cdot\frac{I(G=E)}{p(G=O)}
\{M-\E[M\mid A,X,G]\}\{\frac{1}{p(G=E\mid X)}-1\}\\
&+\frac{I(G=O)}{p(G=O)}\{\E[M\mid A=a,X,G=E]-\psi_1^{(a)}\}
\end{align*}
is the influence function of $\psi_1^{(a)}$.

For $\psi_2^{(a)}$, note that
\begin{equation}
\label{eq:IF:equi:ate:4}
\begin{aligned}
\partial_t\psi^{(a)}_{2_t} 
&=\partial_t \sum_{m,x} mp_t(m\mid a,x,G=O)p_t(x\mid G=O)\\
&= \sum_{m,x} m\partial_t p_t(m\mid a,x,G=O)p(x\mid G=O)\\
&\quad+ \sum_{m,x} m p(m\mid a,x,G=O)\partial_t p_t(x\mid G=O).
\end{aligned}
\end{equation}

For the first term in \eqref{eq:IF:equi:ate:4}, we have
\begin{align*}
&\sum_{m,x} m\partial_t p_t(m\mid a,x,G=O)p(x\mid G=O)\\
&=\sum_{m,x} mS(m\mid a,x,G=O)\frac{1}{p(A=a\mid x,G=O)}\cdot\frac{1}{p(G=O)}p(m,a,x,G=O)\\
&=\sum_{m,\tilde{a},x,g} mS(m\mid \tilde{a},x,g)\frac{I(\tilde{a}=a)}{p(A=a\mid x,G=O)}\cdot\frac{I(g=O)}{p(G=O)}p(m,\tilde{a},x,g)\\
&=\E\Big[
\frac{I(A=a)}{p(A=a\mid X,G=O)}\cdot\frac{I(G=O)}{p(G=O)}
MS(M\mid A,X,G)
\Big]\\
&=\E\Big[
\frac{I(A=a)}{p(A=a\mid X,G=O)}\cdot\frac{I(G=O)}{p(G=O)}
\{M-\E[M\mid A,X,G]\}S(M\mid A,X,G)
\Big],
\end{align*}
Note that
\begin{align*}
\E\Big[
\frac{I(A=a)}{p(A=a\mid X,G=O)}\cdot\frac{I(G=O)}{p(G=O)}
\{M-\E[M\mid A,X,G]\}S(A,X,G)
\Big]=0.
\end{align*}
Therefore,
\begin{equation}
\label{eq:IF:equi:ate:5}
\begin{aligned}
&\sum_{m,x} m\partial_t p_t(m\mid a,x,G=O)p(x\mid G=O)\\
&=\E\Big[
\frac{I(A=a)}{p(A=a\mid X,G=O)}\cdot\frac{I(G=O)}{p(G=O)}
\{M-\E[M\mid A,X,G]\}S(V)
\Big].
\end{aligned}
\end{equation}

For the second term in \eqref{eq:IF:equi:ate:4}, we have
\begin{align*}
&\sum_{m,x} m p(m\mid a,x,G=O)\partial_t p_t(x\mid G=O)\\
&=\sum_{x} \E[M\mid A=a,X=x,G=O]\partial_t p_t(x\mid G=O)\\
&=\sum_{x} \E[M\mid A=a,X=x,G=O] S(x\mid G=O)p(x\mid G=O)\\
&=\sum_{x,g}\frac{I(g=O)}{p(G=O)}\E[M\mid A=a,X=x,G=O]S(x\mid g)p(x,g)\\
&=\E\Big[\frac{I(G=O)}{p(G=O)}\E[M\mid A=a,X,G=O]S(X\mid G)\Big]\\
&=\E\Big[\frac{I(G=O)}{p(G=O)}\{\E[M\mid A=a,X,G=O]-\E[\E[M\mid A=a,X,G=O]\mid G=O]\}S(X\mid G)\Big]\\
&=\E\Big[\frac{I(G=O)}{p(G=O)}\{\E[M\mid A=a,X,G=O]-\psi_2^{(a)}\}S(X\mid G)\Big].
\end{align*}
Note that
\begin{align*}
&\E\Big[\frac{I(G=O)}{p(G=O)}\{\E[M\mid A=a,X,G=O]-\psi_2^{(a)}\}S(G)\Big]=0.
\end{align*}
Therefore,
\begin{equation}
\label{eq:IF:equi:ate:6}
\begin{aligned}
&\sum_{m,x} m p(m\mid a,x,G=O)\partial_t p_t(x\mid G=O)
=\E\Big[\frac{I(G=O)}{p(G=O)}\{\E[M\mid A=a,X,G=O]-\psi_2^{(a)}\}S(V)\Big].
\end{aligned}
\end{equation}

Combining \eqref{eq:IF:equi:ate:5} and \eqref{eq:IF:equi:ate:6} concludes that
\begin{align*}
\partial_t\psi^{(a)}_{2_t} 
&=\E
\Big[\Big\{
\frac{I(A=a)}{p(A=a\mid X,G=O)}\cdot\frac{I(G=O)}{p(G=O)}
\{M-\E[M\mid A,X,G]\}\\
&\qquad+\frac{I(G=O)}{p(G=O)}\{\E[M\mid A=a,X,G=O]-\psi_2^{(a)}\}
\Big\}S(V)\Big].
\end{align*}

Therefore, 
\begin{align*}
\frac{I(G=O)}{p(G=O)}\big\{\frac{I(A=a)}{p(A=a\mid X,G=O)}
\{M-\E[M\mid A,X,G]\}
+\E[M\mid A=a,X,G=O]-\psi_2^{(a)}\big\}
\end{align*}
is the influence function of $\psi_2^{(a)}$.

Similarly, 
\begin{align*}
\frac{I(G=O)}{p(G=O)}\big\{\frac{I(A=a)}{p(A=a\mid X,G=O)}
\{Y-\E[Y\mid A,X,G]\}
+\E[Y\mid A=a,X,G=O]-\psi_3^{(a)}\big\}
\end{align*}
is the influence function of $\psi_3^{(a)}$.

For $i\in\{1,2,3\}$ and $a\in\{0,1\}$, denote the obtained influence functions by $IF_{\psi_i^{(a)}}$.
The influence function for $\psi_{\text{ATE}}^{\text{equi}}$ can be obtained as
$IF_{\psi_{\text{ATE}}^{\text{equi}}}=
IF_{\psi_1^{(1)}}-IF_{\psi_1^{(0)}}
+IF_{\psi_2^{(0)}}-IF_{\psi_2^{(1)}}
+IF_{\psi_3^{(1)}}-IF_{\psi_3^{(0)}}$. Therefore,
\begin{align*}
&\frac{(-1)^{1-A}}{p(A\mid X,G=E)}\cdot\frac{I(G=E)}{p(G=O)}
\{M-\E[M\mid A,X,G]\}\{\frac{1}{p(G=E\mid X)}-1\}\\
&+\frac{I(G=O)}{p(G=O)}\big\{
\frac{(-1)^{1-A}}{p(A\mid X,G=O)}\{Y-\E[Y\mid A,X,G]-M+\E[M\mid A,X,G]\}\\
&+
\E[Y\mid X,A=1,G=O]-\E[Y\mid X,A=0,G=O]\\
&+\E[M\mid X,A=1,G=E]-\E[M\mid X,A=0,G=E]\\
&+\E[M\mid X,A=0,G=O]-\E[M\mid X,A=1,G=O]
-\psi_{\text{ATE}}^{\text{equi}}\big\}
\end{align*}
is the influence function of $\psi_{\text{ATE}}^{\text{equi}}$.

\end{proof}

\begin{proof}[Proof of Proposition \ref{prop:DR:equi-conf} and \ref{prop:DR:equi-conf_ATE}]
	
First, suppose the set $\{\mu^{G=E}_M,\mu^{G=O}_M,\mu^{G=O}_Y\}$ is correctly specified. We have
\begingroup
\allowdisplaybreaks
\begin{align*}
&\mathbb{E}\Big[
\frac{1}{p(A=1,G=O)}\Big\{
\frac{I(G=O)I(A=0)}{1-\hat{\pi}^{G=O}(X)}\{M-\hat{\mu}^{G=O}_M(0,X)\}\\
&\hspace{25mm}+\frac{I(G=E)I(A=0)}{1-\hat{\pi}^{G=E}(X)}\{\frac{1}{\hat{p}(G=E\mid X)}-1\}\{M-\hat{\mu}^{G=E}_M(0,X)\}\\
&\hspace{25mm}+\frac{I(G=O)I(A=0)\hat{\pi}^{G=O}(X)}{1-\hat{\pi}^{G=O}(X)}\{Y-\hat{\mu}^{G=O}_Y(0,X)\}\\
&\hspace{25mm}+I(G=O)\big\{
\hat{\mu}^{G=O}_M(0,X)
+\hat{\mu}^{G=E}_M(0,X)
+I(A=1)\{Y+\hat{\mu}^{G=O}_Y(0,X)\}\big\}\Big\}
\Big]\\
&=\mathbb{E}\Big[
\frac{1}{p(A=1,G=O)}\Big\{
\frac{I(G=O)I(A=0)}{1-\hat{\pi}^{G=O}(X)}\{\E[M\mid X,A=0,G=O]-\hat{\mu}^{G=O}_M(0,X)\}\\
&\hspace{25mm}+\frac{I(G=E)I(A=0)}{1-\hat{\pi}^{G=E}(X)}\{\frac{1}{\hat{p}(G=E\mid X)}-1\}\{\E[M\mid X,A=0,G=E]-\hat{\mu}^{G=E}_M(0,X)\}\\
&\hspace{25mm}+\frac{I(G=O)I(A=0)\hat{\pi}^{G=O}(X)}{1-\hat{\pi}^{G=O}(X)}\{\E[Y\mid X,A=0,G=O]-\hat{\mu}^{G=O}_Y(0,X)\}\\
&\hspace{25mm}+I(G=O)\big\{
\hat{\mu}^{G=O}_M(0,X)
+\hat{\mu}^{G=E}_M(0,X)
+I(A=1)\{Y+\hat{\mu}^{G=O}_Y(0,X)\}\big\}\Big\}
\Big]\\
&=\mathbb{E}\Big[
\frac{1}{p(A=1,G=O)}\Big\{
I(G=O)\big\{
\hat{\mu}^{G=O}_M(0,X)
+\hat{\mu}^{G=E}_M(0,X)
+I(A=1)\{Y+\hat{\mu}^{G=O}_Y(0,X)\}\big\}\Big\}
\Big]\\
&=\theta_{ETT}.
\end{align*}
\endgroup

\begingroup
\allowdisplaybreaks
{\footnotesize \begin{align*}
&\mathbb{E}\Big[
\frac{(-1)^{1-A}}{1-A+(-1)^{1-A}\hat{\pi}^{G=E}(X)}\cdot\frac{I(G=E)}{p(G=O)}
\{M-\hat{\mu}^{G=E}_M(A,X)\}\{\frac{1}{\hat{p}(G=E\mid X)}-1\}\\
&\quad+\frac{I(G=O)}{p(G=O)}\big\{
\frac{(-1)^{1-A}}{1-A+(-1)^{1-A}\hat{\pi}^{G=O}(X)}\{Y-\hat{\mu}^{G=O}_Y(A,X)-M+\hat{\mu}^{G=O}_M(A,X)\}\\
&\quad+
\hat{\mu}^{G=O}_Y(1,X)-\hat{\mu}^{G=O}_Y(0,X)
+\hat{\mu}^{G=E}_M(1,X)-\hat{\mu}^{G=E}_M(0,X)
+\hat{\mu}^{G=O}_M(0,X)-\hat{\mu}^{G=O}_M(1,X)
\big\}
\Big]\\
&=\mathbb{E}\Big[
\frac{(-1)^{1-A}}{1-A+(-1)^{1-A}\hat{\pi}^{G=E}(X)}\cdot\frac{I(G=E)}{p(G=O)}
\{\E[M\mid X,A,G=E]-\hat{\mu}^{G=E}_M(A,X)\}\{\frac{1}{\hat{p}(G=E\mid X)}-1\}\\
&\quad+\frac{I(G=O)}{p(G=O)}\big\{
\frac{(-1)^{1-A}}{1-A+(-1)^{1-A}\hat{\pi}^{G=O}(X)}\{\E[Y\mid X,A,G=O]-\hat{\mu}^{G=O}_Y(A,X)-\E[M\mid X,A,G=O]+\hat{\mu}^{G=O}_M(A,X)\}\\
&\quad+
\hat{\mu}^{G=O}_Y(1,X)-\hat{\mu}^{G=O}_Y(0,X)
+\hat{\mu}^{G=E}_M(1,X)-\hat{\mu}^{G=E}_M(0,X)
+\hat{\mu}^{G=O}_M(0,X)-\hat{\mu}^{G=O}_M(1,X)
\big\}
\Big]\\
&=\mathbb{E}\Big[
\hat{\mu}^{G=O}_Y(1,X)-\hat{\mu}^{G=O}_Y(0,X)
+\hat{\mu}^{G=E}_M(1,X)-\hat{\mu}^{G=E}_M(0,X)
+\hat{\mu}^{G=O}_M(0,X)-\hat{\mu}^{G=O}_M(1,X)
\big\}
\Big]\\
&=\theta_{ATE}.
\end{align*}}
\endgroup

Second, suppose the set $\{\pi^{G=E},\pi^{G=O},p(G=E\mid x)\}$ is correctly specified. We have
\begingroup
\allowdisplaybreaks
\begin{align*}
&\mathbb{E}\Big[
\frac{1}{p(A=1,G=O)}\Big\{
\frac{I(G=O)I(A=0)}{1-\hat{\pi}^{G=O}(X)}\{M-\hat{\mu}^{G=O}_M(0,X)\}\\
&\hspace{25mm}+\frac{I(G=E)I(A=0)}{1-\hat{\pi}^{G=E}(X)}\{\frac{1}{\hat{p}(G=E\mid X)}-1\}\{M-\hat{\mu}^{G=E}_M(0,X)\}\\
&\hspace{25mm}+\frac{I(G=O)I(A=0)\hat{\pi}^{G=O}(X)}{1-\hat{\pi}^{G=O}(X)}\{Y-\hat{\mu}^{G=O}_Y(0,X)\}\\
&\hspace{25mm}+I(G=O)\big\{
\hat{\mu}^{G=O}_M(0,X)
+\hat{\mu}^{G=E}_M(0,X)
+I(A=1)\{Y+\hat{\mu}^{G=O}_Y(0,X)\}\big\}\Big\}
\Big]\\
&=\mathbb{E}\Big[
\frac{1}{p(A=1,G=O)}\Big\{
\frac{I(G=O)I(A=0)}{1-\hat{\pi}^{G=O}(X)}M\\
&\hspace{25mm}+\frac{I(G=E)I(A=0)}{1-\hat{\pi}^{G=E}(X)}\{\frac{1}{\hat{p}(G=E\mid X)}-1\}M\\
&\hspace{25mm}+\frac{I(G=O)I(A=0)\hat{\pi}^{G=O}(X)}{1-\hat{\pi}^{G=O}(X)}Y\\
&\hspace{25mm}+\frac{I(G=O)\E[I(A=0)\mid X,G=O]}{1-\hat{\pi}^{G=O}(X)}\{-\hat{\mu}^{G=O}_M(0,X)\}\\
&\hspace{25mm}+\frac{I(G=E)\E[I(A=0)\mid X,G=E]}{1-\hat{\pi}^{G=E}(X)}\{\frac{1}{\hat{p}(G=E\mid X)}-1\}\{-\hat{\mu}^{G=E}_M(0,X)\}\\
&\hspace{25mm}+\frac{I(G=O)\E[I(A=0)\mid X,G=O]\hat{\pi}^{G=O}(X)}{1-\hat{\pi}^{G=O}(X)}\{-\hat{\mu}^{G=O}_Y(0,X)\}\\
&\hspace{25mm}+I(G=O)\big\{
\hat{\mu}^{G=O}_M(0,X)
+\hat{\mu}^{G=E}_M(0,X)
+I(A=1)\{Y+\hat{\mu}^{G=O}_Y(0,X)\}\big\}\Big\}
\Big]\\
&=\mathbb{E}\Big[
\frac{1}{p(A=1,G=O)}\Big\{
\frac{I(G=O)I(A=0)}{1-\hat{\pi}^{G=O}(X)}M\\
&\hspace{25mm}+\frac{I(G=E)I(A=0)}{1-\hat{\pi}^{G=E}(X)}\{\frac{1}{\hat{p}(G=E\mid X)}-1\}M\\
&\hspace{25mm}+\frac{I(G=O)I(A=0)\hat{\pi}^{G=O}(X)}{1-\hat{\pi}^{G=O}(X)}Y\\
&\hspace{25mm}+I(G=O)\{-\hat{\mu}^{G=O}_M(0,X)\}\\
&\hspace{25mm}+\E[I(G=E)\mid X]\{\frac{1}{\hat{p}(G=E\mid X)}-1\}\{-\hat{\mu}^{G=E}_M(0,X)\}\\
&\hspace{25mm}+I(G=O)\E[I(A=1)\{-\hat{\mu}^{G=O}_Y(0,X)\}\mid X,G=O]\\
&\hspace{25mm}+I(G=O)\big\{
\hat{\mu}^{G=O}_M(0,X)
+\hat{\mu}^{G=E}_M(0,X)
+I(A=1)\{Y+\hat{\mu}^{G=O}_Y(0,X)\}\big\}\Big\}
\Big]\\
&=\mathbb{E}\Big[
\frac{1}{p(A=1,G=O)}\Big\{
\frac{I(G=O)I(A=0)}{1-\hat{\pi}^{G=O}(X)}M\\
&\hspace{25mm}+\frac{I(G=E)I(A=0)}{1-\hat{\pi}^{G=E}(X)}\{\frac{1}{\hat{p}(G=E\mid X)}-1\}M\\
&\hspace{25mm}+\frac{I(G=O)I(A=0)\hat{\pi}^{G=O}(X)}{1-\hat{\pi}^{G=O}(X)}Y\\
&\hspace{25mm}+I(G=O)I(A=1)Y\Big\}
\Big]\\
&=\theta_{ETT}.
\end{align*}
\endgroup

\begingroup
\allowdisplaybreaks
\begin{align*}
&\mathbb{E}\Big[
\frac{(-1)^{1-A}}{1-A+(-1)^{1-A}\hat{\pi}^{G=E}(X)}\cdot\frac{I(G=E)}{p(G=O)}
\{M-\hat{\mu}^{G=E}_M(A,X)\}\{\frac{1}{\hat{p}(G=E\mid X)}-1\}\\
&\quad+\frac{I(G=O)}{p(G=O)}\big\{
\frac{(-1)^{1-A}}{1-A+(-1)^{1-A}\hat{\pi}^{G=O}(X)}\{Y-\hat{\mu}^{G=O}_Y(A,X)-M+\hat{\mu}^{G=O}_M(A,X)\}\\
&\quad+
\hat{\mu}^{G=O}_Y(1,X)-\hat{\mu}^{G=O}_Y(0,X)
+\hat{\mu}^{G=E}_M(1,X)-\hat{\mu}^{G=E}_M(0,X)
+\hat{\mu}^{G=O}_M(0,X)-\hat{\mu}^{G=O}_M(1,X)
\big\}
\Big]\\
&=\mathbb{E}\Big[
\frac{(-1)^{1-A}}{1-A+(-1)^{1-A}\hat{\pi}^{G=E}(X)}\cdot\frac{I(G=E)}{p(G=O)}
M\{\frac{1}{\hat{p}(G=E\mid X)}-1\}\\
&\quad+\frac{I(G=O)}{p(G=O)}\cdot
\frac{(-1)^{1-A}}{1-A+(-1)^{1-A}\hat{\pi}^{G=O}(X)}\{Y-M\}\\
&\quad+\frac{(-1)^{1-A}}{1-A+(-1)^{1-A}\hat{\pi}^{G=E}(X)}\cdot\frac{I(G=E)}{p(G=O)}
\{-\hat{\mu}^{G=E}_M(A,X)\}\{\frac{1}{\hat{p}(G=E\mid X)}-1\}\\
&\quad+\frac{I(G=O)}{p(G=O)}\big\{
\frac{(-1)^{1-A}}{1-A+(-1)^{1-A}\hat{\pi}^{G=O}(X)}\{-\hat{\mu}^{G=O}_Y(A,X)+\hat{\mu}^{G=O}_M(A,X)\}\\
&\quad+
\hat{\mu}^{G=O}_Y(1,X)-\hat{\mu}^{G=O}_Y(0,X)
+\hat{\mu}^{G=E}_M(1,X)-\hat{\mu}^{G=E}_M(0,X)
+\hat{\mu}^{G=O}_M(0,X)-\hat{\mu}^{G=O}_M(1,X)
\big\}
\Big]\\
&=\mathbb{E}\Big[
\frac{(-1)^{1-A}}{1-A+(-1)^{1-A}\hat{\pi}^{G=E}(X)}\cdot\frac{I(G=E)}{p(G=O)}
M\{\frac{1}{\hat{p}(G=E\mid X)}-1\}\\
&\quad+\frac{I(G=O)}{p(G=O)}\cdot
\frac{(-1)^{1-A}}{1-A+(-1)^{1-A}\hat{\pi}^{G=O}(X)}\{Y-M\}\\
&\quad+\frac{I(G=E)}{p(G=O)}
\{-\hat{\mu}^{G=E}_M(1,X)+\hat{\mu}^{G=E}_M(0,X)\}\{\frac{1}{\hat{p}(G=E\mid X)}-1\}\\
&\quad+\frac{I(G=O)}{p(G=O)}\big\{
\{-\hat{\mu}^{G=O}_Y(1,X)+\hat{\mu}^{G=O}_M(1,X)+\hat{\mu}^{G=O}_Y(0,X)-\hat{\mu}^{G=O}_M(0,X)\}\\
&\quad+
\hat{\mu}^{G=O}_Y(1,X)-\hat{\mu}^{G=O}_Y(0,X)
+\hat{\mu}^{G=E}_M(1,X)-\hat{\mu}^{G=E}_M(0,X)
+\hat{\mu}^{G=O}_M(0,X)-\hat{\mu}^{G=O}_M(1,X)
\big\}
\Big]\\
&=\mathbb{E}\Big[
\frac{(-1)^{1-A}}{1-A+(-1)^{1-A}\hat{\pi}^{G=E}(X)}\cdot\frac{I(G=E)}{p(G=O)}
M\{\frac{1}{\hat{p}(G=E\mid X)}-1\}\\
&\quad+\frac{I(G=O)}{p(G=O)}\cdot
\frac{(-1)^{1-A}}{1-A+(-1)^{1-A}\hat{\pi}^{G=O}(X)}\{Y-M\}
\Big]\\
&=\theta_{ATE}.
\end{align*}
\endgroup

Parts 3 and 4 can be proven by combining the techniques used in parts 1 and 2, and thus we omit here.

\end{proof}

\begin{proof}[Proof of Theorem \ref{thm:IF:BSIV-1_ett}: $\psi_{\text{ETT}}^{\text{bsiv1}}$]
	
Define
\begin{align*}
&\psi_1=
\E\big[\frac{\E[Y-M\mid B=1,X,G=O]-\E[Y-M\mid B=0,X,G=O]}{p(A=1\mid B=1,X,G=O)-p(A=1\mid B=0,X,G=O)}~\big|~ A=1,G=O\big],\\
&\psi_2=
\E\big[\frac{\E[M\mid A=0,B,X,G=E]}{p(A=1\mid B,X,G=O)}~\big|~ A=1,G=O\big],\\
&\psi_3=
\E\big[\frac{\E[M\mid A=0,B,X,G=O]p(A=0\mid B,X,G=O)}{p(A=1\mid B,X,G=O)}~\big|~ A=1,G=O\big],\\
&\psi_4=
\E\big[\E[M\mid A=1,B,X,G=O]~\big|~ A=1,G=O\big]\\
&\psi_5=\psi_3+\psi_4.
\end{align*}
We use the notation $\partial_tf(t)$ to denote $\frac{\partial f(t)}{\partial t}\big|_{t=0}$.
For parameter $\psi$, let $\psi_t$ be the parameter under a regular parametric sub-model indexed by $t$, that includes the ground-truth model at $t=0$.
Let $V$ be the set of all observed variables.
In order to obtain an influence function, we need to find a random variable $\Gamma$ with mean zero, that satisfies
\[
\partial_t\psi_t=\E[\Gamma S(V)],
\]
where $S(V)=\partial_t\log p_t(V)$.	

For $\psi_1$, note that
\begin{equation}
\label{eq:IF-BSIV-1-1}
\begin{aligned}
\partial_t{\psi_1}_t	
&=\partial_t\sum_x\frac{\E_t[Y-M\mid B=1,x,G=O]-\E_t[Y-M\mid B=0,x,G=O]}{p_t(A=1\mid B=1,x,G=O)-p_t(A=1\mid B=0,x,G=O)}p_t(x\mid A=1,G=O)\\
&=\sum_x\partial_t\frac{\E_t[Y-M\mid B=1,x,G=O]-\E_t[Y-M\mid B=0,x,G=O]}{p_t(A=1\mid B=1,x,G=O)-p_t(A=1\mid B=0,x,G=O)}p(x\mid A=1,G=O)\\
&\quad+\sum_x\frac{\E[Y-M\mid B=1,x,G=O]-\E[Y-M\mid B=0,x,G=O]}{p(A=1\mid B=1,x,G=O)-p(A=1\mid B=0,x,G=O)}\partial_tp_t(x\mid A=1,G=O).
\end{aligned}
\end{equation}

For the first term in \eqref{eq:IF-BSIV-1-1}, note that
\begingroup
\allowdisplaybreaks
\begin{align*}
&\sum_x\partial_t\frac{\E_t[Y-M\mid B=1,x,G=O]}{p_t(A=1\mid B=1,x,G=O)-p_t(A=1\mid B=0,x,G=O)}p(x\mid A=1,G=O)\\
&=\sum_x\frac{\partial_t\E_t[Y-M\mid B=1,x,G=O]}{p(A=1\mid B=1,x,G=O)-p(A=1\mid B=0,x,G=O)}p(x\mid A=1,G=O)\\
&\quad-\sum_x\frac{\E[Y-M\mid B=1,x,G=O]\partial_tp_t(A=1\mid B=1,x,G=O)}{\{p(A=1\mid B=1,x,G=O)-p(A=1\mid B=0,x,G=O)\}^2}p(x\mid A=1,G=O)\\
&\quad+\sum_x\frac{\E[Y-M\mid B=1,x,G=O]\partial_tp_t(A=1\mid B=0,x,G=O)}{\{p(A=1\mid B=1,x,G=O)-p(A=1\mid B=0,x,G=O)\}^2}p(x\mid A=1,G=O)\\
&=\sum_{y,m,b,x,g}\frac{I(b=1)I(g=O)}{p(A=1\mid B=1,x,G=O)-p(A=1\mid B=0,x,G=O)}\cdot\frac{1}{p(A=1,G=O)}\\
&\qquad\cdot\frac{p(A=1\mid x,G=O)}{p(B=1\mid x,G=O)}\{y-m\}S(y,m\mid b,x,g)p(y,m,b,x,g)\\
&\quad-\sum_{a,b,x,g}\frac{\E[Y-M\mid B=1,x,G=O]}{\{p(A=1\mid B=1,x,G=O)-p(A=1\mid B=0,x,G=O)\}^2}\cdot\frac{I(b=1)I(g=O)}{p(A=1,G=O)}\\
&\qquad\cdot \frac{p(A=1\mid x,G=O)}{p(B=1\mid x,G=O)}I(a=1) S(a\mid b,x,g)p(a,b,x,g)\\
&\quad+\sum_{a,b,x,g}\frac{\E[Y-M\mid B=1,x,G=O]}{\{p(A=1\mid B=1,x,G=O)-p(A=1\mid B=0,x,G=O)\}^2}\cdot\frac{I(b=0)I(g=O)}{p(A=1,G=O)}\\
&\qquad\cdot \frac{p(A=1\mid x,G=O)}{p(B=0\mid x,G=O)}I(a=1) S(a\mid b,x,g)p(a,b,x,g)\\
&=\E\big[\frac{I(B=1)I(G=O)}{p(A=1\mid B=1,X,G=O)-p(A=1\mid B=0,X,G=O)}\cdot\frac{1}{p(A=1,G=O)}\\
&\qquad\cdot\frac{p(A=1\mid X,G=O)}{p(B=1\mid X,G=O)}\{Y-M\}S(Y,M\mid B,X,G)\big]\\
&\quad-\E\big[\frac{\E[Y-M\mid B=1,X,G=O]}{\{p(A=1\mid B=1,X,G=O)-p(A=1\mid B=0,X,G=O)\}^2}\cdot\frac{I(B=1)I(G=O)}{p(A=1,G=O)}\\
&\qquad\cdot \frac{p(A=1\mid X,G=O)}{p(B=1\mid X,G=O)}I(A=1) S(A\mid B,X,G)\big]\\
&\quad+\E\big[\frac{\E[Y-M\mid B=1,X,G=O]}{\{p(A=1\mid B=1,X,G=O)-p(A=1\mid B=0,X,G=O)\}^2}\cdot\frac{I(B=0)I(G=O)}{p(A=1,G=O)}\\
&\qquad\cdot \frac{p(A=1\mid X,G=O)}{p(B=0\mid X,G=O)}I(A=1) S(A\mid B,X,G)\big]\\
&=\E\big[\frac{I(B=1)I(G=O)}{p(A=1\mid B=1,X,G=O)-p(A=1\mid B=0,X,G=O)}\cdot\frac{1}{p(A=1,G=O)}\\
&\qquad\cdot\frac{p(A=1\mid X,G=O)}{p(B=1\mid X,G=O)}\{Y-M-\E[Y-M\mid B=1,X,G=O]\}S(V)\big]\\
&\quad-\E\big[\frac{\E[Y-M\mid B=1,X,G=O]}{\{p(A=1\mid B=1,X,G=O)-p(A=1\mid B=0,X,G=O)\}^2}\cdot\frac{I(B=1)I(G=O)}{p(A=1,G=O)}\\
&\qquad\cdot \frac{p(A=1\mid X,G=O)}{p(B=1\mid X,G=O)}\{I(A=1)-p(A=1\mid B=1,X,G=O)\} S(V)\big]\\
&\quad+\E\big[\frac{\E[Y-M\mid B=1,X,G=O]}{\{p(A=1\mid B=1,X,G=O)-p(A=1\mid B=0,X,G=O)\}^2}\cdot\frac{I(B=0)I(G=O)}{p(A=1,G=O)}\\
&\qquad\cdot \frac{p(A=1\mid X,G=O)}{p(B=0\mid X,G=O)}\{I(A=1)-p(A=1\mid B=0,X,G=O)\} S(V)\big].
\end{align*}
\endgroup
Similarly,
\begingroup
\allowdisplaybreaks
\begin{align*}
&\sum_x\partial_t\frac{\E_t[Y-M\mid B=0,x,G=O]}{p_t(A=1\mid B=1,x,G=O)-p_t(A=1\mid B=0,x,G=O)}p(x\mid A=1,G=O)\\
&=\E\big[\frac{I(B=0)I(G=O)}{p(A=1\mid B=1,X,G=O)-p(A=1\mid B=0,X,G=O)}\cdot\frac{1}{p(A=1,G=O)}\\
&\qquad\cdot\frac{p(A=1\mid X,G=O)}{p(B=0\mid X,G=O)}\{Y-M-\E[Y-M\mid B=0,X,G=O]\}S(V)\big]\\
&\quad-\E\big[\frac{\E[Y-M\mid B=0,X,G=O]}{\{p(A=1\mid B=1,X,G=O)-p(A=1\mid B=0,X,G=O)\}^2}\cdot\frac{I(B=1)I(G=O)}{p(A=1,G=O)}\\
&\qquad\cdot \frac{p(A=1\mid X,G=O)}{p(B=1\mid X,G=O)}\{I(A=1)-p(A=1\mid B=1,X,G=O)\} S(V)\big]\\
&\quad+\E\big[\frac{\E[Y-M\mid B=0,X,G=O]}{\{p(A=1\mid B=1,X,G=O)-p(A=1\mid B=0,X,G=O)\}^2}\cdot\frac{I(B=0)I(G=O)}{p(A=1,G=O)}\\
&\qquad\cdot \frac{p(A=1\mid X,G=O)}{p(B=0\mid X,G=O)}\{I(A=1)-p(A=1\mid B=0,X,G=O)\} S(V)\big].
\end{align*}
\endgroup
Therefore,
\begingroup
\allowdisplaybreaks
{\footnotesize \begin{align*}
&\sum_x\partial_t\frac{\E_t[Y-M\mid B=1,x,G=O]-\E_t[Y-M\mid B=0,x,G=O]}{p_t(A=1\mid B=1,x,G=O)-p_t(A=1\mid B=0,x,G=O)}p(x\mid A=1,G=O)\\
&=\E\Big[\frac{I(G=O)}{p(A=1\mid B=1,X,G=O)-p(A=1\mid B=0,X,G=O)}\cdot\frac{1}{p(A=1,G=O)}
\cdot\frac{p(A=1\mid X,G=O)}{p(B\mid X,G=O)}\\
&\quad\Big\{I(B=1)\{Y-M-\E[Y-M\mid B=1,X,G=O]\}
-I(B=0)\{Y-M-\E[Y-M\mid B=0,X,G=O]\}\\
&\quad+\frac{\E[Y-M\mid B=1,X,G=O]-\E[Y-M\mid B=0,X,G=O]}{p(A=1\mid B=1,X,G=O)-p(A=1\mid B=0,X,G=O)}\\
&\quad\big\{-I(B=1)\{I(A=1)-p(A=1\mid B=1,X,G=O)\}+I(B=0)\{I(A=1)-p(A=1\mid B=0,X,G=O)\}\big\}\Big\}S(V)\Big].
\end{align*}}
\endgroup
For the second term in \eqref{eq:IF-BSIV-1-1}, we have
\begingroup
\allowdisplaybreaks
\begin{align*}
&\sum_x\frac{\E[Y-M\mid B=1,x,G=O]-\E[Y-M\mid B=0,x,G=O]}{p(A=1\mid B=1,x,G=O)-p(A=1\mid B=0,x,G=O)}\partial_tp_t(x\mid A=1,G=O)\\
&=\E\Big[\frac{I(A=1)I(G=O)}{p(A=1,G=O)}\cdot\frac{\E[Y-M\mid B=1,x,G=O]-\E[Y-M\mid B=0,x,G=O]}{p(A=1\mid B=1,x,G=O)-p(A=1\mid B=0,x,G=O)}S(X\mid A,G)\Big]\\
&=\E\Big[\frac{I(A=1)I(G=O)}{p(A=1,G=O)}\{\frac{\E[Y-M\mid B=1,x,G=O]-\E[Y-M\mid B=0,x,G=O]}{p(A=1\mid B=1,x,G=O)-p(A=1\mid B=0,x,G=O)}-\psi_1\}S(V)\Big].
\end{align*}
\endgroup
Therefore,
\begingroup
\allowdisplaybreaks
{\footnotesize \begin{align*}
&\partial_t{\psi_1}_t	\\
&=\E\Bigg[\Bigg\{\frac{I(G=O)}{p(A=1\mid B=1,X,G=O)-p(A=1\mid B=0,X,G=O)}\cdot\frac{1}{p(A=1,G=O)}
\cdot\frac{p(A=1\mid X,G=O)}{p(B\mid X,G=O)}\\
&\quad\Big\{I(B=1)\{Y-M-\E[Y-M\mid B=1,X,G=O]\}
-I(B=0)\{Y-M-\E[Y-M\mid B=0,X,G=O]\}\\
&\quad+\frac{\E[Y-M\mid B=1,X,G=O]-\E[Y-M\mid B=0,X,G=O]}{p(A=1\mid B=1,X,G=O)-p(A=1\mid B=0,X,G=O)}\\
&\quad\big\{-I(B=1)\{I(A=1)-p(A=1\mid B=1,X,G=O)\}+I(B=0)\{I(A=1)-p(A=1\mid B=0,X,G=O)\}\big\}\Big\}\\
&\quad+\frac{I(A=1)I(G=O)}{p(A=1,G=O)}\{\frac{\E[Y-M\mid B=1,x,G=O]-\E[Y-M\mid B=0,x,G=O]}{p(A=1\mid B=1,x,G=O)-p(A=1\mid B=0,x,G=O)}-\psi_1\}\Bigg\}S(V)\Bigg].
\end{align*}}
\endgroup

For $\psi_2$, similar to the parameter $\psi_3$ in the proof of Theorem \ref{thm:IF:equi-conf:ETT}, we have
{\footnotesize \begin{align*}
&\partial_t{\psi_2}_t\\
&=\E\Big[
\frac{1}{p(A=1,G=O)}\Big\{
\frac{I(A=0)I(G=E)}{p(A=0\mid X,B,G=E)}\{\frac{1}{p(G=E\mid X,B)}-1\}\{M-\E[M\mid A=0,X,B,G=E]\}\\
&~~~~~~~~~~+I(G=O)\E[M\mid A=0,X,B,G=E]-I(A=1)I(G=O)\psi_2
\Big\}S(V)
\Big].	
\end{align*}}

For $\psi_3$ and $\psi_4$, note that
\begin{align*}
\psi_3
&=\E\big[\frac{\E[M\mid A=0,B,X,G=O]p(A=0\mid B,X,G=O)}{p(A=1\mid B,X,G=O)}~\big|~ A=1,G=O\big]\\	
&=\sum_{x,b}\E[M\mid A=0,b,x,G=O]\frac{p(A=0\mid b,x,G=O)}{p(A=1\mid b,x,G=O)}p(b,x\mid A=1,G=O)\\
&=\sum_{x,b}\E[M\mid A=0,b,x,G=O]\frac{p(A=0\mid b,x,G=O)}{p(A=1,G=O)}p(b,x,G=O)\\
&=\frac{p(A=0,G=O)}{p(A=1,G=O)}\E[M\mid A=0,G=O],
\end{align*}
and
\begin{align*}
\psi_4
&=\E\big[\E[M\mid A=1,B,X,G=O]~\big|~ A=1,G=O\big]=\E[M\mid A=1,G=O].
\end{align*}
Therefore,
\begin{align*}
\psi_5:=\psi_3+\psi_4
=\frac{1}{p(A=1\mid G=O)}\E[M\mid G=O].
\end{align*}
This implies that
\begin{align*}
&\partial_t{\psi_5}_t=\E\Big[
\frac{I(G=O)}{p(A=1,G=O)}\big\{
M-I(A=1)\psi_5
\big\}S(V)
\Big].	
\end{align*}

For $i\in\{1,2,5\}$, denote the obtained influence functions by $IF_{\psi_i}$. The influence function for $\psi_{\text{ETT}}^{\text{bsiv1}}$ can be obtained as $IF_{\psi_{\text{ETT}}^{\text{bsiv1}}}=IF_{\psi_1}-IF_{\psi_2}+IF_{\psi_5}$. Therefore, using the notations specified in Theorem \ref{thm:IF:BSIV-1_ett},

\begin{align*}
    &\frac{I(G=O)}{P_{11}^O(X)-P_{10}^O(X)}\frac{1}{p(A=1,G=O)}\frac{\pi^O(X)}{\rho_B^O(X)}\bigg\{I(B=1)\{Y-M-e_1^O(X)\}-I(B=0)\{Y-M-e_0^O(X)\}\\
&\quad\quad+\frac{e_1^O(X)-e_0^O(X)}{P_{11}^O(X)-P_{10}^O(X)}\big\{I(B=0)\{I(A=1)-P_{10}^O(X)\}-I(B=1)\{I(A=1)-P_{11}^O(X)\}\big\}\bigg\}\\
&\quad+\frac{I(A=1)I(G=O)}{p(A=1,G=O)}\Big\{\frac{e_1^O(X)-e_0^O(X)}{P_{11}^O(X)-P_{10}^O(X)}\Big\}-\frac{I(A=0)I(G=E)}{p(A=1,G=O)}\cdot\frac{1-\tau(B,X)}{\tau(B,X)}\cdot
\frac{M-M_{0}^{E}(B,X)}{1-P_{1B}^E(X)}\\
&\quad+\frac{I(G=O)}{p(A=1,G=O)}\cdot \{M-M_{0}^{E}(B,X)-I(A=1)\psi^1_{ETT}\}
\end{align*}
is the influence function of $\psi_{\text{ETT}}^{\text{bsiv1}}$.

\end{proof}

\begin{proof}[Proof of Theorem \ref{thm:IF:BSIV-1_ett} $\psi_{\text{ETT}}^{\text{bsiv2}}$]
	
Define
\begin{align*}
    &\psi_1^{(ab)}=\mathbb{E}\big[\mathbb{E}[Y-M\mid A=a,B=b,X,G=O]B\mid A=1,G=O\big],\\
    &\psi_2^{(ab)}=\mathbb{E}\big[\mathbb{E}[Y-M\mid A=a,B=b,X,G=O]\mid A=1,G=O\big],\\
    &\psi_3=\mathbb{E}\bigg[\frac{\mathbb{E}[Y-M\mid A=0,B=1,X,G=O]-\mathbb{E}[Y-M\mid A=0,B=0,X,G=O]}{p(A=0\mid B=1,X,G=O)-p(A=0\mid B=0,X,G=O)}\bigg|A=1,G=O\bigg],\\
    &\psi_4^{(ab)}=\mathbb{E}\bigg[\frac{\mathbb{E}[Y-M\mid A=a,B=b,X,G=O]p(A=1\mid B,X,G=O)}{p(A=0\mid B=1,X,G=O)-p(A=0\mid B=0,X,G=O)}\bigg|G=O\bigg],\\
    &\psi_5=\mathbb{E}\bigg[\frac{\mathbb{E}[M\mid A=0,B,X,G=E]}{p(A=1\mid B,X,G=O)}\mid A=1,G=O\bigg],\\
    &\psi_6=\mathbb{E}\big[\mathbb{E}[M\mid A=1,B,X,G=O]\mid A=1,G=O\big],\\
    &\psi_7=\mathbb{E}\bigg[\frac{\mathbb{E}[M\mid A=0,B,X,G=O]p(A=0\mid B,X,G=O)}{p(A=1\mid B,X,G=O)}\mid A=1,G=O\bigg],\\
    &\psi_8=\psi_6+\psi_7=\frac{1}{P(A=1\mid G=O)}\mathbb{E}[M\mid G=O].\\
\end{align*}
We use the notation $\partial_tf(t)$ to denote $\frac{\partial f(t)}{\partial t}\big|_{t=0}$.
For parameter $\psi$, let $\psi_t$ be the parameter under a regular parametric sub-model indexed by $t$, that includes the ground-truth model at $t=0$.
Let $V$ be the set of all observed variables.
In order to obtain an influence function, we need to find a random variable $\Gamma$ with mean zero, that satisfies
\[
\partial_t\psi_t=\E[\Gamma S(V)],
\]
where $S(V)=\partial_t\log p_t(V)$.	

For $\psi_1^{(ab)}$, note that
\begin{align*}
    \partial_t\psi_{1t}^{(ab)}&=\partial_t\sum_{x,\Tilde{b}}\mathbb{E}_t[Y-M\mid A=a,B=b,x,G=O]\Tilde{b}p_t(x,\Tilde{b}\mid A=1, G=O)\\
    &=\sum_{x,\Tilde{b}}\partial_t\mathbb{E}_t[Y-M\mid A=a,B=b,x,G=O]\Tilde{b}p(x,\Tilde{b}\mid A=1, G=O)\\
    &\ \ \ \ + \sum_{x,\Tilde{b}}\mathbb{E}[Y-M\mid A=a,B=b,x,G=O]\Tilde{b}\partial_tp_t(x,\Tilde{b}\mid A=1, G=O)\\
    &=\sum_{x}\partial_t\mathbb{E}_t[Y-M\mid A=a,B=b,x,G=O]p(x,B=1\mid A=1, G=O)\\
    &\ \ \ \ + \sum_{x,\Tilde{b}}\mathbb{E}[Y-M\mid A=a,B=b,x,G=O]\Tilde{b}\partial_tp_t(x,\Tilde{b}\mid A=1, G=O)\\
\end{align*}

For the first term, note that
{\footnotesize \begin{align*}
    &\sum_{x}\partial_t\mathbb{E}_t[Y-M\mid A=a,B=b,x,G=O]p(x,B=1\mid A=1, G=O)\\
    &=\sum_{y,m,x}(y-m)S(y,m\mid A=a,B=b,x,G=O)p(y,m\mid A=a,B=b,x,G=O)p(x,B=1\mid A=1, G=O)\\
    &=\sum_{y,m,x,\Tilde{b},\Tilde{a},g}I(\Tilde{b}=b)I(\Tilde{a}=a)I(g=O)(y-m)S(y,m\mid \Tilde{a},\Tilde{b},x,g)p(y,m\mid \Tilde{a},\Tilde{b},x,g)p(x,B=1\mid A=1, G=O)\\
    &=\sum_{y,m,x,\Tilde{b},\Tilde{a},g}\frac{I(\Tilde{b}=b)I(\Tilde{a}=a)I(g=O)}{p(A=1,G=O)}(y-m)S(y,m\mid \Tilde{a},\Tilde{b},x,g)\\
    &~~~~~~~~~~~~~~~~~\cdot \frac{p(B=1\mid x,G=O)p(A=1\mid B=1, x, G=O)}{p(\Tilde{b}\mid x,G=O)p(\Tilde{a}\mid\Tilde{b},x,G=O)}p(y,m, \Tilde{a},\Tilde{b},x,g)\\
    &=\mathbb{E}\bigg[\frac{I(B=b)I(A=a)I(G=O)}{p(A=1,G=O)}\frac{p(B=1\mid X,G=O)p(A=1\mid B=1, X, G=O)}{p(B\mid X,G=O)p(A\mid B, X,G=O)}(Y-M)S(Y,M\mid A,B,X,G)\bigg]\\
    &=\mathbb{E}\bigg[\frac{I(B=b)I(A=a)I(G=O)}{p(A=1,G=O)}\frac{p(B=1\mid X,G=O)p(A=1\mid B=1, X, G=O)}{p(B\mid X,G=O)p(A\mid B, X,G=O)}\\
    &\ \ \ \ \ \ \ \ \ \ \{Y-M-\mathbb{E}[Y-M\mid A=a,B=b,X,G=O]\}S(Y,M\mid A,B,X,G)\bigg]
\end{align*}}

Note that,
\begin{align*}
    &\mathbb{E}\bigg[\frac{I(B=b)I(A=a)I(G=O)}{p(A=1,G=O)}\frac{p(B=1\mid X,G=O)p(A=1\mid B=1, X, G=O)}{p(B\mid X,G=O)p(A\mid B, X,G=O)}\\
    &\ \ \ \ \ \ \ \ \ \ \{Y-M-\mathbb{E}[Y-M\mid A=a,B=b,X,G=O]\}S(A,B,X,G)\bigg]=0.
\end{align*}

Therefore,
\begin{align*}
    &\sum_{x}\partial_t\mathbb{E}_t[Y-M\mid A=a,B=b,x,G=O]p(x,B=1\mid A=1, G=O)\\
    &=\mathbb{E}\bigg[\frac{I(B=b)I(A=a)I(G=O)}{p(A=1,G=O)}\frac{p(B=1\mid X,G=O)p(A=1\mid B=1, X, G=O)}{p(B\mid X,G=O)p(A\mid B, X,G=O)}\\
    &\ \ \ \ \ \ \ \ \ \ \{Y-M-\mathbb{E}[Y-M\mid A=a,B=b,X,G=O]\}S(V)\bigg]
\end{align*}

For the second term, note that,
{\footnotesize \begin{align*}
    &\sum_{x,\tilde{b}}\mathbb{E}[Y-M\mid A=a,B=b,x,G=O]\tilde{b}\partial_tp_t(x,\tilde{b}\mid A=1, G=O)\\
    &=\sum_{x,\tilde{b}}\mathbb{E}[Y-M\mid A=a,B=b,x,G=O]\tilde{b}p(x,\tilde{b}\mid A=1, G=O)S(x,\tilde{b}\mid A=1, G=O)\\
    &=\sum_{x,\Tilde{b},\Tilde{a},g}\frac{I(\Tilde{a}=1)I(g=O)}{p(A=1,G=O)}\mathbb{E}[Y-M\mid A=a,B=b,x,G=O]\tilde{b}p(x,\Tilde{b}, \Tilde{a}, g)S(x,\tilde{b}\mid A=1, G=O)\\
    &=\mathbb{E}\bigg[\frac{I(A=1)I(G=O)}{p(A=1,G=O)}\mathbb{E}[Y-M\mid A=a,B=b,X,G=O]ZS(X,B\mid A, G)\bigg]\\
    &=\mathbb{E}\bigg[\frac{I(A=1)I(G=O)}{p(A=1,G=O)}\{\mathbb{E}[Y-M\mid A=a,B=b,X,G=O]B\\
    &~~~~~~~~~~-\mathbb{E}[\mathbb{E}[Y-M\mid A=a,B=b,X,G=O]B\mid A=1, G=O]\}S(X,B\mid A, G)\bigg]\\
    &=\mathbb{E}\bigg[\frac{I(A=1)I(G=O)}{p(A=1,G=O)}\{\mathbb{E}[Y-M\mid A=a,B=b,X,G=O]B-\psi_1^{(az)}\}S(X,B\mid A, G)\bigg]
\end{align*}}

Note that,
\begin{align*}
    \mathbb{E}\bigg[\frac{I(A=1)I(G=O)}{p(A=1,G=O)}\{\mathbb{E}[Y-M\mid A=a,B=b,X,G=O]B-\psi_1^{(ab)}\}S(A, G)\bigg]=0.
\end{align*}

Therefore,
\begin{align*}
    &\sum_{x}\mathbb{E}[Y-M\mid A=a,B=b,x,G=O]\partial_tp_t(x,B=1\mid A=1, G=O)\\
    &=\mathbb{E}\bigg[\frac{I(A=1)I(G=O)}{p(A=1,G=O)}\{\mathbb{E}[Y-M\mid A=a,B=b,X,G=O]B-\psi_1^{(ab)}\}S(V)\bigg]
\end{align*}

Combining concludes that,
\begin{align*}
    \partial_t\psi_{1t}^{(b)}&=\mathbb{E}\bigg[\frac{I(G=O)}{p(A=1,G=O)}\bigg\{I(A=a)I(B=b)\frac{p(B=1\mid X,G=O)p(A=1\mid B=1, X, G=O)}{p(B\mid X,G=O)p(A\mid B, X,G=O)}\\
    &\ \ \ \ \ \ \ \ \ \{Y-M-\mathbb{E}[Y-M\mid A=a,B=b,X,G=O]\}\\
    &\ \ \ \ \ \ \ +I(A=1)\{\mathbb{E}[Y-M\mid A=a,B=b,X,G=O]B-\psi_1^{(ab)}\}\bigg\}S(V)\bigg]
\end{align*}

For $\psi_2^{(ab)}$, note that
\begin{align*}
    \partial_t\psi_{2t}^{(ab)}&=\partial_t\sum_{x}\mathbb{E}_t[Y-M\mid A=a,B=b,x,G=O]p_t(x\mid A=1, G=O)\\
    &=\sum_{x}\partial_t\mathbb{E}_t[Y-M\mid A=a,B=b,x,G=O]p(x\mid A=1, G=O)\\
    &\ \ \ \ + \sum_{x}\mathbb{E}[Y-M\mid A=a,B=b,x,G=O]\partial_tp_t(x\mid A=1, G=O)\\
\end{align*}

For the first term, note that
{\footnotesize \begin{align*}
    &\sum_{x}\partial_t\mathbb{E}_t[Y-M\mid A=a,B=b,x,G=O]p(x\mid A=1, G=O)\\
    &=\sum_{y,m,x}(y-m)S(y,m\mid B=b,A=a,x,G=O)p(y,m\mid B=b,A=a,x,G=O)p(x\mid A=1, G=O)\\
    &=\sum_{y,m,x}\frac{1}{p(A=1, G=O)}\frac{p(x, A=1, G=O)}{p(B=b,A=a,x,G=O)}(y-m)S(y,m\mid B=b,A=a,x,G=O)p(y,m, B=b,A=a,x,G=O)\\
    &=\sum_{y,m,x,\Tilde{b},\Tilde{a},g}\frac{I(\Tilde{b}=b)I(\Tilde{a}=a)I(g=O)}{p(A=1, G=O)}\frac{1}{p(A=a\mid B=b,x,G=O)}\frac{p(A=1\mid x, G=O)}{p(B=b\mid x,G=O)}\\
    &~~~~~~~~~~~~~~~~~~\cdot (y-m)S(y,m\mid \Tilde{b},\Tilde{a},x,g)p(y,m, \Tilde{b},\Tilde{a},x,g)\\
    &=\mathbb{E}\bigg[\frac{I(B=b)I(A=a)I(G=O)}{p(A=1, G=O)}\frac{1}{p(A=a\mid B=b,X,G=O)}\frac{p(A=1\mid X, G=O)}{p(B=b\mid X,G=O)}(Y-M)S(Y,M\mid B,A,X,G)\bigg]\\
    &=\mathbb{E}\bigg[\frac{I(B=b)I(A=a)I(G=O)}{p(A=1, G=O)}\frac{1}{p(A=a\mid B=b,X,G=O)}\frac{p(A=1\mid X, G=O)}{p(B=b\mid X,G=O)}\times\\
    &\ \ \ \ \ \ \ \ (Y-M-\mathbb{E}[Y-M\mid B=b,A=a,X,G=O])S(Y,M\mid B,A,X,G)\bigg].
\end{align*}}

Note that,
\begin{align*}
    &\mathbb{E}\bigg[\frac{I(B=b)I(A=a)I(G=O)}{p(A=1, G=O)}\frac{1}{p(A=a\mid B=b,X,G=O)}\frac{p(A=1\mid X, G=O)}{p(B=b\mid X,G=O)}\times\\
    &\ \ \ \ \ \ \ \ (Y-M-\mathbb{E}[Y-M\mid B=b,A=a,X,G=O])S(B,A,X,G)\bigg]=0.
\end{align*}

Therefore,
\begin{align*}
    &\sum_{x}\partial_t\mathbb{E}_t[Y-M\mid A=a,B=b,x,G=O]p(x\mid A=1, G=O)\\
    &=\mathbb{E}\bigg[\frac{I(B=b)I(A=a)I(G=O)}{p(A=1, G=O)}\frac{1}{p(A=a\mid B=b,X,G=O)}\frac{p(A=1\mid X, G=O)}{p(B=b\mid X,G=O)}\times\\
    &\ \ \ \ \ \ \ \ (Y-M-\mathbb{E}[Y-M\mid B=b,A=a,X,G=O])S(V)\bigg].
\end{align*}

For the second term, note that
\begin{align*}
    &\sum_{x}\mathbb{E}[Y-M\mid A=a,B=b,x,G=O]\partial_tp_t(x\mid A=1, G=O)\\
    &=\sum_{x}\mathbb{E}[Y-M\mid A=a,B=b,x,G=O]S(x\mid A=1, G=O)p(x\mid A=1, G=O)\\
    &=\sum_{x,\tilde{a},g}\frac{I(\tilde{a}=1)I(g=O)}{p(A=1, G=O)}\mathbb{E}[Y-M\mid A=a,B=b,x,G=O]S(x\mid \tilde{a},g)p(x,\tilde{a},g)\\
    &=\mathbb{E}\bigg[\frac{I(A=1)I(G=O)}{p(A=1, G=O)}\mathbb{E}[Y-M\mid A=a,B=b,X,G=O]S(X\mid A,G)\bigg]\\
    &=\mathbb{E}\bigg[\frac{I(A=1)I(G=O)}{p(A=1, G=O)}\bigg\{\mathbb{E}[Y-M\mid A=a,B=b,X,G=O]-\psi_2^{(ab)}\bigg\}S(X\mid A,G)\bigg]
\end{align*}

Note that,
\begin{align*}
    &\mathbb{E}\bigg[\frac{I(A=1)I(G=O)}{p(A=1, G=O)}\bigg\{\mathbb{E}[Y-M\mid A=a,B=b,X,G=O]-\psi_2^{(ab)}\bigg\}S(A,G)\bigg]=0.
\end{align*}

Therefore,
\begin{align*}
    &\sum_{x}\mathbb{E}[Y-M\mid A=a,B=b,x,G=O]\partial_tp_t(x\mid A=1, G=O)\\
    &=\mathbb{E}\bigg[\frac{I(A=1)I(G=O)}{p(A=1, G=O)}\bigg\{\mathbb{E}[Y-M\mid A=a,B=b,X,G=O]-\psi_2^{(ab)}\bigg\}S(V)\bigg]
\end{align*}

Combining concludes that,
\begin{align*}
    \partial_t\psi_{2t}^{(ab)}&=\mathbb{E}\bigg[\frac{I(G=O)}{p(A=1, G=O)}\bigg\{\frac{I(B=b)I(A=a)}{p(A=a\mid B=b,X,G=O)}\frac{p(A=1\mid X, G=O)}{p(B=b\mid X,G=O)}\times\\
    &~~~~~~~~~~~~~~~~~~~~~~~~~\{Y-M-\mathbb{E}[Y-M\mid B=b,A=a,X,G=O]\}\\
    &~~~~~~~~~~~~~~~+I(A=1)\{\mathbb{E}[Y-M\mid A=a,B=b,X,G=O]-\psi_2^{(ab)}\}\bigg\}S(V)\bigg]
\end{align*}

For $\psi_3$, note that,
{\footnotesize \begin{align*}
    \partial_t\psi_{3t}&=\partial_t\sum_{x}\frac{\mathbb{E}_t[Y-M\mid A=0,B=1,x,G=O]-\mathbb{E}_t[Y-M\mid A=0,B=0,x,G=O]}{p_t(A=0\mid B=1,x,G=O)-p_t(A=0\mid B=0,x,G=O)}p_t(x\mid A=1,G=O)\\
    &=\sum_{x}\partial_t\frac{\mathbb{E}_t[Y-M\mid A=0,B=1,x,G=O]-\mathbb{E}_t[Y-M\mid A=0,B=0,x,G=O]}{p_t(A=0\mid B=1,x,G=O)-p_t(A=0\mid B=0,x,G=O)}p(x\mid A=1,G=O)\\
    &\ \ \ \ +\sum_{x}\frac{\mathbb{E}[Y-M\mid A=0,B=1,x,G=O]-\mathbb{E}[Y-M\mid A=0,B=0,x,G=O]}{p(A=0\mid B=1,x,G=O)-p(A=0\mid B=0,x,G=O)}\partial_tp_t(x\mid A=1,G=O)
\end{align*}}

For the first term, note that
{\footnotesize \begin{align*}
    &\sum_{x}\partial_t\frac{\mathbb{E}_t[Y-M\mid A=0,B=1,x,G=O]}{p_t(A=0\mid B=1,x,G=O)-p_t(A=0\mid B=0,x,G=O)}p(x\mid A=1,G=O)\\
    &=\sum_{x}\frac{\partial_t\mathbb{E}_t[Y-M\mid A=0,B=1,x,G=O]}{p(A=0\mid B=1,x,G=O)-p(A=0\mid B=0,x,G=O)}p(x\mid A=1,G=O)\\
    &\ \ \ \ -\sum_{x}\frac{\mathbb{E}[Y-M\mid A=0,B=1,x,G=O]\partial_tp_t(A=0\mid B=1,x,G=O)}{\{p(A=0\mid B=1,x,G=O)-p(A=0\mid B=0,x,G=O)\}^2}p(x\mid A=1,G=O)\\
    &\ \ \ \ +\sum_{x}\frac{\mathbb{E}[Y-M\mid A=0,B=1,x,G=O]\partial_tp_t(A=0\mid B=0,x,G=O)}{\{p(A=0\mid B=1,x,G=O)-p(A=0\mid B=0,x,G=O)\}^2}p(x\mid A=1,G=O)\\
    &=\sum_{y,m,x,a,b,g}\frac{I(a=0)I(b=1)I(g=O)}{p(A=0\mid B=1,x,G=O)-p(A=0\mid B=0,x,G=O)}\cdot\frac{1}{p(A=1,G=O)}\\
    &\ \ \ \ \ \ \ \ \cdot\frac{p(A=1\mid x,G=O)}{p(B=1\mid x,G=O)}\frac{1}{p(A=0\mid B=1,x,G=O)}\{y-m\}S(y,m\mid x,a,b,g)p(y,m,x,a,b,g)\\
    &\ \ \ \ -\sum_{a,b,x,g}\frac{\mathbb{E}[Y-M\mid A=0,B=1,x,G=O]}{\{p(A=0\mid B=1,x,G=O)-p(A=0\mid B=0,x,G=O)\}^2}\frac{I(b=1)I(g=O)}{p(A=1,G=O)}\cdot\\
    &\ \ \ \ \ \ \ \ \cdot\frac{p(A=1\mid x,G=O)}{p(B=1\mid x, G=O)}I(a=0)S(a\mid b,x,g)p(a,b,x,g)\\
    &\ \ \ \ +\sum_{a,b,x,g}\frac{\mathbb{E}[Y-M\mid A=0,B=1,x,G=O]}{\{p(A=0\mid B=1,x,G=O)-p(A=0\mid B=0,x,G=O)\}^2}\frac{I(b=0)I(g=O)}{p(A=1,G=O)}\cdot\\
    &\ \ \ \ \ \ \ \ \cdot\frac{p(A=1\mid x,G=O)}{p(B=0\mid x, G=O)}I(a=0)S(a\mid b,x,g)p(a,b,x,g)\\
    &=\mathbb{E}\bigg[\frac{I(A=0)I(B=1)I(G=O)}{p(A=0\mid B=1,X,G=O)-p(A=0\mid B=0,X,G=O)}\cdot\frac{1}{p(A=1,G=O)}\\
    &\ \ \ \ \ \ \ \ \cdot\frac{p(A=1\mid X,G=O)}{p(B=1\mid X,G=O)}\frac{1}{p(A=0\mid B=1,X,G=O)}\{Y-M\}S(Y,M\mid X,A,B,G)\bigg]\\
    &\ \ \ \ -\mathbb{E}\bigg[\frac{\mathbb{E}[Y-M\mid A=0,B=1,X,G=O]}{\{p(A=0\mid B=1,X,G=O)-p(A=0\mid B=0,X,G=O)\}^2}\frac{I(B=1)I(G=O)}{p(A=1,G=O)}\cdot\\
    &\ \ \ \ \ \ \ \ \cdot\frac{p(A=1\mid X,G=O)}{p(B=1\mid X, G=O)}I(A=0)S(A\mid B,X,G)\bigg]\\
    &\ \ \ \ +\mathbb{E}\bigg[\frac{\mathbb{E}[Y-M\mid A=0,B=1,X,G=O]}{\{p(A=0\mid B=1,X,G=O)-p(A=0\mid B=0,X,G=O)\}^2}\frac{I(B=0)I(G=O)}{p(A=1,G=O)}\cdot\\
    &\ \ \ \ \ \ \ \ \cdot\frac{p(A=1\mid X,G=O)}{p(B=0\mid X, G=O)}I(A=0)S(A\mid B,X,G)\bigg]\\
    &=\mathbb{E}\bigg[\frac{I(A=0)I(B=1)I(G=O)}{p(A=0\mid B=1,X,G=O)-p(A=0\mid B=0,X,G=O)}\cdot\frac{1}{p(A=1,G=O)}\\
    &\ \ \ \ \ \ \ \ \cdot\frac{p(A=1\mid X,G=O)}{p(B=1\mid X,G=O)}\frac{1}{p(A=0\mid B=1,X,G=O)}\{Y-M-\mathbb{E}[Y-M\mid A=0,B=1,X,G=O]\}S(V)\bigg]\\
    &\ \ \ \ -\mathbb{E}\bigg[\frac{\mathbb{E}[Y-M\mid A=0,B=1,X,G=O]}{\{p(A=0\mid B=1,X,G=O)-p(A=0\mid B=0,X,G=O)\}^2}\frac{I(B=1)I(G=O)}{p(A=1,G=O)}\cdot\\
    &\ \ \ \ \ \ \ \ \cdot\frac{p(A=1\mid X,G=O)}{p(B=1\mid X, G=O)}\{I(A=0)-p(A=0\mid B=1,X,G=O)\}S(V)\bigg]\\
    &\ \ \ \ +\mathbb{E}\bigg[\frac{\mathbb{E}[Y-M\mid A=0,B=1,X,G=O]}{\{p(A=0\mid B=1,X,G=O)-p(A=0\mid B=0,X,G=O)\}^2}\frac{I(B=0)I(G=O)}{p(A=1,G=O)}\cdot\\
    &\ \ \ \ \ \ \ \ \cdot\frac{p(A=1\mid X,G=O)}{p(B=0\mid X, G=O)}\{I(A=0)-p(A=0\mid B=0,X,G=O)\}S(V)\bigg].
\end{align*}}

Similarly,
{\footnotesize \begin{align*}
    &\sum_{x}\partial_t\frac{\mathbb{E}_t[Y-M\mid A=0,B=0,x,G=O]}{p_t(A=0\mid B=1,x,G=O)-p_t(A=0\mid B=0,x,G=O)}p(x\mid A=1,G=O)\\
    &=\mathbb{E}\bigg[\frac{I(A=0)I(B=0)I(G=O)}{p(A=0\mid B=1,X,G=O)-p(A=0\mid B=0,X,G=O)}\cdot\frac{1}{p(A=1,G=O)}\\
    &\ \ \ \ \ \ \ \ \cdot\frac{p(A=1\mid X,G=O)}{p(B=0\mid X,G=O)}\frac{1}{p(A=0\mid B=0,X,G=O)}\{Y-M-\mathbb{E}[Y-M\mid A=0,B=0,X,G=O]\}S(V)\bigg]\\
    &\ \ \ \ -\mathbb{E}\bigg[\frac{\mathbb{E}[Y-M\mid A=0,B=0,X,G=O]}{\{p(A=0\mid B=1,X,G=O)-p(A=0\mid B=0,X,G=O)\}^2}\frac{I(B=1)I(G=O)}{p(A=1,G=O)}\cdot\\
    &\ \ \ \ \ \ \ \ \cdot\frac{p(A=1\mid X,G=O)}{p(B=1\mid X, G=O)}\{I(A=0)-p(A=0\mid B=1,X,G=O)\}S(V)\bigg]\\
    &\ \ \ \ +\mathbb{E}\bigg[\frac{\mathbb{E}[Y-M\mid A=0,B=0,X,G=O]}{\{p(A=0\mid B=1,X,G=O)-p(A=0\mid B=0,X,G=O)\}^2}\frac{I(B=0)I(G=O)}{p(A=1,G=O)}\cdot\\
    &\ \ \ \ \ \ \ \ \cdot\frac{p(A=1\mid X,G=O)}{p(B=0\mid X, G=O)}\{I(A=0)-p(A=0\mid B=0,X,G=O)\}S(V)\bigg].
\end{align*}}

Therefore,
{\footnotesize \begin{align*}
    &\sum_{x}\partial_t\frac{\mathbb{E}_t[Y-M\mid A=0,B=1,x,G=O]-\mathbb{E}_t[Y-M\mid A=0,B=0,x,G=O]}{p_t(A=0\mid B=1,x,G=O)-p_t(A=0\mid B=0,x,G=O)}p(x\mid A=1,G=O)\\
    &=\mathbb{E}\bigg[\frac{I(G=O)}{p(A=0\mid B=1,X,G=O)-p(A=0\mid B=0,X,G=O)}\frac{p(A=1\mid X, G=O)}{p(A=1,G=O)}\\
    &\ \ \ \ \ \ \ \bigg\{I(A=0)\big\{\frac{I(B=1)}{p(A=0\mid B=1,X,G=O)p(B=1\mid X,G=O)}\{Y-M-\mathbb{E}[Y-M\mid A=0,B=1,X,G=O]\}\\
    &\ \ \ \ \ \ \ \ \ \ \ -\frac{I(B=0)}{p(A=0\mid B=0,X,G=O)p(B=0\mid X,G=O)}\{Y-M-\mathbb{E}[Y-M\mid A=0,B=0,X,G=O]\}\big\}\\
    &\ \ \ \ \ \ \ \ +\frac{\mathbb{E}[Y-M\mid A=0,B=1,X,G=O]-\mathbb{E}[Y-M\mid A=0,B=0,X,G=O]}{p(A=0\mid B=1,X,G=O)-p(A=0\mid B=0,X,G=O)}\\
    &\ \ \ \ \ \ \ \ \ \ \ \ \cdot\big\{-\frac{I(B=1)}{p(B=1\mid X, G=O)}\{I(A=0)-p(A=0\mid B=1,X,G=O)\}\\
    &\ \ \ \ \ \ \ \ \ \ \ +\frac{I(B=0)}{p(B=0\mid X, G=O)}\{I(A=0)-p(A=0\mid B=0,X,G=O)\}\big\}\bigg\}S(V)\bigg]
\end{align*}}

For the second term, we have
{\footnotesize \begin{align*}
    &\sum_{x}\frac{\mathbb{E}[Y-M\mid A=0,B=1,x,G=O]-\mathbb{E}[Y-M\mid A=0,B=0,x,G=O]}{p(A=0\mid B=1,x,G=O)-p(A=0\mid B=0,x,G=O)}\partial_tp_t(x\mid A=1,G=O)\\
    &=\mathbb{E}\bigg[\frac{I(A=1)I(G=O)}{p(A=1,G=O)}\cdot\frac{\mathbb{E}[Y-M\mid A=0,B=1,X,G=O]-\mathbb{E}[Y-M\mid A=0,B=0,X,G=O]}{p(A=0\mid B=1,X,G=O)-p(A=0\mid B=0,x,G=O)}\cdot S(X\mid A,G)\bigg]\\
    &=\mathbb{E}\bigg[\frac{I(A=1)I(G=O)}{p(A=1,G=O)}\bigg\{\frac{\mathbb{E}[Y-M\mid A=0,B=1,X,G=O]-\mathbb{E}[Y-M\mid A=0,B=0,X,G=O]}{p(A=0\mid B=1,X,G=O)-p(A=0\mid B=0,x,G=O)}-\psi_3\bigg\}\cdot S(V)\bigg].
\end{align*}}

Therefore,
{\footnotesize \begin{align*}
    &\partial_t\psi_{3t}\\
    &=\mathbb{E}\Bigg[\Bigg\{\frac{I(G=O)}{p(A=0\mid B=1,X,G=O)-p(A=0\mid B=0,X,G=O)}\frac{p(A=1\mid X, G=O)}{p(A=1,G=O)}\\
    &\ \ \ \ \ \ \ \ \ \ \bigg\{I(A=0)\big\{\frac{I(B=1)}{p(A=0\mid B=1,X,G=O)p(B=1\mid X,G=O)}\{Y-M-\mathbb{E}[Y-M\mid A=0,B=1,X,G=O]\}\\
    &\ \ \ \ \ \ \ \ \ \ \ \ \ \ \ \ \ \ \ \ \ \ \ \ \ \ \ \ -\frac{I(B=0)}{p(A=0\mid B=0,X,G=O)p(B=0\mid X,G=O)}\{Y-M-\mathbb{E}[Y-M\mid A=0,B=0,X,G=O]\}\big\}\\
    &\ \ \ \ \ \ \ \ \ \ \ \ +\frac{\mathbb{E}[Y-M\mid A=0,B=1,X,G=O]-\mathbb{E}[Y-M\mid A=0,B=0,X,G=O]}{p(A=0\mid B=1,X,G=O)-p(A=0\mid B=0,X,G=O)}\\
    &\ \ \ \ \ \ \ \ \ \ \ \ \ \ \ \ \cdot\big\{-\frac{I(B=1)}{p(B=1\mid X, G=O)}\{I(A=0)-p(A=0\mid B=1,X,G=O)\}\\
    &\ \ \ \ \ \ \ \ \ \ \ \ \ \ \ \ \ \ \ \ +\frac{I(B=0)}{p(B=0\mid X, G=O)}\{I(A=0)-p(A=0\mid B=0,X,G=O)\}\big\}\bigg\}\\
    &\ \ \ \ \ \ \ \ \  +\frac{I(A=1)I(G=O)}{p(A=1,G=O)}\bigg\{\frac{\mathbb{E}[Y-M\mid A=0,B=1,X,G=O]-\mathbb{E}[Y-M\mid A=0,B=0,X,G=O]}{p(A=0\mid B=1,X,G=O)-p(A=0\mid B=0,x,G=O)}-\psi_3\bigg\}\Bigg\}S(V)\Bigg]
\end{align*}}

Finally,
\begin{align*}
    \partial_t\psi_{1t}^{(ab)}&=\mathbb{E}\bigg[\frac{I(G=O)}{p(A=1,G=O)}\bigg\{I(A=a)I(B=b)\frac{p(B=1\mid X,G=O)p(A=1\mid B=1, X, G=O)}{p(B\mid X,G=O)p(A\mid B, X,G=O)}\\
    &\ \ \ \ \ \ \ \ \ \{Y-M-\mathbb{E}[Y-M\mid A=a,B=b,X,G=O]\}\\
    &\ \ \ \ \ \ \ +I(A=1)\{\mathbb{E}[Y-M\mid A=a,B=b,X,G=O]B-\psi_1^{(ab)}\}\bigg\}S(V)\bigg]
\end{align*}

The influence function for $\psi_1^{(11)}$ is:
\begin{align*}
    &\frac{I(G=O)}{p(A=1,G=O)}\bigg\{I(A=1)I(B=1)\{Y-M-\mathbb{E}[Y-M\mid A=1,B=1,X,G=O]\}\\
    &\ \ \ \ \ \ \ +I(A=1)\{\mathbb{E}[Y-M\mid A=1,B=1,X,G=O]B-\psi_1^{(11)}\}\bigg\}
\end{align*}

The influence function for $\psi_1^{(01)}$ is:
\begin{align*}
    &\frac{I(G=O)}{p(A=1,G=O)}\bigg\{I(A=0)I(B=1)\frac{p(A=1\mid B=1, X, G=O)}{p(A=0\mid B=1, X,G=O)}\\
    &\ \ \ \ \ \ \ \ \ \{Y-M-\mathbb{E}[Y-M\mid A=0,B=1,X,G=O]\}\\
    &\ \ \ \ \ \ \ +I(A=1)\{\mathbb{E}[Y-M\mid A=0,B=1,X,G=O]B-\psi_1^{(01)}\}\bigg\}
\end{align*}

The influence function for $\psi_1^{(10)}$ is:
\begin{align*}
    &\frac{I(G=O)}{p(A=1,G=O)}\bigg\{I(A=1)I(B=0)\frac{p(B=1\mid X,G=O)p(A=1\mid B=1, X, G=O)}{p(B=0\mid X,G=O)p(A=1\mid B=0, X,G=O)}\\
    &\ \ \ \ \ \ \ \ \ \{Y-M-\mathbb{E}[Y-M\mid A=1,B=0,X,G=O]\}\\
    &\ \ \ \ \ \ \ +I(A=1)\{\mathbb{E}[Y-M\mid A=1,B=0,X,G=O]B-\psi_1^{(10)}\}\bigg\}
\end{align*}

The influence function for $\psi_1^{(00)}$ is:
\begin{align*}
    &\frac{I(G=O)}{p(A=1,G=O)}\bigg\{I(A=0)I(B=0)\frac{p(B=1\mid X,G=O)p(A=1\mid B=1, X, G=O)}{p(B=0\mid X,G=O)p(A=0\mid B=0, X,G=O)}\\
    &\ \ \ \ \ \ \ \ \ \{Y-M-\mathbb{E}[Y-M\mid A=0,B=0,X,G=O]\}\\
    &\ \ \ \ \ \ \ +I(A=1)\{\mathbb{E}[Y-M\mid A=0,B=0,X,G=O]B-\psi_1^{(00)}\}\bigg\}
\end{align*}

The influence function for $\psi_2^{(10)}$ is:
{\footnotesize \begin{align*}
    &\frac{I(G=O)}{p(A=1, G=O)}\bigg\{\frac{I(A=1)I(B=0)}{p(A=1\mid B=0,X,G=O)}\frac{p(A=1\mid X, G=O)}{p(B=0\mid X,G=O)}\times\\
    &\ \ \ \ \ \ \ \ \{Y-M-\mathbb{E}[Y-M\mid A=1,B=0,X,G=O]\}+I(A=1)\{\mathbb{E}[Y-M\mid A=1,B=0,X,G=O]-\psi_2^{(10)}\}\bigg\}
\end{align*}}

The influence function for $\psi_2^{(00)}$ is:
{\footnotesize \begin{align*}
    &\frac{I(G=O)}{p(A=1, G=O)}\bigg\{\frac{I(A=0)I(B=0)}{p(A=0\mid B=0,X,G=O)}\frac{p(A=1\mid X, G=O)}{p(B=0\mid X,G=O)}\times\\
    &\ \ \ \ \ \ \ \ \{Y-M-\mathbb{E}[Y-M\mid A=0,B=0,X,G=O]\}+I(A=1)\{\mathbb{E}[Y-M\mid A=0,B=0,X,G=O]-\psi_2^{(00)}\}\bigg\}
\end{align*}}

The influence function for $\psi_3$ is:
{\footnotesize \begin{align*}
    &\frac{I(G=O)}{p(A=0\mid B=1,X,G=O)-p(A=0\mid B=0,X,G=O)}\frac{p(A=1\mid X, G=O)}{p(A=1,G=O)}\\
    &\ \ \ \ \ \ \ \ \ \ \bigg\{I(A=0)\big\{\frac{I(B=1)}{p(A=0\mid B=1,X,G=O)p(B=1\mid X,G=O)}\{Y-M-\mathbb{E}[Y-M\mid A=0,B=1,X,G=O]\}\\
    &\ \ \ \ \ \ \ \ \ \ \ \ \ \ \ \ \ \ \ \ \ \ \ \ \ \ -\frac{I(B=0)}{p(A=0\mid B=0,X,G=O)p(B=0\mid X,G=O)}\{Y-M-\mathbb{E}[Y-M\mid A=0,B=0,X,G=O]\}\big\}\\
    &\ \ \ \ \ \ \ \ \ \ \ \ +\frac{\mathbb{E}[Y-M\mid A=0,B=1,X,G=O]-\mathbb{E}[Y-M\mid A=0,B=0,X,G=O]}{p(A=0\mid B=1,X,G=O)-p(A=0\mid B=0,X,G=O)}\\
    &\ \ \ \ \ \ \ \ \ \ \ \ \ \ \ \ \cdot\big\{-\frac{I(B=1)}{p(B=1\mid X, G=O)}\{I(A=0)-p(A=0\mid B=1,X,G=O)\}\\
    &\ \ \ \ \ \ \ \ \ \ \ \ \ \ \ \ \ \ \ \ +\frac{I(B=0)}{p(B=0\mid X, G=O)}\{I(A=0)-p(A=0\mid B=0,X,G=O)\}\big\}\bigg\}\\
    &\ \ \ \ \ \ \ \ \  +\frac{I(A=1)I(G=O)}{p(A=1,G=O)}\bigg\{\frac{\mathbb{E}[Y-M\mid A=0,B=1,X,G=O]-\mathbb{E}[Y-M\mid A=0,B=0,X,G=O]}{p(A=0\mid B=1,X,G=O)-p(A=0\mid B=0,X,G=O)}-\psi_3\bigg\}
\end{align*}}

The influence function for $\psi_5$ is:
\begin{align*}
    &\frac{1}{p(A=1,G=O)}\bigg\{\frac{I(A=0)I(G=E)}{p(A=0\mid X,B,G=E)}\{\frac{1}{p(G=E\mid X,B)}-1\}\{M-\mathbb{E}[M\mid A=0,X,B,G=E]\}\\
    &\ \ \ \ \ \ \ \ \ +I(G=O)\mathbb{E}[M\mid A=0,X,B,G=E]-I(A=1)I(G=O)\psi_5\bigg\}
\end{align*}

The influence function for $\psi_8$ is:
\begin{align*}
    \frac{I(G=O)}{p(A=1,G=O)}\{M-I(A=1)\psi_8\}
\end{align*}

For $i\in\{3,5,6,7,8\}$, denote the obtained influence functions by $IF_{\psi_i}$, for $i\in\{1,2,4\}$ and $a,b\in\{0,1\}$, denote the obtained influence functions by $IF_{\psi_i^{(ab)}}$. The influence function for $\psi_{\text{ETT}}^{\text{bsiv2}}$ can be obtained as $IF_{\psi_{\text{ETT}}^{\text{bsiv2}}}=IF_{\psi_1^{(11)}}-IF_{\psi_1^{(01)}}-IF_{\psi_1^{(10)}}+IF_{\psi_1^{(00)}}+IF_{\psi_2^{(10)}}-IF_{\psi_2^{(00)}}+IF_{\psi_3}-IF_{\psi_5}+IF_{\psi_8}$. Therefore, using the notations specified in Theorem \ref{thm:IF:BSIV-1},
{\footnotesize \begin{align*}
&\frac{I(G=O)}{p(A=1,G=O)}\Bigg\{I(A=1)I(B=1)\{Y-M-E_{01}^O(X)-E_{10}^O(X)+E_{00}^O(X)\}\\
&\quad-I(A=0)I(B=1)\frac{P_{11}^O(X)}{P_{01}^O(X)}\{Y-M-E_{01}^O(X)\}+I(A=1)\{E_{10}^O(X)-E_{00}^O(X)\}\\
&\quad-I(B=0)\frac{\rho_{1}^O(X)P_{11}^O(X)}{\rho_{0}^O(X)P_{10}^O(X)}\big\{I(A=1)\{Y-M-E_{10}^O(X)\}+I(A=0)\{Y-M-E_{00}^O(X)\}\big\}\\
&\quad+\frac{I(B=0)}{P_{10}^O(X)}\frac{\pi^O(X)}{\rho_{0}^O(X)}\cdot\big\{I(A=1)\{Y-M-E_{10}^O(X)\}+I(A=0)\{Y-M-E_{00}^O(X)\}\big\}\\
&\quad+\frac{\pi^O(X)}{P_{01}^O(X)-P_{00}^O(X)}\Big\{I(A=0)\big\{\frac{I(B=1)}{P_{01}^O(X)\rho_1^O(X)}\{Y-M-E_{01}^O(X)\}-\frac{I(B=0)}{P_{00}^O(X)\rho_0^O(X)}\{Y-M-E_{00}^O(X)\}\big\}\\
&\quad\quad+\frac{E_{01}^O(X)-E_{00}^O(X)}{P_{01}^O(X)-P_{00}^O(X)}\big\{-\frac{I(B=1)}{\rho_1^O(X)}\{I(A=0)-P_{01}^O(X)\}+\frac{I(B=0)}{\rho_0^O(X)}\{I(A=0)-P_{00}^O(X)\}\big\}\Big\}\\
&\quad+I(A=1)\{\frac{E_{01}^O(X)-E_{00}^O(X)}{P_{01}^O(X)-P_{00}^O(X)}\}+M-M_0^E(B,X)-I(A=1)\psi^2_{ETT}\Bigg\}\\
&\quad-\frac{1}{p(A=1,G=O)}\cdot\frac{I(A=0)I(G=E)}{1-P_{1B}^E(X)}\cdot\frac{1-\tau(B,X)}{\tau(B,X)}\cdot\{M-M_0^E(B,X)\}
\end{align*}}
is the influence function for $\psi_{\text{ETT}}^{\text{bsiv2}}$.

\end{proof}

\begin{proof}[Proof of Theorem \ref{thm:IF:BSIV-1}: $\psi_{\text{ATE}}^{\text{bsiv1}}$]
	
Define
\begin{align*}
&\psi_1=
\E\big[\frac{\E[Y-M\mid B=1,X,G=O]-\E[Y-M\mid B=0,X,G=O]}{p(A=1\mid B=1,X,G=O)-p(A=1\mid B=0,X,G=O)}~\big|~ G=O\big],\\
&\psi_2^{(a)}=
\E\big[\E[M\mid A=a,B,X,G=E]~\big|~ G=O\big].
\end{align*}
We use the notation $\partial_tf(t)$ to denote $\frac{\partial f(t)}{\partial t}\big|_{t=0}$.
For parameter $\psi$, let $\psi_t$ be the parameter under a regular parametric sub-model indexed by $t$, that includes the ground-truth model at $t=0$.
Let $V$ be the set of all observed variables.
In order to obtain an influence function, we need to find a random variable $\Gamma$ with mean zero, that satisfies
\[
\partial_t\psi_t=\E[\Gamma S(V)],
\]
where $S(V)=\partial_t\log p_t(V)$.	

For $\psi_1$, we have
\begingroup
\allowdisplaybreaks
{\footnotesize \begin{align*}
&\partial_t{\psi_1}_t	\\
&=\E\Bigg[\Bigg\{\frac{I(G=O)}{p(A=1\mid B=1,X,G=O)-p(A=1\mid B=0,X,G=O)}\cdot\frac{1}{p(G=O)}
\cdot\frac{1}{p(B\mid X,G=O)}\\
&\quad\bigg\{I(B=1)\{Y-M-\E[Y-M\mid B=1,X,G=O]\}
-I(B=0)\{Y-M-\E[Y-M\mid B=0,X,G=O]\}\\
&\quad+\frac{\E[Y-M\mid B=1,X,G=O]-\E[Y-M\mid B=0,X,G=O]}{p(A=1\mid B=1,X,G=O)-p(A=1\mid B=0,X,G=O)}\\
&\quad\big\{-I(B=1)\{I(A=1)-p(A=1\mid B=1,X,G=O)\}+I(B=0)\{I(A=1)-p(A=1\mid B=0,X,G=O)\}\big\}\bigg\}\\
&\quad+\frac{I(G=O)}{p(G=O)}\bigg\{\frac{\E[Y-M\mid B=1,x,G=O]-\E[Y-M\mid B=0,x,G=O]}{p(A=1\mid B=1,x,G=O)-p(A=1\mid B=0,x,G=O)}-\psi_1\bigg\}\Bigg\}S(V)\Bigg].
\end{align*}}
\endgroup

For $\psi_2^{(a)}$, similar to the parameter $\psi_1^{(a)}$ in the proof of Theorem \ref{thm:IF:equi-conf:ATE}, we have
{\footnotesize \begin{align*}
\partial_t\psi^{(a)}_{2_t}
&=\E
\Big[\Big\{
\frac{I(A=a)}{p(A=a\mid X,B,G=E)}\cdot\frac{I(G=E)}{p(G=O)}
\{M-\E[M\mid A=a,X,B,G=E]\}\{\frac{1}{p(G=E\mid X,B)}-1\}\\
&\qquad+\frac{I(G=O)}{p(G=O)}\{\E[M\mid A=a,X,B,G=E]-\psi_2^{(a)}\}
\Big\}S(V)\Big].
\end{align*}}

For $i=1$, denote the obtained influence functions by $IF_{\psi_i}$, for $i=2$ and $a\in\{0,1\}$, denote the obtained influence functions by $IF_{\psi_i^{(a)}}$. The influence function for $\psi_{\text{ATE}}^{\text{bsiv1}}$ can be obtained as $IF_{\psi_{\text{ATE}}^{\text{bsiv1}}}=IF_{\psi_1}+IF_{\psi_2^{(1)}}-IF_{\psi_2^{(0)}}$. Therefore, using the notations specified in Theorem \ref{thm:IF:BSIV-1},

\begin{align*}
&\frac{I(G=O)}{p(G=O)}\Bigg\{\frac{1}{P_{11}^O(X)-P_{10}^O(X)}
\frac{1}{\rho_B^O(X)}\bigg\{I(B=1)\{Y-M-e_1^O(X)\}
-I(B=0)\{Y-M-e_0^O(X)\}\\
&\quad\quad\quad\quad\quad+\frac{e_1^O(X)-e_0^O(X)}{P_{11}^O(X)-P_{10}^O(X)}\big\{I(B=0)\{I(A=1)-P_{10}^O(X)\}-I(B=1)\{I(A=1)-P_{11}^O(X)\}\big\}\bigg\}\\
&\quad\quad\quad\quad\quad+\frac{e_1^O(X)-e_0^O(X)}{P_{11}^O(X)-P_{10}^O(X)}+M_{1}^{E}(B,X)-M_{0}^{E}(B,X)-\psi^1_{ATE}\Bigg\}\\
&\quad+\frac{I(G=E)}{p(G=O)}\cdot\frac{1-\tau(B,X)}{\tau(B,X)}\bigg\{\frac{I(A=1)}{P_{1B}^E(X)}\cdot\{M-M_{1}^{E}(B,X)\}-\frac{I(A=0)}{1-P_{1B}^E(X)}\cdot\{M-M_{0}^{E}(B,X)\}\bigg\}
\end{align*}
is the influence function of $\psi_{\text{ATE}}^{\text{bsiv1}}$.

\end{proof}

\begin{proof}[Proof of Theorem \ref{thm:IF:BSIV-1}: $\psi_{\text{ATE}}^{\text{bsiv2}}$]
	
Define
\begin{align*}
        &\psi_1^{(ab)}=\mathbb{E}\big[\mathbb{E}[Y-M\mid A=a,B=b,X,G=O]B\mid G=O\big],\\
    &\psi_2^{(ab)}=\mathbb{E}\big[\mathbb{E}[Y-M\mid A=a,B=b,X,G=O]\mid G=O\big],\\
    &\psi_3=\mathbb{E}\bigg[\frac{\mathbb{E}[Y-M\mid A=1,B=1,X,G=O]-\mathbb{E}[Y-M\mid A=1,B=0,X,G=O]}{p(A=0\mid B=1,X,G=O)-p(A=0\mid B=0,X,G=O)}\bigg|G=O\bigg],\\
    &\psi_4^{(ab)}=\mathbb{E}\bigg[\frac{\mathbb{E}[Y-M\mid A=a,B=b,X,G=O]p(A=1\mid B,X,G=O)}{p(A=0\mid B=1,X,G=O)-p(A=0\mid B=0,X,G=O)}\bigg|G=O\bigg],\\
    &\psi_5^{(a)}=\mathbb{E}[\mathbb{E}[M\mid A=a,B,X,G=E]\mid G=O]
    \end{align*}
We use the notation $\partial_tf(t)$ to denote $\frac{\partial f(t)}{\partial t}\big|_{t=0}$.
For parameter $\psi$, let $\psi_t$ be the parameter under a regular parametric sub-model indexed by $t$, that includes the ground-truth model at $t=0$.
Let $V$ be the set of all observed variables.
In order to obtain an influence function, we need to find a random variable $\Gamma$ with mean zero, that satisfies
\[
\partial_t\psi_t=\E[\Gamma S(V)],
\]
where $S(V)=\partial_t\log p_t(V)$.

    For $\psi_1^{(ab)}$, note that
\begin{align*}
    \partial_t\psi_{1t}^{(ab)}&=\partial_t\sum_{x,\Tilde{b}}\mathbb{E}_t[Y-M\mid A=a,B=b,x,G=O]\Tilde{B}p_t(x,\Tilde{b}\mid G=O)\\
    &=\sum_{x,\Tilde{b}}\partial_t\mathbb{E}_t[Y-M\mid A=a,B=b,x,G=O]\Tilde{b}p(x,\Tilde{b}\mid G=O)\\
    &\ \ \ \ + \sum_{x,\Tilde{b}}\mathbb{E}[Y-M\mid A=a,B=b,x,G=O]\Tilde{b}\partial_tp_t(x,\Tilde{b}\mid G=O)\\
    &=\sum_{x}\partial_t\mathbb{E}_t[Y-M\mid A=a,B=b,x,G=O]p(x,B=1\mid G=O)\\
    &\ \ \ \ + \sum_{x}\mathbb{E}[Y-M\mid A=a,B=b,x,G=O]\partial_tp_t(x,B=1\mid G=O)
\end{align*}

For the first term, note that
{\footnotesize \begin{align*}
    &\sum_{x}\partial_t\mathbb{E}_t[Y-M\mid A=a,B=b,x,G=O]p(x,B=1\mid G=O)\\
    &=\sum_{y,m,x}(y-m)S(y,m\mid A=a,B=b,x,G=O)p(y,m\mid A=a,B=b,x,G=O)p(x,B=1\mid G=O)\\
    &=\sum_{y,m,x,\Tilde{b},\Tilde{a},g}I(\Tilde{b}=b)I(\Tilde{a}=a)I(g=O)(y-m)S(y,m\mid \Tilde{a},\Tilde{b},x,g)p(y,m\mid \Tilde{a},\Tilde{b},x,g)p(x,B=1\mid G=O)\\
    &=\sum_{y,m,x,\Tilde{b},\Tilde{a},g}\frac{I(\Tilde{b}=b)I(\Tilde{a}=a)I(g=O)}{p(G=O)}(y-m)S(y,m\mid \Tilde{a},\Tilde{b},x,g)\frac{p(B=1\mid x,G=O)}{p(\Tilde{b}\mid x,G=O)p(\Tilde{a}\mid\Tilde{b},x,G=O)}p(y,m, \Tilde{a},\Tilde{b},x,g)\\
    &=\mathbb{E}\bigg[\frac{I(B=b)I(A=a)I(G=O)}{p(G=O)}\frac{p(B=1\mid X,G=O)}{p(B\mid X,G=O)p(A\mid B, X,G=O)}(Y-M)S(Y,M\mid A,B,X,G)\bigg]\\
    &=\mathbb{E}\bigg[\frac{I(B=b)I(A=a)I(G=O)}{p(G=O)}\frac{p(B=1\mid X,G=O)}{p(B\mid X,G=O)p(A\mid B, X,G=O)}\\
    &\ \ \ \ \ \ \ \ \ \ \{Y-M-\mathbb{E}[Y-M\mid A=a,B=b,X,G=O]\}S(Y,M\mid A,B,X,G)\bigg]
\end{align*}}

Note that,
\begin{align*}
    &\mathbb{E}\bigg[\frac{I(B=b)I(A=a)I(G=O)}{p(G=O)}\frac{p(B=1\mid X,G=O)}{p(B\mid X,G=O)p(A\mid B, X,G=O)}\\
    &\ \ \ \ \ \ \ \ \ \ \{Y-M-\mathbb{E}[Y-M\mid A=a,B=b,X,G=O]\}S(A,B,X,G)\bigg]=0.
\end{align*}

Therefore,
\begin{align*}
    &\sum_{x}\partial_t\mathbb{E}_t[Y-M\mid A=a,B=b,x,G=O]p(x,B=1\mid A=1, G=O)\\
    &=\mathbb{E}\bigg[\frac{I(B=b)I(A=a)I(G=O)}{p(G=O)}\frac{p(B=1\mid X,G=O)}{p(B\mid X,G=O)p(A\mid B, X,G=O)}\\
    &\ \ \ \ \ \ \ \ \ \ \{Y-M-\mathbb{E}[Y-M\mid A=a,B=b,X,G=O]\}S(V)\bigg]
\end{align*}

For the second term, note that,
{\footnotesize \begin{align*}
    &\sum_{x}\mathbb{E}[Y-M\mid A=a,B=b,x,G=O]\partial_tp_t(x,B=1\mid  G=O)\\
    &=\sum_{x}\mathbb{E}[Y-M\mid A=a,B=b,x,G=O]p(x,B=1\mid G=O)S(x,B=1\mid G=O)\\
    &=\sum_{x,\Tilde{b},g}\frac{I(\Tilde{b}=1)I(g=O)}{p(G=O)}\mathbb{E}[Y-M\mid A=a,B=b,x,G=O]p(x,\Tilde{b}, g)S(x,B=1\mid G=O)\\
    &=\mathbb{E}\bigg[\frac{I(B=1)I(G=O)}{p(G=O)}\mathbb{E}[Y-M\mid A=a,B=b,X,G=O]p(X,B, G)S(X,B\mid G)\bigg]\\
    &=\mathbb{E}\bigg[\frac{I(B=1)I(G=O)}{p(G=O)}\{\mathbb{E}[Y-M\mid A=a,B=b,X,G=O]\\
    &~~~~~~~~-\mathbb{E}[\mathbb{E}[Y-M\mid A=a,B=b,X,G=O]B\mid G=O]\}S(X,B\mid G)\bigg]\\
    &=\mathbb{E}\bigg[\frac{I(B=1)I(G=O)}{p(G=O)}\{\mathbb{E}[Y-M\mid A=a,B=b,X,G=O]-\psi_1^{(ab)}\}S(X,B\mid G)\bigg]
\end{align*}}

Note that,
\begin{align*}
    \mathbb{E}\bigg[\frac{I(B=1)I(G=O)}{p(G=O)}\{\mathbb{E}[Y-M\mid A=a,B=b,X,G=O]-\psi_1^{(ab)}\}S(G)\bigg]=0.
\end{align*}

Therefore,
\begin{align*}
    &\sum_{x}\mathbb{E}[Y-M\mid A=a,B=b,x,G=O]\partial_tp_t(x,B=1\mid G=O)\\
    &=\mathbb{E}\bigg[\frac{I(B=1)I(G=O)}{p(G=O)}\{\mathbb{E}[Y-M\mid A=a,B=b,X,G=O]-\psi_1^{(ab)}\}S(V)\bigg]
\end{align*}

Combining concludes that,
\begin{align*}
    \partial_t\psi_{1t}^{(ab)}&=\mathbb{E}\bigg[\frac{I(G=O)}{p(G=O)}\bigg\{I(A=a)I(B=b)\frac{p(B=1\mid X,G=O)}{p(B\mid X,G=O)p(A\mid B, X,G=O)}\\
    &\ \ \ \ \ \ \ \ \ \{Y-M-\mathbb{E}[Y-M\mid A=a,B=b,X,G=O]\}\\
    &\ \ \ \ \ \ \ +I(B=1)\{\mathbb{E}[Y-M\mid A=a,B=b,X,G=O]-\psi_1^{(ab)}\}\bigg\}S(V)\bigg]
\end{align*}

For $\psi_2^{(ab)}$, note that
\begin{align*}
    \partial_t\psi_{2t}^{(ab)}&=\partial_t\sum_{x}\mathbb{E}_t[Y-M\mid A=a,B=b,x,G=O]p_t(x\mid G=O)\\
    &=\sum_{x}\partial_t\mathbb{E}_t[Y-M\mid A=a,B=b,x,G=O]p(x\mid G=O)\\
    &\ \ \ \ + \sum_{x}\mathbb{E}[Y-M\mid A=a,B=b,x,G=O]\partial_tp_t(x\mid G=O)
\end{align*}

For the first term, note that
{\footnotesize \begin{align*}
    &\sum_{x}\partial_t\mathbb{E}_t[Y-M\mid A=a,B=b,x,G=O]p(x\mid  G=O)\\
    &=\sum_{y,m,x}(y-m)S(y,m\mid B=b,A=a,x,G=O)p(y,m\mid B=b,A=a,x,G=O)p(x\mid G=O)\\
    &=\sum_{y,m,x}\frac{1}{p(G=O)}\frac{p(x, G=O)}{p(B=b,A=a,x,G=O)}(y-m)S(y,m\mid B=b,A=a,x,G=O)p(y,m, B=b,A=a,x,G=O)\\
    &=\sum_{y,m,x,\Tilde{b},\Tilde{a},g}\frac{I(\Tilde{b}=b)I(\Tilde{a}=a)I(g=O)}{p(G=O)}\frac{1}{p(A=a\mid B=b,x,G=O)}\frac{1}{p(B=b\mid x,G=O)}\\
    &~~~~~~~~~~~~~~~~\cdot(y-m)S(y,m\mid \Tilde{b},\Tilde{a},x,g)p(y,m, \Tilde{b},\Tilde{a},x,g)\\
    &=\mathbb{E}\bigg[\frac{I(B=b)I(A=a)I(G=O)}{p(G=O)}\frac{1}{p(A=a\mid B=b,X,G=O)}\frac{1}{p(B=b\mid X,G=O)}(Y-M)S(Y,M\mid B,A,X,G)\bigg]\\
    &=\mathbb{E}\bigg[\frac{I(B=b)I(A=a)I(G=O)}{p(G=O)}\frac{1}{p(A=a\mid B=b,X,G=O)}\frac{1}{p(B=b\mid X,G=O)}\times\\
    &\ \ \ \ \ \ \ \ (Y-M-\mathbb{E}[Y-M\mid B=b,A=a,X,G=O])S(Y,M\mid B,A,X,G)\bigg].
\end{align*}}

Note that,
\begin{align*}
    &\mathbb{E}\bigg[\frac{I(B=b)I(A=a)I(G=O)}{p(G=O)}\frac{1}{p(A=a\mid B=b,X,G=O)}\frac{1}{p(B=b\mid X,G=O)}\times\\
    &\ \ \ \ \ \ \ \ (Y-M-\mathbb{E}[Y-M\mid B=b,A=a,X,G=O])S(B,A,X,G)\bigg]=0.
\end{align*}

Therefore,
\begin{align*}
    &\sum_{x}\partial_t\mathbb{E}_t[Y-M\mid A=a,B=b,x,G=O]p(x\mid G=O)\\
    &=\mathbb{E}\bigg[\frac{I(B=b)I(A=a)I(G=O)}{p(G=O)}\frac{1}{p(A=a\mid B=b,X,G=O)}\frac{1}{p(B=b\mid X,G=O)}\times\\
    &\ \ \ \ \ \ \ \ (Y-M-\mathbb{E}[Y-M\mid B=b,A=a,X,G=O])S(V)\bigg].
\end{align*}

For the second term, note that
\begin{align*}
    &\sum_{x}\mathbb{E}[Y-M\mid A=a,B=b,x,G=O]\partial_tp_t(x\mid  G=O)\\
    &=\sum_{x}\mathbb{E}[Y-M\mid A=a,B=b,x,G=O]S(x\mid  G=O)p(x\mid G=O)\\
    &=\sum_{x,g}\frac{I(g=O)}{p( G=O)}\mathbb{E}[Y-M\mid A=a,B=b,x,G=O]S(x\mid g)p(x,g)\\
    &=\mathbb{E}\bigg[\frac{I(G=O)}{p(G=O)}\mathbb{E}[Y-M\mid A=a,B=b,X,G=O]S(X\mid G)\bigg]\\
    &=\mathbb{E}\bigg[\frac{I(G=O)}{p(G=O)}\bigg\{\mathbb{E}[Y-M\mid A=a,B=b,X,G=O]-\psi_2^{(ab)}\bigg\}S(X\mid G)\bigg]
\end{align*}

Note that,
\begin{align*}
    &\mathbb{E}\bigg[\frac{I(G=O)}{p(G=O)}\bigg\{\mathbb{E}[Y-M\mid A=a,B=b,X,G=O]-\psi_2^{(ab)}\bigg\}S(G)\bigg]=0.
\end{align*}

Therefore,
\begin{align*}
    &\sum_{x}\mathbb{E}[Y-M\mid A=a,B=b,x,G=O]\partial_tp_t(x\mid G=O)\\
    &=\mathbb{E}\bigg[\frac{I(G=O)}{p(G=O)}\bigg\{\mathbb{E}[Y-M\mid A=a,B=b,X,G=O]-\psi_2^{(ab)}\bigg\}S(V)\bigg]
\end{align*}

Combining concludes that,
\begin{align*}
    \partial_t\psi_{2t}^{(ab)}&=\mathbb{E}\bigg[\frac{I(G=O)}{p( G=O)}\bigg\{\frac{I(B=b)I(A=a)}{p(A=a\mid B=b,X,G=O)p(B=b\mid X,G=O)}\times\\
    &~~~~~~~~~~~~~~~~~~~~~~ \{Y-M-\mathbb{E}[Y-M\mid B=b,A=a,X,G=O]\}\\
    &~~~~~~~~~~~~+\{\mathbb{E}[Y-M\mid A=a,B=b,X,G=O]-\psi_2^{(ab)}\}\bigg\}S(V)\bigg]
\end{align*}

For $\psi_3$, note that,
{\footnotesize \begin{align*}
    \partial_t\psi_{3t}&=\partial_t\sum_{x}\frac{\mathbb{E}_t[Y-M\mid A=1,B=1,x,G=O]-\mathbb{E}_t[Y-M\mid A=1,B=0,x,G=O]}{p_t(A=0\mid B=1,x,G=O)-p_t(A=0\mid B=0,x,G=O)}p_t(x\mid G=O)\\
    &=\sum_{x}\partial_t\frac{\mathbb{E}_t[Y-M\mid A=1,B=1,x,G=O]-\mathbb{E}_t[Y-M\mid A=1,B=0,x,G=O]}{p_t(A=0\mid B=1,x,G=O)-p_t(A=0\mid B=0,x,G=O)}p(x\mid G=O)\\
    &\ \ \ \ +\sum_{x}\frac{\mathbb{E}[Y-M\mid A=1,B=1,x,G=O]-\mathbb{E}[Y-M\mid A=1,B=0,x,G=O]}{p(A=0\mid B=1,x,G=O)-p(A=0\mid B=0,x,G=O)}\partial_tp_t(x\mid G=O)
\end{align*}}

For the first term, note that
{\footnotesize \begin{align*}
    &\sum_{x}\partial_t\frac{\mathbb{E}_t[Y-M\mid A=1,B=1,x,G=O]}{p_t(A=0\mid B=1,x,G=O)-p_t(A=0\mid B=0,x,G=O)}p(x\mid G=O)\\
    &=\sum_{x}\frac{\partial_t\mathbb{E}_t[Y-M\mid A=1,B=1,x,G=O]}{p(A=0\mid B=1,x,G=O)-p(A=0\mid B=0,x,G=O)}p(x\mid G=O)\\
    &\ \ \ \ -\sum_{x}\frac{\mathbb{E}[Y-M\mid A=1,B=1,x,G=O]\partial_tp_t(A=0\mid B=1,x,G=O)}{\{p(A=0\mid B=1,x,G=O)-p(A=0\mid B=0,x,G=O)\}^2}p(x\mid G=O)\\
    &\ \ \ \ +\sum_{x}\frac{\mathbb{E}[Y-M\mid A=1,B=1,x,G=O]\partial_tp_t(A=0\mid B=0,x,G=O)}{\{p(A=0\mid B=1,x,G=O)-p(A=0\mid B=0,x,G=O)\}^2}p(x\mid G=O)\\
    &=\sum_{y,m,x,a,b,g}\frac{I(a=1)I(b=1)I(g=O)}{p(A=0\mid B=1,x,G=O)-p(A=0\mid B=0,x,G=O)}\cdot\frac{1}{p(G=O)}\\
    &\ \ \ \ \ \ \ \ \cdot\frac{1}{p(B=1\mid x,G=O)}\frac{1}{p(A=1\mid B=1,x,G=O)}\{y-m\}S(y,m\mid x,a,b,g)p(y,m,x,a,b,g)\\
    &\ \ \ \ -\sum_{a,b,x,g}\frac{\mathbb{E}[Y-M\mid A=1,B=1,x,G=O]}{\{p(A=0\mid B=1,x,G=O)-p(A=0\mid B=0,x,G=O)\}^2}\frac{I(b=1)I(g=O)}{p(G=O)}\cdot\\
    &\ \ \ \ \ \ \ \ \cdot\frac{1}{p(B=1\mid x, G=O)}I(a=0)S(a\mid b,x,g)p(a,b,x,g)\\
    &\ \ \ \ +\sum_{a,b,x,g}\frac{\mathbb{E}[Y-M\mid A=1,B=1,x,G=O]}{\{p(A=0\mid B=1,x,G=O)-p(A=0\mid B=0,x,G=O)\}^2}\frac{I(b=0)I(g=O)}{p(G=O)}\cdot\\
    &\ \ \ \ \ \ \ \ \cdot\frac{1}{p(B=0\mid x, G=O)}I(a=0)S(a\mid b,x,g)p(a,b,x,g)\\
    &=\mathbb{E}\bigg[\frac{I(A=1)I(B=1)I(G=O)}{p(A=0\mid B=1,X,G=O)-p(A=0\mid B=0,X,G=O)}\cdot\frac{1}{p(G=O)}\\
    &\ \ \ \ \ \ \ \ \cdot\frac{1}{p(B=1\mid X,G=O)}\frac{1}{p(A=1\mid B=1,X,G=O)}\{Y-M\}S(Y,M\mid X,A,B,G)\bigg]\\
    &\ \ \ \ -\mathbb{E}\bigg[\frac{\mathbb{E}[Y-M\mid A=1,B=1,X,G=O]}{\{p(A=0\mid B=1,X,G=O)-p(A=0\mid B=0,X,G=O)\}^2}\frac{I(B=1)I(G=O)}{p(G=O)}\cdot\\
    &\ \ \ \ \ \ \ \ \cdot\frac{1}{p(B=1\mid X, G=O)}I(A=0)S(A\mid B,X,G)\bigg]\\
    &\ \ \ \ +\mathbb{E}\bigg[\frac{\mathbb{E}[Y-M\mid A=1,B=1,X,G=O]}{\{p(A=0\mid B=1,X,G=O)-p(A=0\mid B=0,X,G=O)\}^2}\frac{I(B=0)I(G=O)}{p(G=O)}\cdot\\
    &\ \ \ \ \ \ \ \ \cdot\frac{1}{p(B=0\mid X, G=O)}I(A=0)S(A\mid B,X,G)\bigg]\\
    &=\mathbb{E}\bigg[\frac{I(A=1)I(B=1)I(G=O)}{p(A=0\mid B=1,X,G=O)-p(A=0\mid B=0,X,G=O)}\cdot\frac{1}{p(G=O)}\\
    &\ \ \ \ \ \ \ \ \cdot\frac{1}{p(B=1\mid X,G=O)}\frac{1}{p(A=1\mid B=1,X,G=O)}\{Y-M-\mathbb{E}[Y-M\mid A=1,B=1,X,G=O]\}S(V)\bigg]\\
    &\ \ \ \ -\mathbb{E}\bigg[\frac{\mathbb{E}[Y-M\mid A=1,B=1,X,G=O]}{\{p(A=0\mid B=1,X,G=O)-p(A=0\mid B=0,X,G=O)\}^2}\frac{I(B=1)I(G=O)}{p(G=O)}\cdot\\
    &\ \ \ \ \ \ \ \ \cdot\frac{1}{p(B=1\mid X, G=O)}\{I(A=0)-p(A=0\mid B=1,X,G=O)\}S(V)\bigg]\\
    &\ \ \ \ +\mathbb{E}\bigg[\frac{\mathbb{E}[Y-M\mid A=1,B=1,X,G=O]}{\{p(A=0\mid B=1,X,G=O)-p(A=0\mid B=0,X,G=O)\}^2}\frac{I(B=0)I(G=O)}{p(G=O)}\cdot\\
    &\ \ \ \ \ \ \ \ \cdot\frac{1}{p(B=0\mid X, G=O)}\{I(A=0)-p(A=0\mid B=0,X,G=O)\}S(V)\bigg].
\end{align*}}

Similarly,
{\footnotesize \begin{align*}
    &\sum_{x}\partial_t\frac{\mathbb{E}_t[Y-M\mid A=1,B=0,x,G=O]}{p_t(A=0\mid B=1,x,G=O)-p_t(A=0\mid B=0,x,G=O)}p(x\mid G=O)\\
    &=\mathbb{E}\bigg[\frac{I(A=1)I(B=0)I(G=O)}{p(A=0\mid B=1,X,G=O)-p(A=0\mid B=0,X,G=O)}\cdot\frac{1}{p(G=O)}\\
    &\ \ \ \ \ \ \ \ \cdot\frac{1}{p(B=0\mid X,G=O)}\frac{1}{p(A=1\mid B=0,X,G=O)}\{Y-M-\mathbb{E}[Y-M\mid A=1,B=0,X,G=O]\}S(V)\bigg]\\
    &\ \ \ \ -\mathbb{E}\bigg[\frac{\mathbb{E}[Y-M\mid A=1,B=0,X,G=O]}{\{p(A=0\mid B=1,X,G=O)-p(A=0\mid B=0,X,G=O)\}^2}\frac{I(B=1)I(G=O)}{p(G=O)}\cdot\\
    &\ \ \ \ \ \ \ \ \cdot\frac{1}{p(B=1\mid X, G=O)}\{I(A=0)-p(A=0\mid B=1,X,G=O)\}S(V)\bigg]\\
    &\ \ \ \ +\mathbb{E}\bigg[\frac{\mathbb{E}[Y-M\mid A=1,B=0,X,G=O]}{\{p(A=0\mid B=1,X,G=O)-p(A=0\mid B=0,X,G=O)\}^2}\frac{I(B=0)I(G=O)}{p(G=O)}\cdot\\
    &\ \ \ \ \ \ \ \ \cdot\frac{1}{p(B=0\mid X, G=O)}\{I(A=0)-p(A=0\mid B=0,X,G=O)\}S(V)\bigg].
\end{align*}}

Therefore,
{\footnotesize \begin{align*}
    &\sum_{x}\partial_t\frac{\mathbb{E}_t[Y-M\mid A=1,B=1,x,G=O]-\mathbb{E}_t[Y-M\mid A=1,B=0,x,G=O]}{p_t(A=0\mid B=1,x,G=O)-p_t(A=0\mid B=0,x,G=O)}p(x\mid G=O)\\
    &=\mathbb{E}\bigg[\frac{I(G=O)}{p(A=0\mid B=1,X,G=O)-p(A=0\mid B=0,X,G=O)}\frac{1}{p(G=O)}\\
    &\ \ \ \ \ \ \ \bigg\{I(A=1)\big\{\frac{I(B=1)}{p(A=1\mid B=1,X,G=O)p(B=1\mid X,G=O)}\{Y-M-\mathbb{E}[Y-M\mid A=1,B=1,X,G=O]\}\\
    &\ \ \ \ \ \ \ \ \ \ \ -\frac{I(B=0)}{p(A=1\mid B=0,X,G=O)p(B=0\mid X,G=O)}\{Y-M-\mathbb{E}[Y-M\mid A=1,B=0,X,G=O]\}\big\}\\
    &\ \ \ \ \ \ \ \ +\frac{\mathbb{E}[Y-M\mid A=1,B=1,X,G=O]-\mathbb{E}[Y-M\mid A=1,B=0,X,G=O]}{p(A=0\mid B=1,X,G=O)-p(A=0\mid B=0,X,G=O)}\\
    &\ \ \ \ \ \ \ \ \ \ \ \ \cdot\big\{-\frac{I(B=1)}{p(B=1\mid X, G=O)}\{I(A=0)-p(A=0\mid B=1,X,G=O)\}\\
    &\ \ \ \ \ \ \ \ \ \ \ +\frac{I(B=0)}{p(B=0\mid X, G=O)}\{I(A=0)-p(A=0\mid B=0,X,G=O)\}\big\}\bigg\}S(V)\bigg]
\end{align*}}

For the second term, we have
{\footnotesize \begin{align*}
    &\sum_{x}\frac{\mathbb{E}[Y-M\mid A=0,B=1,x,G=O]-\mathbb{E}[Y-M\mid A=0,B=0,x,G=O]}{p(A=0\mid B=1,x,G=O)-p(A=0\mid B=0,x,G=O)}\partial_tp_t(x\mid G=O)\\
    &=\mathbb{E}\bigg[\frac{I(G=O)}{p(G=O)}\cdot\frac{\mathbb{E}[Y-M\mid A=1,B=1,X,G=O]-\mathbb{E}[Y-M\mid A=1,B=0,X,G=O]}{p(A=0\mid B=1,X,G=O)-p(A=0\mid B=0,x,G=O)}\cdot S(X\mid G)\bigg]\\
    &=\mathbb{E}\bigg[\frac{I(G=O)}{p(G=O)}\bigg\{\frac{\mathbb{E}[Y-M\mid A=1,B=1,X,G=O]-\mathbb{E}[Y-M\mid A=1,B=0,X,G=O]}{p(A=0\mid B=1,X,G=O)-p(A=0\mid B=0,x,G=O)}-\psi_3\bigg\}\cdot S(V)\bigg].
\end{align*}}

Therefore,
{\footnotesize \begin{align*}
    &\partial_t\psi_{3t}\\
    &=\mathbb{E}\Bigg[\Bigg\{\frac{I(G=O)}{p(A=0\mid B=1,X,G=O)-p(A=0\mid B=0,X,G=O)}\frac{1}{p(G=O)}\\
    &\ \ \ \ \ \ \ \ \ \ \bigg\{I(A=1)\big\{\frac{I(B=1)}{p(A=1\mid B=1,X,G=O)p(B=1\mid X,G=O)}\{Y-M-\mathbb{E}[Y-M\mid A=1,B=1,X,G=O]\}\\
    &\ \ \ \ \ \ \ \ \ \ \ \ \ \ \ \ \ \ \ \ \ \ \ \ -\frac{I(B=0)}{p(A=1\mid B=0,X,G=O)p(B=0\mid X,G=O)}\{Y-M-\mathbb{E}[Y-M\mid A=1,B=0,X,G=O]\}\big\}\\
    &\ \ \ \ \ \ \ \ \ \ \ \ +\frac{\mathbb{E}[Y-M\mid A=1,B=1,X,G=O]-\mathbb{E}[Y-M\mid A=1,B=0,X,G=O]}{p(A=0\mid B=1,X,G=O)-p(A=0\mid B=0,X,G=O)}\\
    &\ \ \ \ \ \ \ \ \ \ \ \ \ \ \ \ \cdot\big\{-\frac{I(B=1)}{p(B=1\mid X, G=O)}\{I(A=0)-p(A=0\mid B=1,X,G=O)\}\\
    &\ \ \ \ \ \ \ \ \ \ \ \ \ \ \ \ \ \ \ \ +\frac{I(B=0)}{p(B=0\mid X, G=O)}\{I(A=0)-p(A=0\mid B=0,X,G=O)\}\big\}\bigg\}\\
    &\ \ \ \ \ \ \ \ \  +\frac{I(G=O)}{p(G=O)}\bigg\{\frac{\mathbb{E}[Y-M\mid A=1,B=1,X,G=O]-\mathbb{E}[Y-M\mid A=1,B=0,X,G=O]}{p(A=0\mid B=1,X,G=O)-p(A=0\mid B=0,x,G=O)}-\psi_3\bigg\}\Bigg\}S(V)\Bigg]
\end{align*}}

 For $\psi_4^{(ab)}$, note that,
 \begin{align*}
     \partial_t\psi_{4t}^{(ab)}&=\partial_t\sum_{x,\Tilde{b}}\frac{\mathbb{E}_t[Y-M\mid A=a,B=b,x,G=O]p_t(A=1\mid \Tilde{b},x,G=O)}{p_t(A=0\mid B=1,x,G=O)-p_t(A=0\mid B=0,x,G=O)}p_t(x,\Tilde{b}\mid G=O)\\
     &=\sum_{x,\Tilde{b}}\partial_t\frac{\mathbb{E}_t[Y-M\mid A=a,B=b,x,G=O]p_t(A=1\mid \Tilde{b},x,G=O)}{p_t(A=0\mid B=1,x,G=O)-p_t(A=0\mid B=0,x,G=O)}p(x,\Tilde{b}\mid G=O)\\
     &\ \ \ \ +\sum_{x,\Tilde{b}}\frac{\mathbb{E}[Y-M\mid A=a,B=b,x,G=O]p(A=1\mid \Tilde{b},x,G=O)}{p(A=0\mid B=1,x,G=O)-p(A=0\mid B=0,x,G=O)}\partial_tp_t(x,\Tilde{b}\mid G=O)
 \end{align*}

 For the first term, note that
{\footnotesize \begin{align*}
    &\sum_{x,\Tilde{b}}\partial_t\frac{\mathbb{E}_t[Y-M\mid A=a,B=b,x,G=O]p_t(A=1\mid \Tilde{b},x,G=O)}{p_t(A=0\mid B=1,x,G=O)-p_t(A=0\mid B=0,x,G=O)}p(x, \Tilde{b}\mid G=O)\\
    &=\sum_{x,\Tilde{b}}\frac{\partial_t\mathbb{E}_t[Y-M\mid A=a,B=b,x,G=O]p_t(A=1\mid \Tilde{b},x,G=O)}{p(A=0\mid B=1,x,G=O)-p(A=0\mid B=0,x,G=O)}p(x, \Tilde{b}\mid G=O)\\
    &\ \ \ \ -\sum_{x,\Tilde{b}}\frac{\mathbb{E}[Y-M\mid A=a,B=b,x,G=O]p(A=1\mid \Tilde{b},x,G=O)\partial_tp_t(A=0\mid B=1,x,G=O)}{\{p(A=0\mid B=1,x,G=O)-p(A=0\mid B=0,x,G=O)\}^2}p(x, \Tilde{b}\mid G=O)\\
    &\ \ \ \ +\sum_{x,\Tilde{b}}\frac{\mathbb{E}[Y-M\mid A=a,B=b,x,G=O]p(A=1\mid \Tilde{b},x,G=O)\partial_tp_t(A=0\mid B=0,x,G=O)}{\{p(A=0\mid B=1,x,G=O)-p(A=0\mid B=0,x,G=O)\}^2}p(x, \Tilde{b}\mid G=O)\\
    &=\sum_{x,\Tilde{b}}\frac{\partial_t\mathbb{E}_t[Y-M\mid A=a,B=b,x,G=O]p(A=1\mid \Tilde{b},x,G=O)}{p(A=0\mid B=1,x,G=O)-p(A=0\mid B=0,x,G=O)}p(x, \Tilde{b}\mid G=O)\\
    &\ \ \ \ +\sum_{x,\Tilde{b}}\frac{\mathbb{E}[Y-M\mid A=a,B=b,x,G=O]\partial_tp_t(A=1\mid \Tilde{b},x,G=O)}{p(A=0\mid B=1,x,G=O)-p(A=0\mid B=0,x,G=O)}p(x, \Tilde{b}\mid G=O)\\
    &\ \ \ \ -\sum_{x,\Tilde{b}}\frac{\mathbb{E}[Y-M\mid A=a,B=b,x,G=O]p(A=1\mid \Tilde{b},x,G=O)\partial_tp_t(A=0\mid B=1,x,G=O)}{\{p(A=0\mid B=1,x,G=O)-p(A=0\mid B=0,x,G=O)\}^2}p(x, \Tilde{b}\mid G=O)\\
    &\ \ \ \ +\sum_{x,\Tilde{b}}\frac{\mathbb{E}[Y-M\mid A=a,B=b,x,G=O]p(A=1\mid \Tilde{b},x,G=O)\partial_tp_t(A=0\mid B=0,x,G=O)}{\{p(A=0\mid B=1,x,G=O)-p(A=0\mid B=0,x,G=O)\}^2}p(x, \Tilde{b}\mid G=O)\\
    &=\sum_{y,m,x,\Tilde{a},\Tilde{b},\Tilde{\Tilde{b}},g}\frac{I(\Tilde{a}=a)I(\Tilde{\Tilde{b}}=b)I(g=O)}{p(A=0\mid B=1,x,G=O)-p(A=0\mid B=0,x,G=O)}\cdot\frac{1}{p(G=O)}\\
    &\ \ \ \ \ \ \ \ \cdot\frac{1}{p(B=b\mid x,G=O)}\frac{p(A=1\mid \Tilde{b},x,G=O)}{p(A=a\mid B=b,x,G=O)}\{y-m\}S(y,m\mid x,\Tilde{a},\Tilde{\Tilde{b}},g)p(y,m,x,\Tilde{a},\Tilde{b},\Tilde{\Tilde{b}},g)\\
    &\ \ \ \ +\sum_{\Tilde{a},x,\Tilde{b},g}\frac{\mathbb{E}[Y-M\mid A=a,B=b,x,G=O]}{p(A=0\mid B=1,x,G=O)-p(A=0\mid B=0,x,G=O)}\cdot\frac{I(g=O)}{p(G=O)}\\
    &\ \ \ \ \ \ \ \ \cdot\frac{1}{p(B=\Tilde{b}\mid x, G=O)}I(\Tilde{a}=1)S(\Tilde{a}\mid \Tilde{b},x,g)p(\Tilde{a},\Tilde{b},x,g)\\
    &\ \ \ \ -\sum_{\Tilde{a},\Tilde{b},\Tilde{\Tilde{b}},x,g}\frac{\mathbb{E}[Y-M\mid A=a,B=b,x,G=O]p(A=1\mid \Tilde{b},x,G=O)}{\{p(A=0\mid B=1,x,G=O)-p(A=0\mid B=0,x,G=O)\}^2}\frac{I(\Tilde{\Tilde{b}}=1)I(g=O)}{p(G=O)}\cdot\\
    &\ \ \ \ \ \ \ \ \cdot\frac{1}{p(B=1\mid x, G=O)}I(\Tilde{a}=0)S(\Tilde{a}\mid \Tilde{\Tilde{b}},x,g)p(\Tilde{a},x,\Tilde{b},\Tilde{\Tilde{b}},x,g)\\
    &\ \ \ \ +\sum_{\Tilde{a},\Tilde{b},\Tilde{\Tilde{b}},x,g}\frac{\mathbb{E}[Y-M\mid A=a,B=b,x,G=O]p(A=1\mid \Tilde{b},x,G=O)}{\{p(A=0\mid B=1,x,G=O)-p(A=0\mid B=0,x,G=O)\}^2}\frac{I(\Tilde{\Tilde{b}}=0)I(g=O)}{p(G=O)}\cdot\\
    &\ \ \ \ \ \ \ \ \cdot\frac{1}{p(B=0\mid x, G=O)}I(\Tilde{a}=0)S(\Tilde{a}\mid \Tilde{\Tilde{b}},x,g)p(\Tilde{a},\Tilde{b},\Tilde{\Tilde{b}},x,g)\\
    &=\mathbb{E}\bigg[\frac{I(A=a)I(B=b)I(G=O)}{p(A=0\mid B=1,X,G=O)-p(A=0\mid B=0,X,G=O)}\cdot\frac{1}{p(G=O)}\\
    &\ \ \ \ \ \ \ \ \cdot\frac{1}{p(B=b\mid X,G=O)}\frac{p(A=1\mid B,X,G=O)}{p(A=a\mid B=b,X,G=O)}\{Y-M\}S(Y,M\mid X,A,B,G)\bigg]\\
    &\ \ \ \ +\mathbb{E}\bigg[\frac{\mathbb{E}[Y-M\mid A=a,B=b,x,G=O]}{p(A=0\mid B=1,X,G=O)-p(A=0\mid B=0,X,G=O)}\cdot\frac{I(G=O)}{p(G=O)}\\
    &\ \ \ \ \ \ \ \ \cdot\frac{1}{p(B\mid X,G=O)}I(A=1)S(A\mid B,X,G)\bigg]\\
    &\ \ \ \ -\mathbb{E}\bigg[\frac{\mathbb{E}[Y-M\mid A=a,B=b,X,G=O]p(A=1\mid B,X,G=O)}{\{p(A=0\mid B=1,X,G=O)-p(A=0\mid B=0,X,G=O)\}^2}\frac{I(B=1)I(G=O)}{p(G=O)}\cdot\\
    &\ \ \ \ \ \ \ \ \cdot\frac{1}{p(B=1\mid X, G=O)}I(A=0)S(A\mid B,X,G)\bigg]\\
    &\ \ \ \ +\mathbb{E}\bigg[\frac{\mathbb{E}[Y-M\mid A=a,B=b,X,G=O]p(A=1\mid B,X,G=O)}{\{p(A=0\mid B=1,X,G=O)-p(A=0\mid B=0,X,G=O)\}^2}\frac{I(B=0)I(G=O)}{p(G=O)}\cdot\\
    &\ \ \ \ \ \ \ \ \cdot\frac{1}{p(B=0\mid X, G=O)}I(A=0)S(A\mid B,X,G)\bigg]\\
    &=\mathbb{E}\bigg[\frac{I(A=a)I(B=b)I(G=O)}{p(A=0\mid B=1,X,G=O)-p(A=0\mid B=0,X,G=O)}\cdot\frac{1}{p(G=O)}\\
    &\ \ \ \ \ \ \ \ \cdot\frac{1}{p(B=b\mid X,G=O)}\frac{p(A=1\mid B,X,G=O)}{p(A=a\mid B=b,X,G=O)}\{Y-M-\mathbb{E}[Y-M\mid A=a,B=b,X,G=O]\}S(V)\bigg]\\
    &\ \ \ \ +\mathbb{E}\bigg[\frac{\mathbb{E}[Y-M\mid A=a,B=b,X,G=O]}{p(A=0\mid B=1,X,G=O)-p(A=0\mid B=0,X,G=O)}\cdot\frac{I(G=O)}{p(G=O)}\\
    &\ \ \ \ \ \ \ \ \cdot\frac{1}{p(B\mid X,G=O)}\{I(A=1)-p(A=1\mid B,X,G=O)\}S(V)\bigg]\\
    &\ \ \ \ -\mathbb{E}\bigg[\frac{\mathbb{E}[Y-M\mid A=a,B=b,X,G=O]p(A=1\mid B,X,G=O)}{\{p(A=0\mid B=1,X,G=O)-p(A=0\mid B=0,X,G=O)\}^2}\frac{I(B=1)I(G=O)}{p(G=O)}\cdot\\
    &\ \ \ \ \ \ \ \ \cdot\frac{1}{p(B=1\mid X, G=O)}\{I(A=0)-p(A=0\mid B=1,X,G=O)\}S(V)\bigg]\\
    &\ \ \ \ +\mathbb{E}\bigg[\frac{\mathbb{E}[Y-M\mid A=a,B=b,X,G=O]p(A=1\mid B,X,G=O)}{\{p(A=0\mid B=1,X,G=O)-p(A=0\mid B=0,X,G=O)\}^2}\frac{I(B=0)I(G=O)}{p(G=O)}\cdot\\
    &\ \ \ \ \ \ \ \ \cdot\frac{1}{p(B=0\mid X, G=O)}\{I(A=0)-p(A=0\mid B=0,X,G=O)\}S(V)\bigg].
\end{align*}}

For the second term, we have
\begin{align*}
    &\sum_{x,\Tilde{b}}\frac{\mathbb{E}[Y-M\mid A=a,B=b,x,G=O]p(A=1\mid \Tilde{b},x,G=O)}{p(A=0\mid B=1,x,G=O)-p(A=0\mid B=0,x,G=O)}\partial_tp_t(x,\Tilde{b}\mid G=O)\\
    &=\mathbb{E}\bigg[\frac{I(G=O)}{p(G=O)}\cdot\frac{\mathbb{E}[Y-M\mid A=a,B=b,X,G=O]p(A=1\mid B,X,G=O)}{p(A=0\mid B=1,X,G=O)-p(A=0\mid =0,x,G=O)}\cdot S(X,B\mid G)\bigg]\\
    &=\mathbb{E}\bigg[\frac{I(G=O)}{p(G=O)}\bigg\{\frac{\mathbb{E}[Y-M\mid A=a,B=b,X,G=O]p(A=1\mid B,X,G=O)}{p(A=0\mid B=1,X,G=O)-p(A=0\mid B=0,x,G=O)}-\psi_{4}^{(ab)}\bigg\}\cdot S(V)\bigg].
\end{align*}

Therefore,
{\footnotesize \begin{align*}
    &\partial_t\psi_{4t}^{(ab)}\\
    &=\mathbb{E}\Bigg[\Bigg\{\frac{I(G=O)}{p(A=0\mid B=1,X,G=O)-p(A=0\mid B=0,X,G=O)}\frac{1}{p(G=O)}\\
    &\ \ \ \ \ \ \ \ \ \ \cdot\bigg\{\frac{I(A=a)I(B=b)p(A=1\mid B,X,G=O)}{p(A=a\mid B=b,X,G=O)p(B=b\mid X,G=O)}\{Y-M-\mathbb{E}[Y-M\mid A=a,B=b,X,G=O]\}\\
    &\ \ \ \ \ \ \ \ \ \ \ \ \ +\frac{\mathbb{E}[Y-M\mid A=a,B=b,X,G=O]}{p(B\mid X,G=O)}\big\{I(A=1)-p(A=1\mid B,X,G=O)\big\}\\
    &\ \ \ \ \ \ \ \ \ \ \ \ \ +\frac{\mathbb{E}[Y-M\mid A=a,B=b,X,G=O]p(A=1\mid B,X,G=O)}{p(A=0\mid B=1,X,G=O)-p(A=0\mid B=0,X,G=O)}\\
    &\ \ \ \ \ \ \ \ \ \ \ \ \ \ \ \ \cdot\big\{-\frac{I(B=1)}{p(B=1\mid X, G=O)}\{I(A=0)-p(A=0\mid B=1,X,G=O)\}\\
    &\ \ \ \ \ \ \ \ \ \ \ \ \ \ \ \ \ \ \ \ +\frac{I(B=0)}{p(B=0\mid X, G=O)}\{I(A=0)-p(A=0\mid B=0,X,G=O)\}\big\}\bigg\}\\
    &\ \ \ \ \ \ \ \ \  +\frac{I(G=O)}{p(G=O)}\bigg\{\frac{\mathbb{E}[Y-M\mid A=a,B=b,X,G=O]p(A=1\mid B,X,G=O)}{p(A=0\mid B=1,X,G=O)-p(A=0\mid B=0,x,G=O)}-\psi_4^{(ab)}\bigg\}\Bigg\}S(V)\Bigg]
\end{align*}}

For $\psi_5^{(a)}$, we have,
{\footnotesize \begin{align*}
    \partial_{t}\psi_{5t}^{(a)}=\mathbb{E}\Bigg[&\Bigg\{\frac{I(A=a)}{p(A=a\mid X,B, G=E)}\cdot\frac{I(G=E)}{p(G=O)}\{M-\mathbb{E}[M\mid A=a,X,B,G=E]\}\{\frac{1}{p(G=E\mid X,B)}-1\}\\
    &\ \ \ \ \ \ \ \ \ +\frac{I(G=O)}{p(G=O)}\{\mathbb{E}[M\mid A=a,X,B,G=E]-\psi_5^{(a)}\}\Bigg\}S(V)\Bigg]
\end{align*}}

For $i\in\{3\}$, denote the obtained influence functions by $IF_{\psi_i}$, for $i\in\{5\}$ and $a\in\{0,1\}$, denote the obtained influence functions by $IF_{\psi_i^{(a)}}$, for $i\in\{1,2,4\}$ and $a,b\in\{0,1\}$, denote the obtained influence functions by $IF_{\psi_i^{(ab)}}$. The influence function for $\psi_{\text{ATE}}^{\text{bsiv2}}$ can be obtained as $IF_{\psi_1^{(11)}}-IF_{\psi_1^{(01)}}-IF_{\psi_1^{(10)}}+IF_{\psi_1^{(00)}}+IF_{\psi_2^{(10)}}-IF_{\psi_2^{(00)}}-IF_{\psi_3}+IF_{\psi_4^{(11)}}-IF_{\psi_4^{(01)}}-IF_{\psi_4^{(10)}}+IF_{\psi_4^{(00)}}+IF_{\psi_5^{(1)}}-IF_{\psi_5^{(0)}}$. Therefore, using the notations specified in Theorem \ref{thm:IF:BSIV-1},

\begin{align*}
&\frac{I(G=O)}{p(G=O)}\Bigg\{\frac{I(A=1)I(B=1)}{P_{11}^O(X)}\{Y-M-E_{11}^O(X)\}-\frac{I(A=0)I(B=1)}{P_{01}^O(X)}\{Y-M-E_{01}^O(X)\}\\
    &\ \ \ \ \ \ \ \ \ \ \ \ \ \ \ \ \ \ -\frac{I(A=0)I(B=0)}{P_{00}^O(X)}\{Y-M-E_{00}^O(X)\}+\frac{I(A=1)I(B=0)}{P_{10}^O(X)}\{Y-M-E_{10}^O(X)\}\\
    &\ \ \ \ \ \ \ \ \ \ \ \ \ \ \ \ \ \ +I(B=1)\{E_{11}^O(X)-E_{01}^O(X)-E_{10}^O(X)+E_{00}^O(X)\}+E_{10}^O(X)-E_{00}^O(X)\\
    &\ \ \ \ \ \ \ \ \ \ \ \ \ \ \ \ \ \ -\frac{\{E_{01}^O(X)-E_{00}^O(X)\}P_{1B}^O(X)+\{E_{11}^O(X)-E_{10}^O(X)\}(1-P_{1B}^O(X))}{P_{01}^O(X)-P_{00}^O(X)}\\
    &\ \ \ \ \ \ \ \ \ \ \ \ \ \ \ \ \ \ +M_1^E(B,X)-M_0^E(B,X)-\psi_{ATE}^2\\
    &\ \ \ \ \ \ \ \ \ \ \ \ \ \ \ \ \ \ +\frac{1}{P_{01}^O(X)-P_{00}^O(X)}\bigg\{-\frac{I(A=1)I(B=1)P_{01}^O(X)}{P_{11}^O(X)\rho_1^O(X)}\{Y-M-E_{11}^O(X)\}\\
    &\ \ \ \ \ \ \ \ \ \ \ \ \ \ \ \ \ \ \ \ \ \ \ \ \ \ \ \ \ \ \ \ \ -\frac{I(A=0)I(B=1)P_{11}^O(X)}{P_{01}^O(X)\rho_1^O(X)}\{Y-M-E_{01}^O(X)\}\\
    &\ \ \ \ \ \ \ \ \ \ \ \ \ \ \ \ \ \ \ \ \ \ \ \ \ \ \ \ \ \ \ \ \ +\frac{I(A=1)I(B=0)P_{00}^O(X)}{P_{10}^O(X)\rho_0^O(X)}\{Y-M-E_{10}^O(X)\}\\
    &\ \ \ \ \ \ \ \ \ \ \ \ \ \ \ \ \ \ \ \ \ \ \ \ \ \ \ \ \ \ \ \ \ +\frac{I(A=0)I(B=0)P_{10}^O(X)}{P_{00}^O(X)\rho_0^O(X)}\{Y-M-E_{00}^O(X)\}\\
    &\ \ \ \ \ \ \ \ \ \ \ \ \ \ \ \ \ \ \ \ \ \ \ \ \ \ \ \ \ \ \ \ \ +\frac{E_{11}^O(X)-E_{01}^O(X)-E_{10}^O(X)+E_{00}^O(X)}{\rho_B^O(X)}\big\{I(A=1)-P_{1B}^O(X)\big\}\\
    &\ \ \ \ \ \ \ \ \ \ \ \ \ \ \ \ \ \ \ \ \ \ \ \ \ \ \ \ \ \ \ \ \ +\frac{\{E_{11}^O(X)-E_{01}^O(X)-E_{10}^O(X)+E_{00}^O(X)\}P_{1B}^O(X)}{P_{01}^O(X)-P_{00}^O(X)}\\
    &\ \ \ \ \ \ \ \ \ \ \ \ \ \ \ \ \ \ \ \ \ \ \ \ \ \ \ \ \ \ \ \ \ \cdot\big\{-\frac{I(B=1)}{\rho_1^O(X)}\{I(A=0)-P_{01}^O(X)\}+\frac{I(B=0)}{\rho_0^O(X)}\{I(A=0)-P_{00}^O(X)\}\big\}\bigg\}\Bigg\}\\
    &+\frac{I(G=E)}{p(G=O)}\frac{1-\tau(B,X)}{\tau(B,X)}\bigg\{\frac{I(A=1)}{\pi^E(X,B)}\{M-M_1^E(B,X)\}-\frac{I(A=0)}{1-\pi^E(X,B)}\{M-M_0^E(B,X)\}\bigg\}
\end{align*}
is the influence function for $\psi_{\text{ATE}}^{\text{bsiv2}}$.

\end{proof}

\begin{proof}[Proof of Proposition \ref{prop:DR:BSIV_ett}]

\begin{itemize}
    \item $\psi_{\text{ETT}}^{\text{bsiv1}}$:
First, suppose the set $\{P_{AB}^O(X),\mathbb{E}[M\mid B,X,G=O],\mathbb{E}[Y\mid B,X,G=O],\mathbb{E}[M\mid A,B,X,G=E]\}$ is correctly specified. We have
\begingroup
\allowdisplaybreaks
{\footnotesize \begin{align*}
&\E\Big[\frac{I(G=O)}{\hat{P}_{11}^O(X)-\hat{P}_{10}^O(X)}\frac{1}{p(A=1,G=O)}\frac{\hat{\pi}^O(X)}{\hat{\rho}_B^O(X)}\bigg\{I(B=1)\{Y-M-\hat{e}_1^O(X)\}-I(B=0)\{Y-M-\hat{e}_0^O(X)\}\\
&\quad\quad+\frac{\hat{e}_1^O(X)-\hat{e}_0^O(X)}{\hat{P}_{11}^O(X)-\hat{P}_{10}^O(X)}\big\{I(B=0)\{I(A=1)-\hat{P}_{10}^O(X)\}-I(B=1)\{I(A=1)-\hat{P}_{11}^O(X)\}\big\}\bigg\}\\
&\quad+\frac{I(A=1)I(G=O)}{p(A=1,G=O)}\Big\{\frac{\hat{e}_1^O(X)-\hat{e}_0^O(X)}{\hat{P}_{11}^O(X)-\hat{P}_{10}^O(X)}\Big\}-\frac{I(A=0)I(G=E)}{p(A=1,G=O)}\cdot\frac{1-\hat{\tau}(B,X)}{\hat{\tau}(B,X)}\cdot
\frac{M-\hat{M}_{0}^{E}(B,X)}{1-\hat{P}_{AB}^E(X)}\\
&\quad+\frac{I(G=O)}{p(A=1,G=O)}\cdot \{M-\hat{M}_{0}^{E}(B,X)\}\Big]\\
&=\E\Big[\frac{I(G=O)}{\hat{P}_{11}^O(X)-\hat{P}_{10}^O(X)}\frac{1}{p(A=1,G=O)}\frac{\hat{\pi}^O(X)}{\hat{\rho}_B^O(X)}\bigg\{I(B=1)\{\underbrace{\mathbb{E}[Y-M\mid B=1,X,G=O]-\hat{e}_1^O(X)}_{=0}\}\\
&\quad\quad\quad\quad-I(B=0)\{\underbrace{\mathbb{E}[Y-M\mid B=0,X,G=O]-\hat{e}_0^O(X)}_{=0}\}\\
&\quad\quad+\frac{\hat{e}_1^O(X)-\hat{e}_0^O(X)}{\hat{P}_{11}^O(X)-\hat{P}_{10}^O(X)}\big\{I(B=0)\{\underbrace{P(A=1\mid B=0,X,G=O)-\hat{P}_{10}^O(X)}_{=0}\}\\
&\quad\quad\quad\quad-I(B=1)\{\underbrace{P(A=1\mid B=1,X,G=O)-\hat{P}_{11}^O(X)}_{=0}\}\big\}\bigg\}\\
&\quad+\frac{I(A=1)I(G=O)}{p(A=1,G=O)}\Big\{\frac{\hat{e}_1^O(X)-\hat{e}_0^O(X)}{\hat{P}_{11}^O(X)-\hat{P}_{10}^O(X)}\Big\}\\
&\quad\quad-\frac{I(A=0)I(G=E)}{p(A=1,G=O)}\cdot\frac{1-\hat{\tau}(B,X)}{\hat{\tau}(B,X)}\cdot
\underbrace{\frac{\mathbb{E}[M\mid A=0,B,X,G=E]-\hat{M}_{0}^{E}(B,X)}{1-\hat{P}_{AB}^E(X)}}_{=0}\\
&\quad-\frac{I(G=O)}{p(A=1,G=O)}\cdot \hat{M}_{0}^{E}(B,X)+\frac{I(G=O)}{p(A=1,G=O)}\cdot \mathbb{E}[M\mid G=O]\Big]\\
&=\E\Big[\frac{I(A=1)I(G=O)}{p(A=1,G=O)}\Big\{\frac{\hat{e}_1^O(X)-\hat{e}_0^O(X)}{\hat{P}_{11}^O(X)-\hat{P}_{10}^O(X)}\Big\}-\frac{I(G=O)}{p(A=1,G=O)}\cdot \hat{M}_{0}^{E}(B,X)\Big]\\
&\quad\quad+\frac{I(G=O)}{p(A=1,G=O)}\cdot \mathbb{E}[M\mid G=O]\\
&=\E\Big[\frac{I(A=1)I(G=O)}{p(A=1,G=O)}\Big\{\frac{\hat{e}_1^O(X)-\hat{e}_0^O(X)}{\hat{P}_{11}^O(X)-\hat{P}_{10}^O(X)}\Big\}-\frac{I(G=O)I(A=1)}{p(A=1,G=O)}\cdot\frac{\hat{M}_{0}^{E}(B,X)}{P(A=1\mid B,X,G=O)}\Big]\\
&\quad\quad+\frac{I(G=O)}{p(A=1,G=O)}\cdot \mathbb{E}[M\mid G=O]\\
&=\theta_{\text{ETT}}
\end{align*}}
\endgroup

Second, suppose the set $\{\tau(B,X),\rho_B^O(X),\rho_B^E(X),P_{AB}^O(X),P_{AB}^E(X)\}$ is correctly specified. We have
\begingroup
\allowdisplaybreaks
{\footnotesize \begin{align*}
&\E\Big[\frac{I(G=O)}{\hat{P}_{11}^O(X)-\hat{P}_{10}^O(X)}\frac{1}{p(A=1,G=O)}\frac{\hat{\pi}^O(X)}{\hat{\rho}_B^O(X)}\bigg\{I(B=1)\{Y-M-\hat{e}_1^O(X)\}-I(B=0)\{Y-M-\hat{e}_0^O(X)\}\\
&\quad\quad+\frac{\hat{e}_1^O(X)-\hat{e}_0^O(X)}{\hat{P}_{11}^O(X)-\hat{P}_{10}^O(X)}\big\{I(B=0)\{I(A=1)-\hat{P}_{10}^O(X)\}-I(B=1)\{I(A=1)-\hat{P}_{11}^O(X)\}\big\}\bigg\}\\
&\quad+\frac{I(A=1)I(G=O)}{p(A=1,G=O)}\Big\{\frac{\hat{e}_1^O(X)-\hat{e}_0^O(X)}{\hat{P}_{11}^O(X)-\hat{P}_{10}^O(X)}\Big\}-\frac{I(A=0)I(G=E)}{p(A=1,G=O)}\cdot\frac{1-\hat{\tau}(B,X)}{\hat{\tau}(B,X)}\cdot
\frac{M-\hat{M}_{0}^{E}(B,X)}{1-\hat{P}_{AB}^E(X)}\\
&\quad+\frac{I(G=O)}{p(A=1,G=O)}\cdot \{M-\hat{M}_{0}^{E}(B,X)\}\Big]\\
&=\E\Big[\frac{I(G=O)}{\hat{P}_{11}^O(X)-\hat{P}_{10}^O(X)}\frac{1}{p(A=1,G=O)}\underbrace{\frac{\hat{\pi}^O(X)I(A=1)}{P(A=1\mid X,G=O)}}_{=I(A=1)}\bigg\{\underbrace{\frac{I(B=1)}{\hat{\rho}_B^O(X)}}_{=1}\{e_1^O(X)-\hat{e}_1^O(X)\}\\
&\quad\quad\quad\quad-\underbrace{\frac{I(B=0)}{\hat{\rho}_B^O(X)}}_{=1}\{e_0^O(X)-\hat{e}_0^O(X)\}\\
&\quad\quad+\frac{\hat{e}_1^O(X)-\hat{e}_0^O(X)}{\hat{P}_{11}^O(X)-\hat{P}_{10}^O(X)}\big\{I(B=0)\{\underbrace{P(A=1\mid B=0,X,G=O)-\hat{P}_{10}^O(X)}_{=0}\}\\
&\quad\quad\quad\quad-I(B=1)\{\underbrace{P(A=1\mid B=1,X,G=O)-\hat{P}_{11}^O(X)}_{=0}\}\big\}\bigg\}\\
&\quad+\frac{I(A=1)I(G=O)}{p(A=1,G=O)}\Big\{\frac{\hat{e}_1^O(X)-\hat{e}_0^O(X)}{\hat{P}_{11}^O(X)-\hat{P}_{10}^O(X)}\Big\}\\
&\quad-\frac{1}{p(A=1,G=O)}\cdot\underbrace{\frac{P(A=0\mid B,X,G=E)}{1-\hat{P}_{AB}^E(X)}}_{=1}\cdot\underbrace{\Big\{\underbrace{\frac{P(G=E\mid B,X)}{\hat{\tau}(B,X)}}_{=1}-I(G=E)\Big\}}_{=I(G=O)}\cdot
\{\mathbb{E}[M\mid A=0,B,X,G=E]\\
&\quad\quad-\hat{M}_{0}^{E}(B,X)\}\\
&\quad-\frac{I(G=O)}{p(A=1,G=O)}\cdot \hat{M}_{0}^{E}(B,X)+\frac{I(G=O)}{p(A=1,G=O)}\cdot \mathbb{E}[M\mid G=O]\Big]\\
&=\E\Big[\frac{I(G=O)I(A=1)}{p(A=1,G=O)}\frac{1}{\hat{P}_{11}^O(X)-\hat{P}_{10}^O(X)}\bigg\{\{e_1^O(X)-\hat{e}_1^O(X)\}-\{e_0^O(X)-\hat{e}_0^O(X)\}\bigg\}\\
&\quad\quad+\frac{I(A=1)I(G=O)}{p(A=1,G=O)}\Big\{\frac{\hat{e}_1^O(X)-\hat{e}_0^O(X)}{\hat{P}_{11}^O(X)-\hat{P}_{10}^O(X)}\Big\}\\
&\quad-\frac{I(G=O)}{p(A=1,G=O)}\cdot
\{\mathbb{E}[M\mid A=0,B,X,G=E]-\hat{M}_{0}^{E}(B,X)\}-\frac{I(G=O)}{p(A=1,G=O)}\cdot \hat{M}_{0}^{E}(B,X)\Big]\\
&\quad\quad+\frac{\mathbb{E}[M\mid G=O]}{p(A=1\mid G=O)}\\
&=\E\Big[\frac{I(A=1)I(G=O)}{p(A=1,G=O)}\Big\{\frac{e_1^O(X)-e_0^O(X)}{\hat{P}_{11}^O(X)-\hat{P}_{10}^O(X)}\Big\}-\frac{I(G=O)}{p(A=1,G=O)}\cdot
\mathbb{E}[M\mid A=0,B,X,G=E]\Big]\\
&\quad\quad+\frac{\mathbb{E}[M\mid G=O]}{p(A=1\mid G=O)}\\
&=\E\Big[\frac{I(A=1)I(G=O)}{p(A=1,G=O)}\Big\{\frac{e_1^O(X)-e_0^O(X)}{\hat{P}_{11}^O(X)-\hat{P}_{10}^O(X)}\Big\}-\frac{I(A=1)I(G=O)}{p(A=1,G=O)}\cdot\frac{
\mathbb{E}[M\mid A=0,B,X,G=E]}{P(A=1\mid B,X,G=O)}\Big]\\
&\quad\quad+\frac{\mathbb{E}[M\mid G=O]}{p(A=1\mid G=O)}\\
&=\theta_{\text{ETT}}
\end{align*}}
\endgroup

Third, suppose the pair $\{\tau(B,X),P_{AB}^O(X),P_{AB}^E(X),\mathbb{E}[M\mid B,X,G=O],\mathbb{E}[Y\mid B,X,G=O]\}$ is correctly specified. We have
\begingroup
\allowdisplaybreaks
{\footnotesize \begin{align*}
&\E\Big[\frac{I(G=O)}{\hat{P}_{11}^O(X)-\hat{P}_{10}^O(X)}\frac{1}{p(A=1,G=O)}\frac{\hat{\pi}^O(X)}{\hat{\rho}_B^O(X)}\bigg\{I(B=1)\{Y-M-\hat{e}_1^O(X)\}-I(B=0)\{Y-M-\hat{e}_0^O(X)\}\\
&\quad\quad+\frac{\hat{e}_1^O(X)-\hat{e}_0^O(X)}{\hat{P}_{11}^O(X)-\hat{P}_{10}^O(X)}\big\{I(B=0)\{I(A=1)-\hat{P}_{10}^O(X)\}-I(B=1)\{I(A=1)-\hat{P}_{11}^O(X)\}\big\}\bigg\}\\
&\quad+\frac{I(A=1)I(G=O)}{p(A=1,G=O)}\Big\{\frac{\hat{e}_1^O(X)-\hat{e}_0^O(X)}{\hat{P}_{11}^O(X)-\hat{P}_{10}^O(X)}\Big\}-\frac{I(A=0)I(G=E)}{p(A=1,G=O)}\cdot\frac{1-\hat{\tau}(B,X)}{\hat{\tau}(B,X)}\cdot
\frac{M-\hat{M}_{0}^{E}(B,X)}{1-\hat{P}_{AB}^E(X)}\\
&\quad+\frac{I(G=O)}{p(A=1,G=O)}\cdot \{M-\hat{M}_{0}^{E}(B,X)\}\Big]\\
&=\E\Big[\frac{I(G=O)}{\hat{P}_{11}^O(X)-\hat{P}_{10}^O(X)}\frac{1}{p(A=1,G=O)}\frac{\hat{\pi}^O(X)}{\hat{\rho}_B^O(X)}\bigg\{I(B=1)\{\underbrace{\mathbb{E}[Y-M\mid B=1,X,G=O]-\hat{e}_1^O(X)}_{=0}\}\\
&\quad\quad\quad\quad-I(B=0)\{\underbrace{\mathbb{E}[Y-M\mid B=0,X,G=O]-\hat{e}_0^O(X)}_{=0}\}\\
&\quad\quad+\frac{\hat{e}_1^O(X)-\hat{e}_0^O(X)}{\hat{P}_{11}^O(X)-\hat{P}_{10}^O(X)}\big\{I(B=0)\{\underbrace{P(A=1\mid B=0,X,G=O)-\hat{P}_{10}^O(X)}_{=0}\}\\
&\quad\quad\quad\quad-I(B=1)\{\underbrace{P(A=1\mid B=1,X,G=O)-\hat{P}_{11}^O(X)}_{=0}\}\big\}\bigg\}\\
&\quad+\frac{I(A=1)I(G=O)}{p(A=1,G=O)}\Big\{\frac{\hat{e}_1^O(X)-\hat{e}_0^O(X)}{\hat{P}_{11}^O(X)-\hat{P}_{10}^O(X)}\Big\}\\
&\quad-\frac{1}{p(A=1,G=O)}\cdot\underbrace{\frac{P(A=0\mid B,X,G=E)}{1-\hat{P}_{AB}^E(X)}}_{=1}\cdot\underbrace{\Big\{\underbrace{\frac{P(G=E\mid B,X)}{\hat{\tau}(B,X)}}_{=1}-I(G=E)\Big\}}_{=I(G=O)}\cdot
\{\mathbb{E}[M\mid A=0,B,X,G=E]\\
&\quad\quad-\hat{M}_{0}^{E}(B,X)\}\\
&\quad-\frac{I(G=O)}{p(A=1,G=O)}\cdot \hat{M}_{0}^{E}(B,X)+\frac{I(G=O)}{p(A=1,G=O)}\cdot \mathbb{E}[M\mid G=O]\Big]\\
&=\E\Big[\frac{I(A=1)I(G=O)}{p(A=1,G=O)}\Big\{\frac{\hat{e}_1^O(X)-\hat{e}_0^O(X)}{\hat{P}_{11}^O(X)-\hat{P}_{10}^O(X)}\Big\}\\
&\quad-\frac{I(G=O)}{p(A=1,G=O)}\cdot
\{\mathbb{E}[M\mid A=0,B,X,G=E]-\hat{M}_{0}^{E}(B,X)\}-\frac{I(G=O)}{p(A=1,G=O)}\cdot \hat{M}_{0}^{E}(B,X)\Big]\\
&\quad+\frac{\mathbb{E}[M\mid G=O]}{p(A=1\mid G=O)}\\
&=\E\Big[\frac{I(A=1)I(G=O)}{p(A=1,G=O)}\Big\{\frac{\hat{e}_1^O(X)-\hat{e}_0^O(X)}{\hat{P}_{11}^O(X)-\hat{P}_{10}^O(X)}\Big\}-\frac{I(G=O)}{p(A=1,G=O)}\cdot
\mathbb{E}[M\mid A=0,B,X,G=E]\Big]\\
&\quad\quad+\frac{\mathbb{E}[M\mid G=O]}{p(A=1\mid G=O)}\\
&=\E\Big[\frac{I(A=1)I(G=O)}{p(A=1,G=O)}\Big\{\frac{\hat{e}_1^O(X)-\hat{e}_0^O(X)}{\hat{P}_{11}^O(X)-\hat{P}_{10}^O(X)}\Big\}-\frac{I(A=1)I(G=O)}{p(A=1,G=O)}\cdot\frac{
\mathbb{E}[M\mid A=0,B,X,G=E]}{P(A=1\mid B,X,G=O)}\Big]\\
&\quad\quad+\frac{\mathbb{E}[M\mid G=O]}{p(A=1\mid G=O)}\\
&=\theta_{\text{ETT}}
\end{align*}}
\endgroup

Finally, suppose the set $\{\rho_B^O(X),\rho_B^E(X),P_{AB}^O(X),\mathbb{E}[M\mid A,B,X,G=E]\}$ is correctly specified. We have
\begingroup
\allowdisplaybreaks
{\footnotesize \begin{align*}
&\E\Big[\frac{I(G=O)}{\hat{P}_{11}^O(X)-\hat{P}_{10}^O(X)}\frac{1}{p(A=1,G=O)}\frac{\hat{\pi}^O(X)}{\hat{\rho}_B^O(X)}\bigg\{I(B=1)\{Y-M-\hat{e}_1^O(X)\}-I(B=0)\{Y-M-\hat{e}_0^O(X)\}\\
&\quad\quad+\frac{\hat{e}_1^O(X)-\hat{e}_0^O(X)}{\hat{P}_{11}^O(X)-\hat{P}_{10}^O(X)}\big\{I(B=0)\{I(A=1)-\hat{P}_{10}^O(X)\}-I(B=1)\{I(A=1)-\hat{P}_{11}^O(X)\}\big\}\bigg\}\\
&\quad+\frac{I(A=1)I(G=O)}{p(A=1,G=O)}\Big\{\frac{\hat{e}_1^O(X)-\hat{e}_0^O(X)}{\hat{P}_{11}^O(X)-\hat{P}_{10}^O(X)}\Big\}-\frac{I(A=0)I(G=E)}{p(A=1,G=O)}\cdot\frac{1-\hat{\tau}(B,X)}{\hat{\tau}(B,X)}\cdot
\frac{M-\hat{M}_{0}^{E}(B,X)}{1-\hat{P}_{AB}^E(X)}\\
&\quad+\frac{I(G=O)}{p(A=1,G=O)}\cdot \{M-\hat{M}_{0}^{E}(B,X)\}\Big]\\
&=\E\Big[\frac{I(G=O)}{\hat{P}_{11}^O(X)-\hat{P}_{10}^O(X)}\frac{1}{p(A=1,G=O)}\underbrace{\frac{\hat{\pi}^O(X)I(A=1)}{P(A=1\mid X,G=O)}}_{=I(A=1)}\bigg\{\underbrace{\frac{I(B=1)}{\hat{\rho}_B^O(X)}}_{=1}\{e_1^O(X)-\hat{e}_1^O(X)\}\\
&\quad\quad\quad\quad-\underbrace{\frac{I(B=0)}{\hat{\rho}_B^O(X)}}_{=1}\{e_0^O(X)-\hat{e}_0^O(X)\}\\
&\quad\quad+\frac{\hat{e}_1^O(X)-\hat{e}_0^O(X)}{\hat{P}_{11}^O(X)-\hat{P}_{10}^O(X)}\big\{I(B=0)\{\underbrace{P(A=1\mid B=0,X,G=O)-\hat{P}_{10}^O(X)}_{=0}\}\\
&\quad\quad\quad\quad-I(B=1)\{\underbrace{P(A=1\mid B=1,X,G=O)-\hat{P}_{11}^O(X)}_{=0}\}\big\}\bigg\}\\
&\quad+\frac{I(A=1)I(G=O)}{p(A=1,G=O)}\Big\{\frac{\hat{e}_1^O(X)-\hat{e}_0^O(X)}{\hat{P}_{11}^O(X)-\hat{P}_{10}^O(X)}\Big\}\\
&\quad\quad-\frac{I(A=0)I(G=E)}{p(A=1,G=O)}\cdot\frac{1-\hat{\tau}(B,X)}{\hat{\tau}(B,X)}\cdot
\underbrace{\frac{\mathbb{E}[M\mid A=0,B,X,G=E]-\hat{M}_{0}^{E}(B,X)}{1-\hat{P}_{AB}^E(X)}}_{=0}\\
&\quad-\frac{I(G=O)}{p(A=1,G=O)}\cdot \hat{M}_{0}^{E}(B,X)+\frac{I(G=O)}{p(A=1,G=O)}\cdot \mathbb{E}[M\mid G=O]\Big]\\
&=\E\Big[\frac{I(G=O)I(A=1)}{p(A=1,G=O)}\frac{1}{\hat{P}_{11}^O(X)-\hat{P}_{10}^O(X)}\bigg\{\{e_1^O(X)-\hat{e}_1^O(X)\}-\{e_0^O(X)-\hat{e}_0^O(X)\}\bigg\}\\
&\quad\quad+\frac{I(A=1)I(G=O)}{p(A=1,G=O)}\Big\{\frac{\hat{e}_1^O(X)-\hat{e}_0^O(X)}{\hat{P}_{11}^O(X)-\hat{P}_{10}^O(X)}\Big\}\\
&\quad-\frac{I(G=O)}{p(A=1,G=O)}\cdot \hat{M}_{0}^{E}(B,X)\Big]+\frac{\mathbb{E}[M\mid G=O]}{p(A=1\mid G=O)}\\
&=\E\Big[\frac{I(A=1)I(G=O)}{p(A=1,G=O)}\Big\{\frac{e_1^O(X)-e_0^O(X)}{\hat{P}_{11}^O(X)-\hat{P}_{10}^O(X)}\Big\}-\frac{I(G=O)}{p(A=1,G=O)}\cdot \hat{M}_{0}^{E}(B,X)\Big]\\
&\quad\quad+\frac{I(G=O)}{p(A=1,G=O)}\cdot \mathbb{E}[M\mid G=O]\\
&=\E\Big[\frac{I(A=1)I(G=O)}{p(A=1,G=O)}\Big\{\frac{e_1^O(X)-e_0^O(X)}{\hat{P}_{11}^O(X)-\hat{P}_{10}^O(X)}\Big\}-\frac{I(G=O)I(A=1)}{p(A=1,G=O)}\cdot\frac{\hat{M}_{0}^{E}(B,X)}{P(A=1\mid B,X,G=O)}\Big]\\
&\quad\quad+\frac{I(G=O)}{p(A=1,G=O)}\cdot \mathbb{E}[M\mid G=O]\\
&=\theta_{\text{ETT}}
\end{align*}}
\endgroup

    \item $\psi_{\text{ETT}}^{\text{bsiv2}}$:
First, suppose the set $\{P_{AB}^O(X),\mathbb{E}[M\mid A,B,X,G=O],\mathbb{E}[Y\mid A,B,X,G=O],\mathbb{E}[M\mid A,B,X,G=E]\}$ is correctly specified. We have
\begingroup
\allowdisplaybreaks
{\footnotesize \begin{align*}
&\E\Big[\frac{I(G=O)}{p(A=1,G=O)}\Bigg\{I(A=1)I(B=1)\{Y-M-\hat{E}_{01}^O(X)-\hat{E}_{10}^O(X)+\hat{E}_{00}^O(X)\}\\
&\quad-I(A=0)I(B=1)\frac{\hat{P}_{11}^O(X)}{\hat{P}_{01}^O(X)}\{Y-M-\hat{E}_{01}^O(X)\}+I(A=1)\{\hat{E}_{10}^O(X)-\hat{E}_{00}^O(X)\}\\
&\quad-I(B=0)\frac{\hat{\rho}_{1}^O(X)\hat{P}_{11}^O(X)}{\hat{\rho}_{0}^O(X)\hat{P}_{10}^O(X)}I(A=1)\{Y-M-\hat{E}_{10}^O(X)\}+I(B=0)\frac{\hat{\rho}_{1}^O(X)\hat{P}_{11}^O(X)}{\hat{\rho}_{0}^O(X)\hat{P}_{00}^O(X)}I(A=0)\{Y-M-\hat{E}_{00}^O(X)\}\\
&\quad+\frac{I(B=0)}{\hat{P}_{10}^O(X)}\frac{\hat{\pi}^O(X)}{\hat{\rho}_{0}^O(X)}\cdot I(A=1)\{Y-M-\hat{E}_{10}^O(X)\}-\frac{I(B=0)}{\hat{P}_{00}^O(X)}\frac{\hat{\pi}^O(X)}{\hat{\rho}_{0}^O(X)}\cdot I(A=0)\{Y-M-\hat{E}_{00}^O(X)\}\\
&\quad+\frac{\hat{\pi}^O(X)}{\hat{P}_{01}^O(X)-\hat{P}_{00}^O(X)}\Big\{I(A=0)\big\{\frac{I(B=1)}{\hat{P}_{01}^O(X)\hat{\rho}_1^O(X)}\{Y-M-\hat{E}_{01}^O(X)\}-\frac{I(B=0)}{\hat{P}_{00}^O(X)\hat{\rho}_0^O(X)}\{Y-M-\hat{E}_{00}^O(X)\}\big\}\\
&\quad\quad+\frac{\hat{E}_{01}^O(X)-\hat{E}_{00}^O(X)}{\hat{P}_{01}^O(X)-\hat{P}_{00}^O(X)}\big\{-\frac{I(B=1)}{\hat{\rho}_1^O(X)}\{I(A=0)-\hat{P}_{01}^O(X)\}+\frac{I(B=0)}{\hat{\rho}_0^O(X)}\{I(A=0)-\hat{P}_{00}^O(X)\}\big\}\Big\}\\
&\quad+I(A=1)\{\frac{\hat{E}_{01}^O(X)-\hat{E}_{00}^O(X)}{\hat{P}_{01}^O(X)-\hat{P}_{00}^O(X)}\}+M-\hat{M}_0^E(B,X)\Bigg\}\\
&\quad-\frac{1}{p(A=1,G=O)}\cdot\frac{I(A=0)I(G=E)}{1-\hat{P}_{AB}^E(X)}\cdot\frac{1-\hat{\tau}(B,X)}{\hat{\tau}(B,X)}\cdot\{M-\hat{M}_0^E(B,X)\}\Big]\\
&=\E\Big[\frac{I(G=O)}{p(A=1,G=O)}\Bigg\{I(A=1)I(B=1)\{\mathbb{E}[Y-M\mid A=1,B=1,X,G=O]-\hat{E}_{01}^O(X)-\hat{E}_{10}^O(X)+\hat{E}_{00}^O(X)\}\\
&\quad-I(A=0)I(B=1)\frac{\hat{P}_{11}^O(X)}{\hat{P}_{01}^O(X)}\{\underbrace{\mathbb{E}[Y-M\mid A=0,B=1,X,G=O]-\hat{E}_{01}^O(X)}_{=0}\}+I(A=1)\{\hat{E}_{10}^O(X)-\hat{E}_{00}^O(X)\}\\
&\quad-I(B=0)\frac{\hat{\rho}_{1}^O(X)\hat{P}_{11}^O(X)}{\hat{\rho}_{0}^O(X)\hat{P}_{10}^O(X)}I(A=1)\{\underbrace{\mathbb{E}[Y-M\mid A=1,B=0,X,G=O]-\hat{E}_{10}^O(X)}_{=0}\}\\
&\quad+I(B=0)\frac{\hat{\rho}_{1}^O(X)\hat{P}_{11}^O(X)}{\hat{\rho}_{0}^O(X)\hat{P}_{00}^O(X)}I(A=0)\{\underbrace{\mathbb{E}[Y-M\mid A=0,B=0,X,G=O]-\hat{E}_{00}^O(X)}_{=0}\}\\
&\quad+\frac{I(B=0)}{\hat{P}_{10}^O(X)}\frac{\hat{\pi}^O(X)}{\hat{\rho}_{0}^O(X)}\cdot I(A=1)\{\underbrace{\mathbb{E}[Y-M\mid A=1,B=0,X,G=O]-\hat{E}_{10}^O(X)}_{=0}\}\\
&\quad-\frac{I(B=0)}{\hat{P}_{00}^O(X)}\frac{\hat{\pi}^O(X)}{\hat{\rho}_{0}^O(X)}\cdot I(A=0)\{\underbrace{\mathbb{E}[Y-M\mid A=0,B=0,X,G=O]-\hat{E}_{00}^O(X)}_{=0}\}\\
&\quad+\frac{\hat{\pi}^O(X)}{\hat{P}_{01}^O(X)-\hat{P}_{00}^O(X)}\Big\{I(A=0)\big\{\frac{I(B=1)}{\hat{P}_{01}^O(X)\hat{\rho}_1^O(X)}\{\underbrace{\mathbb{E}[Y-M\mid A=0,B=1,X,G=O]-\hat{E}_{01}^O(X)}_{=0}\}\\
&\quad\quad-\frac{I(B=0)}{\hat{P}_{00}^O(X)\hat{\rho}_0^O(X)}\{\underbrace{\mathbb{E}[Y-M\mid A=0,B=0,X,G=O]-\hat{E}_{00}^O(X)}_{=0}\}\big\}\\
&\quad\quad+\frac{\hat{E}_{01}^O(X)-\hat{E}_{00}^O(X)}{\hat{P}_{01}^O(X)-\hat{P}_{00}^O(X)}\big\{-\frac{I(B=1)}{\hat{\rho}_1^O(X)}\{\underbrace{P(A=0\mid B=1,X,G=O)-\hat{P}_{01}^O(X)}_{=0}\}\\
&\quad\quad+\frac{I(B=0)}{\hat{\rho}_0^O(X)}\{\underbrace{P(A=0\mid B=0,X,G=O)-\hat{P}_{00}^O(X)}_{=0}\}\big\}\Big\}\\
&\quad+I(A=1)\{\frac{\hat{E}_{01}^O(X)-\hat{E}_{00}^O(X)}{\hat{P}_{01}^O(X)-\hat{P}_{00}^O(X)}\}+\mathbb{E}[M\mid G=O]-\hat{M}_0^E(B,X)\Bigg\}\\
&\quad-\frac{1}{p(A=1,G=O)}\cdot\frac{I(A=0)I(G=E)}{1-\hat{P}_{AB}^E(X)}\cdot\frac{1-\hat{\tau}(B,X)}{\hat{\tau}(B,X)}\cdot\{\underbrace{\mathbb{E}[M\mid A=0,B,X,G=E]-\hat{M}_0^E(B,X)}_{=0}\}\Big]\\
&=\E\Big[\frac{I(G=O)}{p(A=1,G=O)}\Bigg\{I(A=1)I(B=1)\{E_{11}^O(X)-\hat{E}_{01}^O(X)-\hat{E}_{10}^O(X)+\hat{E}_{00}^O(X)\}+I(A=1)\{\hat{E}_{10}^O(X)-\hat{E}_{00}^O(X)\}\\
&\quad+I(A=1)\{\frac{\hat{E}_{01}^O(X)-\hat{E}_{00}^O(X)}{\hat{P}_{01}^O(X)-\hat{P}_{00}^O(X)}\}+\mathbb{E}[M\mid G=O]-\hat{M}_0^E(B,X)\Bigg\}\Big]\\
&=\E\Big[\frac{I(G=O)}{p(A=1,G=O)}\Bigg\{I(A=1)I(B=1)\{E_{11}^O(X)-\hat{E}_{01}^O(X)-\hat{E}_{10}^O(X)+\hat{E}_{00}^O(X)\}+I(A=1)\{\hat{E}_{10}^O(X)-\hat{E}_{00}^O(X)\}\\
&\quad+I(A=1)\{\frac{\hat{E}_{01}^O(X)-\hat{E}_{00}^O(X)}{\hat{P}_{01}^O(X)-\hat{P}_{00}^O(X)}\}+\mathbb{E}[M\mid G=O]\Bigg\}-\frac{P(G=O\mid B,X)}{p(A=1,G=O)}\mathbb{E}[M\mid A=0,B,X,G=E]\Big]\\
&=\theta_{\text{ETT}}
\end{align*}}
\endgroup

Second, suppose the set $\{\tau(B,X),\rho_B^O(X),\rho_B^E(X),P_{AB}^O(X),P_{AB}^E(X)\}$ is correctly specified. We have
\begingroup
\allowdisplaybreaks
{\footnotesize \begin{align*}
&\E\Big[\frac{I(G=O)}{p(A=1,G=O)}\Bigg\{I(A=1)I(B=1)\{Y-M-\hat{E}_{01}^O(X)-\hat{E}_{10}^O(X)+\hat{E}_{00}^O(X)\}\\
&\quad-I(A=0)I(B=1)\frac{\hat{P}_{11}^O(X)}{\hat{P}_{01}^O(X)}\{Y-M-\hat{E}_{01}^O(X)\}+I(A=1)\{\hat{E}_{10}^O(X)-\hat{E}_{00}^O(X)\}\\
&\quad-I(B=0)\frac{\hat{\rho}_{1}^O(X)\hat{P}_{11}^O(X)}{\hat{\rho}_{0}^O(X)\hat{P}_{10}^O(X)}I(A=1)\{Y-M-\hat{E}_{10}^O(X)\}+I(B=0)\frac{\hat{\rho}_{1}^O(X)\hat{P}_{11}^O(X)}{\hat{\rho}_{0}^O(X)\hat{P}_{00}^O(X)}I(A=0)\{Y-M-\hat{E}_{00}^O(X)\}\\
&\quad+\frac{I(B=0)}{\hat{P}_{10}^O(X)}\frac{\hat{\pi}^O(X)}{\hat{\rho}_{0}^O(X)}\cdot\{I(A=1)\{Y-M-\hat{E}_{10}^O(X)\}-\frac{I(B=0)}{\hat{P}_{00}^O(X)}\frac{\hat{\pi}^O(X)}{\hat{\rho}_{0}^O(X)}\cdot I(A=0)\{Y-M-\hat{E}_{00}^O(X)\}\\
&\quad+\frac{\hat{\pi}^O(X)}{\hat{P}_{01}^O(X)-\hat{P}_{00}^O(X)}\Big\{I(A=0)\big\{\frac{I(B=1)}{\hat{P}_{01}^O(X)\hat{\rho}_1^O(X)}\{Y-M-\hat{E}_{01}^O(X)\}-\frac{I(B=0)}{\hat{P}_{00}^O(X)\hat{\rho}_0^O(X)}\{Y-M-\hat{E}_{00}^O(X)\}\big\}\\
&\quad\quad+\frac{\hat{E}_{01}^O(X)-\hat{E}_{00}^O(X)}{\hat{P}_{01}^O(X)-\hat{P}_{00}^O(X)}\big\{-\frac{I(B=1)}{\hat{\rho}_1^O(X)}\{I(A=0)-\hat{P}_{01}^O(X)\}+\frac{I(B=0)}{\hat{\rho}_0^O(X)}\{I(A=0)-\hat{P}_{00}^O(X)\}\big\}\Big\}\\
&\quad+I(A=1)\{\frac{\hat{E}_{01}^O(X)-\hat{E}_{00}^O(X)}{\hat{P}_{01}^O(X)-\hat{P}_{00}^O(X)}\}+M-\hat{M}_0^E(B,X)\Bigg\}\\
&\quad-\frac{1}{p(A=1,G=O)}\cdot\frac{I(A=0)I(G=E)}{1-\hat{P}_{AB}^E(X)}\cdot\frac{1-\hat{\tau}(B,X)}{\hat{\tau}(B,X)}\cdot\{M-\hat{M}_0^E(B,X)\}\Big]\\
&=\E\Big[\frac{I(G=O)}{p(A=1,G=O)}\Bigg\{I(A=1)I(B=1)E_{11}^O(X)-I(A=1)I(B=1)\{\hat{E}_{01}^O(X)+\hat{E}_{10}^O(X)-\hat{E}_{00}^O(X)\}\\
&\quad-I(A=0)I(B=1)\frac{\hat{P}_{11}^O(X)}{\hat{P}_{01}^O(X)}E_{01}^O(X)+I(A=0)I(B=1)\frac{\hat{P}_{11}^O(X)}{\hat{P}_{01}^O(X)}\hat{E}_{01}^O(X)+I(A=1)\{\hat{E}_{10}^O(X)-\hat{E}_{00}^O(X)\}\\
&\quad+I(B=0)\frac{\hat{\pi}^O(X)-\hat{\rho}_{1}^O(X)\hat{P}_{11}^O(X)}{\hat{\rho}_{0}^O(X)\hat{P}_{10}^O(X)}I(A=1)\{E_{10}^O(X)-\hat{E}_{10}^O(X)\}\\
&\quad-I(B=0)\frac{\hat{\pi}^O(X)-\hat{\rho}_{1}^O(X)\hat{P}_{11}^O(X)}{\hat{\rho}_{0}^O(X)\hat{P}_{00}^O(X)}I(A=0)\{E_{00}^O(X)-\hat{E}_{00}^O(X)\}\\
&\quad+\frac{\hat{\pi}^O(X)}{\hat{P}_{01}^O(X)-\hat{P}_{00}^O(X)}\Big\{\underbrace{\frac{P_{01}^O(X)\rho_1^O(X)}{\hat{P}_{01}^O(X)\hat{\rho}_1^O(X)}}_{=1}\{E_{01}^O(X)-\hat{E}_{01}^O(X)\}-\underbrace{\frac{P_{00}^O(X)\rho_0^O(X)}{\hat{P}_{00}^O(X)\hat{\rho}_0^O(X)}}_{=1}\{E_{00}^O(X)-\hat{E}_{00}^O(X)\}\\
&\quad\quad+\frac{\hat{E}_{01}^O(X)-\hat{E}_{00}^O(X)}{\hat{P}_{01}^O(X)-\hat{P}_{00}^O(X)}\big\{-\frac{I(B=1)}{\hat{\rho}_1^O(X)}\{\underbrace{P(A=0\mid B=1,X,G=O)-\hat{P}_{01}^O(X)}_{=0}\}\\
&\quad\quad+\frac{I(B=0)}{\hat{\rho}_0^O(X)}\{\underbrace{P(A=0\mid B=0,X,G=O)-\hat{P}_{00}^O(X)}_{=0}\}\big\}\Big\}\\
&\quad+\pi^O(X)\{\frac{\hat{E}_{01}^O(X)-\hat{E}_{00}^O(X)}{\hat{P}_{01}^O(X)-\hat{P}_{00}^O(X)}\}+M-\hat{M}_0^E(B,X)\Bigg\}\\
&\quad-\frac{P(A=0\mid B,X,G=O)}{p(A=1,G=O)}\cdot\frac{P(G=E\mid B,X)}{1-\hat{P}_{AB}^E(X)}\cdot\frac{1-\hat{\tau}(B,X)}{\hat{\tau}(B,X)}\cdot\{M-\hat{M}_0^E(B,X)\}\Big]\\
&=\E\Big[\frac{I(G=O)}{p(A=1,G=O)}\Bigg\{I(A=1)I(B=1)E_{11}^O(X)-I(A=1)I(B=1)\{\hat{E}_{01}^O(X)+\hat{E}_{10}^O(X)-\hat{E}_{00}^O(X)\}\\
&\quad-I(A=0)I(B=1)\frac{\hat{P}_{11}^O(X)}{\hat{P}_{01}^O(X)}E_{01}^O(X)+I(A=0)I(B=1)\frac{\hat{P}_{11}^O(X)}{\hat{P}_{01}^O(X)}\hat{E}_{01}^O(X)+I(A=1)\{\hat{E}_{10}^O(X)-\hat{E}_{00}^O(X)\}\\
&\quad+(1-I(B=1))\big\{I(A=1)\{E_{10}^O(X)-\hat{E}_{10}^O(X)\}-I(A=0)I(B=0)\frac{\hat{P}_{10}^O(X)}{\hat{P}_{00}^O(X)}\{E_{00}^O(X)-\hat{E}_{00}^O(X)\}\\
&\quad+\hat{\pi}^O(X)\frac{E_{01}^O(X)-E_{00}^O(X)}{\hat{P}_{01}^O(X)-\hat{P}_{00}^O(X)}+\{\underbrace{\pi^O(X)-\hat{\pi}^O(X)}_{=0}\}\frac{\hat{E}_{01}^O(X)-\hat{E}_{00}^O(X)}{\hat{P}_{01}^O(X)-\hat{P}_{00}^O(X)}\\
&\quad+M-\hat{M}_0^E(B,X)\Bigg\}-\frac{P(A=0\mid B,X,G=O)}{p(A=1,G=O)}\cdot\frac{P(G=E\mid B,X)}{1-\hat{P}_{AB}^E(X)}\cdot\frac{1-\hat{\tau}(B,X)}{\hat{\tau}(B,X)}\cdot\{M-\hat{M}_0^E(B,X)\}\Big]\\
&=\E\Big[\frac{I(G=O)}{p(A=1,G=O)}\Bigg\{I(A=1)I(B=1)E_{11}^O(X)-I(A=1)I(B=1)\{\hat{E}_{01}^O(X)+\hat{E}_{10}^O(X)-\hat{E}_{00}^O(X)\}\\
&\quad-I(A=0)I(B=1)\frac{\hat{P}_{11}^O(X)}{\hat{P}_{01}^O(X)}E_{01}^O(X)+(1-I(A=1))I(B=1)\frac{\hat{P}_{11}^O(X)}{\hat{P}_{01}^O(X)}\hat{E}_{01}^O(X)+I(A=1)\{\hat{E}_{10}^O(X)-\hat{E}_{00}^O(X)\}\\
&\quad+\big\{I(A=1)-I(A=1)I(B=1)\big\}\{E_{10}^O(X)-\hat{E}_{10}^O(X)\}\\
&\quad-\big\{\underbrace{1-I(A=1)-I(B=1)+I(A=1)I(B=1)}_{=(1-I(B=1))P_{00}^O(X)}\big\}\frac{\hat{P}_{10}^O(X)}{\hat{P}_{00}^O(X)}\{E_{00}^O(X)-\hat{E}_{00}^O(X)\}\big\}\\
&\quad+\hat{\pi}^O(X)\frac{E_{01}^O(X)-E_{00}^O(X)}{\hat{P}_{01}^O(X)-\hat{P}_{00}^O(X)}+\mathbb{E}[M\mid G=O]\Bigg\}-\frac{P(G=O\mid B,X)}{p(A=1,G=O)}\mathbb{E}[M\mid A=0,B,X,G=E]\Big]\\
&=\E\Big[\frac{I(G=O)}{p(A=1,G=O)}\Bigg\{I(A=1)I(B=1)\{E_{11}^O(X)-E_{01}^O(X)-E_{10}^O(X)+E_{00}^O(X)\}+I(A=1)\{E_{10}^O(X)-E_{00}^O(X)\}\\
&\quad+I(B=1)\frac{\hat{P}_{11}^O(X)}{\hat{P}_{01}^O(X)}\underbrace{\Big\{P_{11}^O(X)\frac{\hat{P}_{01}^O(X)}{\hat{P}_{11}^O(X)}-P_{01}^O(X)\Big\}}_{=0}\{E_{01}^O(X)-\hat{E}_{01}^O(X)\}\\
&\quad+\hat{\pi}^O(X)\frac{E_{01}^O(X)-E_{00}^O(X)}{\hat{P}_{01}^O(X)-\hat{P}_{00}^O(X)}+\mathbb{E}[M\mid G=O]\Bigg\}-\frac{P(G=O\mid B,X)}{p(A=1,G=O)}\mathbb{E}[M\mid A=0,B,X,G=E]\Big]\\
&=\E\Big[\frac{I(G=O)}{p(A=1,G=O)}\Bigg\{I(A=1)I(B=1)\{E_{11}^O(X)-E_{01}^O(X)-E_{10}^O(X)+E_{00}^O(X)\}+I(A=1)\{E_{10}^O(X)-E_{00}^O(X)\}\\
&\quad+\hat{\pi}^O(X)\{\frac{E_{01}^O(X)-E_{00}^O(X)}{\hat{P}_{01}^O(X)-\hat{P}_{00}^O(X)}\}+\mathbb{E}[M\mid G=O]\Bigg\}-\frac{P(G=O\mid B,X)}{p(A=1,G=O)}\mathbb{E}[M\mid A=0,B,X,G=E]\Big]\\
&=\theta_{\text{ETT}}
\end{align*}}
\endgroup

Third, suppose the pair $\{\tau(B,X),P_{AB}^O(X),P_{AB}^E(X),\mathbb{E}[M\mid A,B,X,G=O],\mathbb{E}[Y\mid A,B,X,G=O]\}$ is correctly specified. We have
\begingroup
\allowdisplaybreaks
{\footnotesize \begin{align*}
&\E\Big[\frac{I(G=O)}{p(A=1,G=O)}\Bigg\{I(A=1)I(B=1)\{Y-M-\hat{E}_{01}^O(X)-\hat{E}_{10}^O(X)+\hat{E}_{00}^O(X)\}\\
&\quad-I(A=0)I(B=1)\frac{\hat{P}_{11}^O(X)}{\hat{P}_{01}^O(X)}\{Y-M-\hat{E}_{01}^O(X)\}+I(A=1)\{\hat{E}_{10}^O(X)-\hat{E}_{00}^O(X)\}\\
&\quad-I(B=0)\frac{\hat{\rho}_{1}^O(X)\hat{P}_{11}^O(X)}{\hat{\rho}_{0}^O(X)\hat{P}_{10}^O(X)}I(A=1)\{Y-M-\hat{E}_{10}^O(X)\}+I(B=0)\frac{\hat{\rho}_{1}^O(X)\hat{P}_{11}^O(X)}{\hat{\rho}_{0}^O(X)\hat{P}_{00}^O(X)}I(A=0)\{Y-M-\hat{E}_{00}^O(X)\}\\
&\quad+\frac{I(B=0)}{\hat{P}_{10}^O(X)}\frac{\hat{\pi}^O(X)}{\hat{\rho}_{0}^O(X)}\cdot I(A=1)\{Y-M-\hat{E}_{10}^O(X)\}-\frac{I(B=0)}{\hat{P}_{00}^O(X)}\frac{\hat{\pi}^O(X)}{\hat{\rho}_{0}^O(X)}\cdot I(A=0)\{Y-M-\hat{E}_{00}^O(X)\}\\
&\quad+\frac{\hat{\pi}^O(X)}{\hat{P}_{01}^O(X)-\hat{P}_{00}^O(X)}\Big\{I(A=0)\big\{\frac{I(B=1)}{\hat{P}_{01}^O(X)\hat{\rho}_1^O(X)}\{Y-M-\hat{E}_{01}^O(X)\}-\frac{I(B=0)}{\hat{P}_{00}^O(X)\hat{\rho}_0^O(X)}\{Y-M-\hat{E}_{00}^O(X)\}\big\}\\
&\quad\quad+\frac{\hat{E}_{01}^O(X)-\hat{E}_{00}^O(X)}{\hat{P}_{01}^O(X)-\hat{P}_{00}^O(X)}\big\{-\frac{I(B=1)}{\hat{\rho}_1^O(X)}\{I(A=0)-\hat{P}_{01}^O(X)\}+\frac{I(B=0)}{\hat{\rho}_0^O(X)}\{I(A=0)-\hat{P}_{00}^O(X)\}\big\}\Big\}\\
&\quad+I(A=1)\{\frac{\hat{E}_{01}^O(X)-\hat{E}_{00}^O(X)}{\hat{P}_{01}^O(X)-\hat{P}_{00}^O(X)}\}+M-\hat{M}_0^E(B,X)\Bigg\}\\
&\quad-\frac{1}{p(A=1,G=O)}\cdot\frac{I(A=0)I(G=E)}{1-\hat{P}_{AB}^E(X)}\cdot\frac{1-\hat{\tau}(B,X)}{\hat{\tau}(B,X)}\cdot\{M-\hat{M}_0^E(B,X)\}\Big]\\
&=\E\Big[\frac{I(G=O)}{p(A=1,G=O)}\Bigg\{I(A=1)I(B=1)\{\mathbb{E}[Y-M\mid A=1,B=1,X,G=O]-\hat{E}_{01}^O(X)-\hat{E}_{10}^O(X)+\hat{E}_{00}^O(X)\}\\
&\quad-I(A=0)I(B=1)\frac{\hat{P}_{11}^O(X)}{\hat{P}_{01}^O(X)}\{\underbrace{\mathbb{E}[Y-M\mid A=0,B=1,X,G=O]-\hat{E}_{01}^O(X)}_{=0}\}+I(A=1)\{\hat{E}_{10}^O(X)-\hat{E}_{00}^O(X)\}\\
&\quad-I(B=0)\frac{\hat{\rho}_{1}^O(X)\hat{P}_{11}^O(X)}{\hat{\rho}_{0}^O(X)\hat{P}_{10}^O(X)}I(A=1)\{\underbrace{\mathbb{E}[Y-M\mid A=1,B=0,X,G=O]-\hat{E}_{10}^O(X)}_{=0}\}\\
&\quad+I(B=0)\frac{\hat{\rho}_{1}^O(X)\hat{P}_{11}^O(X)}{\hat{\rho}_{0}^O(X)\hat{P}_{00}^O(X)}I(A=0)\{\underbrace{\mathbb{E}[Y-M\mid A=0,B=0,X,G=O]-\hat{E}_{00}^O(X)}_{=0}\}\\
&\quad+\frac{I(B=0)}{\hat{P}_{10}^O(X)}\frac{\hat{\pi}^O(X)}{\hat{\rho}_{0}^O(X)}\cdot I(A=1)\{\underbrace{\mathbb{E}[Y-M\mid A=1,B=0,X,G=O]-\hat{E}_{10}^O(X)}_{=0}\}\\
&\quad-\frac{I(B=0)}{\hat{P}_{00}^O(X)}\frac{\hat{\pi}^O(X)}{\hat{\rho}_{0}^O(X)}\cdot I(A=0)\{\underbrace{\mathbb{E}[Y-M\mid A=0,B=0,X,G=O]-\hat{E}_{00}^O(X)}_{=0}\}\\
&\quad+\frac{\hat{\pi}^O(X)}{\hat{P}_{01}^O(X)-\hat{P}_{00}^O(X)}\Big\{I(A=0)\big\{\frac{I(B=1)}{\hat{P}_{01}^O(X)\hat{\rho}_1^O(X)}\{\underbrace{\mathbb{E}[Y-M\mid A=0,B=1,X,G=O]-\hat{E}_{01}^O(X)}_{=0}\}\\
&\quad\quad-\frac{I(B=0)}{\hat{P}_{00}^O(X)\hat{\rho}_0^O(X)}\{\underbrace{\mathbb{E}[Y-M\mid A=0,B=0,X,G=O]-\hat{E}_{00}^O(X)}_{=0}\}\big\}\\
&\quad\quad+\frac{\hat{E}_{01}^O(X)-\hat{E}_{00}^O(X)}{\hat{P}_{01}^O(X)-\hat{P}_{00}^O(X)}\big\{-\frac{I(B=1)}{\hat{\rho}_1^O(X)}\{\underbrace{P(A=0\mid B=1,X,G=O)-\hat{P}_{01}^O(X)}_{=0}\}\\
&\quad\quad+\frac{I(B=0)}{\hat{\rho}_0^O(X)}\{\underbrace{P(A=0\mid B=0,X,G=O)-\hat{P}_{00}^O(X)}_{=0}\}\big\}\Big\}\\
&\quad+I(A=1)\{\frac{\hat{E}_{01}^O(X)-\hat{E}_{00}^O(X)}{\hat{P}_{01}^O(X)-\hat{P}_{00}^O(X)}\}+\mathbb{E}[M\mid G=O]-\hat{M}_0^E(B,X)\Bigg\}\\
&\quad-\frac{P(G=E\mid B,X)}{p(A=1,G=O)}\cdot\frac{P(A=0\mid B,X,G=E)}{1-\hat{P}_{AB}^E(X)}\cdot\frac{1-\hat{\tau}(B,X)}{\hat{\tau}(B,X)}\cdot\{\mathbb{E}[M\mid A=0,B,X,G=E]-\hat{M}_0^E(B,X)\}\Big]\\
&=\E\Big[\frac{I(G=O)}{p(A=1,G=O)}\Bigg\{I(A=1)I(B=1)\{E_{11}^O(X)-\hat{E}_{01}^O(X)-\hat{E}_{10}^O(X)+\hat{E}_{00}^O(X)\}+I(A=1)\{\hat{E}_{10}^O(X)-\hat{E}_{00}^O(X)\}\\
&\quad+I(A=1)\{\frac{\hat{E}_{01}^O(X)-\hat{E}_{00}^O(X)}{\hat{P}_{01}^O(X)-\hat{P}_{00}^O(X)}\}+\mathbb{E}[M\mid G=O]-\hat{M}_0^E(B,X)\Bigg\}\\
&\quad-\frac{1-\hat{\tau}(B,X)}{p(A=1,G=O)}\cdot\{\mathbb{E}[M\mid A=0,B,X,G=E]-\hat{M}_0^E(B,X)\}\Big]\\
&=\E\Big[\frac{I(G=O)}{p(A=1,G=O)}\Bigg\{I(A=1)I(B=1)\{E_{11}^O(X)-\hat{E}_{01}^O(X)-\hat{E}_{10}^O(X)+\hat{E}_{00}^O(X)\}+I(A=1)\{\hat{E}_{10}^O(X)-\hat{E}_{00}^O(X)\}\\
&\quad+I(A=1)\{\frac{\hat{E}_{01}^O(X)-\hat{E}_{00}^O(X)}{\hat{P}_{01}^O(X)-\hat{P}_{00}^O(X)}\}+\mathbb{E}[M\mid G=O]\Bigg\}-\frac{P(G=O\mid B,X)}{p(A=1,G=O)}\mathbb{E}[M\mid A=0,B,X,G=E]\Big]\\
&=\E\Big[\frac{I(G=O)}{p(A=1,G=O)}\Bigg\{I(A=1)I(B=1)\{E_{11}^O(X)-\hat{E}_{01}^O(X)-\hat{E}_{10}^O(X)+\hat{E}_{00}^O(X)\}+I(A=1)\{\hat{E}_{10}^O(X)-\hat{E}_{00}^O(X)\}\\
&\quad+I(A=1)\{\frac{\hat{E}_{01}^O(X)-\hat{E}_{00}^O(X)}{\hat{P}_{01}^O(X)-\hat{P}_{00}^O(X)}\}+\mathbb{E}[M\mid G=O]\Bigg\}-\frac{P(A=1,G=O\mid B,X)}{p(A=1,G=O)}\frac{\mathbb{E}[M\mid A=0,B,X,G=E]}{P(A=1\mid B,X,G=O)}\Big]\\
&=\E\Big[\frac{I(G=O)}{p(A=1,G=O)}\Bigg\{I(A=1)I(B=1)\{E_{11}^O(X)-\hat{E}_{01}^O(X)-\hat{E}_{10}^O(X)+\hat{E}_{00}^O(X)\}+I(A=1)\{\hat{E}_{10}^O(X)-\hat{E}_{00}^O(X)\}\\
&\quad+I(A=1)\{\frac{\hat{E}_{01}^O(X)-\hat{E}_{00}^O(X)}{\hat{P}_{01}^O(X)-\hat{P}_{00}^O(X)}\}+\mathbb{E}[M\mid G=O]\Bigg\}-\frac{P(B,X\mid A=1,G=O)}{P(B,X)}\frac{\mathbb{E}[M\mid A=0,B,X,G=E]}{P(A=1\mid B,X,G=O)}\Big]\\
&=\E\Big[\frac{I(G=O)}{p(A=1,G=O)}\Bigg\{I(A=1)I(B=1)\{E_{11}^O(X)-\hat{E}_{01}^O(X)-\hat{E}_{10}^O(X)+\hat{E}_{00}^O(X)\}+I(A=1)\{\hat{E}_{10}^O(X)-\hat{E}_{00}^O(X)\}\\
&\quad+I(A=1)\{\frac{\hat{E}_{01}^O(X)-\hat{E}_{00}^O(X)}{\hat{P}_{01}^O(X)-\hat{P}_{00}^O(X)}\}+\mathbb{E}[M\mid G=O]\Bigg\}\Big]-\E\Big[\frac{\mathbb{E}[M\mid A=0,B,X,G=E]}{P(A=1\mid B,X,G=O)}\Big|A=1,G=O\Big]\\
&=\theta_{\text{ETT}}
\end{align*}}
\endgroup

Finally, suppose the set $\{\rho_B^O(X),\rho_B^E(X),P_{AB}^O(X),\mathbb{E}[M\mid A,B,X,G=E]\}$ is correctly specified. We have
\begingroup
\allowdisplaybreaks
{\footnotesize \begin{align*}
&\E\Big[\frac{I(G=O)}{p(A=1,G=O)}\Bigg\{I(A=1)I(B=1)\{Y-M-\hat{E}_{01}^O(X)-\hat{E}_{10}^O(X)+\hat{E}_{00}^O(X)\}\\
&\quad-I(A=0)I(B=1)\frac{\hat{P}_{11}^O(X)}{\hat{P}_{01}^O(X)}\{Y-M-\hat{E}_{01}^O(X)\}+I(A=1)\{\hat{E}_{10}^O(X)-\hat{E}_{00}^O(X)\}\\
&\quad-I(B=0)\frac{\hat{\rho}_{1}^O(X)\hat{P}_{11}^O(X)}{\hat{\rho}_{0}^O(X)\hat{P}_{10}^O(X)}I(A=1)\{Y-M-\hat{E}_{10}^O(X)\}+I(B=0)\frac{\hat{\rho}_{1}^O(X)\hat{P}_{11}^O(X)}{\hat{\rho}_{0}^O(X)\hat{P}_{00}^O(X)}I(A=0)\{Y-M-\hat{E}_{00}^O(X)\}\\
&\quad+\frac{I(B=0)}{\hat{P}_{10}^O(X)}\frac{\hat{\pi}^O(X)}{\hat{\rho}_{0}^O(X)}\cdot\{I(A=1)\{Y-M-\hat{E}_{10}^O(X)\}-\frac{I(B=0)}{\hat{P}_{00}^O(X)}\frac{\hat{\pi}^O(X)}{\hat{\rho}_{0}^O(X)}\cdot I(A=0)\{Y-M-\hat{E}_{00}^O(X)\}\\
&\quad+\frac{\hat{\pi}^O(X)}{\hat{P}_{01}^O(X)-\hat{P}_{00}^O(X)}\Big\{I(A=0)\big\{\frac{I(B=1)}{\hat{P}_{01}^O(X)\hat{\rho}_1^O(X)}\{Y-M-\hat{E}_{01}^O(X)\}-\frac{I(B=0)}{\hat{P}_{00}^O(X)\hat{\rho}_0^O(X)}\{Y-M-\hat{E}_{00}^O(X)\}\big\}\\
&\quad\quad+\frac{\hat{E}_{01}^O(X)-\hat{E}_{00}^O(X)}{\hat{P}_{01}^O(X)-\hat{P}_{00}^O(X)}\big\{-\frac{I(B=1)}{\hat{\rho}_1^O(X)}\{I(A=0)-\hat{P}_{01}^O(X)\}+\frac{I(B=0)}{\hat{\rho}_0^O(X)}\{I(A=0)-\hat{P}_{00}^O(X)\}\big\}\Big\}\\
&\quad+I(A=1)\{\frac{\hat{E}_{01}^O(X)-\hat{E}_{00}^O(X)}{\hat{P}_{01}^O(X)-\hat{P}_{00}^O(X)}\}+M-\hat{M}_0^E(B,X)\Bigg\}\\
&\quad-\frac{1}{p(A=1,G=O)}\cdot\frac{I(A=0)I(G=E)}{1-\hat{P}_{AB}^E(X)}\cdot\frac{1-\hat{\tau}(B,X)}{\hat{\tau}(B,X)}\cdot\{M-\hat{M}_0^E(B,X)\}\Big]\\
&=\E\Big[\frac{I(G=O)}{p(A=1,G=O)}\Bigg\{I(A=1)I(B=1)E_{11}^O(X)-I(A=1)I(B=1)\{\hat{E}_{01}^O(X)+\hat{E}_{10}^O(X)-\hat{E}_{00}^O(X)\}\\
&\quad-I(A=0)I(B=1)\frac{\hat{P}_{11}^O(X)}{\hat{P}_{01}^O(X)}E_{01}^O(X)+I(A=0)I(B=1)\frac{\hat{P}_{11}^O(X)}{\hat{P}_{01}^O(X)}\hat{E}_{01}^O(X)+I(A=1)\{\hat{E}_{10}^O(X)-\hat{E}_{00}^O(X)\}\\
&\quad+I(B=0)\frac{\hat{\pi}^O(X)-\hat{\rho}_{1}^O(X)\hat{P}_{11}^O(X)}{\hat{\rho}_{0}^O(X)\hat{P}_{10}^O(X)}I(A=1)\{E_{10}^O(X)-\hat{E}_{10}^O(X)\}\\
&\quad-I(B=0)\frac{\hat{\pi}^O(X)-\hat{\rho}_{1}^O(X)\hat{P}_{11}^O(X)}{\hat{\rho}_{0}^O(X)\hat{P}_{00}^O(X)}I(A=0)\{E_{00}^O(X)-\hat{E}_{00}^O(X)\}\\
&\quad+\frac{\hat{\pi}^O(X)}{\hat{P}_{01}^O(X)-\hat{P}_{00}^O(X)}\Big\{\underbrace{\frac{P_{01}^O(X)\rho_1^O(X)}{\hat{P}_{01}^O(X)\hat{\rho}_1^O(X)}}_{=1}\{E_{01}^O(X)-\hat{E}_{01}^O(X)\}-\underbrace{\frac{P_{00}^O(X)\rho_0^O(X)}{\hat{P}_{00}^O(X)\hat{\rho}_0^O(X)}}_{=1}\{E_{00}^O(X)-\hat{E}_{00}^O(X)\}\\
&\quad\quad+\frac{\hat{E}_{01}^O(X)-\hat{E}_{00}^O(X)}{\hat{P}_{01}^O(X)-\hat{P}_{00}^O(X)}\big\{-\frac{I(B=1)}{\hat{\rho}_1^O(X)}\{\underbrace{P(A=0\mid B=1,X,G=O)-\hat{P}_{01}^O(X)}_{=0}\}\\
&\quad\quad+\frac{I(B=0)}{\hat{\rho}_0^O(X)}\{\underbrace{P(A=0\mid B=0,X,G=O)-\hat{P}_{00}^O(X)}_{=0}\}\big\}\Big\}\\
&\quad+\pi^O(X)\{\frac{\hat{E}_{01}^O(X)-\hat{E}_{00}^O(X)}{\hat{P}_{01}^O(X)-\hat{P}_{00}^O(X)}\}+\mathbb{E}[M\mid G=O]-\hat{M}_0^E(B,X)\Bigg\}\\
&\quad-\frac{1}{p(A=1,G=O)}\cdot\frac{I(A=0)I(G=E)}{1-\hat{P}_{AB}^E(X)}\cdot\frac{1-\hat{\tau}(B,X)}{\hat{\tau}(B,X)}\cdot\{\underbrace{\mathbb{E}[M\mid A=0,B,X,G=E]-\hat{M}_0^E(B,X)}_{=0}\}\Big]\\
&=\E\Big[\frac{I(G=O)}{p(A=1,G=O)}\Bigg\{I(A=1)I(B=1)E_{11}^O(X)-I(A=1)I(B=1)\{\hat{E}_{01}^O(X)+\hat{E}_{10}^O(X)-\hat{E}_{00}^O(X)\}\\
&\quad-I(A=0)I(B=1)\frac{\hat{P}_{11}^O(X)}{\hat{P}_{01}^O(X)}E_{01}^O(X)+I(A=0)I(B=1)\frac{\hat{P}_{11}^O(X)}{\hat{P}_{01}^O(X)}\hat{E}_{01}^O(X)+I(A=1)\{\hat{E}_{10}^O(X)-\hat{E}_{00}^O(X)\}\\
&\quad+(1-I(B=1))\big\{I(A=1)\{E_{10}^O(X)-\hat{E}_{10}^O(X)\}-I(A=0)I(B=0)\frac{\hat{P}_{10}^O(X)}{\hat{P}_{00}^O(X)}\{E_{00}^O(X)-\hat{E}_{00}^O(X)\}\\
&\quad+\hat{\pi}^O(X)\frac{E_{01}^O(X)-E_{00}^O(X)}{\hat{P}_{01}^O(X)-\hat{P}_{00}^O(X)}+\{\underbrace{\pi^O(X)-\hat{\pi}^O(X)}_{=0}\}\frac{\hat{E}_{01}^O(X)-\hat{E}_{00}^O(X)}{\hat{P}_{01}^O(X)-\hat{P}_{00}^O(X)}\\
&\quad+\mathbb{E}[M\mid G=O]-\hat{M}_0^E(B,X)\Bigg\}\Big]\\
&=\E\Big[\frac{I(G=O)}{p(A=1,G=O)}\Bigg\{I(A=1)I(B=1)\{E_{11}^O(X)-E_{01}^O(X)-E_{10}^O(X)+E_{00}^O(X)\}+I(A=1)\{E_{10}^O(X)-E_{00}^O(X)\}\\
&\quad+\hat{\pi}^O(X)\{\frac{E_{01}^O(X)-E_{00}^O(X)}{\hat{P}_{01}^O(X)-\hat{P}_{00}^O(X)}\}+\mathbb{E}[M\mid G=O]\Bigg\}-\frac{P(G=O\mid B,X)}{p(A=1,G=O)}\mathbb{E}[M\mid A=0,B,X,G=E]\Big]\\
&=\theta_{\text{ETT}}
\end{align*}}
\endgroup
\end{itemize}

\end{proof}

\begin{proof}[Proof of Proposition \ref{prop:DR:BSIV_ATE}]

\begin{itemize}
    \item $\psi_{\text{ATE}}^{\text{bsiv1}}$:
    First, suppose the set $\{P_{AB}^O(X),\mathbb{E}[M\mid B,X,G=O],\mathbb{E}[Y\mid B,X,G=O],\mathbb{E}[M\mid A,B,X,G=E]\}$ is correctly specified. We have
\begingroup
\allowdisplaybreaks
{\footnotesize \begin{align*}
&\E\Big[\frac{I(G=O)}{p(G=O)}\Bigg\{\frac{1}{\hat{P}_{11}^O(X)-\hat{P}_{10}^O(X)}
\frac{1}{\hat{\rho}_B^O(X)}\bigg\{I(B=1)\{Y-M-\hat{e}_1^O(X)\}
-I(B=0)\{Y-M-\hat{e}_0^O(X)\}\\
&\quad\quad\quad\quad\quad+\frac{\hat{e}_1^O(X)-\hat{e}_0^O(X)}{\hat{P}_{11}^O(X)-\hat{P}_{10}^O(X)}\big\{I(B=0)\{I(A=1)-\hat{P}_{10}^O(X)\}-I(B=1)\{I(A=1)-\hat{P}_{11}^O(X)\}\big\}\bigg\}\\
&\quad\quad\quad\quad\quad+\frac{\hat{e}_1^O(X)-\hat{e}_0^O(X)}{\hat{P}_{11}^O(X)-\hat{P}_{10}^O(X)}+\hat{M}_{1}^{E}(B,X)-\hat{M}_{0}^{E}(B,X)\Bigg\}\\
&\quad+\frac{I(G=E)}{p(G=O)}\cdot\frac{1-\hat{\tau}(B,X)}{\hat{\tau}(B,X)}\bigg\{\frac{I(A=1)}{\hat{P}_{1B}^E(X)}\cdot\{M-\hat{M}_{1}^{E}(B,X)\}-\frac{I(A=0)}{1-\hat{P}_{1B}^E(X)}\cdot\{M-\hat{M}_{0}^{E}(B,X)\}\bigg\}\Big]\\
&=\E\Big[\frac{I(G=O)}{p(G=O)}\Bigg\{\frac{1}{\hat{P}_{11}^O(X)-\hat{P}_{10}^O(X)}
\frac{1}{\hat{\rho}_B^O(X)}\bigg\{I(B=1)\{\underbrace{\E[Y-M\mid B=1,X,G=O]-\hat{e}_1^O(X)}_{=0}\}\\
&\quad\quad\quad\quad\quad\quad
-I(B=0)\{\underbrace{\E[Y-M\mid B=0,X,G=O]-\hat{e}_0^O(X)}_{=0}\}\\
&\quad\quad\quad\quad\quad+\frac{\hat{e}_1^O(X)-\hat{e}_0^O(X)}{\hat{P}_{11}^O(X)-\hat{P}_{10}^O(X)}\big\{I(B=0)\{\underbrace{P(A=1\mid B=0,X,G=O)-\hat{P}_{10}^O(X)}_{=0}\}\\
&\quad\quad\quad\quad\quad\quad-I(B=1)\{\underbrace{P(A=1\mid B=1,X,G=O)-\hat{P}_{11}^O(X)}_{=0}\}\big\}\bigg\}\\
&\quad\quad\quad\quad\quad+\frac{\hat{e}_1^O(X)-\hat{e}_0^O(X)}{\hat{P}_{11}^O(X)-\hat{P}_{10}^O(X)}+\hat{M}_{1}^{E}(B,X)-\hat{M}_{0}^{E}(B,X)\Bigg\}\\
&\quad+\frac{I(G=E)}{p(G=O)}\cdot\frac{1-\hat{\tau}(B,X)}{\hat{\tau}(B,X)}\bigg\{\frac{I(A=1)}{\hat{P}_{1B}^E(X)}\cdot\{\underbrace{\E[M\mid A=1,B,X,G=E]-\hat{M}_{1}^{E}(B,X)}_{=0}\}\\
&\quad\quad\quad\quad-\frac{I(A=0)}{1-\hat{P}_{1B}^E(X)}\cdot\{\underbrace{\E[M\mid A=0,B,X,G=E]-\hat{M}_{0}^{E}(B,X)}_{=0}\}\bigg\}\Big]\\
&=\E\Big[\frac{I(G=O)}{p(G=O)}\Bigg\{\frac{\hat{e}_1^O(X)-\hat{e}_0^O(X)}{\hat{P}_{11}^O(X)-\hat{P}_{10}^O(X)}+\hat{M}_{1}^{E}(B,X)-\hat{M}_{0}^{E}(B,X)\Bigg\}\Big]\\
&=\theta_{\text{ATE}}
\end{align*}}
\endgroup

Second, suppose the set $\{\tau(B,X),\rho_B^O(X),\rho_B^E(X),P_{AB}^O(X),P_{AB}^E(X)\}$ is correctly specified. We have
\begingroup
\allowdisplaybreaks
{\footnotesize \begin{align*}
&\E\Big[\frac{I(G=O)}{p(G=O)}\Bigg\{\frac{1}{\hat{P}_{11}^O(X)-\hat{P}_{10}^O(X)}
\frac{1}{\hat{\rho}_B^O(X)}\bigg\{I(B=1)\{Y-M-\hat{e}_1^O(X)\}
-I(B=0)\{Y-M-\hat{e}_0^O(X)\}\\
&\quad\quad\quad\quad\quad+\frac{\hat{e}_1^O(X)-\hat{e}_0^O(X)}{\hat{P}_{11}^O(X)-\hat{P}_{10}^O(X)}\big\{I(B=0)\{I(A=1)-\hat{P}_{10}^O(X)\}-I(B=1)\{I(A=1)-\hat{P}_{11}^O(X)\}\big\}\bigg\}\\
&\quad\quad\quad\quad\quad+\frac{\hat{e}_1^O(X)-\hat{e}_0^O(X)}{\hat{P}_{11}^O(X)-\hat{P}_{10}^O(X)}+\hat{M}_{1}^{E}(B,X)-\hat{M}_{0}^{E}(B,X)\Bigg\}\\
&\quad+\frac{I(G=E)}{p(G=O)}\cdot\frac{1-\hat{\tau}(B,X)}{\hat{\tau}(B,X)}\bigg\{\frac{I(A=1)}{\hat{P}_{1B}^E(X)}\cdot\{M-\hat{M}_{1}^{E}(B,X)\}-\frac{I(A=0)}{1-\hat{P}_{1B}^E(X)}\cdot\{M-\hat{M}_{0}^{E}(B,X)\}\bigg\}\Big]\\
&=\E\Big[\frac{I(G=O)}{p(G=O)}\Bigg\{\frac{1}{\hat{P}_{11}^O(X)-\hat{P}_{10}^O(X)}
\bigg\{\underbrace{\frac{I(B=1)}{\hat{\rho}_B^O(X)}}_{=1}\{e_1^O(X)-\hat{e}_1^O(X)\}
-\underbrace{\frac{I(B=0)}{\hat{\rho}_B^O(X)}}_{=1}\{e_0^O(X)-\hat{e}_0^O(X)\}\\
&\quad\quad\quad\quad\quad+\frac{\hat{e}_1^O(X)-\hat{e}_0^O(X)}{\hat{P}_{11}^O(X)-\hat{P}_{10}^O(X)}\big\{I(B=0)\{\underbrace{P(A=1\mid B=0,X,G=O)-\hat{P}_{10}^O(X)}_{=0}\}\\
&\quad\quad\quad\quad\quad\quad-I(B=1)\{\underbrace{P(A=1\mid B=1,X,G=O)-\hat{P}_{11}^O(X)}_{=0}\}\big\}\bigg\}\\
&\quad\quad\quad\quad\quad+\frac{\hat{e}_1^O(X)-\hat{e}_0^O(X)}{\hat{P}_{11}^O(X)-\hat{P}_{10}^O(X)}+\hat{M}_{1}^{E}(B,X)-\hat{M}_{0}^{E}(B,X)\Bigg\}\\
&\quad+\frac{1}{p(G=O)}\cdot\underbrace{\Big\{\underbrace{\frac{P(G=E\mid B,X)}{\hat{\tau}(B,X)}}_{=1}-I(G=E)\Big\}}_{=I(G=O)}\bigg\{\underbrace{\frac{P(A=1\mid B,X,G=E)}{\hat{P}_{1B}^E(X)}}_{=1}\cdot\{\E[M\mid A=1,B,X,G=E]-\hat{M}_{1}^{E}(B,X)\}\\
&\quad\quad\quad\quad\quad-\underbrace{\frac{P(A=0\mid B,X,G=E)}{1-\hat{P}_{1B}^E(X)}}_{=1}\cdot\{\E[M\mid A=0,B,X,G=E]-\hat{M}_{0}^{E}(B,X)\}\bigg\}\Big]\\
&=\E\Big[\frac{I(G=O)}{p(G=O)}\Bigg\{\frac{e_1^O(X)-\hat{e}_1^O(X)-e_0^O(X)+\hat{e}_0^O(X)}{\hat{P}_{11}^O(X)-\hat{P}_{10}^O(X)}+\frac{\hat{e}_1^O(X)-\hat{e}_0^O(X)}{\hat{P}_{11}^O(X)-\hat{P}_{10}^O(X)}\\
&\quad\quad\quad+\hat{M}_{1}^{E}(B,X)-\hat{M}_{0}^{E}(B,X)+\{M-\hat{M}_{1}^{E}(B,X)\}-\{M-\hat{M}_{0}^{E}(B,X)\}\bigg\}\Big]\\
&=\E\Big[\frac{I(G=O)}{p(G=O)}\Bigg\{\frac{e_1^O(X)-e_0^O(X)}{\hat{P}_{11}^O(X)-\hat{P}_{10}^O(X)}+\hat{M}_{1}^{E}(B,X)-\hat{M}_{0}^{E}(B,X)\\
&\quad\quad\quad+\{\E[M\mid A=1,B,X,G=E]-\hat{M}_{1}^{E}(B,X)\}-\{\E[M\mid A=0,B,X,G=E]-\hat{M}_{0}^{E}(B,X)\}\bigg\}\Big]\\
&=\E\Big[\frac{I(G=O)}{p(G=O)}\Bigg\{\frac{e_1^O(X)-e_0^O(X)}{\hat{P}_{11}^O(X)-\hat{P}_{10}^O(X)}+\E[M\mid A=1,B,X,G=E]-\E[M\mid A=0,B,X,G=E]\bigg\}\Big]\\
&=\theta_{\text{ATE}}
\end{align*}}
\endgroup

Parts 3 and 4 can be proven by combining the techniques used in parts 1 and 2, and thus we omit here.

    \item $\psi_{\text{ATE}}^{\text{bsiv2}}$:
    First, suppose the pair $\{P_{AB}^O(X),\mathbb{E}[M\mid A,B,X,G=O],\mathbb{E}[Y\mid A,B,X,G=O],\mathbb{E}[M\mid A,B,X,G=E]\}$ is correctly specified. We have
    \begingroup
\allowdisplaybreaks
{\footnotesize \begin{align*}
&\E\Big[\frac{I(G=O)}{p(G=O)}\Bigg\{\frac{I(A=1)I(B=1)}{\hat{P}_{11}^O(X)}\{Y-M-\hat{E}_{11}^O(X)\}-\frac{I(A=0)I(B=1)}{\hat{P}_{01}^O(X)}\{Y-M-\hat{E}_{01}^O(X)\}\\
    &\ \ \ \ \ \ \ \ \ \ \ \ \ \ \ \ \ \ -\frac{I(A=0)I(B=0)}{\hat{P}_{00}^O(X)}\{Y-M-\hat{E}_{00}^O(X)\}+\frac{I(A=1)I(B=0)}{\hat{P}_{10}^O(X)}\{Y-M-\hat{E}_{10}^O(X)\}\\
    &\ \ \ \ \ \ \ \ \ \ \ \ \ \ \ \ \ \ +I(B=1)\{\hat{E}_{11}^O(X)-\hat{E}_{01}^O(X)-\hat{E}_{10}^O(X)+\hat{E}_{00}^O(X)\}+\hat{E}_{10}^O(X)-\hat{E}_{00}^O(X)\\
    &\ \ \ \ \ \ \ \ \ \ \ \ \ \ \ \ \ \ -\frac{\{\hat{E}_{01}^O(X)-\hat{E}_{00}^O(X)\}\hat{P}_{1B}^O(X)+\{\hat{E}_{11}^O(X)-\hat{E}_{10}^O(X)\}(1-\hat{P}_{1B}^O(X))}{\hat{P}_{01}^O(X)-\hat{P}_{00}^O(X)}\\
    &\ \ \ \ \ \ \ \ \ \ \ \ \ \ \ \ \ \ +\hat{M}_1^E(B,X)-\hat{M}_0^E(B,X)\\
    &\ \ \ \ \ \ \ \ \ \ \ \ \ \ \ \ \ \ +\frac{1}{\hat{P}_{01}^O(X)-\hat{P}_{00}^O(X)}\bigg\{-\frac{I(A=1)I(B=1)\hat{P}_{01}^O(X)}{\hat{P}_{11}^O(X)\hat{\rho}_1^O(X)}\{Y-M-\hat{E}_{11}^O(X)\}\\
    &~~~~~~~~~~~~~~~~~~~~~~~~~~~-\frac{I(A=0)I(B=1)\hat{P}_{11}^O(X)}{\hat{P}_{01}^O(X)\hat{\rho}_1^O(X)}\{Y-M-\hat{E}_{01}^O(X)\}\\
    &~~~~~~~~~~~~~~~~~~~~~~~~~~~+\frac{I(A=1)I(B=0)\hat{P}_{00}^O(X)}{\hat{P}_{10}^O(X)\hat{\rho}_0^O(X)}\{Y-M-\hat{E}_{10}^O(X)\}\\
    &~~~~~~~~~~~~~~~~~~~~~~~~~~~+\frac{I(A=0)I(B=0)\hat{P}_{10}^O(X)}{\hat{P}_{00}^O(X)\hat{\rho}_0^O(X)}\{Y-M-\hat{E}_{00}^O(X)\}\\
    &~~~~~~~~~~~~~~~~~~~~~~~~~~~+\frac{\hat{E}_{11}^O(X)-\hat{E}_{01}^O(X)-\hat{E}_{10}^O(X)+\hat{E}_{00}^O(X)}{\hat{\rho}_B^O(X)}\big\{I(A=1)-\hat{P}_{1B}^O(X)\big\}\\
    &~~~~~~~~~~~~~~~~~~~~~~~~~~~ +\frac{\{\hat{E}_{11}^O(X)-\hat{E}_{01}^O(X)-\hat{E}_{10}^O(X)+\hat{E}_{00}^O(X)\}\hat{P}_{1B}^O(X)}{\hat{P}_{01}^O(X)-\hat{P}_{00}^O(X)}\\
    &~~~~~~~~~~~~~~~~~~~~~\cdot\big\{-\frac{I(B=1)}{\hat{\rho}_1^O(X)}\{I(A=0)-\hat{P}_{01}^O(X)\}+\frac{I(B=0)}{\hat{\rho}_0^O(X)}\{I(A=0)-\hat{P}_{00}^O(X)\}\big\}\bigg\}\Bigg\}\\
    &+\frac{I(G=E)}{p(G=O)}\frac{1-\hat{\tau}(B,X)}{\hat{\tau}(B,X)}\bigg\{\frac{I(A=1)}{\hat{P}_{1B}^E(X)}\{M-\hat{M}_1^E(B,X)\}-\frac{I(A=0)}{1-\hat{P}_{1B}^E(X)}\{M-\hat{M}_0^E(B,X)\}\bigg\}\Big]\\
    &=\E\Big[\frac{I(G=O)}{p(G=O)}\Bigg\{\frac{I(A=1)I(B=1)}{\hat{P}_{11}^O(X)}\{\underbrace{\E[Y-M\mid A=1,B=1,X,G=O]-\hat{E}_{11}^O(X)}_{=0}\}\\
    &\ \ \ \ \ \ \ \ \ \ \ \ \ \ \ \ \ \ -\frac{I(A=0)I(B=1)}{\hat{P}_{01}^O(X)}\{\underbrace{\E[Y-M\mid A=0,B=1,X,G=O]-\hat{E}_{01}^O(X)}_{=0}\}\\
    &\ \ \ \ \ \ \ \ \ \ \ \ \ \ \ \ \ \ -\frac{I(A=0)I(B=0)}{\hat{P}_{00}^O(X)}\{\underbrace{\E[Y-M\mid A=0,B=0,X,G=O]-\hat{E}_{00}^O(X)}_{=0}\}\\
    &\ \ \ \ \ \ \ \ \ \ \ \ \ \ \ \ \ \ +\frac{I(A=1)I(B=0)}{\hat{P}_{10}^O(X)}\{\underbrace{\E[Y-M\mid A=1,B=0,X,G=O]-\hat{E}_{10}^O(X)}_{=0}\}\\
    &\ \ \ \ \ \ \ \ \ \ \ \ \ \ \ \ \ \ +I(B=1)\{\hat{E}_{11}^O(X)-\hat{E}_{01}^O(X)-\hat{E}_{10}^O(X)+\hat{E}_{00}^O(X)\}+\hat{E}_{10}^O(X)-\hat{E}_{00}^O(X)\\
    &\ \ \ \ \ \ \ \ \ \ \ \ \ \ \ \ \ \ -\frac{\{\hat{E}_{01}^O(X)-\hat{E}_{00}^O(X)\}\hat{P}_{1B}^O(X)+\{\hat{E}_{11}^O(X)-\hat{E}_{10}^O(X)\}(1-\hat{P}_{1B}^O(X))}{\hat{P}_{01}^O(X)-\hat{P}_{00}^O(X)}\\
    &\ \ \ \ \ \ \ \ \ \ \ \ \ \ \ \ \ \ +\hat{M}_1^E(B,X)-\hat{M}_0^E(B,X)\\
    &\ \ \ \ \ \ \ \ \ \ \ \ \ \ \ \ \ \ +\frac{1}{\hat{P}_{01}^O(X)-\hat{P}_{00}^O(X)}\bigg\{-\frac{I(A=1)I(B=1)\hat{P}_{01}^O(X)}{\hat{P}_{11}^O(X)\hat{\rho}_1^O(X)}\{\underbrace{\E[Y-M\mid A=1,B=1,X,G=O]-\hat{E}_{11}^O(X)}_{=0}\}\\
    &~~~~~~~~~~~~~~~~~~~~~~~~~~~-\frac{I(A=0)I(B=1)\hat{P}_{11}^O(X)}{\hat{P}_{01}^O(X)\hat{\rho}_1^O(X)}\{\underbrace{\E[Y-M\mid A=0,B=1,X,G=O]-\hat{E}_{01}^O(X)}_{=0}\}\\
    &~~~~~~~~~~~~~~~~~~~~~~~~~~~+\frac{I(A=1)I(B=0)\hat{P}_{00}^O(X)}{\hat{P}_{10}^O(X)\hat{\rho}_0^O(X)}\{\underbrace{\E[Y-M\mid A=1,B=0,X,G=O]-\hat{E}_{10}^O(X)}_{=0}\}\\
    &~~~~~~~~~~~~~~~~~~~~~~~~~~~+\frac{I(A=0)I(B=0)\hat{P}_{10}^O(X)}{\hat{P}_{00}^O(X)\hat{\rho}_0^O(X)}\{\underbrace{\E[Y-M\mid A=0,B=0,X,G=O]-\hat{E}_{00}^O(X)}_{=0}\}\\
    &~~~~~~~~~~~~~~~~~~~~~~~~~~~+\frac{\hat{E}_{11}^O(X)-\hat{E}_{01}^O(X)-\hat{E}_{10}^O(X)+\hat{E}_{00}^O(X)}{\hat{\rho}_B^O(X)}\big\{\underbrace{P(A=1\mid B,X,G=O)-\hat{P}_{1B}^O(X)}_{=0}\big\}\\
    &~~~~~~~~~~~~~~~~~~~~~~~~~~~ +\frac{\{\hat{E}_{11}^O(X)-\hat{E}_{01}^O(X)-\hat{E}_{10}^O(X)+\hat{E}_{00}^O(X)\}\hat{P}_{1B}^O(X)}{\hat{P}_{01}^O(X)-\hat{P}_{00}^O(X)}\\
    &~~~~~~~~~~~~~~~~~~~~~\cdot\big\{-\frac{I(B=1)}{\hat{\rho}_1^O(X)}\{\underbrace{P(A=0\mid B=1,X,G=O)-\hat{P}_{01}^O(X)}_{=0}\}\\
    &~~~~~~~~~~~~~~~~~~~~~+\frac{I(B=0)}{\hat{\rho}_0^O(X)}\{\underbrace{P(A=0\mid B=0,X,G=O)-\hat{P}_{00}^O(X)}_{=0}\}\big\}\bigg\}\Bigg\}\\
    &+\frac{I(G=E)}{p(G=O)}\frac{1-\hat{\tau}(B,X)}{\hat{\tau}(B,X)}\bigg\{\frac{I(A=1)}{\hat{P}_{1B}^E(X)}\{\underbrace{\E[M\mid A=1,B,X,G=E]-\hat{M}_1^E(B,X)}_{=0}\}\\
    &~~~~~~~~~~~~~~~~~~~~~~~~~~~ -\frac{I(A=0)}{1-\hat{P}_{1B}^E(X)}\{\underbrace{\E[M\mid A=0,B,X,G=E]-\hat{M}_0^E(B,X)}_{=0}\}\bigg\}\Big]\\
    &=\E\Big[\frac{I(G=O)}{p(G=O)}\Bigg\{I(B=1)\{\hat{E}_{11}^O(X)-\hat{E}_{01}^O(X)-\hat{E}_{10}^O(X)+\hat{E}_{00}^O(X)\}+\hat{E}_{10}^O(X)-\hat{E}_{00}^O(X)\\
    &\ \ \ \ \ \ \ \ \ \ \ \ \ \ \ \ \ \ -\frac{\{\hat{E}_{01}^O(X)-\hat{E}_{00}^O(X)\}\hat{P}_{1B}^O(X)+\{\hat{E}_{11}^O(X)-\hat{E}_{10}^O(X)\}(1-\hat{P}_{1B}^O(X))}{\hat{P}_{01}^O(X)-\hat{P}_{00}^O(X)}\\
    &\ \ \ \ \ \ \ \ \ \ \ \ \ \ \ \ \ \ +\hat{M}_1^E(B,X)-\hat{M}_0^E(B,X)\Bigg\}\Big]\\
&=\theta_{\text{ATE}}
\end{align*}}
\endgroup

Second, suppose the set $\{\tau(B,X),\rho_B^O(X),\rho_B^E(X),P_{AB}^O(X),P_{AB}^E(X)\}$ is correctly specified. We have
    \begingroup
\allowdisplaybreaks
{\footnotesize \begin{align*}
&\E\Big[\frac{I(G=O)}{p(G=O)}\Bigg\{\frac{I(A=1)I(B=1)}{\hat{P}_{11}^O(X)}\{Y-M-\hat{E}_{11}^O(X)\}-\frac{I(A=0)I(B=1)}{\hat{P}_{01}^O(X)}\{Y-M-\hat{E}_{01}^O(X)\}\\
    &\ \ \ \ \ \ \ \ \ \ \ \ \ \ \ \ \ \ -\frac{I(A=0)I(B=0)}{\hat{P}_{00}^O(X)}\{Y-M-\hat{E}_{00}^O(X)\}+\frac{I(A=1)I(B=0)}{\hat{P}_{10}^O(X)}\{Y-M-\hat{E}_{10}^O(X)\}\\
    &\ \ \ \ \ \ \ \ \ \ \ \ \ \ \ \ \ \ +I(B=1)\{\hat{E}_{11}^O(X)-\hat{E}_{01}^O(X)-\hat{E}_{10}^O(X)+\hat{E}_{00}^O(X)\}+\hat{E}_{10}^O(X)-\hat{E}_{00}^O(X)\\
    &\ \ \ \ \ \ \ \ \ \ \ \ \ \ \ \ \ \ -\frac{\{\hat{E}_{01}^O(X)-\hat{E}_{00}^O(X)\}\hat{P}_{1B}^O(X)+\{\hat{E}_{11}^O(X)-\hat{E}_{10}^O(X)\}(1-\hat{P}_{1B}^O(X))}{\hat{P}_{01}^O(X)-\hat{P}_{00}^O(X)}\\
    &\ \ \ \ \ \ \ \ \ \ \ \ \ \ \ \ \ \ +\hat{M}_1^E(B,X)-\hat{M}_0^E(B,X)\\
    &\ \ \ \ \ \ \ \ \ \ \ \ \ \ \ \ \ \ +\frac{1}{\hat{P}_{01}^O(X)-\hat{P}_{00}^O(X)}\bigg\{-\frac{I(A=1)I(B=1)\hat{P}_{01}^O(X)}{\hat{P}_{11}^O(X)\hat{\rho}_1^O(X)}\{Y-M-\hat{E}_{11}^O(X)\}\\
    &~~~~~~~~~~~~~~~~~~~~~~~~~~~-\frac{I(A=0)I(B=1)\hat{P}_{11}^O(X)}{\hat{P}_{01}^O(X)\hat{\rho}_1^O(X)}\{Y-M-\hat{E}_{01}^O(X)\}\\
    &~~~~~~~~~~~~~~~~~~~~~~~~~~~+\frac{I(A=1)I(B=0)\hat{P}_{00}^O(X)}{\hat{P}_{10}^O(X)\hat{\rho}_0^O(X)}\{Y-M-\hat{E}_{10}^O(X)\}\\
    &~~~~~~~~~~~~~~~~~~~~~~~~~~~+\frac{I(A=0)I(B=0)\hat{P}_{10}^O(X)}{\hat{P}_{00}^O(X)\hat{\rho}_0^O(X)}\{Y-M-\hat{E}_{00}^O(X)\}\\
    &~~~~~~~~~~~~~~~~~~~~~~~~~~~+\frac{\hat{E}_{11}^O(X)-\hat{E}_{01}^O(X)-\hat{E}_{10}^O(X)+\hat{E}_{00}^O(X)}{\hat{\rho}_B^O(X)}\big\{I(A=1)-\hat{P}_{1B}^O(X)\big\}\\
    &~~~~~~~~~~~~~~~~~~~~~~~~~~~ +\frac{\{\hat{E}_{11}^O(X)-\hat{E}_{01}^O(X)-\hat{E}_{10}^O(X)+\hat{E}_{00}^O(X)\}\hat{P}_{1B}^O(X)}{\hat{P}_{01}^O(X)-\hat{P}_{00}^O(X)}\\
    &~~~~~~~~~~~~~~~~~~~~~\cdot\big\{-\frac{I(B=1)}{\hat{\rho}_1^O(X)}\{I(A=0)-\hat{P}_{01}^O(X)\}+\frac{I(B=0)}{\hat{\rho}_0^O(X)}\{I(A=0)-\hat{P}_{00}^O(X)\}\big\}\bigg\}\Bigg\}\\
    &+\frac{I(G=E)}{p(G=O)}\frac{1-\hat{\tau}(B,X)}{\hat{\tau}(B,X)}\bigg\{\frac{I(A=1)}{\hat{P}_{1B}^E(X)}\{M-\hat{M}_1^E(B,X)\}-\frac{I(A=0)}{1-\hat{P}_{1B}^E(X)}\{M-\hat{M}_0^E(B,X)\}\bigg\}\Big]\\
    &=\E\Big[\frac{I(G=O)}{p(G=O)}\Bigg\{\underbrace{\frac{P(A=1\mid B=1,X,G=O)}{\hat{P}_{11}^O(X)}}_{=1}I(B=1)\{E_{11}^O(X)-\hat{E}_{11}^O(X)\}\\
    &\ \ \ \ \ \ \ \ \ \ \ \ \ \ \ \ \ \ -\underbrace{\frac{P(A=0\mid B=1,X,G=O)}{\hat{P}_{01}^O(X)}}_{=1}I(B=1)\{E_{01}^O(X)-\hat{E}_{01}^O(X)\}\\
    &\ \ \ \ \ \ \ \ \ \ \ \ \ \ \ \ \ \ -\underbrace{\frac{P(A=0\mid B=0,X,G=O)}{\hat{P}_{00}^O(X)}}_{=1}I(B=0)\{E_{00}^O(X)-\hat{E}_{00}^O(X)\}\\
    &\ \ \ \ \ \ \ \ \ \ \ \ \ \ \ \ \ \ +\underbrace{\frac{P(A=1\mid B=0,X,G=O)}{\hat{P}_{10}^O(X)}}_{=1}I(B=0)\{E_{10}^O(X)-\hat{E}_{10}^O(X)\}\\
    &\ \ \ \ \ \ \ \ \ \ \ \ \ \ \ \ \ \ +I(B=1)\{\hat{E}_{11}^O(X)-\hat{E}_{01}^O(X)-\hat{E}_{10}^O(X)+\hat{E}_{00}^O(X)\}+\hat{E}_{10}^O(X)-\hat{E}_{00}^O(X)\\
    &\ \ \ \ \ \ \ \ \ \ \ \ \ \ \ \ \ \ -\frac{\{\hat{E}_{01}^O(X)-\hat{E}_{00}^O(X)\}\hat{P}_{1B}^O(X)+\{\hat{E}_{11}^O(X)-\hat{E}_{10}^O(X)\}(1-\hat{P}_{1B}^O(X))}{\hat{P}_{01}^O(X)-\hat{P}_{00}^O(X)}\\
    &\ \ \ \ \ \ \ \ \ \ \ \ \ \ \ \ \ \ +\hat{M}_1^E(B,X)-\hat{M}_0^E(B,X)\\
    &\ \ \ \ \ \ \ \ \ \ \ \ \ \ \ \ \ \ +\frac{1}{\hat{P}_{01}^O(X)-\hat{P}_{00}^O(X)}\bigg\{-\frac{P(A=1\mid B=1,X,G=O)P(B=1\mid X,G=O)\hat{P}_{01}^O(X)}{\hat{P}_{11}^O(X)\hat{\rho}_1^O(X)}\{E_{11}^O(X)-\hat{E}_{11}^O(X)\}\\
    &~~~~~~~~~~~~~~~~~~~~~~~~~~~-\frac{P(A=0\mid B=1,X,G=O)P(B=1\mid X,G=O)\hat{P}_{11}^O(X)}{\hat{P}_{01}^O(X)\hat{\rho}_1^O(X)}\{E_{01}^O(X)-\hat{E}_{01}^O(X)\}\\
    &~~~~~~~~~~~~~~~~~~~~~~~~~~~+\frac{P(A=1\mid B=0,X,G=O)P(B=0\mid X,G=O)\hat{P}_{00}^O(X)}{\hat{P}_{10}^O(X)\hat{\rho}_0^O(X)}\{E_{10}^O(X)-\hat{E}_{10}^O(X)\}\\
    &~~~~~~~~~~~~~~~~~~~~~~~~~~~+\frac{P(A=0\mid B=0,X,G=O)P(B=0\mid X,G=O)\hat{P}_{10}^O(X)}{\hat{P}_{00}^O(X)\hat{\rho}_0^O(X)}\{E_{00}^O(X)-\hat{E}_{00}^O(X)\}\\
    &~~~~~~~~~~~~~~~~~~~~~~~~~~~+\frac{\hat{E}_{11}^O(X)-\hat{E}_{01}^O(X)-\hat{E}_{10}^O(X)+\hat{E}_{00}^O(X)}{\hat{\rho}_B^O(X)}\big\{\underbrace{P(A=1\mid B,X,G=O)-\hat{P}_{1B}^O(X)}_{=0}\big\}\\
    &~~~~~~~~~~~~~~~~~~~~~~~~~~~ +\frac{\{\hat{E}_{11}^O(X)-\hat{E}_{01}^O(X)-\hat{E}_{10}^O(X)+\hat{E}_{00}^O(X)\}\hat{P}_{1B}^O(X)}{\hat{P}_{01}^O(X)-\hat{P}_{00}^O(X)}\\
    &~~~~~~~~~~~~~~~~~~~~~\cdot\big\{-\frac{I(B=1)}{\hat{\rho}_1^O(X)}\{\underbrace{P(A=0\mid B=1,X,G=O)-\hat{P}_{01}^O(X)}_{=0}\}\\
    &~~~~~~~~~~~~~~~~~~~~~~~~~~~+\frac{I(B=0)}{\hat{\rho}_0^O(X)}\{\underbrace{P(A=0\mid B=0,X,G=O)-\hat{P}_{00}^O(X)}_{=0}\}\big\}\bigg\}\Bigg\}\\
    &+\frac{1}{p(G=O)}\underbrace{\big\{\underbrace{\frac{P(G=E\mid B,X)}{\hat{\tau}(B,X)}}_{=1}-I(G=E)\big\}}_{=I(G=O)}\bigg\{\underbrace{\frac{P(A=1\mid B,X,G=E)}{\hat{P}_{1B}^E(X)}}_{=1}\{\E[M\mid A=1,B,X,G=E]-\hat{M}_1^E(B,X)\}\\
    &~~~~~~~~~~~~~~~~~~~~~-\underbrace{\frac{P(A=0\mid B,X,G=E)}{1-\hat{P}_{1B}^E(X)}}_{=1}\{\E[M\mid A=0,B,X,G=E]-\hat{M}_0^E(B,X)\}\bigg\}\Big]\\
    &=\E\Big[\frac{I(G=O)}{p(G=O)}\Bigg\{I(B=1)\big\{\{E_{11}^O(X)-\hat{E}_{11}^O(X)\}-\{E_{01}^O(X)-\hat{E}_{01}^O(X)\}\big\}\\
    &\ \ \ \ \ \ \ \ \ \ \ \ \ \ \ \ \ \ -(1-I(B=1)\big\{\{E_{00}^O(X)-\hat{E}_{00}^O(X)\}-\{E_{10}^O(X)-\hat{E}_{10}^O(X)\}\big\}\\
    &\ \ \ \ \ \ \ \ \ \ \ \ \ \ \ \ \ \ +I(B=1)\{\hat{E}_{11}^O(X)-\hat{E}_{01}^O(X)-\hat{E}_{10}^O(X)+\hat{E}_{00}^O(X)\}+\hat{E}_{10}^O(X)-\hat{E}_{00}^O(X)\\
    &\ \ \ \ \ \ \ \ \ \ \ \ \ \ \ \ \ \ -\frac{\{\hat{E}_{01}^O(X)-\hat{E}_{00}^O(X)\}\hat{P}_{1B}^O(X)+\{\hat{E}_{11}^O(X)-\hat{E}_{10}^O(X)\}(1-\hat{P}_{1B}^O(X))}{\hat{P}_{01}^O(X)-\hat{P}_{00}^O(X)}\\
    &\ \ \ \ \ \ \ \ \ \ \ \ \ \ \ \ \ \ +\hat{M}_1^E(B,X)-\hat{M}_0^E(B,X)\\
    &\ \ \ \ \ \ \ \ \ \ \ \ \ \ \ \ \ \ +\frac{1}{\hat{P}_{01}^O(X)-\hat{P}_{00}^O(X)}\bigg\{-\hat{P}_{01}^O(X)\{E_{11}^O(X)-\hat{E}_{11}^O(X)\}-\hat{P}_{11}^O(X)\{E_{01}^O(X)-\hat{E}_{01}^O(X)\}\\
    &~~~~~~~~~~~~~~~~~~~~~~~~~~~+\hat{P}_{00}^O(X)\{E_{10}^O(X)-\hat{E}_{10}^O(X)\}+\hat{P}_{10}^O(X)\{E_{00}^O(X)-\hat{E}_{00}^O(X)\}\bigg\}\Bigg\}\\
    &+\frac{I(G=O)}{p(G=O)}\bigg\{\{\E[M\mid A=1,B,X,G=E]-\hat{M}_1^E(B,X)\}-\{\E[M\mid A=0,B,X,G=E]-\hat{M}_0^E(B,X)\}\bigg\}\Big]\\
    &=\E\Big[\frac{I(G=O)}{p(G=O)}\Bigg\{I(B=1)\Big\{\{E_{11}^O(X)-\hat{E}_{11}^O(X)\}-\{E_{01}^O(X)-\hat{E}_{01}^O(X)\}\\
    &\ \ \ \ \ \ \ \ \ \ \ \ \ \ \ \ +\{E_{00}^O(X)-\hat{E}_{00}^O(X)\}-\{E_{10}^O(X)-\hat{E}_{10}^O(X)\}\Big\}\\
    &\ \ \ \ \ \ \ -\{E_{00}^O(X)-\hat{E}_{00}^O(X)\}+\{E_{10}^O(X)-\hat{E}_{10}^O(X)\}\\
    &\ \ \ \ \ \ \ \ \ \ \ \ \ \ \ \ \ \ +I(B=1)\{\hat{E}_{11}^O(X)-\hat{E}_{01}^O(X)-\hat{E}_{10}^O(X)+\hat{E}_{00}^O(X)\}+\hat{E}_{10}^O(X)-\hat{E}_{00}^O(X)\\
    &\ \ \ \ \ \ \ \ \ \ \ \ \ \ \ \ \ \ -\frac{\{\hat{E}_{01}^O(X)-\hat{E}_{00}^O(X)\}\hat{P}_{1B}^O(X)+\{\hat{E}_{11}^O(X)-\hat{E}_{10}^O(X)\}(1-\hat{P}_{1B}^O(X))}{\hat{P}_{01}^O(X)-\hat{P}_{00}^O(X)}\\
    &\ \ \ \ \ \ \ \ \ \ \ \ \ \ \ \ \ \ +\frac{1}{\hat{P}_{01}^O(X)-\hat{P}_{00}^O(X)}\bigg\{-(1-\hat{P}_{0B}^O(X))\{E_{11}^O(X)-\hat{E}_{11}^O(X)\}-\hat{P}_{1B}^O(X)\{E_{01}^O(X)-\hat{E}_{01}^O(X)\}\\
    &~~~~~~~~~~~~~~~~~~~~~~~~~~~+(1-\hat{P}_{0B}^O(X))\{E_{10}^O(X)-\hat{E}_{10}^O(X)\}+\hat{P}_{1B}^O(X)\{E_{00}^O(X)-\hat{E}_{00}^O(X)\}\bigg\}\\
    &~~~~~~~~~~~~~~~~~~~~~~~~~~~+\E[M\mid A=1,B,X,G=E]-\E[M\mid A=0,B,X,G=E]\Bigg\}\Big]\\
    &=\E\Big[\frac{I(G=O)}{p(G=O)}\Bigg\{I(B=1)\{E_{11}^O(X)-E_{01}^O(X)-E_{10}^O(X)+E_{00}^O(X)\}+E_{10}^O(X)-E_{00}^O(X)\\
    &\ \ \ \ \ \ \ \ \ \ \ \ \ \ \ \ \ \ -\frac{\hat{P}_{1B}^O(X)(E_{01}^O(X)-E_{00}^O(X))+(1-\hat{P}_{0B}^O(X))(E_{11}^O(X)-E_{10}^O(X))}{\hat{P}_{01}^O(X)-\hat{P}_{00}^O(X)}\\
    &~~~~~~~~~~~~~~~~~~~~~~~~~~~+\E[M\mid A=1,B,X,G=E]-\E[M\mid A=0,B,X,G=E]\Bigg\}\Big]\\
&=\theta_{\text{ATE}}
\end{align*}}
\endgroup

Parts 3 and 4 can be proven by combining the techniques used in parts 1 and 2, and thus we omit here.
    \end{itemize}
    \end{proof}

\begin{proof}[Proof of Theorem \ref{thm:IF:proximal-1_ett}]
    Define
    \begin{align*}
        &\psi_1=\frac{\mathbb{E}[Y\mid G=O]}{p(A=1\mid G=O)}\\
        &\psi_2=\frac{\mathbb{E}[\mathbb{E}[h(M,0,X)\mid A=0,X,G=E]\mid G=O]}{p(A=1\mid G=O)}
    \end{align*}

    We use the notation $\partial_tf(t)$ to denote $\frac{\partial f(t)}{\partial t}\big|_{t=0}$.
For parameter $\psi$, let $\psi_t$ be the parameter under a regular parametric sub-model indexed by $t$, that includes the ground-truth model at $t=0$. Let $V$ be the set of all observed variable.
In order to obtain an influence function, we need to find a random variable $\Gamma$ with mean zero, that satisfies
\[
\partial_t\psi_t=\E[\Gamma S(V)],
\]
where $S(V)=\partial_t\log p_t(V)$.

For $\psi_1$, we have,
\begin{align*}
    \partial_t\psi_{1t}&=\partial_t\mathbb{E}\Big[\frac{I(G=O)}{p(A=1,G=O)}\{Y-I(A=1)\psi_1\}S(V)\Big]
\end{align*}

Therefore,
\begin{align}\label{eq:IF:proxy-ett1}
    \frac{I(G=O)}{p(G=O)p(A=1\mid G=O)}\{Y-I(A=1)\psi_1\}
\end{align}
is the influence function of $\psi_1$.

For $\psi_2$, note that
\begin{equation}
\label{eq:IF:proxy-ett2}
\begin{aligned}
    \partial_t\psi_{2t}&=\partial_t \sum_{m,x} \frac{1}{p_t(A=1\mid G=O)}h_t(m,0,x)p_t(m\mid A=0,x,G=E)p_t(x\mid G=O)\\
&= \sum_{m,x} \partial_t\frac{1}{p_t(A=1,G=O)}h(m,0,x)p(m\mid A=0,x,G=E)p(x, G=O)\\
&= \sum_{m,x} \partial_t\frac{1}{p_t(A=1,G=O)}h(m,0,x)p(m\mid A=0,x,G=E)p(x, G=O)\\
&\quad+\sum_{m,x} \frac{1}{p(A=1, G=O)}\partial_t h_t(m,0,x)p(m\mid A=0,x,G=E)p(x, G=O) \\
&\quad+ \sum_{m,x} \frac{1}{p(A=1, G=O)}h(m,0,x)\partial_t p_t(m\mid A=0,x,G=E)p(x, G=O)\\
&\quad+ \sum_{m,x} \frac{1}{p(A=1, G=O)}h(m,0,x)p(m\mid A=0,x,G=E)\partial_t p_t(x, G=O).
\end{aligned}
\end{equation}

For the first term in \eqref{eq:IF:proxy-ett2}, we have
\begin{equation}
\label{eq:IF:proxy-ett2-1}
\begin{aligned}
    &\sum_{m,x} \partial_t\frac{1}{p_t(A=1, G=O)}h(m,0,x)p(m\mid A=0,x,G=E)p(x, G=O)\\
    &=-\sum_{m,x} \frac{1}{p(A=1,G=O)}h(m,0,x)p(m\mid A=0,x,G=E)p(x, G=O)S(A=1, G=O)\\
    &=-\psi_2S(A=1,G=O)\\
    &=-\E\Big[\frac{I(A=1)I(G=O)}{p(A=1,G=O)}\psi_2S(A,G)\Big]\\
    &=-\E\Big[\frac{I(A=1)I(G=O)}{p(A=1,G=O)}\psi_2S(V)\Big]\\
    &=-\E\Big[\frac{I(A=1)I(G=O)}{p(G=O)p(A=1\mid G=O)}\psi_2S(V)\Big].
\end{aligned}
\end{equation}

For the second term in \eqref{eq:IF:proxy-ett2}, we have
\begin{align*}
&\sum_{m,x} \frac{1}{p(A=1,G=O)}\partial_t h_t(m,0,x)p(m\mid A=0,x,G=E)p(x,G=O) \\
&\sum_{m,x} \frac{1}{p(A=1\mid G=O)}\partial_t h_t(m,0,x)p(m\mid A=0,x,G=E)p(x\mid G=O) \\
&=\sum_{m,x} \frac{1}{p(A=1\mid G=O)}\partial_t h_t(m,0,x)\frac{p(m\mid A=0,x,G=E)}{p(m\mid A=0,x,G=O)p(A=0\mid x,G=O)}p(m,A=0,x\mid G=O)\\
&=\sum_{z,m,x} \frac{1}{p(A=1\mid G=O)}\partial_t h_t(m,0,x)q(z,0,x)p(z,m,A=0,x\mid G=O)\\
&=\sum_{z,m,a,x}\frac{I(a=0)}{p(A=1\mid G=O)} \partial_t h_t(m,a,x)q(z,a,x)p(z,m,a,x\mid G=O)\\
&=\E[ \frac{I(A=0)}{p(A=1\mid G=O)}\partial_t h_t(M,A,X)q(Z,A,X)\mid G=O]\\
&=\E\Big[ \frac{I(A=0)}{p(A=1\mid G=O)}\partial_t \E[h_t(M,A,X)\mid Z,A,X,G=O]q(Z,A,X)\Big| G=O\Big].
\end{align*}

Note that by Assumption \ref{assumption:compexist1} $(ii)$
\begin{align*}
&\E[Y-h(M,A,X)\mid Z,A,X,G=O]=0\\
&\Rightarrow\partial_t\E_t[Y-h_t(M,A,X)\mid Z,A,X,G=O]=0\\
&\Rightarrow\E[\partial_t\{Y-h_t(M,A,X)\}\mid Z,A,X,G=O]\\
&\quad+\E[\{Y-h(M,A,X)\}S(Y,M\mid Z,A,X,G=O)\mid Z,A,X,G=O]=0\\
&\Rightarrow\E[\partial_t h_t(M,A,X)\mid Z,A,X,G=O]\\
&\quad\quad=\E[\{Y-h(M,A,X)\}S(Y,M\mid Z,A,X,G=O)\mid Z,A,X,G=O].
\end{align*}
Therefore,
\begin{align*}
    &\sum_{m,x} \frac{1}{p(A=1,G=O)}\partial_t h_t(m,0,x)p(m\mid A=0,x,G=E)p(x,G=O) \\
    &=\E\Big[ \frac{I(A=0)}{p(A=1\mid G=O)}q(Z,A,X)\\
    &\quad\quad \cdot\E[\{Y-h(M,A,X)\}S(Y,M\mid Z,A,X,G=O)\mid Z,A,X,G=O]\Big| G=O\Big]\\
    &=\E\Big[ \frac{I(A=0)}{p(A=1\mid G=O)}q(Z,A,X)\{Y-h(M,A,X)\}S(Y,M\mid Z,A,X,G=O)\Big| G=O\Big]\\
    &=\E\Big[ \frac{I(A=0)I(G=O)}{p(G=O)p(A=1\mid G=O)}q(Z,A,X)\{Y-h(M,A,X)\}S(Y,M\mid Z,A,X,G=O)\Big].
\end{align*}

Also, note that
\begin{align*}
    \E\Big[ \frac{I(A=0)I(G=O)}{p(G=O)p(A=1\mid G=O)}q(Z,A,X)\{Y-h(M,A,X)\}S(Z,A,X,G=O)\Big]=0.
\end{align*}
Therefore,
\begin{equation}
    \label{eq:IF:proxy-ett2-2}
    \begin{aligned}
        &\sum_{m,x} \frac{1}{p(A=1,G=O)}\partial_t h_t(m,0,x)p(m\mid A=0,x,G=E)p(x,G=O) \\
        &=\E\Big[ \frac{I(A=0)I(G=O)}{p(G=O)p(A=1\mid G=O)}q(Z,A,X)\{Y-h(M,A,X)\}S(V)\Big]\\
        &=\E\Big[ \frac{I(A=0)I(G=O)}{p(G=O)p(A=1\mid G=O)}q(Z,0,X)\{Y-h(M,0,X)\}S(V)\Big].
    \end{aligned}
\end{equation}

For the third term in \eqref{eq:IF:proxy-ett2}, we have
{\footnotesize \begin{align*}
&\sum_{m,x} \frac{1}{p(A=1,G=O)}h(m,0,x)\partial_t p_t(m\mid A=0,x,G=E)p(x, G=O)\\
&=\sum_{m,x} \frac{1}{p(A=1,G=O)}h(m,0,x)\partial_t p_t(m\mid A=0,x,G=E)\{\frac{1}{p(G=E\mid x)}-1\}p(x,G=E)\\
&=\sum_{m,x} \frac{1}{p(A=1,G=O)}h(m,0,x)S(m\mid A=0,x,G=E)\{\frac{1}{p(G=E\mid x)}-1\}\frac{p(m,A=0,x,G=E)}{p(A=0\mid x,G=E)}\\
&=\sum_{m,a,x,g} h(m,a,x)S(m\mid a,x,g)\{\frac{1}{p(G=E\mid x)}-1\}\frac{I(a=0)}{p(A=0\mid x,G=E)}\cdot\frac{I(g=E)}{p(A=1,G=O)}p(m,a,x,g)\\
&=\E\Big[\frac{1}{p(A=1,G=O)}\cdot
\frac{I(G=E)I(A=0)}{1-p(A=1\mid X,G=E)}
h(M,A,X)\{\frac{1}{p(G=E\mid X)}-1\}S(M\mid A,X,G)
\Big]\\
&=\E\Big[\frac{1}{p(A=1,G=O)}\cdot
\frac{I(G=E)I(A=0)}{1-p(A=1\mid X,G=E)}\{
h(M,A,X)-\eta(A,X)\}\{\frac{1}{p(G=E\mid X)}-1\}S(M\mid A,X,G)
\Big],
\end{align*}}
where
\begin{align*}
\eta(a,x)\coloneqq&\E[h(M,A,X)\mid A=a,X=x,G=E]\\
=&\sum_m h(m,a,x)p(m\mid a,x,G=E)\\
=&\sum_{z,m,\tilde{a}} I(\tilde{a}=a)h(m,\tilde{a},x)q(z,\tilde{a},x)p(z,m,\tilde{a}\mid x,G=E)\\
=&\E[ I(A=a)h(M,A,X)q(Z,A,X)\mid X=x,G=E ].
\end{align*}
Note that
{\footnotesize
\begin{align*}
\E\Big[\frac{1}{p(A=1,G=O)}\cdot
\frac{I(G=E)I(A=0)}{1-p(A=1\mid X,G=E)}\{
h(M,A,X)-\eta(A,X)\}\{\frac{1}{p(G=E\mid X)}-1\}S(A,X,G)
\Big]=0.
\end{align*}}
Therefore,
{\footnotesize\begin{equation}
\label{eq:IF:proxy-ett3}
\begin{aligned}
&\sum_{m,x} \frac{1}{p(A=1,G=O)}h(m,0,x)\partial_t p_t(m\mid A=0,x,G=E)p(x, G=O)\\
&=\E\Big[\frac{1}{p(A=1,G=O)}\cdot
\frac{I(G=E)I(A=0)}{1-p(A=1\mid X,G=E)}\{
h(M,A,X)-\eta(A,X)\}\{\frac{1}{p(G=E\mid X)}-1\}S(V)
\Big].
\end{aligned}
\end{equation}}

For the fourth term in \eqref{eq:IF:proxy-ett2}, we have
\begin{align*}
&\sum_{m,x} \frac{1}{p(A=1,G=O)}h(m,0,x)p(m\mid A=0,x,G=E)\partial_t p_t(x,G=O)\\
&=\sum_{x}\sum_m \frac{1}{p(A=1, G=O)}h(m,0,x)p(m\mid A=0,x,G=E)S(x, G=O)p(x, G=O)\\
&=\sum_{x,g}\frac{I(g=O)}{p(G=O)p(A=1\mid G=O)}\eta(0,x)S(x, g)p(x,g)\\
&=\E[\frac{I(G=O)}{p(G=O)p(A=1\mid G=O)}\eta(0,X)S(X, G)]\\
&=\E[\frac{I(G=O)}{p(G=O)p(A=1\mid G=O)}\eta(0,X)S(V)].
\end{align*}
Therefore,
\begin{equation}
\label{eq:IF:proxy-ett4}
\begin{aligned}
&\sum_{m,x}\frac{1}{p(A=1,G=O)}h(m,0,x)p(m\mid A=0,x,G=E)\partial_t p_t(x,G=O)\\
&=\E[\frac{I(G=O)}{p(G=O)p(A=1\mid G=O)}\eta(0,X)S(V)].
\end{aligned}
\end{equation}

Combining \eqref{eq:IF:proxy-ett1}-\eqref{eq:IF:proxy-ett4} concludes that
\begin{align*}
\partial_t\psi_{\text{ETT}}^{\text{proxy}}&=\E
\Big[\Big\{\frac{1}{p(G=O)p(A=1\mid G=O)}\Big\{I(G=O)\big\{Y-I(A=0)q(Z,0,X)\{Y-h(M,0,X)\}\\
&\hspace{90mm}-\eta(0,X)-I(A=1)\psi_{\text{ETT}}^{\text{proxy}}\big\}\\
&\quad-\frac{I(G=E)I(A=0)}{1-p(A=1\mid X,G=E)}
\{h(M,0,X)-\eta(0,X)\}\{\frac{1}{p(G=E\mid X)}-1\}\Big\}
\Big\}S(V)\Big].
\end{align*}

Therefore, 
\begin{align*}
&\frac{1}{p(G=O)p(A=1\mid G=O)}\Big\{I(G=O)\big\{Y-I(A=0)q(Z,0,X)\{Y-h(M,0,X)\}\\
&\hspace{90mm}-\eta(0,X)-I(A=1)\psi_{\text{ETT}}^{\text{proxy}}\big\}\\
&\quad-\frac{I(G=E)I(A=0)}{1-p(A=1\mid X,G=E)}
\{h(M,0,X)-\eta(0,X)\}\{\frac{1}{p(G=E\mid X)}-1\}\Big\}
\end{align*}
is the influence function of $\psi_{\text{ETT}}^{\text{proxy}}$.
\end{proof}

\begin{proof}[Proof of Theorem \ref{thm:IF:proximal}]

Recall the definition that for $a\in\{0,1\}$ 
\[
\psi^a=
\E\big[\E[h(M,A,X)\mid A=a,X,G=E]\big|  G=O\big].
\]
We use the notation $\partial_tf(t)$ to denote $\frac{\partial f(t)}{\partial t}\big|_{t=0}$.
For parameter $\psi^a$, let $\psi^a_t$ be the parameter under a regular parametric sub-model indexed by $t$, that includes the ground-truth model at $t=0$.
In order to obtain an influence function, we need to find a random variable $\Gamma$ with mean zero, that satisfies
\[
\partial_t\psi^a_t=\E[\Gamma S(V)],
\]
where $S(V)=\partial_t\log p_t(V)$.

Note that
\begin{equation}
\label{eq:IF:proxa1}
\begin{aligned}
\partial_t\psi^a_t
&=\partial_t \sum_{m,x} h_t(m,a,x)p_t(m\mid a,x,G=E)p_t(x\mid G=O)\\
&= \sum_{m,x} \partial_t h_t(m,a,x)p(m\mid a,x,G=E)p(x\mid G=O) \\
&\quad+ \sum_{m,x} h(m,a,x)\partial_t p_t(m\mid a,x,G=E)p(x\mid G=O)\\
&\quad+ \sum_{m,x} h(m,a,x)p(m\mid a,x,G=E)\partial_t p_t(x\mid G=O).
\end{aligned}
\end{equation}

For the first term in \eqref{eq:IF:proxa1}, we have
\begin{align*}
&\sum_{m,x} \partial_t h_t(m,a,x)p(m\mid a,x,G=E)p(x\mid G=O) \\
&=\sum_{m,x} \partial_t h_t(m,a,x)\frac{p(m\mid a,x,G=E)}{p(m\mid a,x,G=O)p(a\mid x,G=O)}p(m,a,x\mid G=O)\\
&=\sum_{z,m,x} \partial_t h_t(m,a,x)q(z,a,x)p(z,m,a,x\mid G=O)\\
&=\sum_{z,m,\tilde{a},x}I(\tilde{a}=a) \partial_t h_t(m,\tilde{a},x)q(z,\tilde{a},x)p(z,m,\tilde{a},x\mid G=O)\\
&=\E[ I(A=a)\partial_t h_t(M,A,X)q(Z,A,X)\mid G=O]\\
&=\E\big[ I(A=a)\partial_t \E[h_t(M,A,X)\mid Z,A,X,G=O]q(Z,A,X)\big| G=O\big].
\end{align*}
Note that by Assumption \ref{assumption:compexist1} $(ii)$
\begin{align*}
&\E[Y-h(M,A,X)\mid Z,A,X,G=O]=0\\
&\Rightarrow\partial_t\E_t[Y-h_t(M,A,X)\mid Z,A,X,G=O]=0\\
&\Rightarrow\E[\partial_t\{Y-h_t(M,A,X)\}\mid Z,A,X,G=O]\\
&\quad+\E[\{Y-h(M,A,X)\}S(Y,M\mid Z,A,X,G=O)\mid Z,A,X,G=O]=0\\
&\Rightarrow\E[\partial_t h_t(M,A,X)\mid Z,A,X,G=O]\\
&\quad\quad=\E[\{Y-h(M,A,X)\}S(Y,M\mid Z,A,X,G=O)\mid Z,A,X,G=O].
\end{align*}
Therefore,
\begin{align*}
&\sum_{m,x} \partial_t h_t(m,a,x)p(m\mid a,x,G=E)p(x\mid G=O) \\
&=\E\big[ I(A=a)q(Z,A,X)\E[\{Y-h(M,A,X)\}S(Y,M\mid Z,A,X,G=O)\mid Z,A,X,G=O]\big| G=O\big]\\
&=\E\big[ I(A=a)q(Z,A,X)\{Y-h(M,A,X)\}S(Y,M\mid Z,A,X,G=O)\big| G=O\big]\\
&=\E\big[ I(A=a)\frac{I(G=O)}{p(G=O)}q(Z,A,X)\{Y-h(M,A,X)\}S(Y,M\mid Z,A,X,G) \big].
\end{align*}
Also, note that
\begin{align*}
\E\big[ I(A=a)\frac{I(G=O)}{p(G=O)}q(Z,A,X)\{Y-h(M,A,X)\}S(Z,A,X,G) \big]=0.
\end{align*}
Therefore, 
{\footnotesize
\begin{equation}
\label{eq:eq:IF:proxa2}
\begin{aligned}
\sum_{m,x} \partial_t h_t(m,a,x)p(m\mid a,x,G=E)p(x\mid G=O)=
\E\big[ I(A=a)\frac{I(G=O)}{p(G=O)}q(Z,A,X)\{Y-h(M,A,X)\}S(V) \big].
\end{aligned}
\end{equation}
}

For the second term in \eqref{eq:IF:proxa1}, we have
{\footnotesize \begin{align*}
&\sum_{m,x} h(m,a,x)\partial_t p_t(m\mid a,x,G=E)p(x\mid G=O)\\
&=\sum_{m,x} h(m,a,x)\partial_t p_t(m\mid a,x,G=E)\{\frac{1}{p(G=E\mid x)}-1\}\frac{1}{p(G=O)}p(x,G=E)\\
&=\sum_{m,x} h(m,a,x)S(m\mid a,x,G=E)\{\frac{1}{p(G=E\mid x)}-1\}\frac{1}{p(A=a\mid x,G=E)}\cdot\frac{1}{p(G=O)}p(m,a,x,G=E)\\
&=\sum_{m,\tilde{a},x,g} h(m,\tilde{a},x)S(m\mid \tilde{a},x,g)\{\frac{1}{p(G=E\mid x)}-1\}\frac{I(\tilde{a}=a)}{p(A=a\mid x,G=E)}\cdot\frac{I(g=E)}{p(G=O)}p(m,\tilde{a},x,g)\\
&=\E\Big[
\frac{I(A=a)}{p(A=a\mid X,G=E)}\cdot\frac{I(G=E)}{p(G=O)}
h(M,A,X)\{\frac{1}{p(G=E\mid X)}-1\}S(M\mid A,X,G)
\Big]\\
&=\E\Big[
\frac{I(A=a)}{p(A=a\mid X,G=E)}\cdot\frac{I(G=E)}{p(G=O)}
\{h(M,A,X)-\eta(A,X)\}\{\frac{1}{p(G=E\mid X)}-1\}S(M\mid A,X,G)
\Big],
\end{align*}}
where
\begin{align*}
\eta(a,x)\coloneqq&\E[h(M,A,X)\mid A=a,X=x,G=E]\\
=&\sum_m h(m,a,x)p(m\mid a,x,G=E)\\
=&\sum_{z,m,\tilde{a}} I(\tilde{a}=a)h(m,\tilde{a},x)q(z,\tilde{a},x)p(z,m,\tilde{a}\mid x,G=E)\\
=&\E[ I(A=a)h(M,A,X)q(Z,A,X)\mid X=x,G=E ].
\end{align*}
Note that
\begin{align*}
\E\Big[
\frac{I(A=a)}{p(A=a\mid X,G=E)}\cdot\frac{I(G=E)}{p(G=O)}
\{h(M,A,X)-\eta(A,X)\}\{\frac{1}{p(G=E\mid X)}-1\}S(A,X,G)
\Big]=0.
\end{align*}
Therefore,
\begin{equation}
\label{eq:IF:proxa3}
\begin{aligned}
&\sum_{m,x} h(m,a,x)\partial_t p_t(m\mid a,x,G=E)p(x\mid G=O)\\
&=\E\Big[
\frac{I(A=a)}{p(A=a\mid X,G=E)}\cdot\frac{I(G=E)}{p(G=O)}
\{h(M,A,X)-\eta(A,X)\}\{\frac{1}{p(G=E\mid X)}-1\}S(V)
\Big].
\end{aligned}
\end{equation}

For the third term in \eqref{eq:IF:proxa1}, we have
\begin{align*}
&\sum_{m,x} h(m,a,x)p(m\mid a,x,G=E)\partial_t p_t(x\mid G=O)\\
&=\sum_{x}\sum_m h(m,a,x)p(m\mid a,x,G=E)S(x\mid G=O)p(x\mid G=O)\\
&=\sum_{x,g}\frac{I(g=O)}{p(G=O)}\eta(a,x)S(x\mid g)p(x,g)\\
&=\E[\frac{I(G=O)}{p(G=O)}\eta(a,X)S(X\mid G)]\\
&=\E[\frac{I(G=O)}{p(G=O)}\{\eta(a,X)-\E[\eta(a,X)\mid G=O]\}S(X\mid G)]\\
&=\E[\frac{I(G=O)}{p(G=O)}\{\eta(a,X)-\psi^a\}S(X\mid G)].
\end{align*}
Note that
\begin{align*}
\E[\frac{I(G=O)}{p(G=O)}\{\eta(a,X)-\psi^a\}S(G)]=0.
\end{align*}
Therefore,
\begin{equation}
\label{eq:IF:proxa4}
\begin{aligned}
&\sum_{m,x} h(m,a,x)p(m\mid a,x,G=E)\partial_t p_t(x\mid G=O)
=\E[\frac{I(G=O)}{p(G=O)}\{\eta(a,X)-\psi^a\}S(V)].
\end{aligned}
\end{equation}

Combining \eqref{eq:IF:proxa1}-\eqref{eq:IF:proxa4} concludes that
\begin{align*}
\partial_t\psi^a_t&=\E
\Big[\Big\{
\frac{I(G=O)}{p(G=O)}I(A=a)q(Z,A,X)\{Y-h(M,A,X)\}\\
&\qquad+\frac{I(G=E)}{p(G=O)}\cdot\frac{I(A=a)}{p(A=a\mid X,G=E)}
\{h(M,A,X)-\eta(A,X)\}\{\frac{1}{p(G=E\mid X)}-1\}\\
&\qquad+\frac{I(G=O)}{p(G=O)}\{\eta(a,X)-\psi^a\}
\Big\}S(V)\Big].
\end{align*}

Therefore, 
\begin{align*}
&\frac{I(G=O)}{p(G=O)}I(A=a)q(Z,A,X)\{Y-h(M,A,X)\}\\
&+\frac{I(G=E)}{p(G=O)}\cdot\frac{I(A=a)}{p(A=a\mid X,G=E)}
\{h(M,A,X)-\eta(A,X)\}\{\frac{1}{p(G=E\mid X)}-1\}\\
&+\frac{I(G=O)}{p(G=O)}\{\eta(a,X)-\psi^a\}
\end{align*}
is the influence function of $\psi^a$.

\end{proof}

\begin{proof}[Proof of Proposition \ref{prop:DR:proximal_ett}]
First, suppose the pair $\{h,p(m\mid A=0,x,G=E)\}$ is correctly specified. We have
{\footnotesize \begin{align*}
    &\E\Big[\frac{1}{p(G=O)p(A=1\mid G=O)}\Big\{I(G=O)\big\{Y-I(A=0)\hat{q}(Z,0,X)\{Y-\hat{h}(M,0,X)\}-\hat{\eta}(0,X)\big\}\\
&\quad-\frac{I(G=E)I(A=0)}{1-\hat{p}(A=1\mid X,G=E)}
\{\hat{h}(M,0,X)-\hat{\eta}(0,X)\}\{\frac{1}{\hat{p}(G=E\mid X)}-1\}\Big\}\Big]\\
&=\E\Big[\frac{1}{p(G=O)p(A=1\mid G=O)}\Big\{I(G=O)\big\{Y-\\
&\quad I(A=0)\hat{q}(Z,0,X)\underbrace{\E[Y-\hat{h}(M,0,X)\mid Z,A=0,G=O]}_{=0}-\hat{\eta}(0,X)\big\}\\
&\quad-\frac{I(G=E)I(A=0)}{1-\hat{p}(A=1\mid X,G=E)}
\underbrace{\E[\hat{h}(M,0,X)-\hat{\eta}(0,X)\mid A=0,X,G=E]}_{=0}\{\frac{1}{\hat{p}(G=E\mid X)}-1\}\Big\}\Big]\\
&=\E\Big[\frac{1}{p(G=O)p(A=1\mid G=O)}\Big\{I(G=O)\big\{Y-\hat{\eta}(0,X)\big\}\Big]\\
&=\theta_{\text{ETT}}.
\end{align*}}

Second, suppose the set $\{h,p(A=1\mid x,G=E),p(G=E\mid x)\}$ is correctly specified. We have
\begingroup
\allowdisplaybreaks
{\footnotesize \begin{align*}
&\E\Big[\frac{1}{p(G=O)p(A=1\mid G=O)}\Big\{I(G=O)\big\{Y-I(A=0)\hat{q}(Z,0,X)\{Y-\hat{h}(M,0,X)\}-\hat{\eta}(0,X)\big\}\\
&\quad-\frac{I(G=E)I(A=0)}{1-\hat{p}(A=1\mid X,G=E)}
\{\hat{h}(M,0,X)-\hat{\eta}(0,X)\}\{\frac{1}{\hat{p}(G=E\mid X)}-1\}\Big\}\Big]\\
&=\E\Big[\frac{1}{p(G=O)p(A=1\mid G=O)}\Big\{I(G=O)\big\{Y-I(A=0)\hat{q}(Z,0,X)\underbrace{\E[Y-\hat{h}(M,0,X)\mid Z,A=0,G=O]}_{=0}\big\}\\
&\quad-\frac{I(G=E)I(A=0)}{1-\hat{p}(A=1\mid X,G=E)}
\hat{h}(M,0,X)\{\frac{1}{\hat{p}(G=E\mid X)}-1\}\\
&\quad-\underbrace{\big\{I(G=O)-\frac{I(G=E)I(A=0)}{1-\hat{p}(A=1\mid X,G=E)}\{\frac{1}{\hat{p}(G=E\mid X)}-1\}\big\}}_{=0}\hat{\eta}(0,X)\Big\}\Big]\\
&=\E\Big[\frac{1}{p(G=O)p(A=1\mid G=O)}\Big\{I(G=O)Y-\frac{I(G=E)I(A=0)}{1-\hat{p}(A=1\mid X,G=E)}
\hat{h}(M,0,X)\{\frac{1}{\hat{p}(G=E\mid X)}-1\}\Big\}\Big]\\
&=\theta_{\text{ETT}}.
\end{align*}}
\endgroup

Third, suppose the set $\{q,p(A=1\mid x,G=E),p(G=E\mid x)\}$ is correctly specified. We have
{\footnotesize \begin{align*}
&\E\Big[\frac{1}{p(G=O)p(A=1\mid G=O)}\Big\{I(G=O)\big\{Y-I(A=0)\hat{q}(Z,0,X)\{Y-\hat{h}(M,0,X)\}-\hat{\eta}(0,X)\big\}\\
&\quad-\frac{I(G=E)I(A=0)}{1-\hat{p}(A=1\mid X,G=E)}
\{\hat{h}(M,0,X)-\hat{\eta}(0,X)\}\{\frac{1}{\hat{p}(G=E\mid X)}-1\}\Big\}\Big]\\
&=\E\Big[\frac{1}{p(G=O)p(A=1\mid G=O)}\Big\{I(G=O)\big\{Y-I(A=0)\hat{q}(Z,0,X)Y+I(A=0)\hat{q}(Z,0,X)\hat{h}(M,0,X)\big\}\\
&\quad-\frac{I(G=E)I(A=0)}{1-\hat{p}(A=1\mid X,G=E)}
\hat{h}(M,0,X)\{\frac{1}{\hat{p}(G=E\mid X)}-1\}\\
&\quad-\underbrace{\big\{I(G=O)-\frac{I(G=E)I(A=0)}{1-\hat{p}(A=1\mid X,G=E)}\{\frac{1}{\hat{p}(G=E\mid X)}-1\}\big\}}_{=0}\hat{\eta}(0,X)\Big\}\Big]\\
&=\E\Big[\frac{1}{p(G=O)p(A=1\mid G=O)}\Big\{I(G=O)\big\{Y-I(A=0)\hat{q}(Z,0,X)Y+I(A=0)\hat{q}(Z,0,X)\hat{h}(M,0,X)\big\}\\
&\quad-\frac{I(G=E)I(A=0)}{1-\hat{p}(A=1\mid X,G=E)}
\hat{h}(M,0,X)\{\frac{1}{\hat{p}(G=E\mid X)}-1\}\Big\}\Big]\\
\end{align*}}

Note that
\begin{align*}
&\E\Big[\frac{I(G=O)}{p(G=O)p(A=1\mid G=O)}I(A=0)\hat{q}(Z,0,X)\hat{h}(M,0,X)\Big]\\
&\E\Big[\frac{I(G=O)I(A=0)}{p(A=1,G=O)}\E[\hat{q}(Z,0,X)\mid M,A=0,X,G=O]\hat{h}(M,0,X)\Big]\\
&=\frac{1}{p(A=1,G=O)}\sum_{m,x}\frac{p(m\mid A=0,x,G=E)}{p(m\mid A=0,x,G=O)p(A=0\mid x,G=O)}\hat{h}(m,0,x)p(m,A=0,x,G=O)\\
&=\frac{1}{p(A=1,G=O)}\sum_{m,x}p(m\mid A=0,x,G=E)p(x,G=O)\hat{h}(m,0,x)\\
&=\frac{1}{p(A=1,G=O)}\sum_{m,x}\frac{1}{p(A=0,x,G=E)}p(x,G=O)\hat{h}(m,0,x)p(m,A=0,x,G=E)\\
&=\frac{1}{p(A=1,G=O)}\sum_{m,x}\frac{1}{p(A=0\mid x,G=E)}\frac{p(G=O\mid x)}{p(G=E\mid x)}\hat{h}(m,0,x)p(m,A=0,x,G=E)\\
&=\E\Big[\frac{I(G=E)}{p(A=1,G=O)}\cdot\frac{I(A=0)}{1-\hat{p}(A=1\mid X,G=E)}\cdot\{\frac{1}{\hat{p}(G=E\mid X)}-1\}\hat{h}(M,0,X)\Big].
\end{align*}

Therefore,
\begin{align*}
&\E\Big[\frac{1}{p(G=O)p(A=1\mid G=O)}\Big\{I(G=O)\big\{Y-I(A=0)\hat{q}(Z,0,X)\{Y-\hat{h}(M,0,X)\}-\hat{\eta}(0,X)\big\}\\
&\quad-\frac{I(G=E)I(A=0)}{1-\hat{p}(A=1\mid X,G=E)}
\{\hat{h}(M,0,X)-\hat{\eta}(0,X)\}\{\frac{1}{\hat{p}(G=E\mid X)}-1\}\Big\}\Big]\\
&=\E\Big[\frac{I(G=O)}{p(G=O)p(A=1\mid G=O)}\big\{Y-I(A=0)\hat{q}(Z,0,X)Y\Big]\\
&=\theta_{\text{ETT}}.
\end{align*}
\end{proof}

\begin{proof}[Proof of Proposition \ref{prop:DR:proximal_ate}]
First, suppose the pair $\{h,p(m\mid a,x,G=E)\}$ is correctly specified. We have
\begingroup
\allowdisplaybreaks
{\footnotesize \begin{align*}
&\E\Big[\frac{I(G=O)}{p(G=O)}I(A=a)\hat{q}(Z,a,X)\{Y-\hat{h}(M,a,X)\}\\
&\quad+\frac{I(G=E)}{p(G=O)}\cdot\frac{I(A=a)}{1-a+(-1)^{1-a}\hat{p}(A=1\mid X,G=E)}
\{\hat{h}(M,a,X)-\hat{\eta}(a,X)\}\{\frac{1}{\hat{p}(G=E\mid X)}-1\}\\
&\quad+\frac{I(G=O)}{p(G=O)}\hat{\eta}(a,X)\Big]\\
&=\E\Big[\frac{I(G=O)}{p(G=O)}I(A=a)\hat{q}(Z,a,X)\underbrace{\E[Y-\hat{h}(M,a,X)\mid Z,A=a,G=O]}_{=0}\\
&\quad+\frac{I(G=E)}{p(G=O)}\cdot\frac{I(A=a)}{1-a+(-1)^{1-a}\hat{p}(A=1\mid X,G=E)}
\{\hat{h}(M,a,X)-\hat{\eta}(a,X)\}\{\frac{1}{\hat{p}(G=E\mid X)}-1\}\\
&\quad+\frac{I(G=O)}{p(G=O)}\hat{\eta}(a,X)\Big]\\
&=\E\Big[
\frac{I(G=E)}{p(G=O)}\cdot\frac{I(A=a)}{1-a+(-1)^{1-a}\hat{p}(A=1\mid X,G=E)}
\underbrace{\E[\hat{h}(M,a,X)-\hat{\eta}(a,X)\mid A=a,X,G=E]}_{=0}\{\frac{1}{\hat{p}(G=E\mid X)}-1\}\\
&\quad+\frac{I(G=O)}{p(G=O)}\hat{\eta}(a,X)\Big]\\
&=\E\Big[
\frac{I(G=O)}{p(G=O)}\hat{\eta}(a,X)\Big]\\
&=\theta^{(a)}.
\end{align*}}
\endgroup

Second, suppose the set $\{h,p(A=1\mid x,G=E),p(G=E\mid x)\}$ is correctly specified. We have
\begingroup
\allowdisplaybreaks
{\footnotesize \begin{align*}
&\E\Big[\frac{I(G=O)}{p(G=O)}I(A=a)\hat{q}(Z,a,X)\{Y-\hat{h}(M,a,X)\}\\
&\quad+\frac{I(G=E)}{p(G=O)}\cdot\frac{I(A=a)}{1-a+(-1)^{1-a}\hat{p}(A=1\mid X,G=E)}
\{\hat{h}(M,a,X)-\hat{\eta}(a,X)\}\{\frac{1}{\hat{p}(G=E\mid X)}-1\}\\
&\quad+\frac{I(G=O)}{p(G=O)}\hat{\eta}(a,X)\Big]\\
&=\E\Big[\frac{I(G=O)}{p(G=O)}I(A=a)\hat{q}(Z,a,X)\underbrace{\E[Y-\hat{h}(M,a,X)\mid Z,A=a,G=O]}_{=0}\\
&\quad+\frac{I(G=E)}{p(G=O)}\cdot\frac{I(A=a)}{1-a+(-1)^{1-a}\hat{p}(A=1\mid X,G=E)}
\{\hat{h}(M,a,X)-\hat{\eta}(a,X)\}\{\frac{1}{\hat{p}(G=E\mid X)}-1\}\\
&\quad+\frac{I(G=O)}{p(G=O)}\hat{\eta}(a,X)\Big]\\
&=\E\Big[
\frac{I(G=E)}{p(G=O)}\cdot\frac{I(A=a)}{1-a+(-1)^{1-a}\hat{p}(A=1\mid X,G=E)}
\hat{h}(M,a,X)\{\frac{1}{\hat{p}(G=E\mid X)}-1\}\\
&\quad+\underbrace{\E\big[\frac{I(G=O)}{p(G=O)}-\frac{I(G=E)}{p(G=O)}\cdot\frac{I(A=a)}{1-a+(-1)^{1-a}\hat{p}(A=1\mid X,G=E)}
\cdot\{\frac{1}{\hat{p}(G=E\mid X)}-1\}\big| X\big]}_{=0}\hat{\eta}(a,X)\Big]\\
&=\E\Big[
\frac{I(G=E)}{p(G=O)}\cdot\frac{I(A=a)}{1-a+(-1)^{1-a}\hat{p}(A=1\mid X,G=E)}
\hat{h}(M,a,X)\{\frac{1}{\hat{p}(G=E\mid X)}-1\}\Big]\\
&=\theta^{(a)}.
\end{align*}}
\endgroup

Third, suppose the set $\{q,p(A=1\mid x,G=E),p(G=E\mid x)\}$ is correctly specified. We have
{\footnotesize \begin{align*}
&\E\Big[\frac{I(G=O)}{p(G=O)}I(A=a)\hat{q}(Z,a,X)\{Y-\hat{h}(M,a,X)\}\\
&\quad+\frac{I(G=E)}{p(G=O)}\cdot\frac{I(A=a)}{1-a+(-1)^{1-a}\hat{p}(A=1\mid X,G=E)}
\{\hat{h}(M,a,X)-\hat{\eta}(a,X)\}\{\frac{1}{\hat{p}(G=E\mid X)}-1\}\\
&\quad+\frac{I(G=O)}{p(G=O)}\hat{\eta}(a,X)\Big]\\
&=\E\Big[\frac{I(G=O)}{p(G=O)}I(A=a)\hat{q}(Z,a,X)Y\\
&\quad+\Big\{\frac{I(G=E)}{p(G=O)}\cdot\frac{I(A=a)}{1-a+(-1)^{1-a}\hat{p}(A=1\mid X,G=E)}\cdot\{\frac{1}{\hat{p}(G=E\mid X)}-1\}-\frac{I(G=O)}{p(G=O)}I(A=a)\hat{q}(Z,a,X)\Big\}\hat{h}(M,a,X)\\
&\quad+\underbrace{\E\big[\frac{I(G=O)}{p(G=O)}-\frac{I(G=E)}{p(G=O)}\cdot\frac{I(A=a)}{1-a+(-1)^{1-a}\hat{p}(A=1\mid X,G=E)}
\cdot\{\frac{1}{\hat{p}(G=E\mid X)}-1\}\big| X\big]}_{=0}\hat{\eta}(a,X)\Big].
\end{align*}}

Note that
\begin{align*}
&\E\Big[\frac{I(G=O)}{p(G=O)}I(A=a)\hat{q}(Z,a,X)\hat{h}(M,a,X)\Big]\\
&\E\Big[\frac{I(G=O)}{p(G=O)}I(A=a)\E[\hat{q}(Z,a,X)\mid M,A=a,X,G=O]\hat{h}(M,a,X)\Big]\\
&=\frac{1}{p(G=O)}\sum_{m,x}\frac{p(m\mid a,x,G=E)}{p(m\mid a,x,G=O)p(a\mid x,G=O)}\hat{h}(m,a,x)p(m,a,x,G=O)\\
&=\frac{1}{p(G=O)}\sum_{m,x}p(m\mid a,x,G=E)p(x,G=O)\hat{h}(m,a,x)\\
&=\frac{1}{p(G=O)}\sum_{m,x}\frac{1}{p(a,x,G=E)}p(x,G=O)\hat{h}(m,a,x)p(m,a,x,G=E)\\
&=\frac{1}{p(G=O)}\sum_{m,x}\frac{1}{p(a\mid x,G=E)}\frac{p(G=O\mid x)}{p(G=E\mid x)}\hat{h}(m,a,x)p(m,a,x,G=E)\\
&=\E\Big[\frac{I(G=E)}{p(G=O)}\cdot\frac{I(A=a)}{1-a+(-1)^{1-a}\hat{p}(A=1\mid X,G=E)}\cdot\{\frac{1}{\hat{p}(G=E\mid X)}-1\}\hat{h}(M,a,X)\Big].
\end{align*}

Therefore,
\begin{align*}
&\E\Big[\frac{I(G=O)}{p(G=O)}I(A=a)\hat{q}(Z,a,X)\{Y-\hat{h}(M,a,X)\}\\
&\quad+\frac{I(G=E)}{p(G=O)}\cdot\frac{I(A=a)}{1-a+(-1)^{1-a}\hat{p}(A=1\mid X,G=E)}
\{\hat{h}(M,a,X)-\hat{\eta}(a,X)\}\{\frac{1}{\hat{p}(G=E\mid X)}-1\}\\
&\quad+\frac{I(G=O)}{p(G=O)}\hat{\eta}(a,X)\Big]\\
&=\E\Big[\frac{I(G=O)}{p(G=O)}I(A=a)\hat{q}(Z,a,X)Y\Big]\\
&=\theta^{(a)}.
\end{align*}

\end{proof}

\begin{proof}[Proof of Proposition \ref{prop:qcondmomeq}]
\begin{align*}
&\E[q(Z,A,X)\mid M,A,X,G=O]=\frac{p(M\mid A,X,G=E)}{p(M\mid A,X,G=O)p(A\mid X,G=O)}\\
&\Rightarrow\sum_z q(z,a,x)p(z, m,a\mid x,G=O)=p(m\mid a,x,G=E)\\
&\Rightarrow\sum_z q(z,a,x)\frac{p(z,G=O\mid m,a,x)p(m,a\mid x)}{p(G=O\mid x)}=\frac{p(G=E\mid m,a, x)p(m,a\mid x)}{p(a,G=E\mid x)}\\
&\Rightarrow\sum_{z,g}I(g=O) q(z,a,x)\frac{p(z,g\mid m,a,x)}{p(G=O\mid x)}=\sum_gI(g=E)\frac{p(g\mid m,a, x)}{p(a,G=E\mid x)}\\
&\Rightarrow\E\Big[ \frac{I(G=O)}{p(G=O\mid X)}q(Z,A,X)-\frac{I(G=E)}{p(A,G=E\mid X)}\Big| M,A,X \Big]=0.
\end{align*}

\end{proof}

\end{document}